%% file: Chakrabortty-HDME-Main-R1.tex
\documentclass[aos,preprint]{imsart}

\pdfoutput=1                    %% Need this for accurate compiling of tex files when submitting in ArXiv. %%
\setattribute{journal}{name}{}   %% This will remove the "Submitted to Annals of Statistics" line at the top left corner of the first page

\usepackage[OT1]{fontenc}
\usepackage{amsthm,amsmath,graphicx,enumerate}
\usepackage{natbib}
\usepackage{amssymb,tabularx,multicol,multirow,booktabs}
\usepackage{mhequ}
\usepackage{relsize}
\usepackage[toc,page]{appendix}

\usepackage{xr-hyper} % xr-hyper, unlike xr, makes cross-references to the other file actual hyperlinks that, when clicked, open up through the other file (which needs to be %
\usepackage[colorlinks,citecolor=blue,urlcolor=blue]{hyperref}
\usepackage{color,soul,subfigure}
\usepackage{rotating}

\usepackage{algorithm,algorithmic,comment,bm}  %% From Jiarui's tex file.  %%

\usepackage{float}  %% May need to use this for table repositioning. %%

\usepackage{float}
\restylefloat{table}
\restylefloat{figure}

\startlocaldefs
\numberwithin{equation}{section}
\theoremstyle{plain}

\endlocaldefs

\newtheorem{theorem}[]{Theorem}[section]%\newtheorem{theorem}[thm]{Theorem}
\newtheorem{lemma}[]{Lemma}[section]%\newtheorem{lemma}[thm]{Lemma}
\newtheorem{assumption}[]{Assumption}[section]%\newtheorem{assumption}[thm]{Assumption}
\newtheorem{definition}[]{Definition}[section]%\newtheorem{definition}[thm]{Definition}
\newtheorem{remark}[]{Remark}[section]%\newtheorem{remark}[thm]{Remark}

\def\bx{\mathbf{x}}
\def\bX{\mathbf{X}}

\def\bZ{\mathbf{Z}}
\def\bu{\mathbf{u}}

\def\bv{\mathbf{v}}

\def\bM{\mathbf{M}}

\def\K{\mathbb{K}}
\def\E{\mathbb{E}}

\def\P{\mathbb{P}}
\def\T{\mathbb{T}}

\def\bXv{\overrightarrow{\bX}}

\def\balpha{\boldsymbol{\alpha}}
\def\bbeta{\boldsymbol{\beta}}
\def\btheta{\boldsymbol{\theta}}

\def\bnabla{\boldsymbol{\nabla}}
\def\bpsi{\boldsymbol{\psi}}
\def\bSigma{\boldsymbol{\Sigma}}
\def\bxi{\boldsymbol{\xi}}

\def\balphahat{\widehat{\balpha}}
\def\bbetahat{\widehat{\bbeta}}
\def\bthetahat{\widehat{\btheta}}

\def\mhat{\widehat{m}}
\def\lhat{\widehat{l}}
\def\fhat{\widehat{f}}

\def\mtil{\widetilde{m}}
\def\ltil{\widetilde{l}}

\newcommand\ind{\protect\mathpalette{\protect\independenT}{\perp}}
\def\independenT#1#2{\mathrel{\rlap{$#1#2$}\mkern4mu{#1#2}}}

\newcommand\given[1][]{\:#1\vert\:}
\newcommand\medgiven{\hspace{0.25mm}\vert\hspace{0.25mm}}
\newcommand\smallgiven{\hspace{0.1mm}\vert\hspace{0.1mm}}

\def\tcr{\textcolor{red}}

\def\Lsc{\mathcal{L}}

\def\Isc{\mathcal{I}}
\def\Xsc{\mathcal{X}}

\def\bThat{\widehat{\bT}}

\def\Var{\mbox{Var}}

\def\half{\frac{1}{2}}

\def\nhalf{n^{\half}}
\def\nnhalf{n^{-\half}}

\def\bone{\mathbf{1}}
\def\bzero{\mathbf{0}}

\def\Tsc{\mathcal{T}}

\def\Msc{\mathcal{M}}
\def\bGamma{\boldsymbol{\Gamma}}
\def\bphi{\boldsymbol{\phi}}

\def\bT{\mathbf{T}}
\def\bThat{\widehat{\bT}}
\def\bR{\mathbf{R}}

\def\bDelta{\boldsymbol{\Delta}}
\def\bDeltahat{\widehat{\bDelta}}

\def\bgamma{\boldsymbol{\gamma}}
\def\L{\mathbb{L}}
\def\bmu{\boldsymbol{\mu}}

\def\C{\mathbb{C}}
\def\Asc{\mathcal{A}}
\def\bvj{\bv_{[j]}}
\def\R{\mathbb{R}}
\def\Jsc{\mathcal{J}}
\def\Msc{\mathcal{M}}
\def\Nsc{\mathcal{N}}
\def\T{\mathbb{T}}

\def\Ytil{\widetilde{Y}}

\def\Nsc{\mathcal{N}}

\def\Z{\mathbb{Z}}

\def\Rsc{\mathcal{R}}

\def\bxi{\boldsymbol{\xi}}

\def\Lsc{\mathcal{L}}

\def\bSigma{\boldsymbol{\Sigma}}

\def\bGammahat{\widehat{\bGamma}}

\def\bSigmahat{\widehat{\bSigma}}

\def \hs2{\hspace{2mm}}

\numberwithin{table}{section}
\numberwithin{equation}{section}

\definecolor{jcolor}{RGB}{041,122,000}
\definecolor{darkred}{RGB}{100,000,000}
\definecolor{purple}{RGB}{200,000,200}

\def\boxit#1{\vbox{\hrule\hbox{\vrule\kern6pt  \vbox{\kern6pt#1\kern6pt}\kern6pt\vrule}\hrule}}

\def\bT{\mathbf{T}}

\def\Ysc{\mathcal{Y}}
\def\Y{\mathbb{Y}}
\def\deltapi{\delta_{\pi}}
\def\Dsc{\mathcal{D}}

\def\convP{\stackrel{\P}{\rightarrow}}
\def\convd{\stackrel{d}{\rightarrow}}
\def\bh{\mathbf{h}}

\def\pitil{\widetilde{\pi}}
\def\pihat{\widehat{\pi}}
\def\mhat{\widehat{m}}
\def\phihat{\widehat{\phi}}
\def\phitil{\widetilde{\phi}}
\def\LnDR{\Lsc_{n}^{\mbox{\tiny{DDR}}}}
\def\bthetahatDR{\widehat{\btheta}_{\mbox{\tiny{DDR}}}}
\def\LIPW{\L_{\mbox{\tiny{IPW}}}}
\def\LDR{\L_{\mbox{\tiny{DDR}}}}
\def\LREG{\L_{\mbox{\tiny{REG}}}}
\def\bgammahat{\widehat{\bgamma}}

\def\kappaDR{\kappa_{\mbox{\tiny{DDR}}}}

\def\LntilDR{\widetilde{\Lsc}_n^{\mbox{\tiny{DDR}}}}
\def\Ytil{\widetilde{Y}}

\def\Dscn{\Dsc_n}
\def\Dscnone{\Dsc_{n}^{(1)}}
\def\Dscntwo{\Dsc_{n}^{(2)}}
\def\Dscnk{\Dsc_{n}^{(k)}}
\def\Dscnkp{\Dsc_{n}^{(k')}}

\def\Xscnk{\Xsc_{n,k}^*}
\def\Xscnkp{\Xsc_{n,k'}^*}
\def\Iscone{\Isc_1}
\def\Isctwo{\Isc_2}
\def\Isck{\Isc_k}
\def\Isckp{\Isc_{k'}}
\def\mhatone{\widehat{m}^{(1)}}
\def\mhattwo{\widehat{m}^{(2)}}
\def\mhatk{\widehat{m}^{(k)}}

\def\mtil{\widetilde{m}}
\def\nbar{\bar{n}}

\def\Ekkp{\E_{\Dscnk, \Xscnkp}}

\def\EXscnk{\E_{\Xscnk}}
\def\EXscnkp{\E_{\Xscnkp}}
\def\EDscnk{\E_{\Dscnk}}

\def\Pkkp{\P_{\Dscnk, \Xscnkp}}

\def\PDscnk{\P_{\Dscnk}}

\def\bT{\mathbf{T}}
\def\bTn{\bT_n}
\def\bTzero{\bT_{0}}
\def\bTzeron{\bT_{0,n}}
\def\bTpi{\bT_{\pi}}
\def\bTpin{\bT_{\pi,n}}
\def\bTm{\bT_{m}}
\def\bTmn{\bT_{m,n}}
\def\bR{\mathbf{R}}
\def\bRpim{\bR_{\pi,m}}
\def\bRpimn{\bR_{\pi,m,n}}
\def\bTzeroone{\bTzero^{(1)}}
\def\bTzerotwo{\bTzero^{(2)}}

\def\bTzeroonen{\bTzeron^{(1)}}
\def\bTzerotwon{\bTzeron^{(2)}}
\def\psione{\psi_1}
\def\sigmapsi{\sigma_{\psi}}
\def\sigmaeps{\sigma_{\varepsilon}}
\def\sigmabh{\sigma_{\bh}}
\def\sigmabarone{\bar{\sigma}_1}
\def\Kbarone{\bar{K}_1}
\def\sigmabartwo{\bar{\sigma}_2}
\def\Kbartwo{\bar{K}_2}
\def\bTzeroonej{\bT_{0 [j]}^{(1)}}
\def\bTzerotwoj{\bT_{0 [j]}^{(2)}}
\def\bhj{\bh_{[j]}}

\def\Xscn{\Xsc_n}
\def\Deltapin{\Delta_{\pi,n}}
\def\Deltapininfn{\left\|\Deltapin \right\|_{\infty,n}}
\def\pitiln{\widetilde{\pi}_n}
\def\pitilninfn{\left\| \pitiln \right\|_{\infty,n}}
\def\bvphi{\boldsymbol{\varphi}}
\def\bphij{\bvphi_{[j]}}
\def\bhj{\bh_{[j]}}
\def\bphibarsqnj{\bar{\bvphi}_{n [j]}^{(2)}}
\def\bmusqbphij{\bmu_{\bvphi[j]}^{(2)}}
\def\bmusqbhj{\bmu_{\bh[j]}^{(2)}}
\def\bmusqbhinf{\bmu_{\bh, \infty}^{(2)}}
\def\bTpij{\bT_{\pi [j]}}

\def\vnpi{v_{n,\pi}}

\def\BMC{\mbox{BMC}}
\def\sigmabar{\bar{\sigma}}
\def\Kbar{\bar{K}}
\def\sigmabarbvphi{\sigmabar_{\bvphi}}
\def\Kbarbvphi{\Kbar_{\bvphi}}
\def\bmusqbhinf{\|\bmu^{(2)}_{\bh} \|_{\infty}}
\def\Ascpinj{\Asc_{\pi,n,j}}
\def\Ascpinjc{\Ascpinj^c}

\def\Dscn{\Dsc_n}
\def\Dscnone{\Dsc_{n}^{(1)}}
\def\Dscntwo{\Dsc_{n}^{(2)}}
\def\Dscnk{\Dsc_{n}^{(k)}}
\def\Dscnkp{\Dsc_{n}^{(k')}}

\def\Xscnk{\Xsc_{n,k}^*}
\def\Xscnkp{\Xsc_{n,k'}^*}
\def\Iscone{\Isc_1}
\def\Isctwo{\Isc_2}
\def\Isck{\Isc_k}
\def\Isckp{\Isc_{k'}}
\def\mhatone{\widehat{m}^{(1)}}
\def\mhattwo{\widehat{m}^{(2)}}
\def\mhatk{\widehat{m}^{(k)}}

\def\mtil{\widetilde{m}}
\def\nbar{\bar{n}}

\def\Ekkp{\E_{\Dscnk, \Xscnkp}}

\def\EXscnk{\E_{\Xscnk}}
\def\EXscnkp{\E_{\Xscnkp}}
\def\EDscnk{\E_{\Dscnk}}

\def\Pkkp{\P_{\Dscnk, \Xscnkp}}

\def\PDscnk{\P_{\Dscnk}}

\def\Deltamnk{\Delta^{(k)}_{m, \nbar}}

\def\Deltamnonetwo{\Delta^{(1,2)}_{m, \nbar}}
\def\Deltamntwoone{\Delta^{(2,1)}_{m, \nbar}}
\def\Deltamnkkp{\Delta^{(k,k')}_{m, \nbar}}

\def\Deltamnonetwoin{\left\| \Deltamnonetwo \right\|_{\infty,\nbar}}
\def\Deltamntwoonein{\left\| \Deltamntwoone \right\|_{\infty,\nbar}}
\def\Deltamnkkpin{\left\| \Deltamnkkp \right\|_{\infty,\nbar}}

\def\pitil{\widetilde{\pi}}
\def\pibar{\bar{\pi}}
\def\deltapitil{\widetilde{\delta}_{\pi}}
\def\deltapibar{\bar{\delta}_{\pi}}

\def\bhsqbarkpnj{\bar{\bh}_{\nbar [j]}^{(2, k')}}
\def\sigmabarbhsq{\bar{\sigma}_{\bh}}
\def\Kbarbhsq{\bar{K}_{\bh}}
\def\dnbarj{d_{\nbar, j}}
\def\tnbar{t_{\nbar}}

\def\bTmk{\bTm^{(k)}}
\def\bTmkj{\bT_{m [j]}^{(k)}}
\def\bTmonetwon{\bT_{m,\nbar}^{(1,2)}}
\def\bTmtwoonen{\bT_{m,\nbar}^{(2,1)}}
\def\bTmkkpn{\bT_{m,\nbar}^{(k,k')}}

\def\vnm{v_{n,m}}

\def\vnbarm{v_{\nbar,m}}

\def\bTtilpin{\widetilde{\bT}_{\pi,n}}
\def\bRtilpimn{\widetilde{\bR}_{\pi,m,n}}
\def\bTtilmn{\widetilde{\bT}_{m,n}}
\def\bRdagpimn{\bR_{m,n}^{\dag}}

\def\bPsi{\boldsymbol{\Psi}}
\def\What{\widehat{W}}
\def\what{\widehat{w}}
\def\whatbx{\what_{\bx}}

\def\Stil{\widetilde{S}}
\def\Sbar{\overline{S}}
\def\Rhat{\widehat{R}}
\def\bThat{\widehat{\mathbf{T}}}

\def\bxi{\boldsymbol{\xi}}
\def\bxij{\bxi_{[j]}}
\def\bxibarn{\bar{\bxi}_n}

\def\Deltastarmninfn{\left\|\Delta^*_{m,n}\right\|_{\infty,n}}
\def\bmumodbh{\bmu_{|\bh|}}
\def\bmumodbhj{\bmu_{|\bh_{[j]}|}}
\def\bmumodbhinf{\| \bmumodbh\|_{\infty}}

\def\rbxi{r_*}
\def\sigmabarbxi{\bar{\sigma}_{\bxi}}
\def\Kbarbxi{\bar{K}_{\bxi}}

\def\psitwo{\psi_2}
\def\psitwonorm#1{\| #1\|_{\psitwo}}

\def\psione{\psi_1}
\def\psionenorm#1{\| #1 \|_{\psione}}
\def\psialpha{\psi_{\alpha}}
\def\psialphanorm#1{\| #1 \|_{\psialpha}}
\def\psibeta{\psi_{\beta}}
\def\psibetanorm#1{\| #1 \|_{\psibeta}}
\def\psigamma{\psi_{\gamma}}
\def\psigammanorm#1{\| #1 \|_{\psigamma}}
\def\Lonenorm#1{\| #1 \|_1}

\def\bigpsitwonorm#1{\left\| #1 \right \|_{\psitwo}}

\def\bthetatilDR{\widetilde{\btheta}_{\mbox{\tiny{DDR}}}}
\def\bOmega{\boldsymbol{\Omega}}
\def\bOmegahat{\widehat{\bOmega}}
\def\bUpsilon{\boldsymbol{\Upsilon}}
\def\sigmabUps{\sigma_{\bUpsilon}}
\def\bthetatil{\widetilde{\btheta}}
\def\DDR{{\mbox{\tiny{DDR}}}}
\def\sigmahat{\widehat{\sigma}}
\def\bGammahat{\widehat{\bGamma}}
\def\bpsihat{\widehat{\bpsi}}

\begin{document}

\begin{frontmatter}

% "Title of the paper"
\title{High Dimensional M-Estimation with Missing Outcomes: A Semi-Parametric Framework\thanksref{T1}}
\runtitle{High-Dimensional M-Estimation with Missing Outcomes}
%\thankstext{T1}{This is a working paper and the current draft is certainly not complete. The authors accept full responsibilities for any errors in this incomplete and unpublished manuscript.\smallskip}
\thankstext{T1}{Research supported in part by the NIH grants R01-GM123056 and R01-GM129781.}
\begin{aug}
 \author{\fnms{Abhishek} \snm{Chakrabortty}\corref{}\ead[label=e1]{abhishek@stat.tamu.edu}\thanksref{t1}},
\author{\fnms{Jiarui} \snm{Lu}\corref{}\ead[label=e2]{jiaruilu@pennmedicine.upenn.edu}},
 \author{\fnms{T. Tony} \snm{Cai}\corref{}\ead[label=e3]{tcai@wharton.upenn.edu}}
 \hspace{-0.05in}  %% doing this to make the author names' formatting look slightly better (when all 4 authors are included). %%
 \and
 \author{\fnms{Hongzhe} \snm{Li}\corref{}\ead[label=e4]{hongzhe@pennmedicine.upenn.edu}}
  \thankstext{t1}{Corresponding author; previously at University of Pennsylvania during this work.} %Date of this version: \today.
% {Research was supported in part by the NIH grants R01-GM123056 and R01-GM129781.}
%\thankstext{t1}{\tcr{Editing not finished yet. SEE PAGE 33 FOR REMARKS, and do not edit for now.}}
 \affiliation{Texas A\&M University and University of Pennsylvania}
\runauthor{A. Chakrabortty, J. Lu, T. T. Cai \and H. Li}

\address{Abhishek Chakrabortty\\
Dept. of Statistics\\
Texas A\&M University\\ %University of Pennsylvania\\
College STation, TX 77843, USA.\\ %Philadelphia, PA 19104, USA.\\
\printead{e1}\\
\phantom{E-mail:\ }}

\address{Jiarui Lu\\
Dept. of Biostatistics, Epidemiology \& Informatics\\
University of Pennsylvania\\
Philadelphia, PA 19104, USA.\\
\printead{e2}\\
\phantom{E-mail:\ }}

\address{T. Tony Cai\\
Dept. of Statistics\\
University of Pennsylvania\\
Philadelphia, PA 19104, USA.\\
\printead{e3}\\
\phantom{E-mail:\ }}

\address{Hongzhe Li\\
Dept. of Biostatistics, Epidemiology \& Informatics\\
University of Pennsylvania\\
Philadelphia, PA 19104, USA.\\
\printead{e4}\\
\phantom{E-mail:\ }}
\end{aug}



% NIH grant number: R01-GM123056 (will need to declare this in the final version)
% Potential AMS subject classifications: 62F10, 62F12, 62F25, 62F30, 62J02/05/12, 62J07, 62E17, 62G08 (will choose some among these and put them in the final version)

%\par\smallskip
%{\today}
%{May 29, 2019}

\begin{abstract}
We consider high dimensional $M$-estimation in settings where the response $Y$ is possibly missing at random and the covariates $\bX \in \R^p$ can be high dimensional compared to the sample size $n$. The parameter of interest $\btheta_0 \in \R^d$ is defined as the minimizer of the risk of a convex loss, under a fully non-parametric model, %, similar in spirit to \citet{Negahban_2012},
and $\btheta_0$ \emph{itself} is \emph{high dimensional} which is a key distinction from existing works. %(e.g. \citet{Chern_DDML_2018} and references therein).
Standard high dimensional regression and series estimation with possibly misspecified models and missing $Y$ are included as special cases, as well as their counterparts in causal inference using `potential outcomes'.
%Under an equivalent formulation of this setting based on `potential outcomes' in causal inference, these parameters also have important applications in heterogeneous treatment effects estimation  that are of interest in precision medicine.

Assuming $\btheta_0$ is $s$-sparse ($s \ll n$), we propose an $L_1$-regularized debiased and doubly robust (DDR) estimator of  $\btheta_0$ based on a high dimensional adaptation of the traditional double robust (DR) estimator's construction. Under mild tail assumptions and arbitrarily chosen (working) models for the propensity score (PS) and the outcome regression (OR) estimators,
satisfying \emph{only} some high-level conditions,
we establish \emph{finite sample} performance bounds for the DDR estimator showing its (optimal) $L_2$ error rate to be $\sqrt{s (\log d)/  n}$  when both models are correct, and its consistency and DR properties when only one of them is correct. %The estimators are first order \emph{insensitive} to any estimation errors or knowledge of construction of the PS and OR estimators.
Further, when both the models are correct, %the nuisance function working models are correctly specified,
we propose a \emph{desparsified} version of our DDR estimator that satisfies an \emph{asymptotic linear expansion} and facilitates \emph{inference} on low dimensional components of $\btheta_0$. Finally, we discuss various of choices of high dimensional parametric/semi-parametric working models for the PS and OR estimators. %and establish their properties.
All results are validated via detailed simulations. %Extensive simulation results are presented to validate our claims
\end{abstract}

\end{frontmatter}

%\section{Background and Introduction}\label{intro} To be written.

\section{Introduction}\label{psetup}
Large and complex observational data are commonplace in the modern `big data' era. Statistical analyses of such datasets often poses unique challenges that has led to a plethora of recent work. In particular, two such frequently encountered challenges include: (a) \emph{high dimensional settings,} wherein the dimension of the observed covariates is often comparable to or far exceeds the available sample size, and (b) potential \emph{incompleteness in the data,} especially in the outcome (or response) variable of interest. Both these issues arise naturally whenever observations are easily available for several covariates %$\bX$,
but the corresponding response %$Y$
is difficult and/or expensive to obtain. The latter could %possibly
be due to practical constraints (e.g. time, cost, logistics etc.), or simply by `design' (e.g. any treatment-response data in causal inference, where the response is automatically unobserved for any untreated individual). All these scenarios are routinely encountered in a variety of modern studies involving large databases, including biomedical data like electronic health records, or eQTL (i.e. expression quantitative trait loci) mapping studies in integrative genomics involving gene expression data, as well as in econometrics (e.g. in policy evaluation).
Further, owing to the very \emph{observational nature} of the data, the underlying missingness (or `treatment' assignment) mechanism is often informative (i.e. not randomized) and depends on the covariates, which leads to further complexities of \emph{selection bias} and confounding issues. Appropriate accounting of such biases is \emph{essential} to ensure the validity of any subsequent statistical analyses and inference.

For issue (a) above, both estimation and inference under high dimensional settings, \emph{but} with complete data, are by now quite well studied and equipped with a vast and growing literature centered around regularized methods and sparsity; see \citet{BuhlmannVdG_Book_2011} and \citet{Wainwright_Book_2019} for an overview. For issue (b) as well, under classical (low dimensional) settings, there has been substantial work leading to a rich body of literature on semi-parametric inference for incomplete response data. We refer to \citet{Tsiatis_Book_2007} and \citet{Bang_2005} for a review, as well as the fundamental works of \citet{Robins_1994} and \citet{Robins_1995}. Even under high dimensional settings, there has been a recent surge of work aimed at an analogous treatment of these problems but mostly in cases where the parameter of interest is still low dimensional (typically, the mean of the response) \citep{Farrell_2015, Belloni_ProgEval_2017, Chern_DDML_2018}.

In this paper, we consider a more challenging, and unique, setting that represents a confluence of all the issues highlighted above, combined with the fact the parameter of interest \emph{itself} is high dimensional, something that has received relatively limited attention so far. We first formalize our basic setup and the problem of interest, followed by an overview of our contributions.

\subsection{Problem Setup, Available Data and the Basic Assumptions}\label{avdata}
Let $Y \in \R$ and $\bX \in \R^p$ denote an outcome variable and a covariate vector of interest respectively, with supports $\Ysc \subseteq \R$ and $\Xsc \subseteq \R^p$ neither of which necessarily need to be continuous. In practice, however, $Y$ may not always be observed and let $T \in \{0,1\}$ denotes the indicator of $Y$ being observed. $\Z := (T, Y, \bX)$ is assumed to be defined jointly under some probability measure $\P(\cdot)$, while the \emph{observable} random vector is: $\bZ := (T, TY, \bX)$. The \emph{observed data} $\Dsc_n := $ $\{\bZ_i \equiv (T_i, T_i Y_i, \bX_i): i = 1, \hdots, n \}$ consists of $n$ independent and identically distributed (i.i.d.) realizations of $\bZ$ with joint distribution defined via $\P(\cdot)$.
We emphasize here that our focus is on \emph{high dimensional} settings, where the covariate dimension $p$ is allowed to diverge with $n$ (possibly, faster than $n$). %(including $p \gg n$).
%In general, our methods and theory (which is mostly non-asymptotic) apply to any regime of $(p,n)$ as long as the rate conditions given in our results hold.

\begin{assumption}[Basic assumptions]\label{base:assmpns}\emph{
\hspace{-0.1in}We assume throughout two basic conditions which are both fairly standard in the literature \citep{Imbens_Review_2004}.
\begin{enumerate}[(a)]
\item \emph{Ignorability:} $ T \ind Y \given \bX$, so that the missingness mechanism may depend on $\bX$, but is conditionally independent of $Y$ given $\bX$. This is also referred to often as
the missing at random (MAR) assumption in the literature.
    \vspace{0.05in}
\item \emph{Positivity/overlap:} Let $\pi(\bX) := \P(T = 1 \medgiven \bX)$ %$\forall \; \bx \in \Xsc$,
denote the \emph{propensity score} \citep{Rosenbaum_1983}, %i.e. the probability of $Y$  being observed given $\bX$,
and let $\pi := \P(T = 1)$. Then, we assume: %$\exists$ a universal constant $\deltapi$ with $0 < \deltapi \leq 1$, such that
\begin{equation}
 \pi(\bx) \; \geq \; \deltapi > 0 \;\; \forall \; \bx \in \Xsc, \quad \mbox{for some constant} \;\; \deltapi \in (0,1]. %\quad \mbox{(and hence,} \; \pi \geq \deltapi > 0\; \mbox{as well}).
 \label{pos:eqn}
\end{equation}
Hence, the probability of observing $Y$ given $\bX$ is always strictly positive.
\end{enumerate}
}
\end{assumption}
The MAR assumption in \ref{base:assmpns} (a) also includes the special case $T \ind (Y, \bX)$, commonly known as missing completely at random (MCAR). In such cases, $\pi(\cdot)$ simply equals the constant $\pi$ from part (b). In general, $\pi(\cdot)$ is allowed to depend on $\bX$ and may be unknown in practice when it needs to be estimated.

%\par\smallskip
The framework and notations above are in accordance with the standard treatment in the missing data literature \citep{Tsiatis_Book_2007}. However, the setting also encompasses problems in causal inference under the `potential outcome' framework. These may be equivalently formulated as missing data problems, a fact well known in the literature. We briefly discuss this equivalence below.

\paragraph*{Causal inference  under `potential outcomes' framework} In this setting, the observable vector is $\bZ := (T, \Y, \bX)$, where $T \in \{0,1\}$ denotes a binary `treatment' assignment indicator (can be any kind of assignment or intervention) and $\Y := T Y^{(1)} + (1-T) Y^{(0)}$ denotes the observed outcome with $(Y^{(1)},Y^{(0)})$ being the true `potential outcomes' \citep{Rubin_1974, ImbensRubin_Book_2015} for $T = 1$ and $T = 0$ respectively. Thus, for each potential outcome, this corresponds to our setting if we set $(Y,T) \equiv (Y^{(1)},T)$ or $(Y,T) \equiv (Y^{(0)}, 1-T)$.
%\par\smallskip
%Any \emph{multi-category} treatment assignment setting can also be similarly accommodated in our framework.
It is also worth noting that in the causal inference (CI) literature, $\bX$ is often referred to as `confounders' (in observational studies) or `adjustment' variables (in randomized trials), while the MAR assumption is often known as no unmeasured confounding (NUC) and MCAR as complete randomization.

\subsection{High Dimensional M-Estimation}\label{Mest_prob} We next introduce our \emph{main problem} of interest under this setting. Let $L(Y,\bX,\btheta):  \R \times \R^p \times \R^d \rightarrow \R$ be any `loss' function that is convex and differentiable in $\btheta$, and we assume that $[\E\{ L(Y,\bX,\btheta)\}^2] < \infty$ $\forall \; \btheta \in \R^d$. Then, the $M$-estimation problem considers the estimation of the minimizer $\btheta_0 \in \R^d$ of the risk function defined by $L(\cdot)$. Specifically, we aim to estimate the functional $\btheta_0 \equiv \btheta_0(\P) \in \R^d$ defined as: %parameter vector $\btheta_0 \equiv \btheta_0(\P) \in \R^d$, a functional of the distribution $\P(\cdot)$ of (the unobserved) $\Z$, defined as:
\begin{equation}
\btheta_0 \equiv \; \btheta_{0}(L,\P) \; :=  \underset{\btheta \in \R^d} \arg \min \; \L(\btheta), \;\; \mbox{where} \;\; \L(\btheta) \; := \; \E\left\{ L(Y, \bX, \btheta)\right\}. \label{Mest:prob}
\end{equation}
Here, $d$ is allowed to be high dimensional, i.e. $d$ can diverge with $n$ (possibly faster). %including $d \gg n$.
We assume without loss of generality (w.l.o.g.) that $d \geq 2$. %We assume w.l.o.g. that $d \geq 2$.
The existence and uniqueness of $\btheta_0$ is implicitly assumed given the generality of the framework considered. For most standard examples, this is fairly straightforward to establish with $L(\cdot)$ being convex and sufficiently smooth in $\btheta$. %In general, this can be guaranteed as long as the risk function $\L(\cdot)$ is strongly convex and coercive in $\btheta$.
For convenience of further discussion, let us define: $\forall \; y \in \Ysc$, $\bx \in \Xsc$ and $\btheta \in \R^d$,
\begin{equation*}
\phi(\bx,\btheta) := \E\{  L(Y, \bX, \btheta) \given \bX = \bx\} \;\; \mbox{and} \;\; \bnabla L(y,\bx,\btheta) := \frac{\partial}{\partial \btheta} L(y,\bx,\btheta) \in \R^d.
\end{equation*}
\begin{remark}\label{rem:semiparametric:framework}%[Semi-parametric framework and necessity of accounting for the missingness]
\emph{
It is important to note that $\btheta_0$ in \eqref{Mest:prob} is defined under a fully non-parametric family of $\P$ without any restrictions (upto Assumption \ref{base:assmpns} and basic moment conditions). Hence, the framework is \emph{semi-parametric} and \emph{model free} in this sense with $\btheta_0(\P)$ well-defined for every $\P$ without any model assumptions for $Y \medgiven \bX$ (even though $\btheta_0$ may sometimes be `motivated' by such `working' models for $Y \medgiven \bX$, as in the case of regression problems). % where they are simply interpreted as projection type parameters onto the corresponding covariate space.
}

\emph{
The framework also highlights the \emph{necessity  of accounting for the incompleteness} of $\Dsc_n$. If one ignores it and simply chooses to estimate $\btheta_0$ via risk minimization in the complete part of the data (i.e. observations with $T = 1$), then this `complete case' (CC) estimator will, in general, be \emph{inconsistent} for $\btheta_0$,  since the target parameter for the CC estimator is simply the minimizer of $\E \{ L(Y,\bX, \btheta) \medgiven T = 1\}$ which bears no direct relation to the unconditional minimizer $\btheta_0$ in \eqref{Mest:prob}.
The only cases when the CC estimator happens to be consistent for $\btheta_0$ is if either $T \ind (Y, \bX)$, i.e. MCAR holds (no selection bias), or if $\E \{ \bnabla L(Y,
\bX,\btheta_0)\medgiven \bX\} = \bzero$ almost surely (a.s.) $[\P_{\bX}]$. The latter, in case of regression problems, implies a %correctly specified
parametric model holds for $\E(Y \medgiven \bX)$ with `true' parameter being $\btheta_0$. Both these cases, however,  impose
further restrictions on $\P$. %leading to submodels included in the larger unrestricted family under which $\btheta_0$ is defined. To obtain consistent estimators of $\btheta_0$ uniformly over the entire family, appropriate accounting of the missingness is therefore necessary.
For consistent estimation of $\btheta_0$ over the \emph{entire} family of $\P$ where it is defined,
appropriate accounting of the missingness is thus necessary.
}
\end{remark}
\vspace{-0.101in}
Finally, it is worth mentioning that a special low dimensional case of \eqref{Mest:prob} is the \emph{mean estimation} problem where $\btheta_0 = \E(Y)$ with $L(Y,\bX,\btheta) = (Y - \btheta)^2$ and $d = 1$. In causal inference under the `potential outcome' framework, this also corresponds to the \emph{average treatment effect} (ATE) estimation problem. Both versions of this problem have by now been extensively studied in classical as well as high dimensional settings, especially the latter in recent times.
We defer a detailed literature review to Section \ref{contributions} and only point out here that the \emph{key distinction} between this literature and our setting is that our parameter of interest $\btheta_0$ in \eqref{Mest:prob} is \emph{itself high dimensional} (apart from $\bX$). %We defer a detailed literature review of all these problems to Section \ref{contributions}.

\subsection{Some Applications}\label{Mest:applns} \hspace{-0.1in} The framework \eqref{Mest:prob} encompasses a broad range of important problems. %problems that are useful in practice.
We enlist below a few useful examples for illustration.
%We highlight here a few useful illustrative examples of the $M$-estimation problem \eqref{Mest:prob} that are important and frequently encountered in practice.
%
%\begin{enumerate}[1.]
%\item
\paragraph*{1. High dimensional regression with possibly misspecified models and missing outcomes} \hspace{-0.1in} \eqref{Mest:prob} includes all standard high dimensional regression problems, where we further allow for: (i) potentially misspecified (working) models and (ii) $Y$ to be partly unobserved. For instance, set $\btheta = (a, \mathbf{b}) $ and $L(Y,\bX,\btheta) := l(Y, a + \mathbf{b}'\bX)$ in \eqref{Mest:prob}, with $a \in \R$, $\mathbf{b} \in \R^p$ and $l(u,v): \R \times \R \rightarrow \R$ being some loss function convex and differentiable in $v$. Typical choices of $l(\cdot , \cdot)$ include the `canonical' losses leading to standard regression problems as follows.
\begin{enumerate}[(a)]\label{loss:ex1}
\item The \emph{squared loss:} $l(u,v) \equiv l_{\mbox{\tiny sq}}(u,v) := (u - v)^2$ (for linear regression).
%the most common choice and usually motivated by a linear regression (working) model.

\item The \emph{logistic loss}: $l(u, v) \equiv l_{\log}(u,v) := - u v  + \log \{1 + \exp(v)\}$ (for logistic regression) and \emph{exponential loss}:  $l(u, v) \equiv l_{\exp}(u,v) := - u v  + \exp(v)$ (for Poisson regression), used often for binary or count valued $Y$ respectively.
%which is a typical choice in binary outcome regression and is usually motivated by an underlying logistic regression (working) model.

%\item The \emph{exponential loss}:  $l(u, v) \equiv l_{\exp}(u,v) := - u v  + \exp(v)$, %$\; \forall \; u, v \in \R$
%which is often used for regressing discrete (count) outcomes and is usually motivated by an underlying Poisson regression (working) model.
\end{enumerate}
In all examples, $\btheta_0$ is \emph{model free} and is %As noted in Remark \ref{rem:semiparametric:framework}, $\btheta_0$ is \emph{model free} and well defined
well defined \emph{regardless} of the validity of any motivating parametric (working) model for $Y \medgiven \bX$. In general, it simply corresponds to the `projection' of $\E(Y \medgiven \bX)$ onto that working model space. % Note that $\btheta_0$ is defined under a fully non-parametric family of $\P$. Hence, it is \emph{model-free} and well defined regardless of the validity of any underlying motivating `working' model for $Y \medgiven \bX$ and allows for its misspecification.

%\par\smallskip
%%\item \emph{Series estimation based on high dimensional basis functions.}
As an extension, one may also consider any (model free) \emph{series estimation problem} by replacing $\bX$ above with $\bPsi(\bX) := \{ \psi_j(\bX) \}_{j=1}^d$, a vector (possibly high dimensional)  of $d$ basis functions comprising transformations (possibly non-linear) of $\bX$. We may analogously set $L(Y, \bX, \btheta) := l\{Y,\bPsi(\bX)'\btheta\}$ with the same choices of $l(\cdot,\cdot)$ as above. A frequently used choice of $\bPsi(\cdot)$ includes the polynomial bases: $\bPsi(\bX) := \{1, \bx_j^k: 1 \leq j \leq p, 1 \leq k \leq d_0\}$, for any fixed degree $d_0 \geq 1$ whereby $d = pd_0 + 1$. The special case of $d_0 = 1$ (linear bases) leads to all the earlier examples, while $d_0 = 3$ leads to the cubic spline bases.
%\par\smallskip
%\item
\paragraph*{2. High dimensional single index models (SIMs) with elliptically symmetric designs}
\hspace{-0.1in}
Another interesting application of \eqref{Mest:prob} lies in signal recovery in SIMs with elliptically symmetric designs that satisfy a certain `linearity condition'. To this end, consider the SIM $Y = f(\beta_0'\bX, \epsilon)$, where $f(\cdot): \R^2 \rightarrow \Ysc$ is an \emph{unknown} link function, $\epsilon \ind \bX$ is a random noise (so that $Y \ind \bX \medgiven \bbeta_0'\bX$) and $\bbeta_0$ denotes the unknown index parameter (identifiable \emph{only} upto scalar multiples).
Now, consider any of the regression problems introduced in Example 1 and assume further that $\bX$ has an elliptically symmetric distribution (e.g. Gaussian). Then, $\btheta_0 \equiv (a_0, \mathbf{b}_0)$ defined therein satisfies: $\mathbf{b}_{0} \propto \bbeta_0$. This result, first noted by \citet{Li_Duan_1989}, provides an `easy' route to signal recovery in SIMs, especially in high dimensional settings and with missing outcomes. This also serves as a classic example where the parameter $\btheta_0$ is defined based on a misspecified parametric model, and yet, it has direct interpretability relating it to a parameter characterizing a larger semi-parametric model and allows one to still simply use \eqref{Mest:prob} for signal recovery in a SIM.

%\par\smallskip
%\item
\paragraph*{3. Applications in causal inference (heterogeneous treatment effects)} \hspace{-0.1in} All the problems in Examples 1 and 2 also have equivalent counterparts in causal inference under the `potential outcome'  framework discussed in Section \ref{avdata}. In this setting, these problems have important applications in the estimation of \emph{heterogeneous treatment effects} which is of great interest in personalized medicine. Fundamentally, this problem relates to estimation of the \emph{average conditional treatment effect} (ACTE): $\Delta(\bX) := \E\{ Y_{(1)} - Y_{(0)} \medgiven \bX\}$. In classical settings, estimation of $\Delta(\bX)$ via non-parametric machine learning methods has received considerable attention in recent times, including use of random forests or neural networks \citep{Wager_RFHTE_2017, Farrell_2018}. However, in a `truly' high dimensional setting, wherein $p$ diverges with $n$ (possibly, at a comparable or faster rate), fully non-parametric approaches may not be feasible and/or efficient. In such cases, it is often more reasonable to focus on (model free) projections of $\Delta(\bX)$ on finite (but high) dimensional function spaces. For the space of linear functions of $\bX$, this leads to the \emph{linear heterogeneous treatment effects} estimation problem. Such ideas and problems have indeed been advocated and considered in the recent works of \citet{Chern_HTE_2017} and \citet{Chern_BLP_2017}.

In our framework, this simply corresponds to the linear regression problem %discussed
in Example 1 (adapted to the CI setup). Further, under our setting, one can consider more general problems involving non-linear function spaces (e.g. series estimation) and/or other loss functions (e.g. logistic regression). These problems correspond to the other illustrations in Example 1. On the other hand, using Example 2, one may also consider ACTE estimation via SIMs which provide a clear generalization over standard parametric models and yet, to the best of our knowledge, has received very limited attention so far. %has not been considered been considered so much in the literature. In this case, our  corresponds to the appropriate ACTE estimation problem.  this simply relates to the linear regression problem in Example 1. On the other hand, \tcr{Working on it}
%\end{enumerate}

\subsection{Overview of Related Literature and Summary of Our Contributions}\label{contributions}
Our work contributes to two distinct lines of literature: (i) high dimensional $M$-estimation \emph{and} inference, and (ii) semi-parametric `doubly robust' inference for incomplete (and high dimensional) data. As regards the first line of work, for a \emph{complete} data, $M$-estimation problems are quite well studied in both classical and high dimensional settings; see \citet{VDV_Book_2000} for an overview of the vast classical literature, and \citet{Negahban_2012, Loh_2012, Loh_2015} and \citet{Loh_2017} for some of the more recent advances in high dimensional settings. Relatively little work, however, has been done for the case of \emph{incomplete} (in the response) data, especially in high dimensional settings.
In classical low dimensional settings, inference with incomplete data has a rich literature on semi-parametric methods and so called \emph{doubly robust} inference; see \citet{Bang_2005, Tsiatis_Book_2007,Kang_2007} and \citet{Graham_2011} for a review. Some of the pioneering work in this area were by \citet{Robins_1994,Robins_1995} and their several ensuing papers on related problems. %which we skip here for brevity. %While many of these are typically targeted towards the mean (or ATE) estimation problem, they can be easily adapted to handle vector valued parameters from $M$/$Z$ estimation problems as long as they are low dimensional.

In recent times, there has also been substantial interest in the extension of these approaches to high dimensional settings, leading to a flurry of papers, including \citet{Belloni_TE_2014}, \citet{Farrell_2015}, \citet{Belloni_ProgEval_2017}, \citet{Chern_DDML_2018} and \citet{Athey_ARBATE_2016}, among many other notable ones which we don't attempt to enlist here. However, their focus has \emph{still} mostly been on simple low dimensional parameters, like the mean (or the ATE), and less on cases where the \emph{parameter itself is high dimensional.}
This is one of the key distinctions of our framework. To the best of  our  knowledge, only \citet{Chern_BLP_2017} and \citet{Chern_Plugin_2018} have recently considered settings of a similar sort. While the former considers only the special case of linear regression and that too under a moderate dimensional setting (with $d \ll \sqrt{n}$), the latter certainly allows for a more general framework, but their approach is also somewhat abstract. Our approach is comparatively more detailed and targeted specifically towards the missing data setting, where we provide a complete hands-on solution to the problem \eqref{Mest:prob}. Further, another \emph{key} contribution of our work is to provide inferential tools for our estimator which hasn't been considered therein or any other existing work for that matter. %We summarize our main contributions below.

%Our \emph{main contributions} can be summarized in \emph{three} different facets: (i) \emph{estimation}, (ii) \emph{inference} and (iii) \emph{estimation of the nuisance functions}. Adopting a semi-perspective (as in Remark \ref{rem:semiparametric:framework}) and assuming $\btheta_0$ is $s$-sparse (with $s \ll n$), we propose to estimate $\btheta_0$ via an \emph{$L_1$-regularized debiased and doubly robust (DDR) estimator} based on a high dimensional adaptation of the traditional double robust (DR) estimator's construction, along with careful use of debiasing and sample splitting techniques. The DDR estimator serves as the appropriate \emph{generalization} of standard (low dimensional) DR estimators \citep{Bang_2005, Chern_DDML_2018} for high dimensional parameters. We also present a simple \emph{user friendly implementation algorithm} for these estimators which can be achieved with standard software packages. The ambient high dimensionality (of both $\bX$ and $\btheta_0$) coupled with the missingness of $Y$ and the unavoidable presence of other nuisance function estimators (possibly also high dimensional) makes the analyses challenging and substantially \emph{nuanced} compared to the low dimensional case.

\paragraph*{Main contributions}
Our contributions can be summarized in \emph{two} different facets: \emph{(i) estimation} and \emph{(ii) high dimensional inference} for $\btheta_0$. %and (iii) \emph{estimation of the nuisance functions}.
Adopting a semi-perspective (as in Remark \ref{rem:semiparametric:framework}) and assuming $\btheta_0$ is $s$-sparse ($s \ll n$), we propose to estimate $\btheta_0$ via an \emph{$L_1$-regularized debiased and doubly robust (DDR) estimator} based on a high dimensional adaptation of the traditional double robust (DR) estimator's construction, along with careful use of debiasing and sample splitting techniques. The DDR estimator serves as the appropriate \emph{generalization} of standard (low dimensional) DR estimators \citep{Bang_2005, Chern_DDML_2018} for high dimensional parameters. We also present a simple user friendly \emph{implementation algorithm} for these estimators which can be achieved with standard software packages. The ambient high dimensionality (of both $\bX$ and $\btheta_0$) coupled with the missingness of $Y$ and the unavoidable presence of other nuisance function estimators (possibly also high dimensional) makes the analyses challenging and substantially \emph{nuanced} compared to the low dimensional case.
Under mild tail assumptions and arbitrarily chosen (working) models for estimating the two \emph{nuisance functions}, the propensity score (PS) and the outcome regression (OR) function, satisfying \emph{only} some \emph{high-level} (pointwise) consistency conditions, we establish \emph{finite sample performance bounds} for the DDR estimator showing its (optimal) $L_2$ error rate to be $\sqrt{s (\log d)/  n}$ whenever both working models are correct, and its consistency and DR properties when only one of the models is correct. Further, the estimators are first order \emph{insensitive} to any estimation errors or knowledge of construction of the PS and OR estimators, thereby allowing the use of non-smooth high dimensional and/or adhoc non/semi-parametric estimators with unclear first order properties.

Further, when both models are correct, we propose a \emph{desparsified version} of our DDR estimator that satisfies an \emph{asymptotic linear expansion} (ALE) and facilitates \emph{inference} on low dimensional components of $\btheta_0$. The desparsified DDR estimator is similar (in spirit) to a Debiased Lasso type approach \citep{VdG_2014, Javanmard_2014} and serves as its appropriate generalization in the missing data setting. Furthermore, the ALE it achieves is \emph{semi-parametric optimal} and matches the `efficient' influence function for this problem. Finally, we also discuss a variety of flexible \emph{choices of the nuisance function estimators}, including common high dimensional parametric models, as well as more general semi-parametric models based on series estimators and single index models. We also establish (in the \hyperref[supp_mat]{Supplementary Material}) general results for all these estimators that verify their required properties, %needed for our main results
and these may further be of independent interest. All our results on estimation, inference and the DR properties are validated via extensive simulation studies over various data settings, nuisance function (working) models and comparisons with other (optimal) oracle estimators. %The results in general are satisfactory and match our theoretical results. %where we also include for comparison the oracle (and optimal) estimator (with the nuisance functions presumed known) and also a hypothetical estimator based on a complete data. We find that the estimation performance matches the oracle estimator's performance when both nuisance models are correct, and also exhibits the DR property when only one of them is incorrect. %Extensive simulation results are presented to validate our claims

\paragraph*{Organization}
The rest of this paper is organized as follows. In Section \ref{estimation}, we detail our estimation strategy, including preliminaries on DR estimation, followed by
construction and implementation of the DDR estimator as well as deterministic deviation bounds on its performance. Section \ref{sec:prob:bounds} contains our main results (Theorems \ref{TZERO:THM}-\ref{RMPI:THM}), and the associated high-level assumptions, regarding  convergence rates of the DDR estimator. % via non-asymptotic probabilistic bounds for various error terms.
In Section \ref{sec:inference}, we discuss inference using the desparsified DDR estimator and establish all its properties in Theorem \ref{HDINF:THM}. In Section \ref{NUISANCE}, we discuss various choices of the nuisance function estimators. Finally, the simulation results are presented in Section \ref{sec:sim}, followed by a concluding discussion in Section \ref{Discussion}. In Appendices \ref{draspect}-\ref{pfs:nuisance} of the \hyperref[supp_mat]{Supplementary Material}, we collect several important materials that could not be accommodated in the main manuscript, including discussions on DR properties of the estimator, properties of the nuisance estimators (Theorems \ref{parametric:thm}-\ref{KS:mainthm2}), additional numerical results, and all technical materials, including the proofs of all our main results and the associated supporting lemmas.

\section[Estimation Strategy: A General Approach Based on L1-Regularized DDR Loss Minimization]{Estimation Strategy: A General Approach Based on $L_1$- Regularized Debiased and Doubly Robust (DDR) Loss Minimization}\label{estimation}

%\section[High Dimensional Sparse M-Estimation: A General Framework via Regularized DDR Loss Minimization]{High Dimensional M-Estimation and Sparse Signal Recovery: A General Framework via Regularized DDR Loss Minimization}\label{estimation}

\paragraph*{Notation} %\hspace{-0.1in}
We use the following general notations throughout. For any $\bv \in \mathbb{R}^d$, $\| \bv \|_r$ denotes the $L_r$ vector norm of $\bv$ for any $ r \geq 0$, $\overrightarrow{\bv}$ denotes $(1,\bv')' \in \R^{d+1}$, $\bvj$ denotes the $j^{th}$ coordinate of $\bv$ $\forall \; 1 \leq j \leq d$, $\Asc(\bv) := \{ j: \bvj \neq 0 \}$ denotes the support of $\bv$ and $s_{\bv} := |\Asc(\bv)|$ denotes the cardinality of $\Asc(\bv)$. For any $\Jsc \subseteq \{1, \hdots,d\}$ and $\bv \in \mathbb{R}^d$,  we let $\Pi_{\Jsc}(\bv)$ $: = [\bv_{[j]}1\{j \in \Jsc\}]_{j=1}^d \in \R^d$, $\Msc_{\Jsc} := \{ \bv \in \R^d : \Asc(\bv) \subseteq \Jsc \}$ and $\Msc_{\Jsc}^{\perp} := \{ \bv \in \R^d : \Asc(\bv) \subseteq \Jsc^c \}$, where $\Jsc^c := \{1, \hdots, d\} \backslash \Jsc$ denotes the complement of $\Jsc$. We use the shorthand $\Pi_{\bv}(\cdot)$ and $\Pi_{\bv}^c(\cdot)$ to denote $\Pi_{\Asc(\bv)}(\cdot)$ and $\Pi_{\Asc^c(\bv)}(\cdot)$ respectively. Further, for any measurable (and possibly random) function $f(\cdot)$ of $\bX$, we let $\| f(\cdot)\|_r := [\E_{\bX}\{|f(\bX)|^r\}]^{1/r}$ denote the $L_r$ norm of $f(\cdot)$ with respect to (w.r.t.) $\P_{\bX}$ for any $r \geq 1$ and $\|f(\cdot)\|_{\infty} := \sup_{\bx \in \Xsc}|f(\bx)|$ denote the $L_{\infty}$ norm w.r.t. $\P_{\bX}$. For any sequences $a_n, b_n \geq 0$, we use $a_n \lesssim b_n$ to denote $a_n \leq C b_n$ and $a_n \asymp b_n$ to denote $cb_n \leq a_n \leq Cb_n$ for all $n \geq 1$ and some constants $c, C > 0$. Finally, $a_n \ll b_n$ denotes $a_n = o(b_n)$ and $a_n \gg b_n$ denotes $b_n = o(a_n)$ as $n \rightarrow \infty$. %Finally, we use the abbreviation `w.h.p.'  to denote the phrase `with high probability'.

\subsection{Identification and Alternative Representations of the Expected Loss}\label{id_and_alt_rep}
We next provide three alternative representations of $\L(\cdot)$ in terms of the observables $(T,TY,\bX)$ and some \emph{nuisance functions} identifiable through them. These representations underlie three fundamental estimation strategies typically adopted in the literature, %for these problems,
namely inverse probability weighting (IPW) involving the propensity score $\pi(\cdot)$, regression based imputation (REG) involving the conditional mean $\phi(\cdot,\cdot)$, and \emph{doubly robust} (DR) methods that
combine the IPW and REG approaches %as well as regression based imputation
and provide the benefits of (double) robustness against model misspecification %in the estimation
of either one of the two nuisance functions $\pi(\cdot)$ and $\phi(\cdot,\cdot)$. DR estimators are also known to be (locally) semi-parametric optimal %achieve the semi-parametric efficiency bound
when both nuisance function estimation models are correctly specified; see \citet{Imbens_Review_2004, Bang_2005} for a review. %\citet{Robins_1994, Robins_1995, Imbens_Review_2004, Bang_2005, Kang_2007, Tsiatis_Book_2007} and \citet{Graham_2011} for a detailed review of the related classical literature. %on these methods.

\paragraph*{IPW and regression based representations of $\L(\cdot)$}\label{IPWrep:lemma}
\hspace{-0.1in} For any $\btheta \in \R^d$, we have:
\begin{eqnarray*}
\nonumber \L(\btheta) &\equiv& \E\{L(Y, \bX, \btheta)\} \; = \; \E_{\bX}\{\phi(\bX,\btheta)\} \; =: \; \LREG(\btheta) \;\; \mbox{(say), \;\; and} \\
\nonumber \L(\btheta) &\equiv& \E\{L(Y, \bX, \btheta)\} \; = \; \E\left\{ \frac{T}{\pi(\bX)} L(Y, \bX, \btheta) \right\} \; =: \; \LIPW (\btheta) \;\; \mbox{(say).} %\label{IPWrep:eqn}
\end{eqnarray*}

\paragraph*{Debiased and doubly robust (DDR) representation of $\L(\cdot)$}\label{DRrep:lemma}
\hspace{-0.1in}
It also holds that:
\begin{eqnarray}
%\nonumber \L(\btheta) &\equiv& \E\{L(Y, \bX, \btheta)\} \; \equiv \; \LREG(\btheta) \; \equiv \; \LIPW (\btheta) \\ %\E_{\bX}\{\phi(\bX,\btheta)\} \\
\L(\btheta) &=& \E_{\bX}\{\phi(\bX,\btheta)\} + \E \left[ \frac{T}{\pi(\bX)} \left\{ L(Y, \bX, \btheta) - \phi(\bX,\btheta) \right\} \right] \label{DRrep:eqn1} \\
\nonumber &=:& \LDR(\btheta) \;\; \mbox{(say)} \quad \forall \; \btheta \in \R^d.
\end{eqnarray}
Further, for any functions $\phi^*(\bX,\btheta)$ and $\pi^*(\bX)$ such that $\phi^*(\cdot , \cdot) = \phi(\cdot , \cdot)$ or $\pi^*(\cdot) = \pi(\cdot)$  holds, but \emph{not} necessarily both, it continues to hold that:
\begin{equation}
\LDR(\btheta) \; = \; \E_{\bX}\{\phi^*(\bX,\btheta)\} + \E \left[ \frac{T}{\pi^*(\bX)} \left\{ L(Y, \bX, \btheta) - \phi^*(\bX,\btheta) \right\} \right]. \label{DRrep:eqn2}
\end{equation}

$\LDR(\cdot)$, unlike $\LIPW(\cdot)$ and $\LREG(\cdot)$, is thus DR as it is `protected' against misspecification of either $\pi(\cdot)$ or $\phi(\cdot,\cdot)$, as in (\ref{DRrep:eqn2}).
Further, even when both are correctly specified, it has a naturally `debiased' form owing to the second term in \eqref{DRrep:eqn1}. While this term %(often also called the `augmented IPW' term)
is simply $0$ in the population version, it leads to \emph{crucial} first order benefits in the empirical version of the loss involving the nuisance function estimators, where it has a debiasing effect making the loss first order insensitive to any estimation errors of the nuisance functions. Approaches based on other representations don't enjoy these benefits which can be especially crucial in high dimensional settings. Further discussions on these nuances in a more general context can be found in the recent works of \citet{Chern_LRSE_2016, Chern_DDNML_2017, Chern_DDML_2018,Chern_Plugin_2018} and \citet{Chern_DDMLReisz_2018} on the use of \emph{Neyman orthogonal} scores for semi-parametric inference in the presence of (unknown) high dimensional nuisance functions.

Finally, note that all three identifications above are fully non-parametric. They require no further assumptions on $\P$ (apart from Assumption \ref{base:assmpns}). %They follow from simple uses of Assumption \ref{base:assmpns} (and iterated expectations) and require no further assumptions on $\P$.
%%They simply represent the risk $\L(\cdot)$ as functionals of the observables $(T, TY, \bX)$ and the nuisance functions $\pi(\bX)$ and $\phi(\bX,\btheta)$ which are both estimable from the observed data $\Dsc_n$.
The nuisance functions %$\pi(\bX)$ and $\phi(\bX,\btheta)$
are both estimable from the data. $\pi(\bX)$ is estimable from the data on $(T,\bX)$, while under MAR, $\phi(\bX,\btheta) = \E\{ L(Y,\bX,\btheta) \medgiven \bX, T = 1)$ is estimable from the `complete case' data. Note that in some cases, $\phi(\bX,\btheta)$ may itself involve $\E(Y \medgiven \bX)$. While the latter may sometimes also `motivate' the definition of $\btheta_0$ in \eqref{Mest:prob}, as in parametric regression problems, %based on parametric (working) models for $\E(Y \medgiven \bX)$,
this should \emph{not} be confused with its role as a nuisance function in the identifications of $\L(\cdot)$ above. In fact, it plays the \emph{same} role as a nuisance function here as it does for the special case of the mean/ATE estimation problem, where this role is well understood and commonly utilized. %(and its importance) is very well understood and it is common practice to estimate these nuisance functions and use them to implement the DR type estimators.
We emphasize that the same principle (and practice) continue to apply here for the general problem \eqref{Mest:prob}. %and it should not be confused with the other (unrelated) issue.

\subsection{Simplifying Structural Assumptions}\label{simplifying} For simplicity, we shall assume henceforth a structure on the derivative of $L(Y,\bX,\btheta)$ w.r.t. $\btheta$ as follows. For some functions $\bh(\bX) \in \R^d$ and $g(\bX,\btheta) \in \R$, we assume it takes the form:
\begin{equation}
\bnabla L(Y,\bX,\btheta) \; \equiv \; \frac{\partial}{\partial \btheta} L(Y,\bX,\btheta) \; = \; \bh(\bX) \{ Y - g(\bX,\btheta)\}. \label{lossfn:splform1}
\end{equation}
The structural assumption in \eqref{lossfn:splform1} is mostly for simplicity in the theoretical analyses of %regarding probabilistic bounds for
our proposed estimator. This form is satisfied by most standard loss functions used in practice, including the examples given in Section \ref{Mest_prob}. Extensions of our results to loss functions with more general structures may also be obtained easily albeit at the cost of less tractable technical conditions. %which we avoid here for brevity.

Under \eqref{lossfn:splform1}, the loss function $L(Y,\bX,\btheta)$ therefore takes the form:
\begin{equation}
 L(Y,\bX,\btheta) \; = \; \{\bh(\bX)' \btheta\}Y - f(\bX,\btheta) + C(Y,\bX), \;\; \mbox{where} \label{lossfn:splform1a}
\end{equation}
$f(\bX,\btheta)$ is the anti-derivative of $\bh(\bX) g(\bX,\btheta)$ w.r.t. $\btheta$ and $C(Y,\bX)$ is some function independent of $\btheta$, e.g. $C(Y,\bX) := Y^2$ for the squared loss.
Hence, %under \eqref{lossfn:splform1},
$\phi(\bX,\btheta) = \{\bh(\bX)' \btheta\}\E(Y \medgiven \bX) - f(\bX,\btheta) +  m_C(\bX)$ is convex and differentiable, where %$m(\bX) := \E(Y \medgiven \bX)$ and
$m_C(\bX) := \E\{ C(Y,\bX) \medgiven \bX\}$, %Further, $\phi(\bX,\btheta)$ is convex and differentiable in $\btheta$ with
and $\bnabla \phi(\bX,\btheta) := \frac{\partial}{\partial \btheta} \phi(\bX,\btheta) $ is given by:
\begin{equation}
\bnabla \phi(\bX,\btheta) \; = \; \bh(\bX) \{ m(\bX) - g(\bX,\btheta)\}, \;\; \mbox{where} \;\; m(\bX) \; := \; \E(Y \medgiven \bX). \label{lossfn:splform1b}
\end{equation}

Thus, given any estimates $\{\mhat(\bX),\mhat_C(\bX)\}$ of $\{m(\bX),m_C(\bX)\}$, one can estimate $\phi(\bX,\btheta)$ as: $\phihat(\bX,\btheta) := \{\bh(\bX)' \btheta\}\mhat(\bX) - f(\bX,\btheta) + \mhat_C(\bX)$. Further, $\phihat(\bX,\btheta)$ is also convex and differentiable in $\btheta$ and we have:
\begin{equation}
 \bnabla \phihat(\bX, \btheta) \; := \; \frac{\partial}{\partial \btheta}  \phihat(\bX,\btheta) \; = \; \bh(\bX) \{ \mhat (\bX) - g(\bX, \btheta)\}. \label{lossfn:splform2}
\end{equation}
Note that %to compute $\phihat(\bX,\btheta)$ explicitly, one needs both the estimates $\mhat(\cdot)$ and $\mhat_C(\cdot)$. However,
the part of $\phihat(\bX, \btheta)$ involving $\mhat_C(\cdot)$ is \emph{free} of $\btheta$. Our proposed estimator of $\btheta_0$ in Section \ref{sec:DDRest} is constructed using an $L_1$-regularized minimization (w.r.t. $\btheta$) %of an objective function
involving $\phihat(\cdot,\cdot)$, whereby only its gradient $\bnabla \phihat(\bX,\btheta)$ is of interest, and that depends only on $\mhat(\bX)$ due to \eqref{lossfn:splform2}. Thus, the part of $\phihat(\cdot,\cdot)$ involving $\mhat_C(\cdot)$ may be ignored for all practical purposes and we \emph{only} require an estimator $\mhat(\cdot)$ of $m(\cdot)$ for implementing our final estimator of $\btheta_0$.%wherein we \emph{only} require an estimator $\mhat(\cdot)$ of $m(\cdot)$ and an arbitrary choice of $\mhat_C(\cdot)$ to plug in and obtain the estimator $\phihat(\cdot)$.

\subsection[The L1-Regularized DDR Estimator]{The $L_1$-Regularized DDR Estimator}\label{sec:DDRest}
Let $\{\pihat(\cdot), \mhat(\cdot)\}$ be \emph{any} reasonable estimators of $\{\pi(\cdot), m(\cdot)\}$, %such that at least one (but not necessarily both) of them are correctly specified estimators
and we assume that $\pihat(\cdot)$ is obtained solely from the data $\{(T_i, \bX_i)\} _{i=1}^n$ (see Appendix \ref{discussion:errorterms} for more discussions).
Let $\phihat(\cdot,\cdot)$ be the corresponding estimator of $\phi(\cdot,\cdot)$ based on $\mhat(\cdot)$. We use sample splitting to further construct \emph{cross-fitted} versions of $\mhat(\cdot)$ and $\phihat(\cdot,\cdot)$, as follows.

\paragraph*{{Cross-fitted versions of $\mhat(\cdot)$ and $\phihat(\cdot,\cdot)$ based on sample splitting}}\label{crossfit}
Let $\{\Dscnone,$ $\Dscntwo\}$ denote a random partition (or split) of the original data $\Dsc_n$ into $\K = 2$ equal parts of size $\nbar := n/2$, where %denote their sizes, where without loss of generality (w.l.o.g.),
we assume w.l.o.g. that $n$ is even. Further, let $\Iscone$ and $\Isctwo$ respectively denote the index sets for the observations in $\Dscnone$ and $\Dscntwo$. Hence, %%with $\K \equiv 2$, we have: $\{\Dscnk\}_{k=1}^{\K}$ and $\{\Isck\}_{k = 1}^{\K}$ are disjoint partitions of $\Dsc$ and $\Isc$ respectively with
we have $\bigcup_{k=1}^{\K} \Isck = \Isc := \{1, \hdots, n\}$ and $\bigcup_{k=1}^{\K} \Dscnk = \Dsc_n$. %and $|\Isck| = \nbar \equiv n/\K$ $\; \forall \; k \in \{1, \hdots, \K \equiv 2\}$.

Given any general procedure for obtaining $\mhat(\cdot)$ and $\phihat(\cdot,\cdot)$ based on the full observed data $\Dsc_n$, let $\{\mhatone(\cdot), \phihat^{(1)}(\cdot, \cdot)\}$ and $\{\mhattwo(\cdot), \phihat^{(2)}(\cdot,\cdot)\}$ denote the corresponding versions of these estimators based on $\Dscnone$ and $\Dscntwo$ respectively. Then, we define the \emph{cross-fitted} estimates $\{ \mtil (\bX_i), \phitil(\bX_i,\btheta)\}_{i=1}^n$ of $\{m(\bX_i), \phi(\bX_i, \btheta)\}_{i=1}^n$ at the $n$ training points in $\Dsc_n$ as follows:
\begin{align}
\label{tmcont:mphicrsfit:defn} & \{ \mtil(\bX_i), \phitil(\bX_i, \btheta)\} =
\begin{cases}
\{\mhattwo(\bX_i), \phihat^{(2)}(\bX_i, \btheta)\} & \forall \; i \in \Isc_1, \quad \mbox{and} \\
\{\mhatone(\bX_i), \phihat^{(1)}(\bX_i, \btheta)\} & \forall \; i \in \Isc_2.
\end{cases}
\end{align}
A detailed discussion regarding the benefits (and virtual necessity) of considering these cross-fitted estimators is given in Appendix \ref{discussion:errorterms}. Further insights regarding the benefits of cross-fitting for general semi-parametric estimation problems in the presence of nuisance components can also be found in \citet{Chern_LRSE_2016, Chern_DDML_2018, Chern_Plugin_2018} and \citet{Newey_Crossfit_2018}. However, note also that we do \emph{not} require sample splitting for constructing the estimates $\{\pihat(\bX_i)\}_{i=1}^n$ as long as $\pihat(\cdot)$ is obtained only from the data on $\{(T_i, \bX_i)\}_{i=1}^n$.

%\begin{remark}\label{tmcont:genK:rem}
%\emph{
%While we focus here on the simple case of sample splitting with $\K = 2$, our notations and analyses are designed to easily accommodate the general case of $\K$-fold cross fitting for any fixed $\K \geq 2$. We stick to $\K = 2$ for simplicity and brevity of our arguments. Finally, note that the estimator $\bthetahatDR$ obtained via this cross-fitting procedure can also be replicated several times over different splittings of $\Dsc_n$, and then suitably combined over these replications to average out the (minor) randomness due to sample splitting.
%}
%\end{remark}

While we focus on $\K = 2$ for simplicity, our analyses can easily accommodate any (fixed) $\K \geq 2$.
Note also that the final estimator can be replicated several times %over different splittings %of $\Dsc_n$, and then suitably combined over these replications
to average out the (minor) randomness due to the cross-fitting.

\paragraph*{The estimator} Recall the DDR representation of the expected loss $\L(\btheta)$:
\begin{equation*}
 \LDR(\btheta) \; = \;\E_{\bX}\{\phi(\bX;\btheta)\} + \E \left[ \frac{T}{\pi(\bX)} \left\{ L(Y, \bX, \btheta) - \phi(\bX;\btheta) \right\} \right], %\\
\end{equation*}
and define its empirical version, based on the estimates $\{\phitil(\bX,\btheta), \pihat(\bX_i)\}_{i=1}^n$ plugged in, as follows. For any $\btheta \in \R^d$, let us define the \emph{empirical DDR loss} %
\begin{eqnarray}
%\nonumber \LnDR(\btheta) &:=& \frac{1}{n}\sum_{i=1}^n \phitil(\bX_i,\btheta) + \frac{1}{n}\sum_{i=1}^n \frac{T_i}{\pihat(\bX_i)} \left\{L(\Y_i,\bX_i, \btheta_i) - \phitil(\bX_i,\btheta)\right\} \\
&& \; \hspace{0.023in}\LnDR(\btheta) \; := \; \frac{1}{n}\sum_{i=1}^n \phitil(\bX_i,\btheta) + \frac{1}{n}\sum_{i=1}^n \frac{T_i}{\pihat(\bX_i)} \left\{L(Y_i,\bX_i, \btheta_i) - \phitil(\bX_i,\btheta)\right\}. \label{empDRloss}
\end{eqnarray}

With $\btheta_0$ (and $\bX$) possibly high dimensional, we shall need to assume that $\btheta_0$ is sparse with sparsity much smaller than $d$ when $d \gg n$. In general, we denote the sparsity of $\btheta_0$ as $s := \| \btheta_0\|_0 $ with $1 \leq s \leq d$. We
now propose to estimate $\btheta_0$ using the \emph{$L_1$-regularized DDR estimator}, $\bthetahatDR$, given by: %defined as follows.
\begin{equation}
\bthetahatDR \; \equiv \; \bthetahatDR(\lambda_n) \; = \; \underset{\btheta \in \R^d}{\arg \min} \; \left\{\LnDR(\btheta) + \lambda_n \| \btheta\|_1 \right\}, \label{penDRest}
\end{equation}
where
$\LnDR(\cdot)$ is as in \eqref{empDRloss} and $\lambda_n \geq 0$ denotes the regularization (or tuning) parameter. (For a classical setting with $d \ll n$, $\lambda_n$ may be set to $0$ if desired).%, although we do not pursue the theoretical analysis for this case)%(Under a classical setting with $p \ll n$, $\lambda_n$ may be set to $0$, if desired, although we do not pursue the theoretical analysis for this case).

\subsection{Simple Algorithm for Implementation}\label{sec:simple:algorithm}
The estimator $\bthetahatDR$ in \eqref{penDRest} can be implemented using a simple user-friendly imputation type algorithm.

Given the observed data $\Dsc_n$ and the estimates $\{ \pihat(\bX_i), \mtil(\bX_i) \}_{i=1}^n$,  define a set of \emph{pseudo outcomes} $\{\Ytil_i\}_{i=1}^n$ and the \emph{pseudo loss} $\LntilDR(\btheta)$ as follows:
\begin{equation}
\Ytil_i := \mtil(\bX_i) + \frac{T_i}{\pihat(\bX_i)}\{Y_i -  \mtil(\bX_i)\} \; \mbox{and} \; \LntilDR(\btheta) := \frac{1}{n}\sum_{i=1}^n L(\Ytil_i,\bX_i, \btheta). \label{eq:pseudo:outcome:def}
\end{equation}
Clearly $\LntilDR(\cdot)$ is convex and differentiable, and under \eqref{lossfn:splform1}-\eqref{lossfn:splform2}, it is easy to see that $\bnabla \LntilDR(\btheta) = \bnabla \LnDR(\btheta)$, where for any $f(\cdot)$, $\bnabla f(\btheta) := \frac{\partial}{\partial \btheta} f(\btheta)$. %$\bnabla \LntilDR(\btheta) := \frac{\partial}{\partial \btheta} \LntilDR(\btheta)$.

Further, observe that the solution for the minimization in \eqref{penDRest} is uniquely determined by the underlying normal equations (the KKT conditions) which \emph{only} depend on the gradient of $\LnDR(\cdot)$ and the subgradient of $\| \cdot \|_1$. Hence, the solution stays \emph{unchanged} if $\LnDR(\btheta)$ in \eqref{penDRest} is replaced by $\LntilDR(\btheta)$ which has the same gradient. Consequently, $\bthetahatDR$ in \eqref{penDRest} may also be defined as:
\begin{equation}
\bthetahatDR \; \equiv \; \bthetahatDR(\lambda_n) %\; \equiv \; \arg \min_{\btheta \in \R^d} \; \{\LnDR(\btheta) + \lambda_n\| \btheta \|_1 \}
\; := \; \arg \min_{\btheta \in \R^d} \; \{\LntilDR(\btheta) + \lambda_n\| \btheta \|_1 \}. \label{eq:simplealgo}
\end{equation}

Thus, if one `pretends' to have a fully observed data $\widetilde{\Dsc}_n := \{ (\Ytil_i, \bX_i)\}_{i=1}^n$ in terms of the pseudo outcomes $\Ytil_i$, then $\bthetahatDR$ can be simply obtained by a $L_1$-penalized minimization of the corresponding empirical risk for $L(\cdot)$ based on $\widetilde{\Dsc}_n$. This minimization is quite straightforward to implement and can be done so using standard statistical software packages (e.g. \texttt{`glmnet'} in \texttt{R}).

Note also that \eqref{eq:simplealgo} confirms our earlier claim that although the estimator $\phitil(\bX, \btheta)$ involved in the definition \eqref{empDRloss} of $\LnDR(\btheta)$ may require estimation of other nuisance functions (independent of $\btheta$) apart from $m(\bX)$, the implementation of $\bthetahatDR$ via the minimization in \eqref{penDRest}, or equivalently the one in \eqref{eq:simplealgo}, requires \emph{only} an estimator of $m(\bX)$, along with that of $\pi(\bX)$.

\subsection{Performance Guarantees: Deviation Bounds}\label{sec:dev:bound}

We next provide a \emph{deterministic} deviation bound regarding the finite sample performance of $\bthetahatDR$ that serves as the backbone for most of our main theoretical analyses. We begin with an assumption. Recall the notations introduced in Section \ref{estimation}.

\begin{assumption}[Restricted strong convexity]\label{strngconv_assmpn}
\emph{
We assume that the loss function $\LnDR(\btheta)$ satisfies a restricted strong convexity (RSC) property at $\btheta = \btheta_0$, as follows: $\exists$ a (non-random) constant $\kappaDR > 0$ such that
\begin{align}
& \delta \LnDR(\btheta_0; \bv) \; \geq \; \kappaDR \|\bv\|_2^2 \quad \forall \; \bv \in \C(\btheta_0), \quad \mbox{where} \;\; \forall \; \btheta,\bv \in \R^d,\label{strngconv_eqn} \\
& \;\; \delta \LnDR(\btheta; \bv)  \; := \; \LnDR(\btheta +\bv) - \LnDR(\btheta) - \bv'\{\bnabla\LnDR(\btheta)\} \nonumber \\
&  \;\; \mbox{and} \quad \C(\btheta_0)  \; := \; \{ \bv \in \R^d: \; \| \Pi_{\btheta_0}^c(\bv)\|_1  \leq   3 \| \Pi_{\btheta_0}(\bv)\|_1 \} \; \subseteq \; \R^d.  \nonumber
\end{align}
}
\end{assumption}
Assumption \ref{strngconv_assmpn}, largely adopted from \citet{Negahban_2012}, is one of the several restricted eigenvalue type assumptions that are standard in the high dimensional statistics literature. While we assume (\ref{strngconv_eqn}) %to hold
deterministically for any realization of $\Dsc_n$, %it only needs to hold a.s. $[\P]$ for some $\kappaDR > 0$. With
it can be relaxed with appropriate modifications to only hold with high probability (w.h.p.). It is important to note that owing to the very structure of $\LnDR(\cdot)$ in \eqref{empDRloss} and the assumed structures in \eqref{lossfn:splform1}-\eqref{lossfn:splform2} for $L(\cdot)$ and $\phitil(\cdot)$, the RSC condition \eqref{strngconv_eqn} is completely \emph{independent} of the quantities depending on the missingness aspect of the problem, i.e. $\delta \LnDR(\btheta_0; \bv)$ in \eqref{strngconv_eqn} is independent of $\{T_i, Y_i\}_{i=1}^n$ as well as the nuisance function estimates $\{ \pihat(\bX_i), \mtil(\bX_i)\}_{i=1}^n$. In fact, it is the \emph{same} as the corresponding version one would obtain in the case of a fully observed data. This fact also follows from the alternative definition of $\bthetahatDR$ in \eqref{eq:simplealgo} based on the pseudo outcomes and the pseudo loss $\LntilDR(\cdot)$ in \eqref{eq:pseudo:outcome:def}. % which satisfies $\delta\LntilDR(\btheta,\bv) = \delta \LnDR(\btheta,\bv)$.
Thus, verifying \eqref{strngconv_eqn} is \emph{equivalent} to verifying the same for a fully observed data which is quite well studied \citep{Negahban_2012, RudZhou_2013, Lecue_2014, KC_MBSG_2018, Vershynin_2018} for several standard problems under fairly mild conditions. This thereby provides an easy route to verifying the RSC condition \eqref{strngconv_eqn} under our setting.

%In general, for standard choices of the underlying loss function $L(u,v)$ and with $\bX$ sufficiently well behaved (e.g. sub-Gaussian) and $\Var(\bX) \succ \bzero$, (\ref{strngconv_eqn}) can be shown \citep{Vershynin_2012, Negahban_2012, Rud_Zhou_2013} to hold w.h.p. for some constant $\kappaDR > 0$ depending only on the sub-Gaussian norm of $\bX$ and the minimum eigenvalue of $\Var(\bX)$. We shall however stick to the deterministic formulation in (\ref{strngconv_eqn}) for simplicity.

\begin{lemma}[Deterministic deviation bounds for $\bthetahatDR$]\label{DEV:BOUND}
Assume $L(\cdot)$ is convex and differentiable in $\btheta$ and %its derivative
satisfies the form \eqref{lossfn:splform1}. %Let $\bthetahatDR$ be as defined in \eqref{penDRest} %or \eqref{eq:simplealgo},
Let Assumption \ref{strngconv_assmpn} hold, with $\kappaDR > 0$ as defined therein, and recall that $s := \| \btheta_0\|_0$. Then, for any realization of $\Dsc_n$ and for any choice of $\lambda \equiv \lambda_n \geq 2 \left\| \bnabla \LnDR (\btheta_0)\right\|_{\infty}$, %$\bthetahatDR$ in \eqref{penDRest} satisfies %we have:
\begin{equation}
\| \bthetahatDR - \btheta_0 \|_2 \; \leq \; 3 \sqrt{s} \frac{\lambda_n}{\kappaDR} \;\; \mbox{and} \;\;
\| \bthetahatDR - \btheta_0 \|_1 \; \leq \; 12s \frac{\lambda_n}{\kappaDR}. \label{det:bound}
\end{equation}
\emph{Convergence rates (informal statement).} %In Section \ref{sec:prob:bounds},
We establish via Theorems \ref{TZERO:THM}-\ref{RMPI:THM} later that under suitable assumptions (given in Section \ref{assumptions}), $\left\| \bnabla \LnDR (\btheta_0)\right\|_{\infty}$ $\lesssim \sqrt{(\log d)/n}$ w.h.p. Hence, choosing $\lambda \equiv \lambda_n \asymp \sqrt{(\log d)/n}$,  it follows that  %w.h.p.,
\begin{equation*}
\| \bthetahatDR - \btheta_0 \|_2 \; \lesssim \; \sqrt{\frac{s \log d}{n}} \;\; \mbox{and} \;\;
\| \bthetahatDR - \btheta_0 \|_1 \; \lesssim \; s\sqrt{\frac{ \log d}{n}} \;\; \mbox{w.h.p.}  \label{informal:rate:bound}
\end{equation*}
\end{lemma}
The deviation bounds \eqref{det:bound}, essentially an easy consequence of the results of \citet{Negahban_2012}, deterministically relate the $L_2$ and $L_1$ error rates of the estimator to the chosen $\lambda_n$ and provides an easy recipe for establishing its convergence rates by studying the same for the (random) lower bound of $\lambda_n$ given in Lemma \ref{DEV:BOUND}. %Hence, the \emph{main goal} from hereon is to analyze this (random) lower bound of $2 \left\| \bnabla \LnDR (\btheta_0)\right\|_{\infty}$ in Lemma \ref{DEV:BOUND} regarding the choice of $\lambda_n$.
This is the main goal of Section \ref{sec:prob:bounds}, where we obtain sharp non-asymptotic upper bounds for $\left\| \bnabla \LnDR (\btheta_0)\right\|_{\infty}$ converging to $0$ at satisfactory rates w.h.p. A choice of $\lambda_n$ of the order of this bound guarantees the requirement of $\lambda_n \geq 2 \left\| \bnabla \LnDR (\btheta_0)\right\|_{\infty}$ in Lemma \ref{DEV:BOUND} to hold w.h.p. and establishes the convergence rates, defined by the $\lambda_n$, for the bounds in \eqref{det:bound}.

Finally, note also that the (informal) bounds in the second part of Lemma \ref{DEV:BOUND} establish the obvious \emph{rate optimality} of the estimator since it matches the (well known) optimal estimation error rate for a fully observed data. % (that is included here a s a special case).

\section[The Main Results for the DDR Estimator: Convergence Rates and Probabilistic Bounds]{The Main Results for the DDR Estimator: Convergence Rate and Probabilistic Bounds for $\left\| \bnabla \LnDR (\btheta_0)\right\|_{\infty}$}\label{sec:prob:bounds}

For most of our theoretical analyses of $\left\| \bnabla \LnDR (\btheta_0)\right\|_{\infty}$, we will assume that $\{\pihat(\cdot), \mhat(\cdot)\}$ are both correctly specified estimators of $\{\pi(\cdot),m(\cdot)\}$. The analysis even for this case is involved (and necessarily non-asymptotic) due to the high dimensionality. %Even for this case, the analysis is quite involved
%to the presence of the nuisance function estimators %(that leads to controlling averages of dependent variables)
%and the inherent high dimensional setting. %which necessitates non-asymptotic bounds.

Under possible misspecification of one of the estimators, the DR property (in terms of consistency) of $\left\| \bnabla \LnDR (\btheta_0)\right\|_{\infty}$ and that of $\bthetahatDR(\lambda_n)$, for a suitably chosen $\lambda_n$ under Lemma \ref{DEV:BOUND}, indeed follows due to the very nature of construction of $\LnDR (\cdot)$ and its population version $\LDR (\cdot)$ outlined in \eqref{DRrep:eqn1}-\eqref{DRrep:eqn2}. This DR property is well known in classical settings \citep{Robins_1994, Robins_1995, Bang_2005} and should also be expected to hold in high-dimensional settings under suitable conditions. We discuss these DR properties further in Appendix \ref{draspect}. %Section \ref{draspect}.

One of the reasons behind considering the DDR representations $\LDR(\btheta)$ and $\LnDR(\btheta)$ is that apart from the obvious benefits of double robustness, the DDR loss has a naturally `debiased' form that provides \emph{crucial} technical benefits in controlling the associated error terms which are naturally `centered' in a certain sense (see Appendix \ref{discussion:errorterms} for more details on these technical aspects). %when both $\pihat (\cdot)$ and $\mhat(\cdot)$ are correctly specified. %, a setting where other approaches such as IPW and REG type estimators are also applicable in principle, but they don't enjoy such technical benefits.
The advantages of such debiased representations, especially in high dimensional settings, have also been studied in a more general context under the name of \emph{Neyman orthogonalization} in the recent works of \citet{Chern_LRSE_2016, Chern_DDNML_2017, Chern_DDML_2018, Chern_Plugin_2018} and \citet{Chern_DDMLReisz_2018}. %The DDR representation indeed (naturally) satisfies such an `orthogonal' structure.

%Lastly, even under possible misspecification of one of $\{\pihat(\cdot), \mhat(\cdot)\}$, we still sketch the DR property of $\left\| \bnabla \LnDR (\btheta_0)\right\|_{\infty}$ in Section \ref{draspect} later, but only in terms of consistency with the convergence rates being obtained based on a general but crude analysis. To obtain possibly faster convergence rates, one needs a more nuanced and careful analysis which, even in a classical setting, has to be done in a case-specific manner as the analysis and the first order properties (and convergence rates) now \emph{will} depend on the exact nature of construction of the estimators and their corresponding first order properties. Considering the main goals and scope of this paper, we suppress such finer analysis under those cases for the sake of simplicity and brevity.
%
%\paragraph*{The basic decomposition}
\subsection{The Basic Decomposition}\label{basic:decomp}
Let $\bTn := \bnabla \LnDR (\btheta_0) \in \R^d$ with $\| \bTn \|_{\infty}$ being our quantity of interest. We first note a decomposition of $\bTn$ as follows.
\begin{eqnarray}
\nonumber \bTn &=& \bTzeron + \bTpin - \bTmn - \bRpimn \\
     &:=& \frac{1}{n}\sum_{i=1}^n \bTzero(\bZ_i) + \frac{1}{n}\sum_{i=1}^n \bTpi(\bZ_i) - \frac{1}{n} \sum_{i=1}^n \bTm(\bZ_i) - \frac{1}{n}\sum_{i=1}^n \bRpim(\bZ_i), \label{decomp:eqn}
\end{eqnarray}
where $\bTzero(\bZ)$, $\bTpi(\bZ)$, $\bTm(\bZ)$ and $\bRpim(\bZ)$ with $\bZ = (T,\Y,\bX)$ are given by:
\begin{alignat}{2}
 & \bTzero(\bZ) && \; := \; \{m(\bX) - g(\bX,\btheta_0)\} \bh(\bX) + \frac{T}{\pi(\bX)} \{ Y - m(\bX)\} \bh(\bX) \label{tzero} \\
%\nonumber & \bTzero(\bZ) && \; := \; \bTzeroone(\bZ) \; + \; \bTzerotwo(\bZ), \quad \mbox{where} \\
%& \bTzeroone(\bZ) && \; := \; \{m(\bX) - g(\bX,\btheta_0)\} \bh(\bX) \quad \mbox{and} \label{tzero1} \\
%& \bTzerotwo(\bZ) && \; := \; \frac{T}{\pi(\bX)} \{ Y - m(\bX)\} \bh(\bX), \label{tzero2} \\
& \bTpi(\bZ) && \; := \; \left\{\frac{T}{\pihat(\bX)} - \frac{T}{\pi(\bX)} \right\} \{ Y - m(\bX)\}  \bh(\bX), \label{tpi}\\
& \bTm(\bZ) && \; := \; \left\{ \frac{T}{\pi(\bX)} - 1 \right\} \{\mtil(\bX) - m(\bX)\} \bh(\bX), \quad \mbox{and} \label{tm}\\
& \bRpim(\bZ) && \; := \; \left\{\frac{T}{\pihat(\bX)} - \frac{T}{\pi(\bX)}\right\} \{\mtil(\bX) - m(\bX)\} \bh(\bX). \label{rmpi}
\end{alignat}
In the decomposition \eqref{decomp:eqn}, $\bTzeron$ denotes the leading (first order) term, while $\bTpin$ and  $\bTmn$ denote the main error terms accounting for the estimation errors of $\pihat(\cdot)$ and $\mhat(\cdot)$ respectively, and $\bRpimn$ is a second order bias term involving the product of the estimation errors of $\pihat(\cdot)$ and $\mhat(\cdot)$.

%We next proceed towards our main goal of obtaining non-asymptotic probabilistic bounds for $\| \bTn \|_{\infty}$.
\paragraph*{Summary of results} We control $\| \bTn \|_{\infty}$ by separately controlling $\|\bTzeron\|_{\infty}$, $\|\bTpin \|_{\infty}$, $\|\bTmn\|_{\infty}$ and $\|\bRpimn\|_{\infty}$ through Theorems \ref{TZERO:THM}-\ref{RMPI:THM}. The results show that the convergence rate of $\| \bTn \|_{\infty}$ is determined primarily by that of the leading term $\| \bTzeron \|_{\infty}$ while the rates of the other three terms are of a (faster) lower order. In particular, we show that under suitable assumptions,
\[
\|\bTzeron\|_{\infty} \lesssim \sqrt{\frac{\log d}{n}} \;\; \mbox{and} \;\; \|\bTpin \|_{\infty} + \|\bTmn\|_{\infty}+ \|\bRpimn\|_{\infty} \lesssim \sqrt{\frac{\log d}{n}} o(1)
\]
w.h.p. The results (proved in Appendices \ref{pf:tzero:thm}-\ref{pf:rmpi:thm}) are all non-asymptotic (with precise constants) and involve careful analyses via concentration inequalities to account for the nuisance function estimators and the high dimensionality. %These technical tools are collected in Appendix \ref{techtools}.

\begin{remark}[Generality of the results]\label{rem:firstorderinsensitive}
\emph{
It is important to note that our results here are completely \emph{free} in terms of choice of the nuisance function estimators. The analysis and the convergence rates are \emph{first order insensitive} to any estimation errors of the nuisance functions and hold \emph{regardless} of any knowledge of the construction and/or first order properties of the estimators, %how these estimators are obtained and what their first order properties are,
as long as they satisfy some basic high-level conditions on their convergence rates. This is also largely an artifact of the debiased form of the DDR loss.
}
\end{remark}

\subsection{The Assumptions Required}\label{assumptions}
We first summarize the main assumptions required for controlling the various terms in \eqref{decomp:eqn}. We begin with a few standard assumptions on the tail behaviors of some key random variables. % involved in our analyses.

\begin{assumption}[Sub-Gaussian tail behaviors]\label{subgaussian:assmpn}
\emph{\noindent (a) We assume that $\varepsilon(\Z) := Y - m(\bX)$, $\psi(\bX) := m(\bX) - g(\bX, \btheta_0)$ and $\bh(\bX)$ are sub-Gaussian (as per Definition \ref{orlicz:def} in Appendix \ref{techtools}, with $\alpha = 2$ therein) with $\psitwonorm{\varepsilon} \leq \sigmaeps$,  $\psitwonorm{\psi(\bX)} \leq \sigmapsi$ and $\psitwonorm{\bh(\bX)} \leq \sigmabh$, for some constants $\sigmaeps, \sigmapsi, \sigmabh \geq 0$. \smallskip %where the sub-Gaussian behaviors are as defined in Definition \ref{subgauss:def}.
}

\emph{
(b) For controlling $\bTpin$, we additionally assume that $\{\varepsilon(\Z) \given \bX\}$ is (conditionally) sub-Gaussian with $\psitwonorm{\varepsilon(\Z) \given \bX} \leq \sigmaeps(\bX)$ for some function $\sigmaeps(\cdot) \geq 0$ such that $\| \sigmaeps(\cdot) \|_{\infty}$ $\leq \sigmaeps < \infty$ with $\sigmaeps$ being as in part (a) above.
}
\end{assumption}
Next, we discuss the basic high-level conditions required on $\pihat(\cdot)$ and $\mhat(\cdot)$. %the behavior and convergence rates of the nuisance function estimators $\pihat(\cdot)$ and $\mhat(\cdot)$. Further discussions on the assumptions are given in Remarks \ref{assmpns:rem1}-\ref{assmpns:rem2}.

\begin{assumption}[Tail bounds on the pointwise behavior of $\pihat(\cdot) - \pi(\cdot)$]\label{tpicont:assmpn}
\emph{
We assume that $\pihat(\cdot)$ is obtained solely from the data $\Xsc_n := \{ (T_i, \bX_i)\}_{i=1}^n \subseteq \Dsc_n$, and for some sequences $\vnpi \geq 0$ and $q_{n,\pi} \in [0,1]$, with $\vnpi = o(1)$ and $q_{n,\pi} = o(1)$, $\pihat(\cdot) - \pi(\cdot)$ satisfies a (pointwise) tail bound at the $n$ training points $\{\bX_i\}_{i=1}^n$ as follows: for any $t \geq 0$ and for some constant $C \geq 0$,
\begin{equation}
\P \{ | \pihat(\bX_i) - \pi(\bX_i) | > t \vnpi \} \; \leq \; C \exp( - t^2) + q_{n, \pi} \;\; \forall \; 1 \leq i \leq n, \label{pi:tailbound2}
\end{equation}
and we further assume that $ v_{n,\pi}\sqrt{\log (nd)} = o(1)$ and $q_{n,\pi} = o(n^{-1}d^{-1})$. %The bound \eqref{pi:tailbound2} can also be equivalently stated as
}
\end{assumption}

\begin{assumption}[Pointwise tail bounds on $\mhat(\cdot) - m(\cdot)$]\label{tmcont:assmpn}
\emph{
For a generic version of $\mhat(\cdot)$ obtained from a data of size $n$ (e.g. $\Dscn$), we assume that for some sequences  $v_{n,m} \geq 0$ and $q_{n,m} \in [0,1]$, with $v_{n,m} = o(1)$ and $q_{n,m}  = o(1)$,
$\mhat(\cdot) - m(\cdot)$ satisfies a (pointwise) tail bound  at any fixed $\bx \in \Xsc$ as follows: for any $t \geq  0$ and for some constant $C > 0$,
\begin{align}
& \P \{ | \mhat(\bx) - m(\bx) | > t \vnm \} \; \leq \;  C \exp ( -t^2) + q_{n,m}, \;\; \mbox{so that} \label{m:tailbound1} \\
& \P \{ | \mhatk(\bX_i) - m(\bX_i) | > t \vnbarm \} \; \leq \; C \exp ( -t^2) + q_{\nbar,m}, \;\; \forall \; k = 1,2 \label{m:tailbound2}
\end{align}
and $\bX_i \in \Dscnkp \ind \Dscnk$ with $k' \neq k \in \{1,2\}$, where $\nbar := n/2$ and $\mhatk(\cdot)$ denotes the version of $\mhat(\cdot)$ obtained from $\Dscnk$ with size $\nbar \equiv n/2$.
%}
%\emph{
Further, we assume that $v_{\nbar,m} \sqrt{\log (nd)} = o(1)$ and $q_{\nbar,m} = o(n^{-1}d^{-1})$. %and that $\vnpi \vnbarm (\log n)  = o\{\sqrt{(\log d)/n}\}$, where $\vnpi$ is as in Assumption \ref{tpicont:assmpn}.
}
\end{assumption}

%\subsection{Discussion on the Assumptions for the Nuisance Function Estimates}\label{discussion:assmpn}

\begin{remark}\label{assmpns:rem1}
\emph{
Assumptions \ref{tpicont:assmpn} and \ref{tmcont:assmpn} are both fairly mild and general (high-level) conditions that should be expected to hold for most reasonable estimators $\{\pihat(\cdot), \mhat(\cdot)\}$ of $\{\pi(\cdot), m(\cdot)\}$. Note that \eqref{pi:tailbound2}, \eqref{m:tailbound1} and \eqref{m:tailbound2} are all conditions on the \emph{pointwise}  behaviors of  $\pihat(\cdot) - \pi(\cdot)$ and $\mhat(\cdot) - m(\cdot)$, and do \emph{not} require any uniform tail bounds over all $\bx \in \Xsc$, such as bounds on the $L_{\infty}$ or $L_2$ errors of $\{\pihat(\cdot),\mhat(\cdot)\}$.
%$L_{\infty} errors  $\|\pihat(\cdot) - \pi(\cdot)\|_{\infty}$ and $\|\mhat(\cdot) - m(\cdot) \|_{\infty}$, or $L_2$ errors $\|\pihat(\cdot) - \pi(\cdot)\|_2$ and $\|\mhat(\cdot) - m(\cdot) \|_2$.
Such conditions are much stronger and also generally harder to verify in high dimensional settings.
%}
%
%\emph{
We simply require pointwise tail bounds for the errors $\pihat(\bX_i) - \pi(\bX_i)$ and $\mhat(\bx) - m(\bx)$, ensuring that they have well-behaved tails. The sequences $\{\vnpi, \vnm\}$ indicate the convergence rates of the estimators, while $\{q_{n,\pi}, q_{n,m}\}$ in the probability bounds allow to rigourously account for any potential lower order terms.
%Finally, note that we require explicit tail bounds, as opposed to assumptions only on the asymptotic rates, since $\bT_n$ is high dimensional and these error terms are directly involved in analyzing the (non-asymptotic) behavior of $\| \bT_n\|_{\infty}$ that we need to control. %behavior and rate of $\| \bTn \|_{\infty}$ with $\bT_n$ being high dimensional.
}
\end{remark}

\begin{remark}[Sufficient conditions for Assumptions \ref{tpicont:assmpn} and \ref{tmcont:assmpn}]\label{assmpns:rem:extra}
\emph{
In general, for any estimator $\pihat(\cdot)$ satisfying a high probability guarantee of the form: $|\pihat(\bX_i) - \pi(\bX_i)| \leq v_n$ with probability at least $1 - q_n$, the bound \eqref{pi:tailbound2} %in Assumption \ref{tpicont:assmpn}
can be shown to hold with $\{\vnpi, q_{n,\pi} \} \propto \{v_n, q_n\}$, through a simple use of Hoeffding's inequality (see Lemma \ref{lem:7:hptosgtypebound} in this regard). Similarly, for any estimator $\mhat(\cdot)$ satisfying a high probability bound: $|\mhat(\bx) - m(\bx)| \leq v_n$ with probability at least $1 - q_n$, the bounds \eqref{m:tailbound1}-\eqref{m:tailbound2} %in Assumption \ref{tpicont:assmpn}
can be shown to hold with $\{\vnm, q_{n,m} \} \propto \{v_n, q_n\}$.
These high probability bounds are expected to be satisfied by most reasonable estimators and hence, so are our assumptions.
}
\end{remark}

\begin{remark}[Examples]\label{assmpns:rem2}
\emph{
In Section \ref{NUISANCE}, we discuss several choices of the estimators $\pihat(\cdot)$ and $\mhat(\cdot)$ based on parametric families, `extended' parametric families (series estimators) and semi-parametric single index families. For all these estimators, we establish precise tail bounds (see Theorems \ref{parametric:thm}, \ref{KS:mainthm1} and \ref{KS:mainthm2}) that are generally useful and should be of independent interest. Among other implications, they also verify the bounds in Assumptions \ref{tpicont:assmpn} and \ref{tmcont:assmpn}.
}
\end{remark}

\subsection{Controlling the Leading Order Term}\label{cont:tzero}

The following result quantifies the behavior and convergence rate of the first order term $\|\bTzeron\|_{\infty}$ in \eqref{decomp:eqn}. %which is essentially the `leading order' term  but also has the simplest structure and is easiest to control among all terms in \eqref{decomp:eqn}.

\begin{theorem}[Control of $\| \bTzeron\|_{\infty}$]\label{TZERO:THM}
\hspace{-0.13in} Under Assumptions \ref{base:assmpns} and \ref{subgaussian:assmpn} (a), %and with the constants $\sigma_0, K_0 > 0$ as defined above, we have: for any $\epsilon > 0$,
\begin{equation*}
\P \left( \left\| \bTzeron \right\|_{\infty} > \sqrt{2} \sigma_0 \epsilon + K_0 \epsilon^2 \right) \; \leq \; 4 \exp\left( - n \epsilon^2 + \log d \right) \quad \mbox{for any} \;\; \epsilon \geq 0, \label{btzeron:thm:finalbound}
\end{equation*}
where $\sigma_0 := 2\sqrt{2}\sigmabh(\sigmapsi + \sigmaeps \deltapi^{-1})$, $K_0 := 2\sigmabh(\sigmapsi + \sigmaeps \deltapi^{-1}) $ and $(\deltapi, \sigmaeps, \sigmabh, \sigmapsi)$ are as defined in the assumptions. In particular, setting $\epsilon = c \sqrt{(\log d)/n}$ for any constant $c > 1$, we have: with probability at least $1 - 4 d^{-(c^2-1)}$,
\begin{equation*}
%\nonumber && \mbox{With probability} \;\; \geq \;\; 1 - \frac{4}{d^{c^2-1}},\\
\left\| \bTzeron \right\|_{\infty} \; \leq \; c\sqrt{\frac{\log d}{n}} \sqrt{2}\sigma_0 + c^2\frac{\log d}{n} K_0  \;\; \lesssim \; \sqrt{\frac{\log d}{n}}. \label{btzeron:easybound}
\end{equation*}
\end{theorem}
%Here and throughout, all our `high probability' bounds are presented with a logarithmic factor of $\sqrt{\log (nd)}$ in the rates to ensure that the probability guarantee of $1 - O\{ (nd)^{-1}\}$ converges to $1$ \emph{regardless} of whether $d$ diverges with $n$ or not. The more familiar rate involving $\sqrt{\log d}$ and a probability guarantee of $1 - O(d^{-1})$ may also be recovered by an analogous choice of $\epsilon$. This argument also applies to all subsequent results and will not be repeated.

\subsection{Controlling the Error Term from the Propensity Score's Estimation}\label{cont:tpi}

Next, we propose the following result to control the error term $\bTpin$ in (\ref{decomp:eqn}).

\begin{theorem}[Control of $\| \bTpin \|_{\infty}$]\label{TPI:THM}
Let Assumptions \ref{base:assmpns}, \ref{subgaussian:assmpn} and \ref{tpicont:assmpn} hold with $(\vnpi, q_{n,\pi})$ and $(\deltapi,\sigmaeps, \sigmabh, C)$ as defined therein. %and let $\bmusqbhinf := \max\{ \E\{\bhj^2(\bX)\}: j = 1, \hdots, d\}$.
Then, for any constants $c, c_1, c_2, c_3 > 1$, where we assume $ c_2 \vnpi \sqrt{\log (nd)}  \leq \deltapi/2 < \deltapi$ and $c_3 \sqrt{(\log d)/n} < 1$ w.l.o.g., we have: with probability at least $1 - 2d^{-(c^2-1)} - 4d^{-(c_3^2-1)} -  2C (nd)^{-(c_1^2-1)} - 2C (nd)^{-(c_2^2-1)} - 4q_{n,\pi}(nd)$, %$1 - \frac{2}{d^{c^2-1}} - \frac{4}{d^{c_3^2-1}} - \sum_{j=1}^2\frac{2C}{(nd)^{c_j^2-1}} - 4q_{n,\pi}(nd)$,
%\begin{equation*}
% \mbox{With probability} \; \geq \; 1 - \frac{2}{d^{c^2-1}} - \frac{4}{d^{c_3^2-1}} - \sum_{j=1}^2\frac{2C}{(nd)^{c_j^2-1}} - 4q_{n,\pi}(nd),
%\end{equation*}
\begin{equation*}
\|\bTpin\|_{\infty} \; \leq \; c\sqrt{\frac{\log d}{n}} \{\vnpi \sqrt{\log (nd)}\} C_1 \left( \frac{\bmusqbhinf}{\deltapi} + C_2 \sqrt{\frac{\log d}{n}} \right)^{\half},
\end{equation*}
where $\bmusqbhinf := \max_{1 \leq j \leq d}\E\{\bhj^2(\bX)\}$, $C_1 := c_1 (4\sqrt{2} \sigmaeps/\deltapi)$ and $C_2 := c_3 (\sqrt{2}\sigma_{\pi} + K_{\pi})$ with $\sigma_{\pi} := 2 \sqrt{2}\sigmabh^2 \deltapi^{-2}$ and $K_{\pi} := 2\sigmabh^2 \deltapi^{-2}$ being constants.
\end{theorem}
\begin{remark}\label{tpi:rem}
\emph{
Theorem \ref{TPI:THM} therefore shows that $\|\bTpin\|_{\infty} \lesssim \sqrt{(\log d)/ n}$ $\{\vnpi \sqrt{\log (nd)}\}  = o\{\sqrt{(\log d)/ n}\}$ w.h.p. The proof is given in Appendix \ref{pf:tpi:thm}. %In the proof of Theorem \ref{TPI:THM}, we also provide a general result (Theorem \ref{tpi:thm:tailbound}) on tail bounds for $\bTpin$.
}
\end{remark}

\subsection{Controlling the Error Term from the Conditional Mean's Estimation}\label{cont:tm}
We now control the error term $\bTmn$ in (\ref{decomp:eqn}) involving the cross-fitted estimates $\{\mtil(\bX_i)\}_{i=1}^n$ obtained via sample splitting, through the result below.

\begin{theorem}[Control of $\| \bTmn \|_{\infty}$]\label{TM:THM}
%Assume the setup, conditions and results of Theorem \ref{tm:thm:tailbound}, and adopt all notations used therein.
Let Assumptions \ref{base:assmpns}, \ref{subgaussian:assmpn} (a) and \ref{tmcont:assmpn} hold, with $(\vnbarm, q_{\nbar,m})$, $\nbar \equiv n/2$ and %the constants
$(\deltapi, \sigmabh, C)$ as defined therein.
Then, for any constants $c, c_1, c_2 > 1$, where we assume $c_2 \sqrt{(\log d)/\nbar} < 1$ w.l.o.g., %we have:
with probability at least $1- 4d^{-(c^2-1)} - 8d^{-(c_2^2-1)} - 4C(\nbar d)^{-(c_1^2-1)} - 4 q_{\nbar,m} (\nbar d)$, %$1- \frac{4}{d^{c^2-1}} - \frac{8}{d^{c_2^2-1}} - \frac{4C}{(\nbar d)^{c_1^2-1}} - 4 q_{\nbar,m} (\nbar d)$,
%\begin{equation*}
% \mbox{With probability} \; \geq \; 1 - \frac{4}{d^{c^2-1}} - \frac{8}{d^{c_2^2-1}} - \frac{4C}{(\nbar d)^{c_1^2-1}} - 4 q_{\nbar,m} (\nbar d),\\
%\end{equation*}
\begin{equation*}
\|\bTmn\|_{\infty} \; \leq \; c \sqrt{\frac{\log d}{n}} \{ \vnbarm \sqrt{\log (n d)}\} C_1^* \left( \bmusqbhinf + C^*_2 \sqrt{\frac{\log d}{n}}\right)^{\half},
\end{equation*}
where $\bmusqbhinf$ is as in Theorem \ref{TPI:THM}, $C_1^* := 4 c_1 \deltapibar$ and $C_2^* := \sqrt{2}c_2 (\sqrt{2}\sigma_m + K_m)$, with $\sigma_m := 2 \sqrt{2} \sigmabh^2$, $K_m := 2 \sigmabh^2$ and $\deltapibar \leq \deltapi^{-1}$ being constants.
\end{theorem}
\begin{remark}\label{tm:rem}
\hspace{-0.05in}
\emph{
Theorem \ref{TM:THM} therefore shows that $\|\bTmn\|_{\infty} \lesssim \sqrt{(\log d )/ n} $ $\{\vnbarm \sqrt{\log (n d)} \} = o\{\sqrt{(\log d)/n}\}$ w.h.p. The proof is given in Appendix \ref{pf:tm:thm}.%In the proof of Theorem \ref{TM:THM}, we also provide a general result (Theorem \ref{tm:thm:tailbound}) on tail bounds for $\bTmn$.
}
\end{remark}

\subsection{Controlling The Lower Order Term}\label{cont:rmpi}
Finally, we control the second order error (or bias) term $\bRpimn$ in (\ref{decomp:eqn}) through the following result. %involving the random variable $\bRpim(\bZ)$ defined in (\ref{rmpi}).

\begin{theorem}[Control of $\| \bRpimn \|_{\infty}$]\label{RMPI:THM}
Let Assumptions \ref{base:assmpns}, \ref{subgaussian:assmpn}, \ref{tpicont:assmpn} and \ref{tmcont:assmpn} hold with $(\vnpi,  q_{n,\pi})$, $(\vnbarm, q_{\nbar,m}, \nbar)$ and $(\deltapi, C)$ as defined therein, and assume that $\vnpi \vnbarm (\log n) = o\{ \sqrt{(\log d)/n}\}$. Then, for any constants $c_1,c_2,c_3,c_4 > 1$ with $ c_2 \vnpi \sqrt{\log n} \leq \deltapi/2 < \deltapi$ and $c_4 \sqrt{(\log d)/n } < 1$, we have: with probability at least $1 - \sum_{j=1}^3 C n^{-(c_j^2 - 1)} - 2d^{-(c_4^2 - 1)} - 2nq_{n,\pi} - n q_{\nbar,m}$, %$1 - \sum_{j=1}^3\frac{C}{n^{c_j^2 - 1}} - \frac{2}{d^{c_4^2 - 1}} - 2nq_{n,\pi} - n q_{\nbar,m},$
\begin{equation*}
%&& \mbox{With probability} \; \geq \; 1 - \sum_{j=1}^3\frac{C}{n^{c_j^2 - 1}} - \frac{2}{d^{c_4^2 - 1}} - 2nq_{n,\pi} - n q_{\nbar,m}, \\
\left\| \bRpimn \right\|_{\infty}\; \leq  \; c_1 c_3 \bar{C}_1  \{ \vnpi \vnbarm (\log n)\} %\rpin \rmn
\left(\bmumodbhinf + c_4\bar{C}_2\sqrt{\frac{\log d}{n}}\right),  \;\; \mbox{where}
\end{equation*}
%$\rpin := \vnpi \sqrt{\log n} $ and $\rmn := \vnbarm \sqrt{\log n }$ with $\rpin \rmn = o\{\sqrt{(\log d)/n}\}$,
$\bmumodbhinf := \max_{1 \leq j \leq d}\E\{ |\bhj(\bX)| \}$ and $\bar{C}_1 := 2/\deltapi$, $\bar{C}_2 := \sqrt{2}\sigma_{\pi,m} + K_{\pi,m}$ are constants with $\sigma_{\pi,m} := 4 \sigmabh\deltapi^{-1} $ and $K_{\pi,m} := 2 \sqrt{2} \sigmabh \deltapi^{-1}$.
\end{theorem}
\begin{remark}\label{rmpi:rem}
\hspace{-0.115in}
\emph{
Thus, Theorem \ref{RMPI:THM} shows $\|\bRpimn\|_{\infty} \lesssim \vnpi \vnbarm (\log n) %\rpin \rmn
= o\{\sqrt{(\log d)/n}\}$ w.h.p. where the last step is by assumption, a sufficient condition for which is $\max\{\vnpi, \vnbarm \} (\log n)^{1/2} \lesssim \{(\log d)/n \}^{1/4}$.
%$\max\{\rpin, \rmn\} \lesssim \{(\log d)/n\}^{0.25}$.
Conditions of this flavor are well known and standard in the mean (or ATE) estimation literature, where they are routinely adopted to control these kind of second order (product-type) bias terms \citep{Farrell_2015, Chern_DDML_2018}. In Theorem \ref{rmpi:thm:tailbound}, %the proof of Theorem \ref{RMPI:THM},
we provide a more general result on tail bounds for $\bRpimn$.
}
\end{remark}

\section[High Dimensional Inference via the DDR Estimator: Asymptotic Linear Expansions]{High Dimensional Inference via the DDR Estimator: Desparsification and Asymptotic Linear Expansion}\label{sec:inference}

We next discuss a debiasing/desparsification approach for the DDR estimator $\bthetahatDR$ which is useful for establishing an estimator with an \emph{asymptotic linear expansion} (ALE), a property not possessed by the $L_1$-regularized shrinkage type estimator $\bthetahatDR$. Such expansions form the fundamental ingredients for high dimensional inference as they automatically lead to asymptotic normality (and hence confidence intervals, p-values, tests etc.) for low dimensional components of $\btheta_0$, thus paving the way for inference on $\btheta_0$ among many other implications. For a fully observed data (much unlike the setting we have here), such problems have received substantial attention in recent times \citep{VdG_2014, Javanmard_2014, Javanmard_2018, Cai_2017}. %much unlike the case of incomplete data as we have here.

For simplicity, we restrict our discussion here to the case of the squared loss: $L(Y, \bX, \btheta) =  \{ Y - \bPsi(\bX)'\btheta \}^2$ where $\bPsi(\bX) \equiv \{ \bPsi_{[j]}(\bX) \}_{j=1}^d \in \R^d$ denotes some basis functions (possibly high dimensional) of $\bX$. While more general loss functions can also be handled similarly, the corresponding results and conditions can be technically more involved. We choose to skip such analyses here given the scope and content of the current work. Note that $L(\cdot)$  satisfies \eqref{lossfn:splform1} with $\bh(\bx) = -2 \bPsi(\bx)$ and $g(\bx,\btheta) = \bPsi(\bx)'\btheta$. The case $\bPsi(\bX) = (1, \bX')'$ corresponds to standard linear regression.  For convenience, let us also define:
\begin{equation*}
 \bSigma \; := \; \E \{ \bPsi(\bX) \bPsi(\bX)' \}, \;\; \bSigmahat \; := \; \frac{1}{n} \sum_{i=1}^n \bPsi(\bX_i) \bPsi(\bX_i)' \;\; \mbox{and} \;\; \bOmega \; := \; \bSigma^{-1}, %\\
\end{equation*}
where we assume that $\E\{ \|\bPsi(\bX) \|_2^2\} < \infty$ and $\bSigma$ is positive definite, so that $\bSigma$ and the precision matrix
$\bOmega$ are both well-defined and well-conditioned.
With $L(\cdot)$ as above, note that we have: $\E\{ \bnabla^2 L(Y,\bX, \btheta)\} = 2 \bSigma$ and its inverse is $\frac{1}{2} \bOmega$ for any $\btheta$,
where for any function $f(\btheta)$, $\bnabla^2 f(\btheta)$ denotes its Hessian matrix w.r.t. $\btheta$. Further, we also have: $\bnabla^2 \LnDR(\btheta) = \bnabla^2 \LntilDR(\btheta) = 2\bSigmahat$.

\subsection{The Desparsified DDR Estimator}\label{desparsified:est}
Let $\bOmegahat$ be \emph{any} reasonable estimator of the precision matrix $\bOmega$ based on
the observed data $\Dsc_n$. Then, given the original $L_1$-regularized DDR estimator $\bthetahatDR$ in \eqref{penDRest}, or as %equivalently
in \eqref{eq:simplealgo}, we define the corresponding \emph{desparsified DDR estimator} $\bthetatilDR$  as follows.
\begin{align}
 \bthetatilDR & \; := \; \bthetahatDR - \frac{1}{2} \bOmegahat \bnabla \LnDR( \bthetahatDR) \; \equiv \; \bthetahatDR - \frac{1}{2} \bOmegahat \bnabla \LntilDR( \bthetahatDR) \label{desparsify:est:def}  \\
 & \;\; = \;  \bthetahatDR + \bOmegahat \frac{1}{n} \sum_{i=1}^n \{ \Ytil_i - \bPsi(\bX_i)'\bthetahatDR\}  \bPsi(\bX_i), \;\; \mbox{where}  \nonumber %\label{desparsify:est:altdef}
\end{align}
$\Ytil_i \equiv \mtil(\bX_i) + \{T_i /\pihat(\bX_i)\}  \{Y_i - \mtil(\bX_i)\}$ are the \emph{pseudo outcomes} as in \eqref{eq:pseudo:outcome:def}. %Section \ref{sec:simple:algorithm}.

The desparsification step in \eqref{desparsify:est:def} is similar in spirit to that of \citet{VdG_2014}, while accounting for a more general and complex setting here involving missing responses. It serves as the appropriate \emph{generalization} of their approach when adapted to this setting. As seen from the representation in the final step, the debiasing step \emph{still} uses the full data \emph{but} with the pseudo outcomes $\Ytil_i$ instead of the true $Y_i$. For a fully observed data with $\Ytil_i = Y_i$, this indeed reduces to the usual Debiased Lasso estimator of \citet{Javanmard_2014}.
In addition, we also allow for misspecified models, non-Gaussian settings and covariate transformations, unlike most of the relevant existing %high dimensional inference
literature (with the exception of \citet{Buhlmann_2015}).

It should be noted that the principle of debiasing has also been used extensively %is not entirely new or restricted to high dimensional inference. They have been used extensively
in the classical semi-parametric inference literature, where it is often called \emph{one-step update} \citep{VDV_Book_2000} and is used to obtain efficient estimators starting from an initial (inefficient) estimator. In our setting, the `update' is used more as a bias correction to obtain an estimator with an ALE starting from a shrinkage estimator that has no such desirable properties. In classical settings, such ALEs are also known as \emph{Bahadur representations.}

\paragraph*{Choice of $\bOmegahat$} Since the debiasing still involves the full data (with the pseudo outcomes), the estimator $\bOmegahat$ is exactly the \emph{same} as that used for a standard fully observed data. This is again largely due to the structure of the DDR loss (and the debiasing term therein). Consequently, one pays no price for the missing outcomes as far as the estimation of $\bOmega$ and the associated conditions are concerned, and can borrow any standard precision matrix estimator from the literature. Several such examples exist depending on the setting (low or high dimensional). In the former case, one can simply choose $\bSigmahat^{-1}$, while for the latter, under sparsity assumptions on $\bOmega$, one can use the Nodewise Lasso estimator of \citet{VdG_2014}, among other choices. For our results on $\bthetatilDR$, we only assume some high-level conditions on $\{\bOmegahat, \bOmega \}$ and one is free to use \emph{any} estimator of $\bOmega$ as long as those conditions are satisfied. We next discuss these conditions (and some notations) followed by our results. %We start with some notations.

\par\smallskip
For any matrix $\bM_{d \times d}$, let $\bM_{[i \cdot]} \in \R^d$ denote its $i^{th}$ row and $\bM_{[ij]}$ denote its $(i, j)^{th}$ entry. %for any $ 1 \leq i,j \leq d$.
Let $\| \bM \|_1 := \max_{1 \leq i \leq d} \sum_{j=1}^d |\bM_{[ij]}|$, $\| \bM \|_2 = \lambda_{\max}^{1/2}(\bM'\bM)$ and $\| \bM \|_{\max} := \max_{1 \leq i,j \leq d} |\bM_{[ij]} |$ denote the maximum rowwise $L_1$ norm, the spectral norm and the elementwise maximum norm of $\bM$ respectively, where $\lambda_{\max(\cdot)}$ denotes the maximum eigenvalue. Finally, %let us also
recall the notations $\bTzeron, \bTpin, \bTmn$ and $\bRpimn$ defined in the decomposition \eqref{decomp:eqn} of $\bT_n \equiv \bnabla \LnDR (\btheta_0)$ and for convenience of further discussion, define:
%and also observe that under the specific choice of $L(\cdot)$ here, we have:
%\begin{equation}
%- \; \frac{1}{2} \bTzero(\bZ) \; = \;  \{ m(\bX) - \bPsi(\bX)'\btheta_0 \} \bPsi(\bX) + \frac{T}{\pi(\bX)} \{ Y - m(\bX)\} \bPsi(\bX). \label{tzero:sqloss:def}
%\end{equation}
\begin{eqnarray}
&& \hspace{0.07in} \bR_{n,1}  :=  - \frac{1}{2}(\bOmegahat - \bOmega) \bnabla \LnDR(\btheta_0), \;\; \bR_{n,2}  :=  - \frac{\bOmega}{2} (\bTpin - \bTmn - \bRpimn) \nonumber \\
&& \;\; \bR_{n,3}  :=  (I_d - \bOmegahat \bSigmahat) (\bthetahatDR - \btheta_0) \;\; \mbox{and let} \;\; \bDelta_n  :=  (\bR_{n,1} + \bR_{n,2} + \bR_{n,3}). %\\
\label{eq:HDinf:Deltandefn}
\end{eqnarray}

\begin{assumption}[High-level conditions on $\bOmega$ and $\bOmegahat\}$] \label{HDinf:assmpn}
\hspace{-0.1in} \emph{We assume that:}
%\par\smallskip
\emph{
(a) $\| \bOmegahat - \bOmega \|_1 = O_{\P}(r_{n})$ and $\|I_d - \bOmegahat\bSigmahat \|_{\max} = O_P(\omega_n)$ for some sequences %non-negative sequences
$\{r_n, \omega_n\} \equiv \{r_{n,\bOmega}, \omega_{n, \bOmega}\} \geq 0$ with $r_n  \sqrt{\log d} = o_{\P}(1)$ and $\omega_n (s \sqrt{\log d}) = o_{\P}(1)$, where $s = \| \btheta_0\|_0$ and $I_d$ denotes the $d \times d$ identity matrix.
}
\par\smallskip
\emph{
\hspace{-0.25in}
(b) $\bUpsilon(\bX) := \bOmega \bPsi(\bX)$ is sub-Gaussian (as per Definition \ref{orlicz:def} with $\alpha = 2$) with $\psitwonorm{\bUpsilon(\bX)} \leq \sigmabUps < \infty$, for some constant $\sigmabUps \geq 0$. Further, we assume that $v_n^* = o_{\P}(1)$, where $v_n^* := (\vnpi + \vnbarm) \sqrt{(\log d) \log (nd)} + \nhalf \vnpi \vnbarm (\log n)$  %$v_n^* := \sigmabUps\{ (\vnpi + \vnbarm) \sqrt{(\log d) \log (nd)} + \nhalf \vnpi \vnbarm (\log n) \} $
and $\{\vnpi, \vnbarm \}$ are the rates of $\{\pihat(\cdot), \mhat(\cdot)\}$ defined in Assumptions \ref{tpicont:assmpn}-\ref{tmcont:assmpn}.
}
\end{assumption}
Assumption \ref{HDinf:assmpn} (a) imposes some general rate conditions on $\bOmegahat$. For most common choices of $\bOmegahat$, including those discussed earlier, these lead to fairly standard conditions.  %that should be satisfied by most common choices of $\bOmegahat$ (under suitable sparsity assumptions on $\bOmega$, if required), including those discussed earlier. For these choices, the corresponding rate conditions are fairly standard in the literature.
Under a low dimensional setting with  $\bOmegahat = \bSigmahat^{-1}$, $\omega_n = 0$ trivially and $r_n = d/\sqrt{n}$ under suitable assumptions; see \citet{Vershynin_2018} for relevant results. Under high dimensional settings, with $\bOmega$ assumed to be sparse and $\bOmegahat$ chosen to be the Nodewise Lasso estimator, $\omega_n = \sqrt{(\log d)/n}$ and $r_n = s_{\bOmega}\sqrt{(\log d)/n}$; see \citet{VdG_2014} for relevant results. In this case, the conditions read as: $s_{\bOmega}(\log d) = o(\sqrt{n})$ and $s(\log d) = o(\sqrt{n})$. These are all familiar (often unavoidable) conditions in the high dimensional inference literature \citep{Cai_2017, Javanmard_2018}. %even for a full data setting.

The sub-Gaussianity condition on $\bUpsilon (\bX)$ in Assumption \ref{HDinf:assmpn} (b) is needed to control the term $\bR_{n,2}$ in \eqref{eq:HDinf:Deltandefn}. Conditions of a similar flavor have also been adopted implicitly or explicitly in \citet{VdG_2014} and \citet{Javanmard_2014}. The condition holds with $\sigmabUps$ to be a constant if either $\| \bOmega \|_2 = O(1)$ and $\bPsi(\bX)$ is (vector) sub-Gaussian in the sense of \citet{Vershynin_2018} with a $O(1)$ norm, or if $\|\bOmega \|_1 = O(1)$ and $\bPsi(\bX)$ is sub-Gaussian in the (weaker) sense of Definition \ref{orlicz:def} with a $O(1)$ norm. %In general, we allow $\sigmabUps$ to depend on $d$ and include it in the rate $v_n^*$.
Finally, the condition on $v_n^*$ is the same (upto a $\sqrt{\log d}$ factor) as those needed for Theorems \ref{TPI:THM}-\ref{RMPI:THM}.

\begin{theorem}[ALE and entrywise asymptotic normality of $\bthetatilDR$]\label{HDINF:THM}
Under Assumptions \ref{base:assmpns}, \ref{strngconv_assmpn}, \ref{subgaussian:assmpn}-\ref{tmcont:assmpn} %, \ref{tpicont:assmpn}, \ref{tmcont:assmpn}
and \ref{HDinf:assmpn}, and with $\bDelta_n$ as defined in \eqref{eq:HDinf:Deltandefn}, $L(\cdot)$ assumed to be the squared loss and $\bthetahatDR$ constructed using a choice of $\lambda_n \asymp \sqrt{(\log d)/n}$, the desparsified DDR estimator $\bthetatilDR$ satisfies the ALE:
%Using \eqref{desparsify:est:def}, it then follows easily that
\begin{equation}
\;(\bthetatilDR - \btheta_0) \; = \; \frac{1}{n} \sum_{i=1}^n \bOmega\{\bpsi_0(\bZ_i)\} + \bDelta_n, \;\; \mbox{where} \;\; \E\{\bpsi_0(\bZ)\} =  \bzero \;\; \mbox{with}  \label{desparsifiedest:ALE}
\end{equation}
\vspace{-0.25in}
\begin{align*}
& \bpsi_0(\bZ) \; = \;  \{ m(\bX) - \bPsi(\bX)'\btheta_0 \} \bPsi(\bX) + \frac{T}{\pi(\bX)} \{ Y - m(\bX)\} \bPsi(\bX), \;\; \mbox{and} \\
& \|\bDelta_n\|_{\infty} \; = \; O_{\P}\left( r_n \sqrt{\frac{\log d}{n}} + v_n^* \nnhalf + \omega_n s  \sqrt{\frac{\log d}{n}}  \right) \; = \; o_{\P}\left(\nnhalf\right).
\end{align*}
Consequently, letting $\bGamma_0(\bZ) := \bOmega \bpsi_0(\bZ)$, $\sigma_{0, j}^2 := \E\{\bGamma_{0 [j]}^2(\bZ) \}$ and assuming that $\sigma_{0,j} > c_0$ $\forall \; j$, for some constant $c_0 > 0$, we have: for each $1 \leq j \leq d$,
\begin{align*}
& \sqrt{n} \sigma_{0, j}^{-1} ( \bthetatil_{\DDR [j]} - \btheta_{0 [j]} )  \convd  \Nsc(0,1) \;\; \mbox{and} \;\;  \sqrt{n} \sigmahat_{0, j}^{-1} ( \bthetatil_{\DDR [j]} - \btheta_{0 [j]})  \convd  \Nsc(0,1), \\
& \mbox{where} \;\; \sigmahat_{0,j}^2 := \frac{1}{n}\sum_{i=1}^n \bGammahat^2_{0 [j]}(\bZ_i) \;\; \mbox{satisfying} \;\; \max_{1 \leq j \leq d} | \sigmahat_{0, j}^2 - \sigma_{0, j}^2| = o_{\P}(1).
\end{align*}
Here $\bGammahat_{0 [j]}(\bZ_i) := \bOmegahat_{[j\cdot]}' \bpsihat_0(\bZ_i)$, where $\bpsihat_0(\bZ_i)$ denotes the estimated version of $\bpsi_0(\bZ)$ in \eqref{desparsifiedest:ALE} with $\{\pi(\bX_i),  m(\bX_i), \btheta_0\}$ plugged in as $\{ \pihat(\bX_i), \mtil(\bX_i), \bthetahatDR\}$.
\end{theorem}
%\begin{remark}
Theorem \ref{HDINF:THM} therefore provides all the necessary inferential tools for $\bthetatilDR$. The ALE \eqref{desparsifiedest:ALE} is also \emph{optimal,} in a certain sense, since the function $\bGamma_0(\bZ) \equiv \bOmega \bpsi_0(\bZ)$ defining the i.i.d. summand (also known as the influence function) in the ALE is known to be the \emph{efficient influence function} for estimating $\btheta_0$ in a classical setting ($d$ fixed) under a fully non-parametric (i.e. unrestricted, upto Assumption \ref{base:assmpns}) family of $\P$, and its variance %$\Var\{\bGamma_0(\bZ)\}$
equals the semi-parametric optimal variance \citep{Robins_1994, Robins_1995, Graham_2011}. The same conclusions continue to hold in high dimensional settings for low-dimensional components (e.g. each coordinate) of $\btheta_0$. Thus, $\bthetatilDR$ achieves the (coordinatewise) \emph{semi-parametric efficiency bound} and is optimal among all achievable estimators of $\btheta_0$ admitting ALEs under a non-parametric family of $\P$.
Further, the asymptotic normality results %along with the consistency results for the standard error estimator $\sigmahat_{0,j}$
also allow one to construct asymptotically valid $(1-\alpha)$ level confidence intervals (CIs): $CI_j := \bthetatil_{\DDR [j]} \pm z_{\alpha/2} \sigmahat_{0,j}$, for each coordinate $\btheta_{0 [j]}$ of $\btheta_0$, where $z_{\alpha/2}$ denotes the $(1-\alpha/2)^{th}$ quantile of the $\Nsc(0,1)$ distribution with $\alpha \in (0,1)$.
%\end{remark}
%\tcr{Rest of the details for this section to be added soon. Please check the slides for an overview of the results.}

\section{Estimation of the Nuisance Functions}\label{NUISANCE}
In Sections \ref{nuisance:pi}-\ref{nuisance:m}, we discuss  various choices for the nuisance function estimators $\{\pihat(\cdot), \mhat(\cdot)\}$ required for implementing our proposed methods. Our entire approach so far does \emph{not} require any specific knowledge of the construction or properties of these estimators as long as they satisfy the high-level conditions in Assumption \ref{tpicont:assmpn}-\ref{tmcont:assmpn}. Hence, one is free to use \emph{any} choice of these estimators based on high dimensional parametric or semi-parametric models, or even non-parametric machine learning based estimators, as has been advocated in many recent works for other related problems in similar settings \citep{Farrell_2015, Chern_DDML_2018, Farrell_2018}. However, a fully non-parametric and/or machine learning based approach may not be feasible or efficient in `truly' high dimensional settings where $p$ diverges with $n$. In this section, we discuss a few novel, principled, and yet, flexible families of choices for $\pihat(\cdot)$ and $\mhat(\cdot)$, including common parametric models, as well as series estimators and single index models. In Appendices \ref{parametricfamilies:convergencerates}-\ref{sec:KS}, we establish general results for these estimators under high dimensional settings that verify our basic assumptions and may also be of independent interest.

\subsection{Propensity Score Estimation: A Few Choices and Their Properties}\label{nuisance:pi}
In some cases, $\pi(\cdot)$ may be known whereby $\pihat(\cdot) \equiv \pi(\cdot)$ trivially. When $\pi(\cdot)$ is unknown, we consider the following (class of) choices for estimating $\pi(\cdot)$. %as follows.
%\begin{enumerate}[1.]%[(i)]
%\item
%\par\smallskip
\paragraph*{`Extended' parametric families (or high dimensional series estimators)} We assume that $\pi(\cdot)$ belongs to the family: $\pi(\bx) \equiv \E(T \smallgiven \bX = \bx) = g\{ \balpha'\bPsi(\bx)\}$, where $g(\cdot)$ $\in$ $[0,1]$ is a \emph{known} `link' function,
    $\bPsi(\bx) := \{ \psi_k(\bx) \}_{k=1}^K$ is any set of $K$ (known) basis functions, possibly high dimensional, with $K$ allowed to depend on $n$ (including $K \gg n$), and $\balpha \in \R^K$ is an \emph{unknown} parameter vector that is further assumed to be sparse (if required).

\par\smallskip
     \emph{Estimator.} $\pi(\bx)$ is then estimated as: $\pihat(\bx) = g\{\balphahat'\bPsi(\bx)\}$, where $\balphahat$ denotes \emph{some} given estimator of $\balpha$ obtained via \emph{any} suitable estimation procedure based on the observed data for $(T,\bX)$ that \emph{only} satisfies a basic high-level requirement that $\| \balphahat - \balpha\|_1 \leq a_n$ w.h.p. for some sequence $a_n = o(1)$.

\par\smallskip
\emph{Examples.} The models above include, as a special case, any logistic regression model for $T | \bX$ given by: $\pi(\bx) = g\{\balpha'\bPsi(\bx)\}$, where $g(u) \equiv g_{\mbox{\tiny{expit}}}(a) := \exp(a)/\{1+\exp(a)\}$. The estimator $\balphahat$ in this case maybe obtained using a simple $L_1$-penalized logistic regression of $T$ vs. $\bPsi(\bX)$ based on the observed data $\{T_i, \bPsi(\bX_i) \}_{i=1}^n$. Using standard results from the theory of high dimensional regression %theory
\citep{BuhlmannVdG_Book_2011, Negahban_2012, Wainwright_Book_2019}, %{HasTibWainwright_Book_2015},
it can be shown that under suitable assumptions (e.g. RSC and exponential tail conditions), $\| \balphahat - \balpha\|_1 \lesssim a_n \equiv a_n(s_{\balpha},K) := s_{\balpha} \sqrt{(\log K) /n}$ w.h.p., where $s_{\balpha} := \| \balpha \|_0$ denotes the sparsity of $\balpha$.

As for the basis functions $\bPsi(\bx)$, some reasonable choices include the polynomial bases given by: $\bPsi(\bx) := \{1, \bx_j^k: 1 \leq j \leq p, 1 \leq k \leq d_0\}$ for any degree $d_0 \geq 1$. The special case $d_0 = 1$ corresponds to the linear bases which leads to all standard parametric models that are commonly used in practice.

\par\smallskip
\emph{The case when $\pi(\cdot)$ is constant.} Note that the extended parametric framework above also includes the special case where $\pi(\cdot)$ is unknown but constant (i.e. the case of MCAR or complete randomization), in which case $g(\balpha'\bX)$ simply equals the constant $\pi$, and $\balpha$ is just an unknown parameter in $\R$ that can be estimated at the rate of $O(n^{-1/2})$ via the usual sample mean of $T$.

\subsection{Estimation of the Conditional Mean: Choices and Their Properties}\label{nuisance:m}
We consider the following two (class of) choices for estimating $m(\cdot)$.

%\begin{enumerate}[1.]
%\item
\paragraph*{1. `Extended' parametric families (high dimensional series estimators)} We assume that $m(\cdot)$ belongs to the family: $g\{\bgamma'\bPsi(\bX)\}$ where $g(\cdot)$ is a (known) `link' function (e.g. `canonical' link functions), $\bPsi(\bX) := \{ \psi_k(\bX) \}_{k=1}^K$ is any set of $K$ (known) basis functions, with $K$ possibly high dimensional and is allowed to depend on $n$ (including $K \gg n$), and $\bgamma \in \R^K$ is an unknown parameter vector that is further assumed to be sparse (if required).

\par\smallskip
\emph{Estimator.} We estimate $m(\bx) \equiv \E(Y \smallgiven \bX) \equiv \E(Y \smallgiven \bX, T=1) = g\{\bgamma'\bPsi(\bX)\}$ as: $\mhat(\bx) = g\{\bgammahat'\bPsi(\bX)\}$, where $\bgammahat$ denotes \emph{some} given estimator of $\bgamma$ obtained via \emph{any} suitable estimation procedure based on the `complete case' data $\Dsc_n^{(c)} := \{(Y_i,\bX_i) \given T_i = 1\}_{i=1}^n$ that \emph{only} satisfies a basic high-level requirement that $\| \bgammahat - \bgamma\|_1 \leq a_n$ w.h.p. for some sequence $a_n = o(1)$.

\par\smallskip
\emph{Examples.} These models include, as special cases, all standard parametric regression models with `canonical' link functions, through suitable choices of $g(\cdot)$ depending on the nature of $Y$ (continuous, binary or discrete). Specifically, $g(u) \equiv g_{\mbox{\tiny{id}}} = u$ (identity link), %corresponds to linear regression,
$g(u) \equiv g_{\mbox{\tiny{expit}}} = \exp(u)/\{ 1 + \exp(u)\}$ (expit/logit link) and %corresponds to logistic regression and
$g(u) \equiv g_{\mbox{\tiny{exp}}} = \exp(u)$ (exponential/log link) correspond to the linear, logistic and Poisson regression models respectively.

As for the basis functions $\bPsi(\bx)$, some reasonable choices include the polynomial bases given by: $\bPsi(\bx) := \{1, \bx_j^k: 1 \leq j \leq p, 1 \leq k \leq d_0\}$ for any degree $d_0 \geq 1$. The special case $d_0 = 1$ corresponds to the linear bases which leads to all standard parametric models, while $d_0 = 3$ leads to cubic splines. %that are commonly used in practice.

\par\smallskip
\emph{Examples of $\bgammahat$.} For all the examples above, with $g(\cdot)$ being any `canonical' link function, the estimator $\bgammahat$ of $\bgamma$ may be simply obtained through a corresponding $L_1$ penalized `canonical' link based regression (e.g. linear, logistic or Poisson regression) of $Y$ vs. $\bX$ in the `complete case' data $\Dsc_n^{(c)}$ under Assumption \ref{base:assmpns} (a). Using standard results from high dimensional regression \citep{BuhlmannVdG_Book_2011, Negahban_2012, Wainwright_Book_2019}, %{HasTibWainwright_Book_2015}
it can be shown that under suitable assumptions (e.g. RSC and exponential tail conditions) and Assumption \ref{base:assmpns}, $\| \bgammahat - \bgamma\|_1 \lesssim a_n \equiv a_n(s_{\bgamma},K) := s_{\bgamma} \sqrt{(\log K) /n}$ w.h.p., where $s_{\bgamma} := \| \bgamma \|_0$ denotes the sparsity of $\bgamma$.
%\par\smallskip
%\item

\paragraph*{2. Semi-parametric single index models} We assume that $m(\cdot)$ satisfies the SIM: $m(\bX) \equiv \E(Y \smallgiven \bX) \equiv \E(Y \smallgiven \bX, T =1 ) = g(\bgamma'\bX)$, where $g(\cdot) \in \R$ is some \emph{unknown} `link' function and $\bgamma \in \R^p$ is an unknown parameter (identifiable only upto scalar multiples) that is further assumed to be sparse (if required).
\par\smallskip
\emph{Estimator.} Given \emph{any} reasonable estimator $\bgammahat$ of the $\bgamma$ `direction' obtained from %based on some suitable procedure on the observed dat $\Dsc_n$
$\Dsc_n$, we estimate $m(\bX)  \equiv \E(Y \given \bgamma'\bX) \equiv \E(Y \given \bgamma'\bX, T = 1) = g(\bgamma'\bX)$ via a one-dimensional kernel smoothing (KS) over the estimated scores $\{\bgammahat'\bX_i\}_{i=1}^n$, under appropriate smoothness and regularity assumptions, as follows.
\begin{equation*}
\mhat(\bx) \; \equiv \; \mhat(\bgammahat'\bx) \; \equiv \; \mhat(\bgammahat,\bx) \; := \; \frac{\frac{1}{nh} \sum_{i=1}^n T_i Y_i K\left( \frac{\bgammahat'\bX_i - \bgammahat'\bx}{h} \right)}{\frac{1}{nh} \sum_{i=1}^n T_i K\left( \frac{\bgammahat'\bX_i - \bgammahat'\bx}{h} \right)} \quad \forall \; \bx \in \Xsc, %\; \equiv \; \frac{\lhat(\bgammahat, \bx)}{\fhat(\bgammahat,\bx)},
\end{equation*}
 where $K(\cdot): \R \rightarrow \R$ is some suitable `kernel' function and $h \equiv h_n > 0$ denotes a bandwidth sequence with $h_n = o(1)$. Here, we \emph{only} assume that $\bgammahat$ is \emph{some} reasonable estimator of the $\bgamma$ `direction' satisfying a basic high-level condition: $\|\bgammahat - \bgamma_0\|_1 \leq a_n $ w.h.p. for some $\bgamma_0 \propto \bgamma$ and $a_n = o(1)$.
\par\smallskip
\emph{Estimation of $\bgammahat$.} Under Assumption \ref{base:assmpns} (a) and the SIM framework we have adopted here, $\E(Y \smallgiven \bX) \equiv \E(Y \smallgiven \bX, T = 1) = g(\bgamma'\bX)$. Hence,  in general, one may use \emph{any} standard method available in the literature for signal recovery in SIMs \citep{Horowitz_2009, Alquier_2013, Radchenko_2015}, %{Yi_2015, Yang_2017}
and apply it to the `complete case' data $\Dsc_n^{(c)}$ to obtain a reasonable estimator $\bgammahat$ of $\bgamma$. Under some additional design restrictions and model assumptions, however, one may also estimate $\bgamma$ by even simpler approaches, as follows.
%\begin{enumerate}[(a)]
%\item

\par\smallskip
\noindent (a) Suppose $Y$ satisfies the (slightly) stronger SIM formulation: $(Y \smallgiven \bX) \equiv  (Y \smallgiven \bX, T =1) = f(\bgamma'\bX; \epsilon)$ for some unknown function $f: \R^2 \rightarrow \Ysc$ and some noise $\epsilon \ind (T,\bX)$, and assume further that the distribution of $(\bX \smallgiven T = 1)$ is elliptically symmetric. Then, owing to the results of \citet{Li_Duan_1989}, one can \emph{still} estimate $\bgamma$ with a rate guarantee of $a_n  = s_{\bgamma} \sqrt{(\log p)/n}$ using a simple $L_1$ penalized `canonical' link based regression (e.g. linear, logistic or Poisson regression) of $Y$ vs. $\bX$ in the `complete case' data $\Dsc_n^{(c)}$, as discussed in the previous example. %The same conclusion continues to hold if $T \ind \bX$ and the distribution of $\bX$ is elliptically symmetric.
    Similar approaches %based on similar ideas
    have been used extensively in recent years for sparse signal recovery in high dimensional SIMs with fully observed data and elliptically symmetric designs \citep{Plan_2013, Plan_2016, Goldstein_2016, Genzel_2017, Wei_2018} %Plan_2017,

%\item
\par\smallskip
\noindent (b) Suppose $Y$ satisfies the same SIM as in part (a) above, and assume now that the distribution of $\bX$ is elliptically symmetric. Then, using the results of \citet{Li_Duan_1989}, along with our discussions in Section \ref{id_and_alt_rep} regarding IPW representations, it follows that one can also estimate $\bgamma$ using a \emph{weighted}  $L_1$-penalized regression based on any `canonical' link (e.g. linear, logistic or Poisson regression) of $Y$ vs. $\bX$ in the `complete case' data $\Dsc_n^{(c)}$. The weights are given by: $\pi^{-1}(\bX)$, if $\pi(\cdot)$ is known, and by $\pihat^{-1}(\bX)$, if $\pi(\cdot)$ is unknown and estimated via $\pihat(\cdot)$ (assumed to be correctly specified) through any of the choices discussed in Section \ref{nuisance:m}. Using the results of \citet{Negahban_2012}, along with the techniques used in our proofs of Lemma \ref{DEV:BOUND} and Theorems \ref{TZERO:THM} and \ref{RMPI:THM}, it can be shown that the resulting IPW estimator $\bgammahat$ satisfies an $L_1$ norm bound: $\|\bgammahat - \bgamma\|_1 \lesssim a_n \equiv s_{\bgamma}\sqrt{(\log p)/n}$ w.h.p. in the case when $\pi(\cdot)$ is known, and $\|\bgammahat - \bgamma\|_1 \lesssim a_n \equiv  s_{\bgamma}\max\{\sqrt{(\log p)/n}, \vnpi \sqrt{\log n} \} $ when $\pi(\cdot)$ is unknown, where $\vnpi = o(1)$ denotes the (pointwise) convergence rate of $\pihat(\cdot)$ as given in Assumption \ref{tpicont:assmpn}. Given the main goals of this paper, we skip the technical details and proofs of these claims for the sake of brevity.
%\end{enumerate}
%\end{enumerate}

\section{Simulation Studies}\label{sec:sim}
% \tcr{Results are available and will be added soon.} %Please check the slides for an overview of all the results.}

%\input{Simulation-Main.tex}

%structure of the simulation sections

%simulation setup (models and choices of parameters)
%estimators implemented (estimators considered, comparison criteria and comments)
%simulation results  (tables and figures)
%Comments on the results
%Double robustness : cc estimator and large sample setting

%additional results in the sup.

%\section{Simulation Studies}\label{Simulation}
We conducted extensive simulations to examine the performances of our proposed estimation and inference procedures under various
%In this section, we perform a group of simulations to examine the performances of our method under different
data generating processes (DGPs) and parameter settings. We set $n=1000$, and $p=50$ or $500$ reflecting moderate and high dimensional settings, respectively. (In Appendix \ref{Sim. double-robustness}, we also conduct a large sample analysis with $n=50000$ to investigate the DR properties of our estimator(s), as well as the performance of the CC estimator). %study the double-robustness of our proposed estimator and the performance of the complete case estimator (see  for details).
%For DGPs, the observed data $\Dsc_n := $ $\{\bZ_i \equiv (T_i, T_i Y_i, \bX_i): i = 1, \hdots, n \}$ is given by $\bX \sim N(\mathbf{0}, \bSigma_p)$ (the choices of $\bSigma_p$ will be discussed later) and three models for $Y|\bX$ and $T|\bX$ : a logistic model for $T|\bX$ and a linear model for $Y|\bX$ (denoted as ``linear-linear'' DGP), a logistic model with both linear and quadratic terms for $T|\bX$ and a linear model with both linear and quadratic terms for $Y|\bX$ (denoted as ``quad-quad'' DGP) and a single index model (SIM) for both $T|\bX$ and $Y|\bX$ (denoted as ``SIM-SIM'' DGP). These models are formalized as following:
%$\bX$ is generated  as:
We used $\bX \sim N(\mathbf{0}, \bSigma_p)$, for $3$ choices of $\bSigma_p$ discussed below, %and used three model combinations or DGPs for $Y|\bX$ and $T|\bX$ as follows:
and $3$ DGPs for $Y|\bX$ and $T|\bX$ as follows:%,  given by: : a logistic model for $T|\bX$ and a linear model for $Y|\bX$ (denoted as ``linear-linear'' DGP), a logistic model with both linear and quadratic terms for $T|\bX$ and a linear model with both linear and quadratic terms for $Y|\bX$ (denoted as ``quad-quad'' DGP) and a single index model (SIM) for both $T|\bX$ and $Y|\bX$ (denoted as ``SIM-SIM'' DGP). These models are formalized as following:
\begin{enumerate}[(a)]
	\item \emph{``Linear-linear'' DGP:} $Y   = \gamma_{0} + \bgamma'\bX  + \varepsilon$ with $\varepsilon | \bX \sim N(0, 1)$, and $\textrm{logit}\{\pi(\bX)\}  \equiv \textrm{logit} \{\E (T | \bX )\}= \alpha_{0} +\balpha'\bX.$ These represent standard linear and logistic regression models for $Y|\bX$ and $T|\bX$ respectively.
\par\smallskip
	
	%\begin{align*} %\label{sim: DGP_linear_linear}
	%& Y   = \bgamma_{0} + \bgamma'\bX  + \varepsilon, \quad \varepsilon | \bX \sim N(0, 1),  \\
	%&  \textrm{logit}\{\pi(\bX)\}  = \textrm{logit} \{\E (T | \bX )\}= \balpha_{0} +\balpha'\bX.
	%\end{align*}

	\item  \emph{``Quad-quad'' DGP:} $Y  = \gamma_{0} + \bgamma'\bX + \sum_{j=1}^{p} \bgamma^{*}_{[j]} \bX_{[j]}^2 + \varepsilon$ with $\varepsilon | \bX \sim N(0,1)$, and $\textrm{logit} \{\pi(\bX)\} \equiv \textrm{logit}\{\E (T | \bX )\}  =  \alpha_{0} +  \balpha'\bX + \sum_{j=1}^{p}\balpha^{*}_{[j]} \bX_{[j]}^2.$ These allow both linear and quadratic effects of $\bX$ in $\E(Y| \bX)$ and $\textrm{logit} \{\pi(\bX)\}$. %represent logistic and linear models with , for $T|\bX$ and $Y|\bX$ respectively.
\par\smallskip

	%\begin{align*} %\label{sim: DGP_quad_quad}
	%& Y  = \bgamma_{0} + \bgamma'\bX + \sum_{j=1}^{p} \bgamma^{*}_{j} \bX_{j}^2 + \varepsilon, \quad \varepsilon | \bX \sim N(0,1),  \\
	%& \textrm{logit} \{\pi(\bX)\} = \textrm{logit}\{\E (T | \bX )\}  =  \balpha_{0} +  \balpha'\bX + \sum_{j=1}^{p}\balpha^{*}_{j} \bX_{j}^2.
	%\end{align*}
	
\item  \emph{``SIM-SIM'' DGP:} $Y = \gamma_{0} + \bgamma'\bX  + c_Y(\bgamma'\bX)^2+\varepsilon$ with $\varepsilon | \bX \sim N(0, 1)$, and $\textrm{logit}\{\pi(\bX) \} \equiv   \textrm{logit} \{\E (T  | \bX )\} =  \alpha_{0} + \balpha'\bX  + c_T(\balpha'\bX )^2$. These represent standard single index models (SIMs) for both $\E(Y|\bX)$ and $\textrm{logit} \{\pi(\bX)\}$.

	%\begin{align*} %\label{sim: DGP_sim_sim}
	%& Y = \bgamma_{0} + \bgamma'\bX  + c_Y(\bgamma'\bX)^2+\varepsilon, \quad \varepsilon | \bX \sim N(0, 1),\\
	%& \textrm{logit}\{\pi(\bX) \} =   \textrm{logit} \{\E (T  | \bX )\} =  \balpha_{0} + \balpha'\bX  + c_T(\balpha'\bX )^2.
	%\end{align*}
\end{enumerate}
The \emph{choices of $\bSigma_p$} were: (a) $\bSigma_p = \textbf{I}_p$, the \emph{identity} matrix, or (b) $\bSigma_{ij} = \rho^{|i-j|}$, the first order \emph{autoregressive} (AR1) matrix, or (c) $\bSigma_p = \rho\bone_p\bone_p'+ (1 - \rho)\textbf{I}_p$, the \emph{compound symmetry} (CS) matrix, where we set $\rho = 0.2$. These choices %of the covariance structure of $\bX$
exhibit a variety of correlation and sparsity structures in $\bSigma_p$, %for the covariance matrix, %different correlations among the coordinates of $\bX$ and different sparsities of the covariance matrices,
ranging from independent and sparse (i.e. $\textbf{I}_p$) to correlated and not sparse (i.e. CS matrix).

%using minimizing mean squared errors (MSE) as criterion.
%(100 times for $n=50000$).

We also manually truncated $\pi(\cdot)$ to lie in $[0.1, 0.9]$ to avoid extreme values. By choice of our model parameters, the proportion of data being truncated was roughly around $1\%$ and the proportion of observations with missing $Y$ was around $40\%$ for all model settings. The tuning parameter $\lambda_n$ in \eqref{eq:simplealgo} for obtaining $\bthetahatDR$  %in the penalized regression
 was selected using 10-fold cross validation of the loss $L(\cdot)$. All simulations were replicated 500 times. Further details on our parameter choices and other implementation details are provided in Appendix %Section
\ref{Simulation: Tech. detail}.

\subsection{Target Parameter and Choices of the Working Nuisance Models} \label{Target para.}
We considered the linear regression problem, where the target parameter $\btheta_0$ is:% is: % $\btheta_0$ is the best linear estimator:
\begin{equation*}
\btheta_0 := \underset{\btheta \in \R^d}{\arg \min} \; \E (Y - \bXv'\btheta)^2  =  \bSigma^{-1}\E (\overrightarrow{\bX} Y), \;\; \mbox{with} \;\; d = p+1, \;\; \bSigma := \E(\bXv\bXv'),
\end{equation*}
%with $d = p+1$ and the definition of
and $\overrightarrow{\bv}$ being as in Section \ref{estimation}. Note that $\btheta_0$ is the (model-free) target parameter %that is always the target
for linear regression \emph{regardless} of whether $\E(Y | \bX)$ is truly linear or not. For the ``linear-linear'' DGP, $\btheta_0$ matches the parameters $(\gamma_0,\bgamma)$ therein. %introduced previously.
For the other non-linear DGPs, $\btheta_0$ is, in general, \emph{different} from the parameters introduced therein. %we introduced in the working nuisance models.
By choices of our model parameters, this $\btheta_0$ is \emph{still} guaranteed to be sparse. For all the DGPs (linear or non-linear), we computed (and fixed) $\btheta_0$ via Monte-Carlo based on a large dataset with size $200000$.

To implement the estimator $\bthetahatDR$ of $\btheta_0$, we considered the following (combination of) choices of (working) models for the nuisance estimators $\pihat(\cdot)$ and $\mhat(\cdot)$. %two choices of the working nuisance model for $\pihat(\cdot)$ and three choices of $\mhat(\cdot)$ are considered.
Specifically, we considered \emph{two choices for the PS estimator $\pihat(\cdot)$}:
\begin{enumerate}[1.]
	\item  \emph{``$\pihat$: linear''} - obtained via a standard $L_1$-penalized logistic regression with linear covariates (i.e. the Logistic Lasso) fitted to the data $\{T_i, \bX_i\}_{i=1}^{n}$.
	\item  \emph{``$\pihat$: quad''} - obtained via an $L_1$-penalized logistic regression with both linear and quadratic covariates fitted to the data $\{T_i, \bX_i\}_{i=1}^{n}$.
\end{enumerate}
For \emph{each} choice of $\pihat(\cdot)$, we used \emph{three choices of the OR estimator $\mhat(\cdot)$}:
\begin{enumerate}[1.]
\item \emph{``$\mhat$: linear''} - obtained via a standard $L_1$-penalized linear regression (i.e. the Lasso) of $Y$ vs. $\bX$ fitted to the `complete case' data $\Dsc_n^{(c)}$.

\vspace{0.05in}
\item \emph{``$\mhat$: quad''} - obtained via an $L_1$-penalized linear regression with both linear and quadratic covariates fitted to the `complete case' data $\Dsc_n^{(c)}$.

\vspace{0.05in}
\item \emph{``$\mhat$: SIM''} - obtained by fitting a SIM to the `complete case' data $\Dsc_n^{(c)}$ with the index parameter estimated via an IPW Lasso (as discussed in Method 2(b) in Section \ref{nuisance:m}, which applies under our assumptions on $\bX$).
\end{enumerate}
Thus, we had 6 different combinations of $\{\pihat(\cdot), \mhat(\cdot)\}$ \emph{each} of which were used to implement $\bthetahatDR$ on data generated from \emph{any} given DGP. It is important to note that the names used to denote these choices have \emph{no relation} to those of the true DPGs for $\pi(\cdot)$ and $m(\cdot)$. For each DGP, there exists a combination of $\{\pihat(\cdot), \mhat(\cdot)\}$ that correctly specifies at least one of $\{\pi(\cdot),m(\cdot)\}$.
For the ``linear-linear'' DGP, all 6 choices of $\{\pihat(\cdot), \mhat(\cdot)\}$ are correct. For the ``quad-quad'' DGP, only ``$\pihat$: quad-$\mhat$: quad'' is correct for both, while there are some combinations that are correct for only one, e.g. ``$\pihat$: linear-$\mhat$: quad'' is correct for $m(\cdot)$ but misspecifies $\pi(\cdot)$.
Finally, note that for the ``SIM-SIM'' DGP, we do \emph{not} include any case where $\pihat(\cdot)$ is correct. This, in some sense, serves as a test of robustness for our estimator. As the results in Section \ref{Sim. Results} will reveal, the performance improves significantly (and is nearly optimal) whenever $\mhat(\cdot)$ is correct, and is quite robust to any misspecification of $\pihat(\cdot)$.

%From the results presented in Section \ref{Sim. Results}, we would see that correctly specifying the conditional mean $m(\cdot)$ would largely reduce the estimation errors comparing to that of $\pi(\cdot)$.

\subsection{Estimators Implemented and Criteria for Comparison} \label{Est. Implemented}
Apart from $\bthetahatDR$, we also considered the following two oracle estimators for comparison:
\begin{enumerate}[(a)]
	\item $ \widehat{\btheta}_{orac} $ \emph{(oracle)}: The version of $\bthetahatDR$ %An estimator obtained
 assuming $\pi(\cdot)$ and $m(\cdot)$ to be known.
 	\item $\widehat{\btheta}_{full}$ \emph{(`super'-oracle)}: The version of $\bthetahatDR$ %An estimator
obtained assuming the full dataset is observed, i.e. no missing $Y$ and no involvement of $\{\pihat(\cdot),\mhat(\cdot)\}$.
\end{enumerate}
The implementation of these estimators is similar (in spirit) to that of $\bthetahatDR$ after making the appropriate adjustments, as detailed above. %the same as our proposed estimator $\bthetahatDR$, using the corresponding working nuisance functions or dataset.
The oracle estimator $\widehat{\btheta}_{orac} $ is considered to examine the impact of estimating the nuisance functions involved in $\bthetahatDR$. Moreover, it is also known (at least under classical settings) to  achieve  the semi-parametric \emph{optimal} performance \citep{Graham_2011} for this problem. The `super'-oracle $\widehat{\btheta}_{full}$ is an \emph{ideal-case} estimator, of course, obtained assuming the full data is observed. We consider it mainly as a benchmark to get a sense of the best performance one can hope to achieve.

We compared the estimators based on the following \emph{performance criteria:}
%We evaluate the simulation performances by:

\begin{enumerate}[1.]
	\item \emph{Estimation:} For all the estimators, we report their average $L_2$ errors for estimating $\btheta_0$,
%Measuring the $L_2$ norms of the differences between the estimators and the true parameter $\btheta_0$.
defined as the average of $\|\widetilde{\btheta} - \btheta_0 \|_2$ over the 500 replications, where $\widetilde{\btheta}$ is any candidate estimator.
%The reported values are the average $L_2$ norms over all replications %with
In addition, we report the standard errors (SEs) of these $L_2$ errors over the 500 replications in the parentheses.

\vspace{0.05in}
\item \emph{Inference:} We calculate the empirical coverage probabilities (CovPs) and the mean lengths of the (coordinatewise) 95\% CIs of $\btheta_0$ obtained via $\bthetatilDR$, over all replications. We report the average and median of these empirical CovPs, denoted ``$A$-CovP'' and ``$M$-CovP'' respectively, and the average of the mean CI lengths, all separated over the truly zero and non-zero coefficients of $\btheta_0$. We also report their respective SEs in the subscripts.
%Using the 95\% (coordinatewise) CIs of $\btheta_0$ obtained via $\bthetatilDR$, we calculate the empirical coverage probabilities (CovPs) and the mean length of the CIs over all replications. We report the average and median of these empirical CovPs, denoted ``A-CovP'' and ``M-CovP'' respectively, and the average of the mean CI lengths, all separated over the truly zero and non-zero coefficients (coeffs.) of $\btheta_0$. %We also report the SEs (for A-CovP and CI length) and the median absolute deviation (MAD) (for M-CovP) in the subscripts.

%Calculating the average coverage probabilities (CovP) and average lengths of the confidence intervals (CIs). We calculate the empirical coverage probability and average length of the CI for each coefficient and the reported values are means and medians together with standard errors and median absolute deviation (MAD) as subscripts taken over the truly zero and non-zero coefficients.
%The expected  coverage level of the CIs is set as $95\%$.

%	\item We investigate double-robustness of our estimator and the performance of the complete case estimator (samples with $T = 1$) in Section \ref{Sim. double-robustness}.
\end{enumerate}
\begin{table}[H]%[!htbp]
	\centering
	\caption{Average $L_2$ errors of $\bthetahatDR$, obtained via various combinations of the nuisance estimators $\{\pihat(\cdot), \mhat(\cdot)\}$, and those of the oracle  estimators $\widehat{\btheta}_{orac}$ and $\widehat{\btheta}_{full}$, for $n = 1000$, $\bSigma_p = I_p$ and all three choices of the \emph{true} DGPs.} %The values in parentheses denote the SEs (over iterations) of the $L_2$ errors.}
	\label{table: DGP_id_n1000_p50}
	{\bf(I)} $p=50$. \\
	(a) DGP: ``Linear-linear'' for $\pi(\cdot)$ and $m(\cdot)$.
	\vspace{0.05in}
	\resizebox{\textwidth}{!}{
		\begin{tabular}{ll||ccc}\hline
			\multicolumn{2}{l}{Working nuisance model}& $\bthetahatDR$  & $ \widehat{\btheta}_{orac} $ &  $\widehat{\btheta}_{full}$    \\
			\hline
			\multirow{2}{*}{$\mhat$: linear}&$\pihat$: logit & 0.222 (0.035) & 0.223 (0.036) & 0.168 (0.027)   \\
			& $\pihat$: quad & 0.221 (0.035) & 0.223 (0.036) & 0.168 (0.027)  \\
			\hline
			\multirow{2}{*}{$\mhat$: quad}&$\pihat$: logit & 0.224 (0.035)  & 0.223 (0.036) & 0.168 (0.027)  \\
			&$\pihat$: quad &0.224 (0.035) & 0.223 (0.036) & 0.168 (0.027)\\
			\hline
			\multirow{2}{*}{$\mhat$: SIM}&	$\pihat$: logit & 0.222 (0.036) &  0.223 (0.036) & 0.168 (0.027)  \\
			&$\pihat$: quad & 0.222 (0.036) & 0.223 (0.036) & 0.168 (0.027) \\
			\hline
	\end{tabular}}
	(b) DGP: ``Quad-quad'' for $\pi(\cdot)$ and $m(\cdot)$.
	\vspace{0.05in}
	\resizebox{\textwidth}{!}{
		\begin{tabular}{ll||ccc}\hline
			\multicolumn{2}{l}{Working nuisance model}& $\bthetahatDR$  & $ \widehat{\btheta}_{orac} $ &  $\widehat{\btheta}_{full}$   \\
			\hline
			\multirow{2}{*}{$\mhat$: linear}&$\pihat$: logit & 0.682 (0.115) & 0.478 (0.076) & 0.453 (0.074)  \\
			& $\pihat$: quad & 0.638 (0.105) & 0.478 (0.076) & 0.453 (0.074)    \\
			\hline
			\multirow{2}{*}{$\mhat$: quad}&$\pihat$: logit & 0.475 (0.077) & 0.478 (0.076) & 0.453 (0.074)  \\
			&$\pihat$: quad & 0.475 (0.077) & 0.478 (0.076) & 0.453 (0.074) \\
			\hline
			\multirow{2}{*}{$\mhat$: SIM}&	$\pihat$: logit &  0.683 (0.116) & 0.478 (0.076) & 0.453 (0.074)   \\
			&$\pihat$: quad & 0.640 (0.108)  & 0.478 (0.076) & 0.453 (0.074) \\
			\hline
	\end{tabular}}
	(c) DGP: ``SIM-SIM'' for $\pi(\cdot)$ and $m(\cdot)$.
	\vspace{0.05in}
	\resizebox{\textwidth}{!}{
		\begin{tabular}{ll||ccc}\hline
			\multicolumn{2}{l}{Working nuisance model}& $\bthetahatDR$  & $ \widehat{\btheta}_{orac} $ &  $\widehat{\btheta}_{full}$    \\
			\hline
			\multirow{2}{*}{$\mhat$: linear}&$\pihat$: logit & 0.618 (0.138) &  0.517 (0.125) & 0.499 (0.121) \\
			& $\pihat$: quad & 0.613 (0.137)  & 0.517 (0.125) & 0.499 (0.121)   \\
			\hline
			\multirow{2}{*}{$\mhat$: quad} &$\pihat$: logit &  0.616 (0.141) & 0.517 (0.125) & 0.499 (0.121)\\
			&$\pihat$: quad & 0.612 (0.140) &  0.517 (0.125) & 0.499 (0.121)   \\
			\hline
			\multirow{2}{*}{$\mhat$: SIM}&	$\pihat$: logit &  0.553 (0.132) &  0.517 (0.125) & 0.499 (0.121) \\
			&$\pihat$: quad &  0.550 (0.131) &  0.517 (0.125) & 0.499 (0.121)  \\
			\hline
	\end{tabular}}
\end{table}

\subsection{Simulation Results} \label{Sim. Results}
We only present here the simulation results for the case $\bSigma_p = I_p$. %using identity covariance matrix are provided in this section.
The results for $\bSigma_p = $ AR1 or CS matrices exhibit similar patterns.  %share similar results and hence
These are given in Appendix \ref{sec:sim:othercov} %\ref{sec:sim:supp}
of the \hyperref[supp_mat]{Supplementary Material}.

\begin{table}[!ht]%tbp]
	\centering
	\caption{See caption of Table \ref{table: DGP_id_n1000_p50}. (Only change: $p= 500$ instead of $50$)  }
	\label{table: DGP_id_n1000_p500}
	{\bf (II)} $p=500$. \\
	(a) DGP: ``Linear-linear'' for $\pi(\cdot)$ and $m(\cdot)$.
	\vspace{0.05in}
	\resizebox{\textwidth}{!}{
		\begin{tabular}{ll||ccc}\hline
			\multicolumn{2}{l}{Working nuisance model}& $\bthetahatDR$  & $ \widehat{\btheta}_{orac} $ &  $\widehat{\btheta}_{full}$   \\
			\hline
			\multirow{2}{*}{$\mhat$: linear}&$\pihat$: logit & 0.448 (0.047) & 0.424 (0.042) & 0.317 (0.028)  \\
			& $\pihat$: quad & 0.448 (0.046) & 0.424 (0.042) & 0.317 (0.028)   \\
			\hline
			\multirow{2}{*}{$\mhat$: quad}&$\pihat$: logit & 0.461 (0.050) & 0.424 (0.042) & 0.317 (0.028)   \\
			&$\pihat$: quad & 0.461 (0.050) & 0.424 (0.042) & 0.317 (0.028)   \\
			\hline
			\multirow{2}{*}{$\mhat$: SIM}&	$\pihat$: logit & 0.436 (0.045) & 0.424 (0.042) & 0.317 (0.028)  \\
			&$\pihat$: quad & 0.436 (0.045) & 0.424 (0.042) & 0.317 (0.028) \\
			\hline
	\end{tabular}}
	(b) DGP: ``Quad-quad'' for $\pi(\cdot)$ and $m(\cdot)$.
	\vspace{0.05in}
	\resizebox{\textwidth}{!}{
		\begin{tabular}{ll||ccc}\hline
			\multicolumn{2}{l}{Working nuisance model}& $\bthetahatDR$  & $ \widehat{\btheta}_{orac} $ &  $\widehat{\btheta}_{full}$    \\
			\hline
			\multirow{2}{*}{$\mhat$: linear}&$\pihat$: logit & 1.153 (0.122) & 0.866 (0.082) & 0.811 (0.078)   \\
			& $\pihat$: quad & 1.141 (0.121) & 0.866 (0.082) & 0.811 (0.078)   \\
			\hline
			\multirow{2}{*}{$\mhat$: quad}&$\pihat$: logit & 0.887 (0.088) & 0.866 (0.082) & 0.811 (0.078) \\
			&$\pihat$: quad & 0.887 (0.088) & 0.866 (0.082) & 0.811 (0.078) \\
			\hline
			\multirow{2}{*}{$\mhat$: SIM}&	$\pihat$: logit & 1.151 (0.117) & 0.866 (0.082) & 0.811 (0.078)   \\
			&$\pihat$: quad & 1.136 (0.117) & 0.866 (0.082) & 0.811 (0.078)  \\
			\hline
	\end{tabular}}
	(c) DGP: ``SIM-SIM'' for $\pi(\cdot)$ and $m(\cdot)$.
	\vspace{0.05in}
	\resizebox{\textwidth}{!}{
		\begin{tabular}{ll||ccc}\hline
			\multicolumn{2}{l}{Working nuisance model}& $\bthetahatDR$  & $ \widehat{\btheta}_{orac} $ &  $\widehat{\btheta}_{full}$  \\
			\hline
			\multirow{2}{*}{$\mhat$: linear}&$\pihat$: logit &  1.103 (0.158) & 1.116 (0.168) & 1.087 (0.165)  \\
			& $\pihat$: quad & 1.090 (0.149) & 1.116 (0.168) & 1.087 (0.165)      \\
			\hline
			\multirow{2}{*}{$\mhat$: quad} &$\pihat$: logit &1.108 (0.159) & 1.116 (0.168) & 1.087 (0.165)  \\
			&$\pihat$: quad & 1.095 (0.151) & 1.116 (0.168) & 1.087 (0.165)    \\
			\hline
			\multirow{2}{*}{$\mhat$: SIM}&	$\pihat$: logit &  1.034 (0.161) & 1.116 (0.168) & 1.087 (0.165)  \\
			&$\pihat$: quad &  1.021 (0.153)  & 1.116 (0.168) & 1.087 (0.165) \\
			\hline
	\end{tabular}}
\end{table}

Tables \ref{table: DGP_id_n1000_p50} and \ref{table: DGP_id_n1000_p500} provide the $L_2$ error comparison for all estimators under $p= 50$ and $500$ respectively.
The results, in general, exhibit a similar pattern across the two tables, with the errors being only  higher (and understandably so) for $p = 500$ than the corresponding case for $p=50$. Overall, we observe that whenever  $\{\pihat(\cdot), \mhat(\cdot)\}$ are both correctly specified, $\bthetahatDR$ closely matches the performance of the oracle $\bthetahat_{orac}$, which validates our claims of optimality, and also, first order insensitivity (Remark \ref{rem:firstorderinsensitive}) of $\bthetahatDR$ to nuisance function estimation errors.
Interestingly, over all the DGP settings, the performance of $\bthetahatDR$, in fact, continues to remain nearly as good when only $\mhat(\cdot)$, but \emph{not} necessarily $\pihat(\cdot)$, is correctly specified. This indicates that it is fairly robust to any misspecification of $\pihat(\cdot)$. On the other hand, it is also %appears somewhat
more sensitive to $\mhat(\cdot)$, since the errors do tend to increase somewhat %appreciably
when $\mhat(\cdot)$ is misspecified for some of the non-linear DGPs, except Table \ref{table: DGP_id_n1000_p500}(c). Nevertheless, it is still expected to be consistent whenever only one of $\{\pihat(\cdot), \mhat(\cdot)\}$ is correct, and we validate this DR property via a large sample analysis in Appendix \ref{Sim. double-robustness}.
Lastly, it is interesting to note that for all the non-linear DGPs, the $L_2$ errors of $\bthetahat_{orac}$ and $\bthetahat_{full}$ are relatively close, indicating that there is little loss due to missing $Y$ in estimating $\btheta_0$. This is possibly due to the non-linear form of $m(\bX)$ %a joint effect of missing $Y$ and the non-linear nature of $\E(Y|X)$ %and the high dimensionality all of
which contributes towards reducing the gap between the two oracles. %and `super'-oracle estimators.
\begin{table}[H]%[!htbp]
	\centering
	\caption{Average ($A$-CovP) and median ($M$-CovP) of the empirical coverage probabilities (CovPs) for the (coordinatewise) 95\% CIs of $\btheta_0$ obtained via $\bthetatilDR$ (based on various combinations of the nuisance estimators $\{\pihat(\cdot), \mhat(\cdot)\}$) %over all replications,
for $n = 1000$, $\bSigma_p = I_p$ and all three choices of the \emph{true} DGPs. Shown also are the corresponding average %(over coordinates) of the mean
lengths of these CIs. All values are reported separately for the truly zero and non-zero coefficients of $\btheta_0$ (see Section \ref{Est. Implemented}).}
%%and lengths of the CIs built upon the desparsified estimator under the setting of $n=1000$ using identity covariance matrix. Different working nuisance models for $\pi(\cdot)$ and $m(\cdot)$ are compared. We report the means and medians together with standard errors and MAD as subscripts.  }
	\label{table: DGP_id_n1000_p50_infer}
	{\bf (I)} $p=50$. \\
	(a) DGP: ``Linear-linear'' for $\pi(\cdot)$ and $m(\cdot)$.
		\vspace{0.05in}
		\resizebox{\textwidth}{!}{
			\begin{tabular}{ll||ccc|ccc}\hline
				\multicolumn{2}{l}{Working nuisance model}  & \multicolumn{3}{c}{Zero coefficients} &  \multicolumn{3}{c}{Non-zero coefficients} \\
				\hline
				\multicolumn{2}{l}{} &$A$-CovP & $M$-CovP  & Length &$A$-CovP & $M$-CovP  & Length \\
				\hline
				\multirow{2}{*}{$\mhat$: linear}&$\pihat$: logit & $0.94_{0.01}$ & $(0.94_{0.01})$ & $0.16_{0}$ & $0.94_{0.01}$ & $(0.94_{0.01})$ & $0.16_{0}$ \\
				&$\pihat$: quad & $0.94_{0.01}$ & $(0.94_{0.01})$ & $0.16_{0}$ & $0.94_{0.01}$ & $(0.94_{0.01})$ & $0.16_{0}$  \\
				\hline
				\multirow{2}{*}{$\mhat$: quad}&$\pihat$: logit & $0.94_{0.01}$ & $(0.94_{0.01})$ & $0.16_{0}$ & $0.94_{0.01}$ & $(0.94_{0.01})$ & $0.16_{0}$   \\
				&$\pihat$: quad &  $0.95_{0.01}$ & $(0.95_{0.01})$ & $0.16_{0}$ & $0.94_{0.01}$ & $(0.94_{0.01})$ & $0.16_{0}$  \\
				\hline
				\multirow{2}{*}{$\mhat$: SIM}&	$\pihat$: logit &  $0.95_{0.01}$ & $(0.95_{0.01})$ & $0.16_{0}$ & $0.93_{0.01}$ & $(0.93_{0.01})$ & $0.16_{0}$  \\
				&$\pihat$: quad & $0.95_{0.01}$ & $(0.95_{0.01})$ & $0.16_{0}$ & $0.94_{0.01}$ & $(0.94_{0.01})$ & $0.16_{0}$\\
				\hline
		\end{tabular}}
	(b) DGP: ``Quad-quad'' for $\pi(\cdot)$ and $m(\cdot)$.
	\vspace{0.05in}
	\resizebox{\textwidth}{!}{
		\begin{tabular}{ll||ccc|ccc}\hline
			\multicolumn{2}{l}{Working nuisance model} & \multicolumn{3}{c}{Zero coefficients} &  \multicolumn{3}{c}{Non-zero coefficients} \\
		\hline
		\multicolumn{2}{l}{} &$A$-CovP & $M$-CovP  & Length &$A$-CovP & $M$-CovP  & Length \\
		\hline
			\multirow{2}{*}{$\mhat$: linear}&$\pihat$: logit & $0.94_{0.01}$ & $(0.94_{0.01})$ & $0.41_{0}$ & $0.88_{0.16}$ & $(0.93_{0.02})$ & $0.46_{0.08}$ \\
			& $\pihat$: quad & $0.95_{0.01}$ & $(0.95_{0.01})$ & $0.41_{0}$ & $0.89_{0.12}$ & $(0.93_{0.02})$ & $0.46_{0.07}$   \\
			\hline
			\multirow{2}{*}{$\mhat$: quad}&$\pihat$: logit &  $0.94_{0.01}$ & $(0.94_{0.01})$ & $0.34_{0}$ & $0.94_{0.01}$ & $(0.94_{0.01})$ & $0.38_{0.06}$\\
			&$\pihat$: quad &$0.94_{0.01}$ & $(0.95_{0.01})$ & $0.34_{0}$ & $0.94_{0.01}$ & $(0.94_{0.01})$ & $0.38_{0.06}$   \\
			\hline
			\multirow{2}{*}{$\mhat$: SIM}&	$\pihat$: logit &  $0.95_{0.01}$ & $(0.94_{0.01})$ & $0.41_{0}$ & $0.88_{0.16}$ & $(0.94_{0.02})$ & $0.46_{0.08}$ \\
			&$\pihat$: quad & $0.95_{0.01}$ & $(0.95_{0.01})$ & $0.41_{0}$ & $0.89_{0.12}$ & $(0.93_{0.03})$ & $0.47_{0.07}$ \\
			\hline
	\end{tabular}}
	(c) DGP: ``SIM-SIM'' for $\pi(\cdot)$ and $m(\cdot)$.
	\vspace{0.05in}
	\resizebox{\textwidth}{!}{
		\begin{tabular}{ll||ccc|ccc}\hline
			\multicolumn{2}{l}{Working nuisance model} & \multicolumn{3}{c}{Zero coefficients} &  \multicolumn{3}{c}{Non-zero coefficients} \\
		\hline
		\multicolumn{2}{l}{}  &$A$-CovP & $M$-CovP  & Length &$A$-CovP & $M$-CovP  & Length \\
		\hline
		\multirow{2}{*}{$\mhat$: linear}&$\pihat$: logit &  $0.94_{0.01}$ & $(0.94_{0.01})$ & $0.46_{0}$ & $0.94_{0.01}$ & $(0.94_{0.01})$ & $0.52_{0.04}$  \\
			& $\pihat$: quad & $0.94_{0.01}$ & $(0.94_{0.01})$ & $0.45_{0}$ & $0.94_{0.01}$ & $(0.94_{0.01})$ & $0.52_{0.04}$   \\
			\hline
			\multirow{2}{*}{$\mhat$: quad}&$\pihat$: logit & $0.95_{0.01}$ & $(0.95_{0.01})$ & $0.45_{0}$ & $0.94_{0.01}$ & $(0.94_{0.01})$ & $0.52_{0.04}$   \\
			&$\pihat$: quad & $0.95_{0.01}$ & $(0.95_{0.01})$ & $0.45_{0}$ & $0.94_{0.01}$ & $(0.94_{0.01})$ & $0.52_{0.04}$\\
			\hline
			\multirow{2}{*}{$\mhat$: SIM}&	$\pihat$: logit &$0.95_{0.01}$ & $(0.95_{0.01})$ & $0.40_{0}$ & $0.94_{0.01}$ & $(0.94_{0.01})$ & $0.46_{0.03}$  \\
			&$\pihat$: quad &  $0.95_{0.01}$ & $(0.95_{0.01})$ & $0.40_{0}$ & $0.94_{0.01}$ & $(0.94_{0.01})$ & $0.45_{0.03}$ \\
			\hline
	\end{tabular}}
\end{table}
%For example in Table \ref{table: DGP_id_n1000_p50}(b), under the ``quad-quad'' DGP, the error for ``$\pihat$: quad-$\mhat$: linear'' is larger than that of ``$\pihat$: linear-$\mhat$: quad''. Similar patterns can also be observed from ``SIM-SIM'' DGP. In fact, when the working model for $m(\cdot)$ is correct, the estimation performances are equally good for both choices of the working model for $\pi(\cdot)$. In the ``SIM-SIM'' DGP, the estimation errors for using different working nuisance model are relatively close to each other. This is due to the parameter that we choose for the ``SIM-SIM'' DGP. In addition, comparing the estimation errors of $\btheta_{orac}$ and $\btheta_{full}$ across different DGPs, we notice that they are relatively closer for ``quad-quad'' and ``SIM-SIM'' DGPs than those of the ``linear-linear'' DGPs. The joint effect of missing outcomes and non-linear DPGs reduces the gap between the estimation errors using oracle and super oracle estimator.
%
\begin{table}[!ht]%[!htbp]
	\centering
	\caption{See caption of Table \ref{table: DGP_id_n1000_p50_infer}. (Only change: $p= 500$ instead of $50$) }
	\label{table: DGP_id_n1000_p500_infer}
	{\bf (II)} $p=500$. \\
	(a) DGP: ``Linear-linear'' for $\pi(\cdot)$ and $m(\cdot)$.
	\vspace{0.05in}
	\resizebox{\textwidth}{!}{
		\begin{tabular}{ll||ccc|ccc}\hline
		\multicolumn{2}{l}{Working nuisance model}  & \multicolumn{3}{c}{Zero coefficients} &  \multicolumn{3}{c}{Non-zero coefficients} \\
		\hline
		\multicolumn{2}{l}{}  &$A$-CovP & $M$-CovP  & Length &$A$-CovP & $M$-CovP  & Length \\
		\hline
			\multirow{2}{*}{$\mhat$: linear}&$\pihat$: logit &$0.94_{0.01}$ & $(0.94_{0.01})$ & $0.16_{0}$ & $0.92_{0.01}$ & $(0.92_{0.01})$ & $0.16_{0}$ \\
			&$\pihat$: quad &  $0.94_{0.01}$ & $(0.94_{0.01})$ & $0.16_{0}$ & $0.91_{0.02}$ & $(0.92_{0.01})$ & $0.16_{0}$  \\
			\hline
			\multirow{2}{*}{$\mhat$: quad}&$\pihat$: logit &  $0.94_{0.01}$ & $(0.95_{0.01})$ & $0.17_{0}$ & $0.91_{0.02}$ & $(0.91_{0.01})$ & $0.17_{0}$ \\
			&$\pihat$: quad &  $0.94_{0.01}$ & $(0.95_{0.01})$ & $0.17_{0}$ & $0.91_{0.02}$ & $(0.91_{0.01})$ & $0.17_{0}$  \\
			\hline
			\multirow{2}{*}{$\mhat$: SIM}&	$\pihat$: logit &  $0.94_{0.01}$ & $(0.94_{0.01})$ & $0.16_{0}$ & $0.92_{0.01}$ & $(0.92_{0.01})$ & $0.16_{0}$  \\
			&$\pihat$: quad & $0.94_{0.01}$ & $(0.95_{0.01})$ & $0.16_{0}$ & $0.92_{0.01}$ & $(0.92_{0.01})$ & $0.16_{0}$ \\
			\hline
	\end{tabular}}
	(b) DGP: ``Quad-quad'' for $\pi(\cdot)$ and $m(\cdot)$.
	\vspace{0.05in}
	\resizebox{\textwidth}{!}{
		\begin{tabular}{ll||ccc|ccc}\hline
		\multicolumn{2}{l}{Working nuisance model}  & \multicolumn{3}{c}{Zero coefficients} &  \multicolumn{3}{c}{Non-zero coefficients} \\
		\hline
		\multicolumn{2}{l}{}  &$A$-CovP & $M$-CovP  & Length &$A$-CovP & $M$-CovP  & Length \\
		\hline
			\multirow{2}{*}{$\mhat$: linear}&$\pihat$: logit & $0.95_{0.01}$ & $(0.95_{0.01})$ & $0.44_{0}$ & $0.91_{0.03}$ & $(0.92_{0.02})$ & $0.46_{0.07}$  \\
			& $\pihat$: quad & $0.95_{0.01}$ & $(0.95_{0.01})$ & $0.43_{0}$ & $0.91_{0.03}$ & $(0.92_{0.01})$ & $0.46_{0.06}$   \\
			\hline
			\multirow{2}{*}{$\mhat$: quad}&$\pihat$: logit & $0.94_{0.01}$ & $(0.95_{0.01})$ & $0.33_{0}$ & $0.92_{0.01}$ & $(0.92_{0.01})$ & $0.35_{0.04}$ \\
			&$\pihat$: quad & $0.94_{0.01}$ & $(0.95_{0.01})$ & $0.33_{0}$ & $0.92_{0.01}$ & $(0.92_{0.01})$ & $0.35_{0.04}$   \\
			\hline
			\multirow{2}{*}{$\mhat$: SIM}&	$\pihat$: logit &  $0.95_{0.01}$ & $(0.95_{0.01})$ & $0.44_{0}$ & $0.91_{0.05}$ & $(0.93_{0.02})$ & $0.47_{0.07}$  \\
			&$\pihat$: quad & $0.95_{0.01}$ & $(0.95_{0.01})$ & $0.43_{0}$ & $0.91_{0.04}$ & $(0.92_{0.01})$ & $0.46_{0.06}$ \\
			\hline
	\end{tabular}}
	(c) DGP: ``SIM-SIM'' for $\pi(\cdot)$ and $m(\cdot)$.
	\vspace{0.05in}
	\resizebox{\textwidth}{!}{
		\begin{tabular}{ll||ccc|ccc}\hline
			\multicolumn{2}{l}{Working nuisance model}  & \multicolumn{3}{c}{Zero coefficients} &  \multicolumn{3}{c}{Non-zero coefficients} \\
			\hline
			\multicolumn{2}{l}{}  &$A$-CovP & $M$-CovP  & Length &$A$-CovP & $M$-CovP  & Length \\
			\hline
			\multirow{2}{*}{$\mhat$: linear}&$\pihat$: logit & $0.94_{0.01}$ & $(0.95_{0.01})$ & $0.53_{0}$ & $0.87_{0.05}$ & $(0.88_{0.06})$ & $0.57_{0.03}$  \\
			& $\pihat$: quad & $0.94_{0.01}$ & $(0.95_{0.01})$ & $0.53_{0}$ & $0.87_{0.05}$ & $(0.86_{0.07})$ & $0.57_{0.03}$  \\
			\hline
			\multirow{2}{*}{$\mhat$: quad}&$\pihat$: logit & $0.95_{0.01}$ & $(0.95_{0.01})$ & $0.53_{0}$ & $0.88_{0.04}$ & $(0.88_{0.05})$ & $0.57_{0.03}$  \\
			&$\pihat$: quad & $0.95_{0.01}$ & $(0.95_{0.01})$ & $0.53_{0}$ & $0.87_{0.05}$ & $(0.87_{0.06})$ & $0.57_{0.03}$\\
			\hline
			\multirow{2}{*}{$\mhat$: SIM}&	$\pihat$: logit &$0.95_{0.01}$ & $(0.95_{0.01})$ & $0.50_{0}$ & $0.93_{0.02}$ & $(0.93_{0.01})$ & $0.54_{0.03}$  \\
			&$\pihat$: quad &  $0.95_{0.01}$ & $(0.95_{0.01})$ & $0.50_{0}$ & $0.93_{0.02}$ & $(0.93_{0.01})$ & $0.54_{0.03}$ \\
			\hline
	\end{tabular}}
\end{table}
Tables \ref{table: DGP_id_n1000_p50_infer} and \ref{table: DGP_id_n1000_p500_infer} summarize the CovPs and lengths of the CIs obtained via %the desparsified estimator
$\bthetatilDR$. We first observe that across \emph{all} DGPs and choices of $p$, the CovPs for the truly zero coefficients of $\btheta_0$ are always close to the desired 95\% level \emph{regardless} of the choice of the working nuisance models, which is (pleasantly) surprising. While our theoretical results in Section \ref{sec:inference} on $\bthetatilDR$ do require both $\pihat(\cdot)$ and $\mhat(\cdot)$ to be correct, the empirical results seem to be quite robust in this regard, at least for the zero coefficients. Among the non-zero coefficients of $\btheta_0$ (which are much fewer in number), the results are more along expected lines. For $p = 50$, the CovPs are all close to 95\% whenever $\mhat(\cdot)$ is correctly specified which, again, demonstrates the robustness of the results (this time in inference) towards misspecififcation of $\pihat(\cdot)$. On the other hand, when $\mhat(\cdot)$ is misspecified, the average CovPs could often be much lower than 95\%, e.g. Table \ref{table: DGP_id_n1000_p50_infer}(b), which should not be unexpected. However, for these same CIs, the median CovPs are considerably better and \emph{still} reasonably close to 95\%, thus indicating that for only a few of these coefficients, the corresponding CIs have low CovPs when $\mhat(\cdot)$ is misspecified. For the SIM-SIM DGP, however, the results seem to be quite good and invariant to the choices of $\{\pihat(\cdot),\mhat(\cdot)\}$.

%For the truly non-zero coefficients, the results are determined by the working models and $n,p$. When $m(\cdot)$ is correctly specified, the CIs would provide correct coverage probabilities when $p=50$. When $m(\cdot)$ is not correctly specified, the performance of the CIs is in general not good in terms of mean coverage probabilities, but these CIs still have a reasonable median coverage probabilities (see Table \ref{table: DGP_id_n1000_p50_infer}(b)). This implies that in the moderate dimensional case, there are some particular coefficients whose corresponding CIs have low coverage probabilities when the model is misspecified while some others still have desired value. When the true DGP is ``SIM-SIM'', different working models provide similar results. The pattern is the same as the one observed in the estimation error table.

For $p = 500$, the results for the zero coefficients are similar to $p = 50$. For the non-zero coefficients, however, the CovPs generally tend to be a little bit below 95\%, and at the same time, are more similar (except for Table \ref{table: DGP_id_n1000_p500_infer}(c)) across different choices of $\{\pihat(\cdot), \mhat(\cdot)\}$ regardless of their correctness. This is possibly due to a combination of the price we pay in estimating the influence function for $\bthetatilDR$ and the precision matrix $\bOmega$ under such high dimensional settings, as well as the (well known) bias inherent in the non-zero coefficients of shrinkage estimators like $\bthetahatDR$. These finite sample biases are expected to be reduced with larger sample sizes. Indeed, in our large sample analyses in Appendix \ref{Sim. double-robustness} with $n=50000$, the patterns are much clearer and the results much improved, wherein we show that the CovPs achieved are fairly close to the nominal level of 95\% whenever at least one of $\{\pihat(\cdot),\mhat(\cdot)\}$ is correct.
Lastly, for both $p = 50$ and $500$, the average lengths of the CIs are in general shorter for the cases where $\mhat(\cdot)$ %if not both $\{\pihat(\cdot),\mhat(\cdot)\}$,
is correctly specified (e.g., Table \ref{table: DGP_id_n1000_p50_infer}(b)-(c) and Table \ref{table: DGP_id_n1000_p500_infer}(b)-(c)). Whenever $\mhat(\cdot)$ is misspecified, the corresponding CIs tend to be larger and also provide low coverage in some cases, as expected. %Those misspecified working models provide CIs that have poor coverage probabilities even with larger lengths.

As mentioned before, the results for $\bSigma_p =$ AR1 or CS, given in Appendix \ref{sec:sim:othercov}, are mostly similar to the corresponding results for $\bSigma_p = I_p$,
indicating (empirically) that our procedures are fairly \emph{robust} to the underlying correlation structure of $\bX$, as well as the degree of sparsity of $\bOmega$, in high dimensional settings. Finally, in Appendix \ref{Sim. double-robustness}, apart from validating the DR properties (both in estimation and inference) of our estimators, we also demonstrate, via a large sample analysis for a non-linear DGP, that the CC estimator can be \emph{inconsistent}, thus showing its unsuitability as a general estimator of $\btheta_0$.

\section{Discussion}\label{Discussion}  %% Needs to be edited %%

%high dimensional M estimation problem with missing response
%high dimensional target parameter under fully non-parametric model
%estimation error: finite sample bounds
%DDR property
%inference of the DDR estimator
%estimation of the nuisance functions

%sharp bounds when only one correct

\hspace{-0.053in}In this paper, we studied high dimensional $M$-estimation problems with missing outcomes under a model-free semi-parametric framework. Our parameter of interest itself is high dimensional which is a key distinction from most of the existing literature. A variety of important problems were discussed as special cases of this framework,
along with their counterparts in causal inference based on the `potential outcomes' framework. %with useful applications in heterogeneous treatment effects estimation and precision medicine. %With the response $Y$ possibly missing at random and high dimensional covariates, we consider estimation of and inference for the target parameter $\btheta_0$, which is defined as the minimizer of the risk of a convex loss. This parameter of interest is defined in such a way that it is a high dimensional parameter under a fully non-parametric model. This framework includes standard regression with missing outcomes and is also applicable to causal inference such as heterogeneous treatment effects estimation.

We proposed the $L_1$-regularized DDR estimator of $\btheta_0$ which serves as a \emph{generalization} of traditional DR estimators for high dimensional parameters. We studied its properties in detail via non-asymptotic bounds under mild tail assumptions and \emph{only} high-level rate conditions on the nuisance estimators, showing its rate optimality, first order insensitivity and DR properties under appropriate conditions. %when both nuisance models are correct, and consistency %(in Appendix \ref{draspect})%showing its rate optimality and first order insensitivity when both nuisance models are correct, and consistency when only one is correct.
Our other main contribution is the desparsified DDR estimator which admits a \emph{semi-parametric optimal} ALE %that is  %, under suitable conditions, a ALE
and also facilitates \emph{inference} on $\btheta_0$. %ALE , % and also facilitates .
It serves as the appropriate \emph{generalization} of Debiased Lasso type estimators for high dimensional inference with missing $Y$. Further, %As a third contribution,
we also discuss various choices of the nuisance estimators %based on high dimensional models %parametric or semi-parametric
and their properties. However, one is also free to use \emph{any} other choice as long as it satisfies our basic conditions. All our results were validated via extensive simulations, and while not presented due to limited space, we also have promising real data analysis results for our methods. Lastly, we have only investigated the DR properties of our estimator in terms of consistency. Getting the sharp rates (and inferential tools) under these general settings is much more challenging and requires a case-by-case analysis. We leave this for future research. %and remains an open problem.%(in Appendix \ref{SEC:NUISANCE:SEPARATE:SUPP}. %Meanwhile, we provide theoretical results on estimating the nuisance functions (propensity score and outcome models) using linear and non-linear, parametric and semi-parametric models, which expands the current literature.

%We also investigate the double robustness of our estimator showing its consistency even if only one of the working models on the propensity score and the outcome models is correctly specified. We include both the theoretical results and the simulations that examine this property in the appendix. The sharp rates of our proposed estimator under more general settings requires case-by-case studies and remains an open problem.

\appendix

%\vspace{-0.014in}
\section*{Supplementary Material}\label{supp_mat}
\textbf{Supplementary Materials for ``High Dimensional $M$-Estimation with Missing Outcomes: A Semi-Parametric Framework''}
(.pdf file).
In the \hyperref[supp_mat]{Supplementary Material} (Appendices \ref{draspect}-\ref{pfs:nuisance}), we collect several important materials that could not be accommodated in the main manuscript, %, including discussions on DR properties of the estimator, additional numerical results and all technical materials and discussions, including proofs for all the main results and associated supporting lemmas.
including: discussions on the DR properties of our estimator (Appendix \ref{draspect}), properties of the nuisance function estimators (Appendix \ref{SEC:NUISANCE:SEPARATE:SUPP}; Theorems \ref{parametric:thm}-\ref{KS:mainthm2}) and their proofs (Appendix \ref{pfs:nuisance}), additional numerical results (Appendix \ref{sec:sim:supp}), %including specifically the validation of DR properties and performance of the CC estimator (Appendix \ref{Sim. double-robustness}),
 supporting lemmas and a few related definitions (Appendix \ref{techtools}), some key technical discussions on the error terms (Appendix \ref{discussion:errorterms}),
 and the proofs of all our main results (Lemma \ref{DEV:BOUND} in Appendix \ref{pf:dev:bound}, Theorem \ref{TZERO:THM} in Appendix \ref{pf:tzero:thm}, Theorems \ref{TPI:THM}-\ref{RMPI:THM} in Appendices \ref{pf:tpi:thm}-\ref{pf:rmpi:thm}, and Theorem \ref{HDINF:THM} in Appendix \ref{pf:HDinf:thm}).

%\bibliographystyle{imsart-nameyear}
%%\bibliographystyle{apalike}
%\bibliography{P1-HDME-Abhishek-Biblio}

%\end{document}

\pagebreak

\setcounter{section}{0}

\begin{center}
%\large
{\bf \uppercase{Supplementary Materials for ``High Dimensional $M$-Estimation with Missing Outcomes: A Semi-Parametric Framework''}}
\end{center}%\label{suppl}%\vspace{0.1in} %\vspace{0.8in}
\vspace{0.03in}

\begin{center}
\uppercase{By Abhishek Chakrabortty, Jiarui Lu, T. Tony Cai} \\
\uppercase{and Hongzhe Li} \par\smallskip
{\em Texas A\&M University and  University of Pennsylvania}
\end{center}
\vspace{0.02in}

\input{Chakrabortty-HDME-Supp-R1.tex}

\bibliographystyle{imsart-nameyear}
\bibliography{P1-HDME-Abhishek-Biblio}

\end{document}

%% file: Chakrabortty-HDME-Supp-R1.tex
\paragraph*{Organization}
In this supplement (Appendices \ref{draspect}-\ref{pfs:nuisance}), %of the \hyperref[supp_mat]{Supplementary Material},
we collect several important materials that could not be accommodated in the main manuscript, including: discussions on the DR properties of our estimator (Appendix \ref{draspect}), properties of the nuisance function estimators (Appendix \ref{SEC:NUISANCE:SEPARATE:SUPP}; Theorems \ref{parametric:thm}-\ref{KS:mainthm2}) and their proofs (Appendix \ref{pfs:nuisance}), additional numerical results (Appendix \ref{sec:sim:supp}), %including specifically the validation of DR properties and performance of the CC estimator (Appendix \ref{Sim. double-robustness}),
 supporting lemmas and a few related definitions (Appendix \ref{techtools}), some key technical discussions on the error terms (Appendix \ref{discussion:errorterms}),
 and the proofs of all our main results (Lemma \ref{DEV:BOUND} in Appendix \ref{pf:dev:bound}, Theorem \ref{TZERO:THM} in Appendix \ref{pf:tzero:thm}, Theorems \ref{TPI:THM}-\ref{RMPI:THM} in Appendices \ref{pf:tpi:thm}-\ref{pf:rmpi:thm}, and Theorem \ref{HDINF:THM} in Appendix \ref{pf:HDinf:thm}). %technical materials proofs of all our main results and the associated supporting lemmas.

\appendix

\section{Double Robustness of the DDR Estimator}\label{draspect}

Our probabilistic analysis of $\|\bT_n\|_{\infty}$ %$\|\bT_n\|_{\infty} \equiv \| \bnabla \LnDR (\btheta_0)\|_{\infty}$
for establishing the convergence rate of $\bthetahatDR$ (in the light of Lemma \ref{DEV:BOUND}) has so far assumed that both the nuisance functions $\{\pi(\cdot), m(\cdot)\}$ are correctly estimated via $\{\pihat(\cdot),\mhat(\cdot)\}$ satisfying Assumptions \ref{tpicont:assmpn}-\ref{tmcont:assmpn}. As noted in \eqref{DRrep:eqn2}, the nature of the population DDR loss $\LDR(\cdot)$ and the empirical version $\LnDR(\cdot)$ is such that consistency of $\|\bT_n\|_{\infty}$ (and hence $\bthetahatDR$) should hold even if only one of $\{\pihat(\cdot),\mhat(\cdot)\}$ is correct. %, and not necessarily both,

%\subsection{Consistency under Misspecification of one of the Nuisance Functions}\label{subsec:draspect}
In this section, we briefly sketch the arguments that ensure \emph{consistency} of $\|\bT_n\|_{\infty}$ even if \emph{only one} of $\{\pihat(\cdot), \mhat(\cdot)\}$ is correctly specified but \emph{not} necessarily both. The convergence rates underlying this consistency, while reasonable, are not necessarily sharp, however. To obtain sharper rates (if possible at all) under these general situations, one needs a more nuanced case-by-case analysis which \emph{will depend} now on the first order properties, rates and nature of construction of the estimators, unlike the case when both estimators are correctly specified and the results are first order insensitive (see Remark \ref{rem:firstorderinsensitive}), i.e. require no specific knowledge about the estimators except for some high-level convergence properties. This is true even for classical settings, and the high dimensional setting here only lends further complexity and subtlety to the issue. Considering the main goals and scope of this paper, we suppress such finer analysis under those cases, for simplicity and brevity.

\paragraph*{Case 1} Suppose that $\pihat(\cdot)$ is misspecified, such that $\pihat(\bx) \convP \pi^*(\bx) \neq \pi(\bx)$ following Assumption \ref{tpicont:assmpn} with $\pi(\cdot)$ therein replaced by a general $\pi^*(\cdot)$, %possibly not  equal to $ \pi(\cdot)$,
while $\mhat(\cdot)$ is still correctly specified with $\mhat(\bx) \convP m(\bx)$ following Assumption \ref{tmcont:assmpn}. In this case, the terms $\bTzeron$ and $\bTmn$ in the decomposition (\ref{decomp:eqn}) of $\bT_n$ will stay unaffected and their properties still governed by the results of Theorems \ref{TZERO:THM} and \ref{TM:THM} respectively, while the error terms $\bTpin$ and $\bRpimn$ involving $\pihat(\cdot)$  would be affected and need to be appropriately analyzed as follows.

$\bTpin$ should be further decomposed into two terms as: $\bTpin = \bTtilpin + \bTpin^*$,
\begin{align}
\mbox{where} \;\; &\bTtilpin \; := \; \frac{1}{n} \sum_{i=1}^n \left\{\frac{T_i}{\pihat(\bX_i)} - \frac{T_i}{\pi^*(\bX_i)}\right\} \left\{ Y_i - m(\bX_i)\right\} \bh(\bX_i) \nonumber  \\
\mbox{and} \;\; &\bTpin^* \; := \; \frac{1}{n} \sum_{i=1}^n \left\{\frac{T_i}{\pi^*(\bX_i)} - \frac{T_i}{\pi(\bX_i)}\right\} \left\{ Y_i - m(\bX_i)\right\} \bh(\bX_i),  \nonumber %\label{dr:pi:eqn1} %\\
\end{align}
while $\bRpimn$ should be decomposed further as: $\bRpimn = \bRtilpimn + \bRpimn^*$,
\begin{align}
\mbox{where} \;\; & \bRtilpimn \; := \; \frac{1}{n} \sum_{i=1}^n  \left\{\frac{T_i}{\pihat(\bX_i)} - \frac{T_i}{\pi^*(\bX_i)}\right\} \left\{ \mtil(\bX_i) - m(\bX_i)\right\} \bh(\bX_i)  \nonumber \\
\mbox{and} \;\; &  \bRpimn^* \; := \; \frac{1}{n} \sum_{i=1}^n \left\{\frac{T_i}{\pi^*(\bX_i)} - \frac{T_i}{\pi(\bX_i)}\right\} \left\{  \mtil(\bX_i) - m(\bX_i)\right\} \bh(\bX_i). \nonumber %\label{dr:pi:eqn2}
\end{align}
Suppose Assumption \ref{tpicont:assmpn} is modified appropriately with $\pi(\cdot)$ therein replaced throughout by $\pi^*(\cdot)$, the true target function of $\pihat(\cdot)$ in this case, and assume also that $\pi^*(\bX) > \deltapi^* > 0$ for some constant $\deltapi^*$, and $\pi^*(\bX) - \pi(\bX)$ is bounded (or sub-Gaussian).
Then, under Assumptions \ref{base:assmpns} and \ref{subgaussian:assmpn}-\ref{tmcont:assmpn}, using similar arguments as those used in the proofs of Theorems \ref{TZERO:THM}-\ref{TPI:THM} (for $\bTpin^*$ and $\bTtilpin$ respectively) and Theorem \ref{RMPI:THM} (for $\bRtilpimn$ and $\bRpimn^*$), it can be shown that
\begin{align}
& \| \bTtilpin \|_{\infty} \; \lesssim \; \vnpi \sqrt{\log(nd)} \sqrt{\frac{\log d}{n}} \;\; \mbox{and} \;\; \| \bTpin^* \|_{\infty} \; \lesssim \; \sqrt{\frac{\log d}{n}} \;\; \mbox{w.h.p.}, \;\; \mbox{and} \nonumber \\
& \| \bRtilpimn \|_{\infty} \; \lesssim \; \vnpi \vnbarm (\log n) \;\; \mbox{and} \;\; \| \bRpimn^* \|_{\infty} \; \lesssim \; \vnbarm \sqrt{\log n} \;\; \mbox{w.h.p.} \nonumber
\end{align}
\paragraph*{Case 2} Suppose $\mhat(\cdot)$ is misspecified instead with $\mhat(\bx) \convP m^*(\bx)\neq m(\bx)$ according to Assumption \ref{tmcont:assmpn} with $m(\cdot)$ replaced by a general $m^*(\cdot)$ therein, %possibly not  equal to $ \m(\cdot)$,
while $\pihat(\cdot)$ is still correctly specified with $\pihat(\bx) \convP \pi(\bx)$ following Assumption \ref{tpicont:assmpn}. In this case, the terms $\bTzeron$ and $\bTpin$ in the decomposition (\ref{decomp:eqn}) of $\bT_n$ stay unaffected and their properties still governed by the results of Theorems \ref{TZERO:THM} and \ref{TPI:THM} respectively, while the error terms $\bTmn$ and $\bRpimn$ involving $\mhat(\cdot)$  would be affected and need to be appropriately analyzed as follows.

$\bTmn$ may be further decomposed into two terms as: $\bTmn =  \bTtilmn  + \bTmn^*$,
\begin{align}
\mbox{where} \;\; & \bTtilmn \; := \; \frac{1}{n} \sum_{i=1}^n \left\{\frac{T_i}{\pi(\bX_i)} - 1 \right\} \left\{ \mtil(\bX_i) - m^*(\bX_i)\right\} \bh(\bX_i) \nonumber  \\
\mbox{and} \;\; & \bTmn^* \; := \; \frac{1}{n} \sum_{i=1}^n \left\{\frac{T_i}{\pi(\bX_i)} - 1 \right\} \left\{ m^*(\bX_i) - m(\bX_i)\right\} \bh(\bX_i), \nonumber %\label{dr:m:eqn1}
\end{align}
while $\bRpimn$ should be decomposed further as: $\bRpimn = \bRdagpimn   + \bRpimn^{**}$,
\begin{align}
\mbox{where} \;\; & \bRdagpimn \; := \; \frac{1}{n} \sum_{i=1}^n  \left\{\frac{T_i}{\pihat(\bX_i)} - \frac{T_i}{\pi(\bX_i)}\right\} \left\{ \mtil(\bX_i) - m^*(\bX_i)\right\} \bh(\bX_i) \nonumber \\
\mbox{and} \;\; & \bRpimn^{**} \; := \;  \frac{1}{n} \sum_{i=1}^n \left\{\frac{T_i}{\pihat(\bX_i)} - \frac{T_i}{\pi(\bX_i)}\right\} \left\{  m^*(\bX_i) - m(\bX_i)\right\} \bh(\bX_i). \nonumber %\label{dr:m:eqn2}
\end{align}
Suppose Assumption \ref{tmcont:assmpn} is modified appropriately whereby $m(\cdot)$ is replaced throughout by $m^*(\cdot)$, the true target function of $\mhat(\cdot)$ in this case. Further, assume also that $m^*(\bX) - m(\bX)$ is sub-Gaussian.
Then, under Assumptions \ref{base:assmpns} and \ref{subgaussian:assmpn}-\ref{tmcont:assmpn}, using similar arguments as those in the proofs of Theorems \ref{TZERO:THM} and \ref{TM:THM} (for $\bTmn^*$ and $\bTtilmn$ respectively) and Theorem \ref{RMPI:THM} (for $\bRdagpimn$ and $\bRpimn^{**}$), it is not difficult to show that the following hold:
\begin{align}
& \| \bTtilmn \|_{\infty} \; \lesssim \; \vnbarm \sqrt{\log(nd)} \sqrt{\frac{\log d}{n}} \;\; \mbox{and} \;\; \| \bTmn^* \|_{\infty} \; \lesssim \; \sqrt{\frac{\log d}{n}} \;\; \mbox{w.h.p.}, \; \mbox{and} \nonumber \\
& \| \bRdagpimn \|_{\infty} \; \lesssim \; \vnpi \vnbarm (\log n) \;\; \mbox{and} \;\; \| \bRpimn^{**} \|_{\infty} \; \lesssim \; \vnpi \sqrt{\log n} \;\; \mbox{w.h.p.} \nonumber
\end{align}

Combining the results over the two cases, under a general setting allowing for misspecification of either $\pihat(\cdot)$ or  $\mhat(\cdot)$, the terms in \eqref{decomp:eqn} therefore satisfy:
\begin{align}
& \| \bTzeron\|_{\infty} + \|\bTpin\|_{\infty} + \|\bTmn \|_{\infty} \; \lesssim \; \sqrt{\frac{\log d}{n}}\{1 + 1_{(\pi^*, m^*) \neq (\pi, m)} + o(1)\} \label{dr:final:rates} \\
& \mbox{and} \;\; \| \bRpimn \|_{\infty} \; \lesssim \; \{\vnpi 1_{(m^* \neq m)} + \vnbarm  1_{(\pi^* \neq \pi)}\} \sqrt{\log n} + \vnpi \vnbarm (\log n). \nonumber
\end{align}
Hence, even under possible misspecification of one of the nuisance function estimators, $\|\bT_n\|_{\infty}$ is certainly $o_{\P}(1)$ and thus double robust (in terms of consistency). %Hence, and hence double robust (in terms of consistency).
Consequently, $\bthetahatDR$ is also double robust (in terms of consistency)  in the light of Lemma \ref{DEV:BOUND} for an appropriately chosen $\lambda_n \geq 2 \| \bT_n \|_{\infty} = o_{\P}(1)$ as long as the corresponding the deviation bounds in \eqref{det:bound} involving $\sqrt{s \lambda_n}$ (for $L_2$ consistency) and $s\sqrt{\lambda_n}$ (for $L_1$ consistency) are assumed to be $o(1)$.

It is important to note from \eqref{dr:final:rates} that under the misspecification of either $\pihat(\cdot)$ or  $\mhat(\cdot)$, at least one among $\|\bTpin\|_{\infty}$ and $\|\bTmn \|_{\infty}$ is no longer a lower order term, but instead contributes an extra term of order $\sqrt{(\log d)/n}$, same as the main term $\bTzeron$, while the other one stays to be of lower order. More importantly, however, the behavior of the %second order
product-type bias (or `drift') term $\bRpimn$ changes dramatically! From being a lower order term involving the products of the rates of $\pihat(\cdot)$ and $\mhat(\cdot)$, it now involves the individual rates themselves appearing as leading order terms in a complementary manner, i.e. $\vnbarm$ appears if $\pihat(\cdot)$ is misspecified and $\vnpi$ appears if $\mhat(\cdot)$ is misspecified. This is mainly due to the unavoidable appearance of the additional terms $\bRpimn^*$ or $\bRpimn^{**}$, and their control inevitably requires use of the first order properties and rates of $\{\pihat(\cdot), \mhat(\cdot)\}$. In general, these rates are not necessarily of faster (or even same) order than $\sqrt{(\log d)/n}$. In fact, they are quite likely to be slower in most cases, especially if $\pihat(\cdot)$ and/or $\mhat(\cdot)$ are obtained based on non/semi-parametric models or high dimensional parametric models, in all of which cases the convergence rates are typically slower than $\sqrt{(\log d)/n}$.

Hence, under misspecification of $\pihat(\cdot)$ or $\mhat(\cdot)$, the $L_2$ convergence rate of $\bthetahatDR$ is likely to be slower than the usual benchmark rate of $\sqrt{s (\log d)/n}$. To achieve estimators with faster rates, one needs to carefully incorporate further bias corrections while constructing the estimator itself, given a choice of $\{\pihat(\cdot),\mhat(\cdot)\}$. This is quite a challenging problem in high dimensional settings, even for the simple case of mean (or ATE) estimation and with $\{\pihat(\cdot),\mhat(\cdot)\}$ obtained using standard high dimensional sparse parametric models. This case has been considered only recently by \citet{Avagyan_ATE-Mispecified_2017} and \citet{Andrea_Unified_2019}, where the methods and the associated analyses are evidently quite involved. We refer the interested reader to these papers for further insights on the problem and the ensuing challenges and nuances. However, given the scope of this paper, we do not delve further into such analyses for brevity, especially since in our case, the parameter is also high dimensional which leads to further complexity. % to the problem.
Nevertheless, we do empirically investigate in detail and validate the double robustness of $\bthetahatDR$ and $\bthetatilDR$ in our simulation studies; see Appendix \ref{Sim. double-robustness} %\ref{sec:sim}
for the results.

%\subsection{Double-robustness of the estimator and the performance of complete case estimator} \label{Sim. double-robustness}

\section{Results on Nuisance Function Estimators}\label{SEC:NUISANCE:SEPARATE:SUPP}

%\subsection{Convergence Rates for the `Extended' Parametric Families}\label{parametricfamilies:convergencerates}
\subsection{Convergence Rates for `Extended' Parametric Families}\label{parametricfamilies:convergencerates}
\hspace{-0.07in}We establish here tail bounds and convergence rates for estimators based on the `extended' parametric families discussed in Sections \ref{nuisance:pi}-\ref{nuisance:m}. For notational simplicity, we derive the results for a general outcome which may be assigned to be $T$ for estimation of $\pi(\cdot)$, or $TY$ for estimation of $m(\cdot)$. Let $(Z,\bX)$ denote a generic random vector where $Z \in \R$ and $\bX \in \R^p$ with support $\Xsc \subseteq \R^p$. %Let $Z \in \R$ be a generic random variable and $\bX \in \R^p$ be a random vector of covariates with support $\Xsc \subseteq \R^p$.
Consider an `extended' parametric family of (working) models for estimating $\E(Z \medgiven \bX)$ given by: $g\{\bbeta'\bPsi(\bX)\}$ where $\bPsi(\bX) \in \R^K$ is some vector of basis functions. Let $\bbeta_0$ denote the `target' parameter corresponding to this working model and let $\bbetahat$ be \emph{any} estimator of $\bbeta_0$ based on any suitable procedure applied to the observed data: $\{Z_i, \bX_i\}_{i=1}^n$. Then, we estimate $\E(Z \medgiven \bX = \bx)$ based on the working model as: $g\{\bbetahat'\bPsi(\bx)\}$. The result below establishes a tail bound for this estimator w.r.t. its target $g\{\bbeta_0'\bPsi(\bx)\}$.
\begin{theorem}\label{parametric:thm}
Suppose $\bbetahat$ satisfies a basic high-level $L_1$ error guarantee: %$\|\bbetahat - \bbeta_0 \|_1 \leq a_n$ with probability at least $1- q_n$ for some sequences $a_n \geq 0$ and $q_n \in [0,1]$ such that $a_n, q_n = o(1)$.
\begin{equation*}
\P (\| \bbetahat - \bbeta \|_1 > a_n) \; \leq \; q_n \;\; \mbox{for some} \;\; a_n, q_n = o(1), \;\; a_n \geq 0, \; q_n \in [0,1].
\end{equation*}
Suppose further that $g(\cdot)$ is Lipschitz continuous with $|g(u) - g(v)| \leq C_g |u-v|$ $\forall \; u,v \in \R$ and that $\bPsi(\bX)$ is uniformly bounded, i.e. $\max_{1 \leq j \leq K} | \bPsi_{[j]} (\bX)| \leq C_{\bPsi} < \infty $ a.s. $[\P]$, for some constants $C_g,  C_{\bPsi} \geq 0$.
%$\max_{1 \leq j \leq K} | \bPsi_{[j]} (\bX)| \leq C_{\bPsi} $ almost surely (a.s.) $[\P]$, for some constants $C_g,  C_{\bPsi}$.
Then, for any $t \geq 0$,
\begin{equation*}
\P \left[ \sup_{\bx \in \Xsc} | g\{\bbetahat'\bPsi(\bx) \} - g\{\bbeta_0'\bPsi(\bx) \}| \; > \; (\sqrt{2}C_g C_{\bPsi}) a_n t \right] \; \leq \; 2 \exp(-t^2) + q_n.
\end{equation*}
\end{theorem}
\noindent Theorem \ref{parametric:thm} establishes a bound for the supremum which is much stronger than what we need to verify our basic assumptions. Nevertheless, as a consequence, it establishes that when one uses any of these `extended' parametric families for constructing $\{\pihat(\cdot), \mhat(\cdot)\}$, then the pointwise tail bounds required in our basic Assumptions \ref{tpicont:assmpn}-\ref{tmcont:assmpn} hold with the choices of $\{\vnpi, \vnm\}$ $\propto a_n$
and $\{q_{n,\pi}, q_{n,m} \} \propto q_n$. Further, as discussed in Sections \ref{nuisance:pi} and \ref{nuisance:m}, for most common choices of $\bbetahat$ based on penalized estimators from high dimensional models, the $L_1$ error rate $a_n$ should behave as: $a_n \propto s_{\bbeta_0} \sqrt{(\log {K})/n}$ w.h.p.

\subsection[High Dimensional Single Index KS: Non-Asymptotic Bounds and Rates]{High Dimensional Single Index Models: Non-Asymptotic Bounds and Rates for KS over Estimated Index Parameters}\label{sec:KS}

%In this section,
Here, we study the properties of single index KS estimators involving high dimensional covariates with the index parameter being (possibly) unknown and estimated. The underlying high dimensionality and the non-ignorable index estimation error makes the analyses nuanced and different from most existing results in the literature under classical settings. We consider both linear kernel average estimators (e.g. density estimators) as well as ratio form estimators (e.g. conditional mean estimators) and develop a non-asymptotic theory that establishes concrete tail bounds and pointwise convergence rates for such estimators. The results apply equally to both classical and high dimensional regimes, and while obtained in course of characterizing our nuisance function estimators' properties, may also be useful in other applications and should be of independent interest. We therefore
present the results under a generic framework and a set of notations independent of the main paper. %that is  independent of the rest of the paper.

%Let $Z \in \R$ be a generic random variable and $\bX \in \R^p$ be a random vector of covariates with support $\Xsc \subseteq \R^p$. We assume that $(Z,\bX)$ has finite $2^{nd}$ moments. Let $\{(Z_i, \bX_i): i = 1, \hdots, n\}$ be a sample of $n \geq 2$ i.i.d. realizations of $(Z,\bX)$, where we note that $p \geq 1$ is allowed to be high dimensional w.r.t.  $n$, i.e. $p$ is allowed to diverge with $n$. Let $\bbeta \in \R^p$ be any (unknown) parameter vector of interest and let $\bbetahat$ denote \emph{any} reasonable estimator of $\bbeta$ that satisfies a basic high-level guarantee (to be made precise shortly) in terms of an $L_1$ norm bound given by: $\| \bbetahat - \bbeta \|_1 \leq a_n$ w.h.p. for some sequence $a_n = o(1)$.

Let $\{(Z_i, \bX_i): i = 1, \hdots, n\}$ denote a sample of $n \geq 2$ i.i.d. realizations of a generic random vector $(Z,\bX)$ assumed to have finite $2^{nd}$ moments, where $Z \in \R$, $\bX \in \R^p$ with support $\Xsc \subseteq \R^p$ and $p \geq 1$ is allowed to be high dimensional compared to the sample size, i.e. $p$ is allowed to diverge with $n$.

Let $\bbeta \in \R^p$ be any (unknown) `parameter' of interest and let $\bbetahat$ denote \emph{any} reasonable estimator of $\bbeta$ that satisfies a basic high-level $L_1$ error guarantee:
\begin{equation}
\P (\| \bbetahat - \bbeta \|_1 > a_n) \; \leq \; q_n \;\; \mbox{for some} \;\; a_n, q_n = o(1), \;\; a_n \geq 0, \; q_n \in [0,1]. \label{eq:index:L1bound}
\end{equation}
\eqref{eq:index:L1bound} is a reasonable high-level requirement that should hold in most cases. It is important to note that \eqref{eq:index:L1bound} is the \emph{only} condition we require on $\{ \bbeta, \bbetahat\}$ for all our results and nothing specific regarding their construction or properties.
\par\smallskip
Let $W \equiv W_{\bbeta} := \bbeta'\bX$ and $\What := \bbetahat'\bX$. %$W_i \equiv W_{\bbeta,i} := \bbeta'\bX_i$ and $\What_i := \bbetahat'\bX_i$ for $ i = 1, \hdots, n$, and
For any $\bx \in \R^p$, let $w_{\bx} \equiv w_{\bx,\bbeta} := \bbeta'\bx$ and $\what_{x} := \bbetahat'\bx$. For any $w \in \R$, let $m_{\bbeta}(w) := \E(Z \given W = w)$ and $l_{\bbeta}(w) := m_{\bbeta}(w)f_{\bbeta}(w)$, where $f_{\bbeta}(\cdot)$ denotes the density of $W \equiv \bbeta'\bX$. Finally, for any $\bx\in \Xsc$, let $m (\bbeta, \bx) := m_{\bbeta}(\bbeta'\bx)$, $f (\bbeta, \bx) := f_{\bbeta}(\bbeta'\bx)$ and $l (\bbeta, \bx) := l_{\bbeta}(\bbeta'\bx)$.
\par\smallskip
Given \emph{any} estimator $\bbetahat$ of $\bbeta$ satisfying \eqref{eq:index:L1bound}, consider the following single index KS estimators of $l(\bbeta,\bx)$, $f(\bbeta,\bx)$ and $m(\bbeta,\bx)$ for any \emph{fixed} $\bx \in \Xsc$,
\begin{align*}
& \lhat(\bbetahat, \bx) \; := \; \frac{1}{nh} \sum_{i=1}^n Z_i K\left( \frac{\bbetahat'\bX_i - \bbetahat'\bx}{h} \right) \; \equiv \; \frac{1}{nh} \sum_{i=1}^n Z_i K\left( \frac{\What_i - \whatbx}{h} \right), \nonumber \\
& \fhat(\bbetahat, \bx) \; := \; \frac{1}{nh} \sum_{i=1}^n K\left( \frac{\bbetahat'\bX_i - \bbetahat'\bx}{h} \right)  \quad \mbox{and} \quad \mhat(\bbetahat, \bx) \; := \; \frac{\lhat(\bbetahat, \bx)}{\fhat(\bbetahat, \bx)}, \nonumber  %\quad \equiv \; \frac{1}{nh} \sum_{i=1}^n K\left( \frac{\What_i - \whatbx}{h} \right), \;\; \mbox{and} \nonumber
\end{align*}
where $K(\cdot): \R \rightarrow \R$ denotes any suitable kernel function (e.g. the Gaussian kernel) and $h \equiv h_n > 0$ denotes the bandwidth sequence with $h_n = o(1)$.

%Note that $\fhat(\cdot)$ is a special case of $\lhat(\cdot)$ with $Z \equiv 1$, and
$\lhat(\cdot)$ and $\fhat(\cdot)$ are both linear kernel average (LKA) estimators while $\mhat(\cdot)$ is a ratio type %(Nadaraya-Watson)
KS estimator. We obtain non-asymptotic tail bounds and (pointwise) convergence rates for these estimators in Theorems \ref{KS:mainthm1}-\ref{KS:mainthm2} below. The Assumptions \ref{KS:assmpn1}-\ref{KS:assmpn2} for these results are given separately in Appendix \ref{sec:KS:assmpns}.
%We mainly focus on the LKA estimator $\lhat(\cdot)$. $\fhat(\cdot)$ is a special case (with $Z \equiv 1$). These can be then combined to obtain results for the ratio estimator $\mhat(\cdot)$. We summarize our assumptions first.

\begin{theorem}[Tail bounds for LKA estimators]\label{KS:mainthm1}
\hspace{-0.1in} Consider the estimator $\lhat(\bbetahat, \bx)$ of $l(\bbeta, \bx)$. Assume \eqref{eq:index:L1bound} and Assumptions \ref{KS:assmpn1}-\ref{KS:assmpn2} (in Appendix \ref{sec:KS:assmpns})
and that $h = o(1)$, $\log (np)/(nh) = o(1)$ and $(a_n/h)\sqrt{\log p} = o(1)$. Then, for any fixed $\bx \in \Xsc$ and any $t \geq 0$, with probability at least $1 - 9 \exp(-t^2) - 2 q_n$,
\begin{equation*}
|l(\bbetahat, \bx) - l(\bbeta, \bx)|  \hspace{0.028in} \leq \hspace{0.028in} C_1\left(\frac{t+1}{\sqrt{nh}} + \frac{t^2 \sqrt{\log n}}{nh}\right) + C_2 \left(h^2 + a_n + \frac{a_n^2}{h^2} + \frac{\log (np)}{nh}\right)
\end{equation*}
for some constants $C_1, C_2 > 0$ depending only on those in the assumptions.

%Further, a similar tail bound also holds at the training points $\{\bX_i\}_{i=1}^n$. For any $1 \leq i \leq n$ and any $t \geq 0$, with probability $\geq 1 - 11 \exp(-t^2) - 2 q_n$,
%\begin{equation*}
%|\lhat(\bbetahat, \bX_i) - l(\bbeta, \bX_i)|  \leq  C_1^*\left(\frac{t+1}{\sqrt{nh}} + \frac{t^2 \sqrt{\log n}}{nh}\right) + C_2^* \left(h^2 + a_n + \frac{a_n^2}{h^2} + \frac{\log np}{nh}\right)
%\end{equation*}
%for some constants $C_1^*, C_2^* > 0$ depending only on those in the assumptions.
\end{theorem}
Apart from an explicit tail bound, Theorem \ref{KS:mainthm1} also establishes the convergence rate of $\lhat(\bbetahat,\bx)$ to be $O({nh}^{-\half} + h^2 + a_n + a_n^2 h^{-2})$ which quantifies the additional price one pays for estimating the high dimensional index parameter $\bbeta$ apart from the error rate of a standard one dimensional KS. This is highlighted through all the terms in the bound involving the $L_1$ error rate $a_n$ of $\bbetahat$. For a given $a_n$, one can also optimize the choice of $h = O(n^{-a})$ over $a > 0$  by minimizing the convergence rate above whose terms behave differently with $h$, similar to a variance-bias tradeoff phenomenon typically observed in KS regression. We skip these technical discussions here for brevity.

\begin{theorem}[Tail bounds for ratio type KS estimators]\label{KS:mainthm2}
 Consider the ratio type KS estimator $\mhat(\bbetahat,\bx)$ of $m(\bbeta, \bx)$ and assume that $|m(\bbeta,\bx)| \leq \delta_m$ and $f(\bbeta,\bx) \geq \delta_f > 0$ for some constants $\delta_m, \delta_f > 0$. For any $t \geq 0$, define:
$$
\epsilon_n(t) := C_1\frac{t+1}{\sqrt{nh}} + C_2 \frac{t^2 \sqrt{\log n}}{nh} + C_3 b_n, \;\; \mbox{where} \;\; b_n := h^2 + a_n + \frac{a_n^2}{h^2} + \frac{\log (np)}{nh}
$$ and $C_1, C_2, C_3 > 0$ are the same as in Theorem \ref{KS:mainthm1}. Assume \eqref{eq:index:L1bound}, Assumptions \ref{KS:assmpn1}-\ref{KS:assmpn2} (in Appendix \ref{sec:KS:assmpns})
and that
$h = o(1)$, $\log (np)/(nh) = o(1)$, $(a_n/h)\sqrt{\log p} = o(1)$ and $b_n = o(1)$.
Then, for any fixed $\bx \in \Xsc$ and any $t, t_* \geq 0$ with $t_*$ further assumed w.l.o.g. to satisfy $\epsilon_n(t_*) \leq \delta_f/2 < \delta_f$, we have: with probability at least $1 - 18 \exp(-t^2) - 9 \exp(-t_*^2) - 6 q_n$,
\begin{align*}
& |\mhat(\bbetahat, \bx) - m(\bbeta,\bx) | \; \leq \; \frac{2(1 + \delta_m)}{\delta_f} \epsilon_n(t)  \; \lesssim \;  \frac{t+1}{\sqrt{nh}} + \frac{t^2 \sqrt{\log n}}{nh} + b_n,
\end{align*}
where `$\lesssim$' denotes inequality upto multiplicative contants (possibly depending on those introduced in the assumptions). In particular, assuming further that $\{\log (np) \log n\}/(nh) = o(1)$ and choosing $t = t_* = c\sqrt{\log np}$ for any $c > 0$ (assuming w.l.o.g. the chosen $t_*$ satisfies the required condition), we have:
\begin{align*}
& |\mhat(\bbetahat, \bx) - m(\bbeta,\bx) | \; \lesssim \; (c+1)\sqrt{\frac{\log (np)}{nh}} \left( 1 + c \sqrt{\frac{\log(np) \log n}{nh}} \right)  + b_n \\
& \;\; \lesssim \; c \sqrt{\frac{\log (np)}{nh}} + b_n \;\;\; \mbox{with probability at least} \;\; 1 - 27 (np)^{-{c^2}} - 6q_n.
\end{align*}

\end{theorem}
%% The commented part below is the version of the verification remark if we also include SIM as a choice for estimating the propensity score. %%
%
%As a consequence, Theorem \ref{KS:mainthm2} verifies our basic Assumptions \ref{tpicont:assmpn} and \ref{tmcont:assmpn} regarding $\pihat(\cdot)$ and $\mhat(\cdot)$ when one chooses to estimate these based on single index models. In particular, it establishes that the tail bounds in Assumptions \ref{tpicont:assmpn} and \ref{tmcont:assmpn} hold with the choices $\{\vnpi, \vnm\} \propto v_n$, $\{\bnpi, \bnm\} \propto b_n$ and $\{ q_{n,\pi}, q_{n,m} \propto \exp(-t_*^2) + q_n$  with $v_n$, $b_n$, $q_n$ and $t_*$ as in Theorem \ref{KS:mainthm2} above and $t_*$ can be chosen as large as possible as long as it still satisfies the required conditions. Finally, as discussed in Sections \ref{nuisance:pi} and \ref{nuisance:m}, for most common choices of the estimator $\bbetahat$, the $L_1$ error rate $a_n$ is expected to behave as: $a_n \propto s_{\bbeta_0} \sqrt{(\log {p})/n}$ w.h.p..

Theorem \ref{KS:mainthm2} establishes explicit tail bounds and convergence rates for the ratio-type KS estimator $\mhat(\bbetahat,\bx)$. As a consequence, it also verifies our basic Assumption \ref{tmcont:assmpn} regarding $\mhat(\cdot)$ when one chooses to estimate it using SIMs. In particular, in view of Remark \ref{assmpns:rem:extra}, it establishes that the tail bound \eqref{m:tailbound1} %in Assumption \ref{tmcont:assmpn}
holds with the choices $\vnm \propto \sqrt{\log (np)/(nh)} + b_n$ and $q_{n,m} \propto (np)^{-c} + q_n$,  for some $c > 0$, with $b_n$ and $q_n$ as above. Finally, as discussed in Sections \ref{nuisance:pi} and \ref{nuisance:m}, for most common choices of the estimator $\bbetahat$, the $L_1$ error rate $a_n$ is expected to behave as: $a_n \propto s_{\bbeta} \sqrt{(\log {p})/n}$ w.h.p., where $s_{\bbeta} := \|\bbeta \|_0$.

\subsection{Assumptions for Theorems \ref{KS:mainthm1} and \ref{KS:mainthm2}}\label{sec:KS:assmpns}
We summarize here the smoothness and regularity assumptions required for Theorems \ref{KS:mainthm1}-\ref{KS:mainthm2}.

\begin{assumption}[Standard smoothness assumptions and conditions on $K(\cdot)$ and the tail behavior of $Z$]
\label{KS:assmpn1}
\emph{We assume the following conditions.}
\begin{enumerate}[(a)]
\item \emph{$Z$ is sub-Gaussian with $\psitwonorm{Z} \leq \sigma_Z$ for some constant $\sigma_Z \geq 0$. %$\psialphanorm{Z} \leq \sigma_Z$ for some $\alpha \geq 1$ and some constant $\sigma_Z \geq 0$, with the $\psitwonorm{\cdot}$ norm being as in Definition \ref{orlicz:def}. Special cases include sub-exponential ($\alpha = 1$), sub-Gaussian ($\alpha  = 2$) and bounded ($\alpha = \infty$) random variables.
    }
\par\smallskip
\item \emph{$K(\cdot)$ is bounded and integrable with $\| K(\cdot)\|_{\infty} \leq M_K$ and $\int_{\R} |K(u)| du \leq C_K$ for some constants $M_K, C_K \geq 0$.}
\par\smallskip
\item \emph{Let $m_{\bbeta}^{(2)}(w) := \E\{ Z^2 \given \bbeta'\bX = w\}$ for any $w \in \R$. Then, $m_{\bbeta}^{(2)}(w)f_{\bbeta}(w)$ is bounded in $w \in \R$ and $\| m_{\bbeta}^{(2)}(\cdot)f_{\bbeta}(\cdot)\|_{\infty} \leq B_1$ for some constant $B_1 \geq 0$.}
\par\smallskip
\item \emph{$K(\cdot)$ is a second order kernel satisfying: $\int_{\R} K(u) d(u) = 1$, $\int_{R} u K(u) du = 0$ and $\int_{\R} u^2 | K(u) | du \leq R_K < \infty$ for some constant $R_K \geq 0$.
%\item
$l_{\bbeta}(\cdot) \equiv m_{\bbeta}(\cdot) f_{\bbeta}(\cdot)$ is twice continuously differentiable with bounded second derivatives $l_{\bbeta}^{''}(\cdot)$ satisfying: $\| l_{\bbeta}^{''}(\cdot)\|_{\infty} \leq B_2 $ for some constant $B_2 \geq 0$.}
\end{enumerate}
%}
\end{assumption}

\begin{assumption}[Further conditions on $K(\cdot)$ and other assumptions to account for the estimation error of $\bbeta$]\label{KS:assmpn2}
\emph{
We also assume the following.
\begin{enumerate}[(a)]
\item $K(\cdot)$ is continuously differentiable with a bounded and integrable derivative $K'(\cdot)$ satisfying $\| K'(\cdot) \|_{\infty} \leq M_{K'}$ and $\int_{\R} | K'(u) | du \leq C_{K'}$ for some constants $M_{K'}, C_{K'} \geq 0$. Further, $K(u) \rightarrow 0$ as $u \rightarrow \infty$ or $u \rightarrow - \infty$.
    \vspace{0.05in}
\item Let $\boldsymbol{\eta}_{\bbeta}(w) := \E(Z \bX \given \bbeta'\bX = w) f_{\bbeta}(w)$ for any $w \in \R$, and let $\boldsymbol{\eta}_{\bbeta [j]} (\cdot)$ denote the $j^{th}$ coordinate of $\boldsymbol{\eta}_{\bbeta}(\cdot)$ for $j = 1, \hdots, d$. Then, for each $j$, $\boldsymbol{\eta}_{\bbeta [j]} (\cdot)$ is continuously differentiable with derivative $\boldsymbol{\eta}_{\bbeta [j]}' (\cdot)$ that is bounded uniformly in $j = 1,\hdots, d$. Further, $l_{\bbeta}(\cdot)$ is also continuously differentiable with a bounded derivative $l_{\bbeta}'(\cdot)$. Thus, $\max_{1 \leq j \leq d} \| \boldsymbol{\eta}_{\bbeta [j]}' (\cdot) \|_{\infty} \leq B^*_1$ and $\| l_{\bbeta}'(\cdot)\|_{\infty} \leq B^*_2$ for some constants $B^*_1, B^*_2 \geq 0$.
\vspace{0.05in}
\item $K'(\cdot)$ satisfies a `local' Lipschitz property as follows. There exists a constant $L > 0$ such that for all $u, v \in \R$ with $| u - v| \leq L$, $|K'(u) - K'(v)| \leq \varphi(u) |u-v| $ for some bounded and integrable function $\varphi(\cdot): \R \rightarrow \R^+$ with $\| \varphi(\cdot) \|_{\infty} \leq M_{\varphi}$ and $\int_{\R} \varphi(u) du \leq C_{\varphi}$ for some constants $M_{\varphi}, C_{\varphi} \geq 0$.
\vspace{0.05in}
\item $\bX$ is bounded, i.e. $\|\bX \|_{\infty} \leq M_{\bX}$ a.s. $[\P]$ for some constant $M_{\bX} \geq 0$, and $\bbetahat$ satisfies the high-level guarantee \eqref{eq:index:L1bound}. %: $\| \bbetahat - \bbeta\|_1 \leq a_n$ with probability $\geq 1 - q_n$, for some $a_n, q_n \geq 0$ with $a_n = o(1)$ and $q_n = o(1)$.
    Further, we assume $a_n/h = o(1)$ and $2M_{\bX}(a_n / h) \leq L$, where $L$ is as in (c) above and $a_n$ is as in \eqref{eq:index:L1bound}.
\end{enumerate}
}
\end{assumption}
Most of the smoothness assumptions and the conditions on $K(\cdot)$ in Assumptions \ref{KS:assmpn1} and \ref{KS:assmpn2} are fairly mild and standard in the non-parametric statistics literature. Similar or equivalent versions of these assumptions can be found in a variety of references including \citet{Newey_Book_1994, Andrews_1995, Masry_1996} and \citet{Hansen_2008}, among others.

Assumption \ref{KS:assmpn2} (c) imposes a `local' Lipschitz property of sorts on $K'(\cdot)$, where the Lipschitz `constant' is a bounded function that also decays quickly enough to be integrable. This is satisfied by the Gaussian kernel in particular. In general, it holds for any $K(\cdot)$ where $K'(\cdot)$ has a compact support and is Lipschitz continuous, or $K'(\cdot)$ is differentiable with a bounded derivative $K''(\cdot)$ that has a polynomially integrable tail, i.e. $|K''(u)| \leq |u|^{-\rho}$ for some $\rho > 1$ and all $u \in \R$ such that $|u| > L^*$ for some $L^* > 0$ (see \citet{Hansen_2008}).

Finally, the boundedness assumption on $\bX$ is mostly for the simplicity of our exposition. With appropriate modifications in the proofs, this can be relaxed to allow for more general tail behaviors of $\bX$ (e.g. $\bX$ is sub-Gaussian), although the corresponding technical analyses can be more involved. %We therefore stick to a bounded $\bX$ to avoid such technicalities.

\section{Supplementary Numerical Results}\label{sec:sim:supp}

\subsection{Simulation Setting: Technical Details} \label{Simulation: Tech. detail}
%\tcr{**To be edited.}
We summarize here a few relevant details regarding our simulation studies, including in particular, the parameter choices for all the DGPs, along with other technical details of the implementations. The DGP parameters are specified as follows. %In this section we provide more technical details of the simulation. The parameters in the DGPs are specified as the following:
\begin{enumerate}[(a)]
\item For the DGPs of $T|\bX$, we set $\alpha_0 = 0.5$, and chose $\balpha$ and $\balpha^{*}$ as follows.
 %Choices of $\balpha$, $\balpha^{*}$ and $\alpha_{0}$:
\begin{enumerate}[(i)]
\item When $p=50$, we set $\|\balpha\|_0  = 5$ and $\|\balpha^{*}\|_0 = 2$ with:
\begin{align*}
&\balpha  = 1/ \sqrt{5}(1,-1,0.5,-0.5,0.5,\mathbf{0}_{p-5}), \\
& \balpha^{*}  = (0.25,-0.25,\mathbf{0}_{p-2}).
\end{align*}
\item When $p =500$, we set $\|\balpha\|_0  = 10$ and $\|\balpha^{*}\|_0 = 4$ with:
\begin{align*}
 &\balpha = 1/ \sqrt{10}(\mathbf{1}_3, \mathbf{-1}_2, \mathbf{0.5}_2,\mathbf{-0.5}_3,\mathbf{0}_{p-10}), \\
 &\balpha^{*} = (\mathbf{0.25}_2,\mathbf{-0.25}_2,\mathbf{0}_{p-4}).
\end{align*}
\end{enumerate}
Note that in both sets of choices, $\balpha$ is normalized by $\sqrt{\|\balpha\|_0}$ to ensure that the likelihood of $\pi(\bX)$ being too close to 0 or 1 is small, in practice. %and enables the validity of the positivity assumption.

\par\smallskip
\item For the DGPs of $Y|\bX$, we set $\gamma_0 = 1$, and chose $\bgamma$ and $\bgamma^{*}$  as follows. %Choices of $\bgamma$, $\bgamma^{*}$ and $\gamma_{0}$:
\begin{enumerate}[(i)]
\item When $p=50$, we set $\|\bgamma\|_0  = 10$ and $\|\bgamma^{*}\|_0 = 5$ with:
\begin{align*}
& \bgamma =(\mathbf{1}_3,\mathbf{-1}_2,\mathbf{0.5}_2,\mathbf{-0.5}_3,\mathbf{0}_{p-10}), \\
 &\bgamma^{*} = (1,-1,0.5,0.5,-0.5,\mathbf{0}_{p-5}).
\end{align*}
\item When $p =500$, we set $\|\bgamma\|_0  = 20$, $\|\bgamma^{*}\|_0 = 5$ with:
\begin{align*}
 &\bgamma =(\mathbf{1}_3,\mathbf{-1}_2,\mathbf{0.5}_5,\mathbf{-0.5}_5,\mathbf{0.25}_2,\mathbf{-0.25}_3, \mathbf{0}_{p-20}),\\
 &\bgamma^{*} = (1,-1,0.5,0.5,-0.5,\mathbf{0}_{p-5}).
\end{align*}
\end{enumerate}
\end{enumerate}
In addition, for the SIM DGPs, we set $c_T = 0.2$ and $c_Y = 0.3/\sqrt{\lambda_{max}(\bSigma_p)}$, %and the intercepts $\bgamma_{0}$ and $\balpha_{0}$ are set to be $1$ and $0.5$.
where $\lambda_{max}(\bSigma_p)$ is the largest eigenvalue of the matrix $\bSigma_p$. Throughout, in the above, we have used the notation $\mathbf{a}_d := (a,a,\ldots,a) \in \R^d$ for any $a \in \R$. %$\mathbf{a}_d := ( \underbrace{a,a, \ldots,a}_{d})$.  %The parameters in the DGPs for $T | \bX$ are normalized by $\sqrt{\|\balpha\|_0}$ so that the proportion of $\pi(\bX)$ that is close to 0 or 1 is small.
Lastly, for implementing $\bthetatilDR$ and the associated CIs, we choose $\bOmegahat$ as $\bSigmahat^{-1}$ when $p \ll n$, or as the Nodewise Lasso estimator otherwise (see Section \ref{sec:inference}).

The sample splitting and cross-fitting required for our estimator was performed with $\K = 2$ folds. %The number of folds in the sample splitting is 2.
The tuning parameter for any penalized logistic regression involved in obtaining $\pihat(\cdot)$ was chosen via minimizing the Bayes Information Criteria (BIC), and %the tuning parameters
  that for any penalized linear regression involved in obtaining $\mhat(\cdot)$ was chosen via 10-fold least squares cross validation. %using 10-fold cross validation with minimizing mean squared errors (MSE) as criterion.
  The bandwidth in the kernel smoothing required for fitting any SIM was chosen based on %the nonparametric regression
  %for SIM is chosen using
  least square cross-validation, as suggested in the \texttt{`np'} package in \texttt{R}. All codes were implemented in \texttt{R} and are available upon request. %\par

\subsection{Investigating Double Robustness of the DDR Estimator and Performance of the Complete Case Estimator}\label{Sim. double-robustness}
%\tcr{**To be edited.}
We present here a large sample analysis of one of our simulation settings, with $n = 50000$, $p = 50$ or $500$, and the true DGP for $\{\pi(\cdot), m(\cdot)\}$ chosen to be ``quad-quad'' for illustration. We study the asymptotic properties of our estimators $\bthetahatDR$ and $\bthetatilDR$, specifically, their DR properties (for both estimation and inference), whereby they should remain consistent when at least one of the two nuisance estimators $\pihat(\cdot)$ and $\mhat(\cdot)$ is correctly specified, but \emph{not} necessarily both. %To investigate whether our proposed estimator has the desired double-robustness property, we study a large sample setting where $n=50000$, $p=50$ and $500$.
%Consequently, when either the propensity score $\pi(\cdot)$ or the conditional mean $m(\cdot)$ is correctly specified, the final estimators should be consistent.
In addition, apart from the oracle estimators $\bthetahat_{orac}$ and $\bthetahat_{full}$, we also implement the complete case (CC) estimator, $\bthetahat_{cc}$, obtained via a simple Lasso of $Y$ vs. $\bX$ in the complete case data (i.e. samples with $T = 1$), in order to investigate its estimation performance. %of the aside from the oracle and super oracle estimators, we also consider the complete case estimator $\widehat{\btheta}_{cc}$, which is an estimator obtained by using only the complete data (samples with $T = 1$).
This estimator is expected to be consistent \emph{only} when the true DGP for $Y|X$ is linear %is ``linear-linear''
which, by choice, is not the case here.
\begin{table}[!htbp]
	\centering
	\caption{A large sample analysis of the performance of all estimators with $n = 50000$, $\bSigma_p = I_p$, DGP for $\{\pi(\cdot), m(\cdot)\}$ = ``Quad-quad'', and using various combinations of the nuisance estimators $\{\pihat(\cdot), \mhat(\cdot)\}$.
{\it Table (a):} Comparison of the $L_2$ errors of $\bthetahatDR$, $\widehat{\btheta}_{orac}$, $\widehat{\btheta}_{full}$ and the CC estimator $\bthetahat_{cc}$.
{\it Table (b):} Average ($A$-CovP) and median ($M$-CovP) of the empirical CovPs for the 95\% CIs of $\btheta_0$ obtained via $\bthetatilDR$, as well as the average lengths of these CIs, all reported separately for the truly zero and non-zero coefficients of $\btheta_0$.} %Average coverage probabilities and lengths of the CIs built upon the desparsified estimator. We report the means and medians together with standard errors and MADs as subscripts. The reported values are separated into truly zero and non-zero coefficients.}
%Average $L_2$ errors of $\bthetahatDR$, obtained via various combinations of the nuisance estimators $\{\pihat(\cdot), \mhat(\cdot)\}$, and those of the oracle  estimators $\widehat{\btheta}_{orac}$ and $\widehat{\btheta}_{full}$, for $n = 1000$, $\bSigma_p = I_p$ and all three choices of the \emph{true} DGPs.
%Average ($A$-CovP) and median ($M$-CovP) of the empirical coverage probabilities (CovPs) for the (coordinatewise) 95\% CIs of $\btheta_0$ obtained via $\bthetatilDR$ (based on various combinations of the nuisance estimators $\{\pihat(\cdot), \mhat(\cdot)\}$) for $n = 1000$, $\bSigma_p = I_p$ and all three choices of the \emph{true} DGPs. Shown also are the corresponding average length of these CIs. All values are reported separately for the truly zero and non-zero coefficients of $\btheta_0$ (see Section \ref{Est. Implemented}).
	\label{table: DGP_id_n50000_p50}
	{\bf (I)} $p=50$. \\
	(a) Comparison of $L_2$ errors for the estimators.
\vspace{0.05in}
	\resizebox{\textwidth}{!}{
	\begin{tabular}{ll||cccc}\hline
		\multicolumn{2}{l}{Working nuisance model}& $\bthetahatDR$  & $ \widehat{\btheta}_{orac} $ &  $\widehat{\btheta}_{full}$  &  $\widehat{\btheta}_{cc}$  \\
		\hline
		\multirow{2}{*}{$\mhat$: linear}&$\pihat$: logit & 0.460 (0.026) &  0.072 (0.011) & 0.069 (0.01) & 0.528 (0.021)  \\
		& $\pihat$: quad & 0.204 (0.137) & 0.072 (0.011) & 0.069 (0.01) & 0.528 (0.021)   \\
		\hline
		\multirow{2}{*}{$\mhat$: quad}&$\pihat$: logit & 0.071 (0.010) &  0.072 (0.011) & 0.069 (0.01) & 0.528 (0.021)  \\
		&$\pihat$: quad & 0.072 (0.011)  &  0.072 (0.011) & 0.069 (0.01) & 0.528 (0.021) \\
		\hline
		\multirow{2}{*}{$\mhat$: SIM}&	$\pihat$: logit & 0.323 (0.019)  &  0.072 (0.011) & 0.069 (0.01) & 0.528 (0.021)  \\
		&$\pihat$: quad & 0.175 (0.079) &  0.072 (0.011) & 0.069 (0.01) & 0.528 (0.021) \\
		\hline
\end{tabular}}
	(b) Average (and median) CovPs and lengths of the CIs from $\bthetatilDR$.
\vspace{0.05in}
\resizebox{\textwidth}{!}{
	\begin{tabular}{ll||ccc|ccc}\hline
		\multicolumn{2}{l}{Working nuisance model}  & \multicolumn{3}{c}{Zero coefficients} &  \multicolumn{3}{c}{Non-zero coefficients} \\
		\hline
		\multicolumn{2}{l}{} &$A$-CovP & $M$-CovP  & Length &$A$-CovP & $M$-CovP  & Length \\
		\hline
		\multirow{2}{*}{$\mhat$: linear}&$\pihat$: logit & $0.94_{0.03}$ & $(0.95_{0.03})$ & $0.06_{0}$ & $0.68_{0.39}$ & $(0.84_{0.19})$ & $0.07_{0.02}$ \\
		& $\pihat$: quad & $0.96_{0.02}$ & $(0.96_{0.01})$ & $0.12_{0}$ %& $0.12_{0.01}$
& $0.96_{0.02}$ & $(0.96_{0.03})$ & $0.14_{0.08}$   \\
		\hline
		\multirow{2}{*}{$\mhat$: quad}&$\pihat$: logit &  $0.94_{0.03}$ & $(0.95_{0.02})$ & $0.05_{0}$ & $0.93_{0.03}$ & $(0.95_{0.01})$ & $0.05_{0.01}$   \\
		&$\pihat$: quad & $0.94_{0.03}$ & $(0.95_{0.03})$ & $0.05_{0}$ & $0.94_{0.02}$ & $(0.95_{0.01})$ & $0.05_{0.01}$  \\
		\hline
		\multirow{2}{*}{$\mhat$: SIM}&	$\pihat$: logit &  $0.94_{0.03}$ & $(0.94_{0.01})$ & $0.06_{0}$ & $0.80_{0.19}$ & $(0.88_{0.13})$ & $0.07_{0.01}$  \\
		&$\pihat$: quad & $0.95_{0.02}$ & $(0.95_{0.03})$ & $0.10_{0}$ & $0.95_{0.02}$ & $(0.95_{0.02})$ & $0.12_{0.06}$ \\
		\hline
\end{tabular}}
\end{table}

\begin{table}[!htbp]
	\centering
	\caption{See caption of Table \ref{table: DGP_id_n50000_p50}. (Only change: $p= 500$ instead of $50$) }
	\label{table: DGP_id_n50000_p500}
	{\bf (II)} $p=500$. \\
	(a) Comparison of $L_2$ errors for the estimators.
	\vspace{0.05in}
	\resizebox{\textwidth}{!}{
	\begin{tabular}{ll||cccc}\hline
		\multicolumn{2}{l}{Working nuisance model}& $\bthetahatDR$  & $ \widehat{\btheta}_{orac} $ &  $\widehat{\btheta}_{full}$  &  $\widehat{\btheta}_{cc}$  \\
		\hline
		\multirow{2}{*}{$\mhat$: linear}&$\pihat$: logit & 0.297 (0.017) & 0.178 (0.009) & 0.173 (0.007) & 0.325 (0.018) \\
		& $\pihat$: quad & 0.282 (0.113) & 0.178 (0.009) & 0.173 (0.007) & 0.325 (0.018)  \\
		\hline
		\multirow{2}{*}{$\mhat$: quad}&$\pihat$: logit &  0.177 (0.008) & 0.178 (0.009) & 0.173 (0.007) & 0.325 (0.018) \\
		&$\pihat$: quad & 0.180 (0.01) & 0.178 (0.009) & 0.173 (0.007) & 0.325 (0.018) \\
		\hline
		\multirow{2}{*}{$\mhat$: SIM}&	$\pihat$: logit &  0.407 (0.022) & 0.178 (0.009) & 0.173 (0.007) & 0.325 (0.018) \\
		&$\pihat$: quad & 0.294 (0.045) & 0.178 (0.009) & 0.173 (0.007) & 0.325 (0.018)\\
		\hline
\end{tabular}}
	(b) Average (and median) CovPs and lengths of the CIs from $\bthetatilDR$.
	\vspace{0.05in}
	\resizebox{\textwidth}{!}{
	\begin{tabular}{ll||ccc|ccc}\hline
		\multicolumn{2}{l}{Working nuisance model}  & \multicolumn{3}{c}{Zero coefficients} &  \multicolumn{3}{c}{Non-zero coefficients} \\
		\hline
		\multicolumn{2}{l}{} &$A$-CovP & $M$-CovP  & Length &$A$-CovP & $M$-CovP  & Length \\
		\hline
		\multirow{2}{*}{$\mhat$: linear}&$\pihat$: logit & $0.95_{0.02}$ & $(0.95_{0.03})$ & $0.07_{0}$ & $0.78_{0.32}$ & $(0.94_{0.04})$ & $0.07_{0.01}$ \\
		& $\pihat$: quad & $0.95_{0.02}$ & $(0.96_{0.01})$ & $0.09_{0}$ & $0.94_{0.04}$ & $(0.96_{0.03})$ & $0.10_{0.03}$   \\
		\hline
		\multirow{2}{*}{$\mhat$: quad}&$\pihat$: logit &  $0.95_{0.02}$ & $(0.95_{0.01})$ & $0.05_{0}$ & $0.94_{0.02}$ & $(0.94_{0.02})$ & $0.05_{0.01}$  \\
		&$\pihat$: quad & $0.95_{0.02}$ & $(0.95_{0.01})$ & $0.05_{0}$ & $0.94_{0.02}$ & $(0.94_{0.02})$ & $0.05_{0.01}$  \\
		\hline
		\multirow{2}{*}{$\mhat$: SIM}&	$\pihat$: logit &   $0.95_{0.02}$ & $(0.95_{0.03})$ & $0.08_{0}$ & $0.75_{0.38}$ & $(0.94_{0.05})$ & $0.09_{0.01}$  \\
		&$\pihat$: quad & $0.95_{0.02}$ & $(0.95_{0.01})$ & $0.08_{0}$ & $0.88_{0.12}$ & $(0.92_{0.04})$ & $0.09_{0.02}$ \\
		\hline
\end{tabular}}
\end{table}

The results are presented separately for $p = 50$ and $500$ in Tables \ref{table: DGP_id_n50000_p50} and \ref{table: DGP_id_n50000_p500} respectively. Tables \ref{table: DGP_id_n50000_p50}(a) and \ref{table: DGP_id_n50000_p500}(a)  summarize the $L_2$ estimation error comparison for all the estimators, while Tables \ref{table: DGP_id_n50000_p50}(b) and \ref{table: DGP_id_n50000_p500}(b) provide all the inference related results based on $\bthetatilDR$. %Tables \ref{table: DGP_id_n50000_p50} and \ref{table: DGP_id_n50000_p500} summarize the estimation errors of $\bthetahatDR$, coverage probabilities and lengths of the CIs.
%We could see the double robustness of our estimator through the estimation errors.
The results, for both estimation and inference, and for each $p$, clearly validate the DR properties of $\bthetahatDR$ and $\bthetatilDR$. Whenever both working nuisance models are correct, %ly specified,
the achieved $L_2$ %estimation
errors of $\bthetahatDR$ are very close to those of the oracle estimators. %and are close to the super oracle estimator.
In addition, whenever $\mhat(\cdot)$ is correct, the results are similar (and near optimal) regardless of $\pihat(\cdot)$, %when only the conditional mean $m(\cdot)$ is correctly specified, it could have similar performance comparing to both correctly specified.
which is consistent with the results for $n=1000$. Further, when only $\pihat(\cdot)$ is correct, the $L_2$ errors are still smaller than %the case
when both are misspecified but cannot reach the same level as the correctly specified case, showing that consistency still holds but possibly at a slower convergence rate, %This still shows consistency, but possibly with slower convergence rates,
as discussed in Appendix \ref{draspect}. %due to the convergence rate is slow in such cases (see the discussion in Section \ref{draspect}).
Finally, when both are misspecified, the $L_2$ errors of $\bthetahatDR$ are much higher, indicating its inconsistency, as expected. On the same vein, the last columns of Tables \ref{table: DGP_id_n50000_p50}(a) and \ref{table: DGP_id_n50000_p500}(a) show that the $L_2$ errors of $\bthetahat_{cc}$ are also quite high (and different from the oracles) even at this sample size, thereby clearly showing that it is \emph{inconsistent}, as expected under a non-linear DGP for $Y|X$, and hence, is unsuitable as a general estimator of $\btheta_0$. %similar (if not larger) to these values %As we stated, only in ``linear-linear'' DGP is the complete case estimator consistent.
%This is clearly revealed in the last column of the Table \ref{table: DGP_id_n50000_p50} and \ref{table: DGP_id_n50000_p500}.

As regards the inference results, across all settings, the CovPs for the zero coefficients of $\btheta_0$
are always close to the expected 95\% level, similar to the results for $n = 1000$. For the non-zero coefficients, the results for both $p = 50$ and $500$ now demonstrate a clear pattern, whereby they are close to 95\% as soon as at least one of $\{\pihat(\cdot), \mhat(\cdot)\}$ is correct, and considerably lower when both are misspecified (indicating inconsistency). This therefore validates the DR property, \emph{even} for $\sqrt{n}$-rate inference via $\bthetatilDR$. It is interesting to note that while our theoretical results on $\bthetatilDR$ do require both $\{\pihat(\cdot),\mhat(\cdot)\}$ to be correct, the empirical results seem to be quite robust in this regard, achieving 95\% CovPs in large samples via $\sqrt{n}$-rate CIs whenever at least one, but \emph{not} necessarily both, working nuisance model is correct. Finally, the lengths of the CIs also seem to be small across all settings, thus indicating consistency. However, for the cases where only one of $\pihat(\cdot)$ and $\mhat(\cdot)$ is correct, especially the former, the CIs have the desired CovPs but are wider than those obtained when both are correct (possibly due to larger biases in variance estimation).

%\input{Simulation-Supp.tex}

%\section{Supplementary Numerical Results}

\subsection{Simulation Results for Non-Identity Covariance Matrices}\label{sec:sim:othercov}
%\tcr{** To be edited.}
 We present here additional simulation results for cases when $\bSigma_p$, the covariance matrix of $\bX$, corresponds to other correlation structures (possibly not sparse), specifically $\bSigma_p = $ AR1 (autoregressive) or CS (compund symmetry). %Aside from identity covariance matrix, we also study the case when the covariance matrix $\bSigma_p$ is AR1 and CS.

When $\bSigma_p = $ AR1, the results (for both estimation and inference, and for $p = 50$ and $500$) are presented in Tables \ref{table: DGP_auto_n1000_p50}, \ref{table: DGP_auto_n1000_p500}, \ref{table: DGP_auto_n1000_p50_infer} and \ref{table: DGP_auto_n1000_p500_infer}. Overall, the results are fairly consistent with those for the case when $\bSigma_p = I_p$ (identity matrix). The estimation errors as well as the inference results are quite close for both choices of $\bSigma_p$, thereby drawing similar conclusions as discussed in Section \ref{Sim. Results}. %from before. %as the identity case.
This is also possibly because the AR1 matrix with a relatively small $\rho=0.2$ is fairly close to the identity matrix $I_p$. Laslty, we also note that in Table \ref{table: DGP_auto_n1000_p500}(c), the estimation errors for the ``$\mhat$: SIM'' case are interestingly slightly better than the oracles. This, however, is not the case in general.

\begin{table}[!ht]%[!htbp]
	\centering
	\caption{Average $L_2$ errors of $\bthetahatDR$, obtained via various combinations of the nuisance estimators $\{\pihat(\cdot), \mhat(\cdot)\}$, and those of the oracle  estimators $\widehat{\btheta}_{orac}$ and $\widehat{\btheta}_{full}$, for $n = 1000$, $\bSigma_p = \mbox{AR1}$ and all $3$ choices of the \emph{true} DGPs.}
	\label{table: DGP_auto_n1000_p50}
	{\bf (I)} $p=50$. \\
	(a) DGP: ``Linear-linear'' for $\pi(\cdot)$ and $m(\cdot)$.
	\vspace{0.05in}
	\resizebox{\textwidth}{!}{
		\begin{tabular}{ll||ccc}\hline
			\multicolumn{2}{l}{Working nuisance model}& $\bthetahatDR$  & $ \widehat{\btheta}_{orac} $ &  $\widehat{\btheta}_{full}$    \\
			\hline
			\multirow{2}{*}{$\mhat$: linear}&$\pihat$: logit & 0.222 (0.038) & 0.223 (0.038) & 0.169 (0.028)  \\
			& $\pi$: quad & 0.222 (0.038) & 0.223 (0.038) & 0.169 (0.028)  \\
			\hline
			\multirow{2}{*}{$\mhat$: quad}&$\pihat$: logit & 0.224 (0.038) &  0.223 (0.038) & 0.169 (0.028)  \\
			&$\pihat$: quad & 0.223 (0.038) & 0.223 (0.038) & 0.169 (0.028)\\
			\hline
			\multirow{2}{*}{$\mhat$: SIM}&	$\pihat$: logit & 0.222 (0.038) & 0.223 (0.038) & 0.169 (0.028)  \\
			&$\pihat$: quad &  0.222 (0.038) & 0.223 (0.038) & 0.169 (0.028) \\
			\hline
	\end{tabular}}
	(b) DGP: ``Quad-quad'' for $\pi(\cdot)$ and $m(\cdot)$.
	\vspace{0.05in}
	\resizebox{\textwidth}{!}{
		\begin{tabular}{ll||ccc}\hline
			\multicolumn{2}{l}{Working nuisance model}& $\bthetahatDR$  & $ \widehat{\btheta}_{orac} $ &  $\widehat{\btheta}_{full}$   \\
			\hline
			\multirow{2}{*}{$\mhat$: linear}&$\pihat$: logit & 0.664 (0.107) & 0.469 (0.075) & 0.445 (0.074)  \\
			& $\pi$: quad & 0.625 (0.104) & 0.469 (0.075) & 0.445 (0.074)    \\
			\hline
			\multirow{2}{*}{$\mhat$: quad}&$\pihat$: logit & 0.464 (0.075) &  0.469 (0.075) & 0.445 (0.074)  \\
			&$\pihat$: quad & 0.464 (0.075) & 0.469 (0.075) & 0.445 (0.074) \\
			\hline
			\multirow{2}{*}{$\mhat$: SIM}&	$\pihat$: logit & 0.671 (0.109) & 0.469 (0.075) & 0.445 (0.074)   \\
			&$\pihat$: quad & 0.631 (0.106) & 0.469 (0.075) & 0.445 (0.074) \\
			\hline
	\end{tabular}}
	(c) DGP: ``SIM-SIM'' for $\pi(\cdot)$ and $m(\cdot)$.
	\vspace{0.05in}
	\resizebox{\textwidth}{!}{
		\begin{tabular}{ll||ccc}\hline
			\multicolumn{2}{l}{Working nuisance model}& $\bthetahatDR$  & $ \widehat{\btheta}_{orac} $ &  $\widehat{\btheta}_{full}$    \\
			\hline
			\multirow{2}{*}{$\mhat$: linear}&$\pihat$: logit & 0.569 (0.127 ) & 0.478 (0.112) & 0.459 (0.109) \\
			& $\pi$: quad & 0.567 (0.127) & 0.478 (0.112) & 0.459 (0.109)   \\
			\hline
			\multirow{2}{*}{$\mhat$: quad} &$\pihat$: logit &  0.562 (0.126) & 0.478 (0.112) & 0.459 (0.109)\\
			&$\pi$: quad & 0.562 (0.126) & 0.478 (0.112) & 0.459 (0.109)   \\
			\hline
			\multirow{2}{*}{$\mhat$: SIM}&	$\pihat$: logit &  0.499 (0.119) & 0.478 (0.112) & 0.459 (0.109) \\
			&$\pihat$: quad &  0.498 (0.120) & 0.478 (0.112) & 0.459 (0.109) \\
			\hline
	\end{tabular}}
\end{table}

\begin{table}[!ht]%[!htbp]
	\centering
	\caption{See caption of Table \ref{table: DGP_auto_n1000_p50}. (Only change: $p= 500$ instead of $50$)}
	\label{table: DGP_auto_n1000_p500}
	{\bf (II)} $p=500$. \\
	(a) DGP: ``Linear-linear'' for $\pi(\cdot)$ and $m(\cdot)$.
	\vspace{0.05in}
	\resizebox{\textwidth}{!}{
		\begin{tabular}{ll||ccc}\hline
			\multicolumn{2}{l}{Working nuisance model}& $\bthetahatDR$  & $ \widehat{\btheta}_{orac} $ &  $\widehat{\btheta}_{full}$   \\
			\hline
			\multirow{2}{*}{$\mhat$: linear}&$\pihat$: logit &  0.420 (0.045) & 0.401 (0.043) & 0.295 (0.029)  \\
			& $\pihat$: quad & 0.419 (0.044) & 0.401 (0.043) & 0.295 (0.029)   \\
			\hline
			\multirow{2}{*}{$\mhat$: quad}&$\pihat$: logit & 0.430 (0.046)  & 0.401 (0.043) & 0.295 (0.029)   \\
			&$\pihat$: quad & 0.430 (0.046)  & 0.401 (0.043) & 0.295 (0.029)    \\
			\hline
			\multirow{2}{*}{$\mhat$: SIM}&	$\pihat$: logit & 0.409 (0.044) & 0.401 (0.043) & 0.295 (0.029)  \\
			&$\pihat$: quad & 0.408 (0.044) & 0.401 (0.043) & 0.295 (0.029) \\
			\hline
	\end{tabular}}
	(b) DGP: ``Quad-quad'' for $\pi(\cdot)$ and $m(\cdot)$.
	\vspace{0.05in}
	\resizebox{\textwidth}{!}{
		\begin{tabular}{ll||ccc}\hline
			\multicolumn{2}{l}{Working nuisance model}& $\bthetahatDR$  & $ \widehat{\btheta}_{orac} $ &  $\widehat{\btheta}_{full}$    \\
			\hline
			\multirow{2}{*}{$\mhat$: linear}&$\pihat$: logit & 1.060 (0.112)  & 0.797 (0.084) & 0.743 (0.077)   \\
			& $\pihat$: quad & 1.049 (0.109) & 0.797 (0.084) & 0.743 (0.077)   \\
			\hline
			\multirow{2}{*}{$\mhat$: quad}&$\pihat$: logit & 0.814 (0.083) & 0.797 (0.084) & 0.743 (0.077) \\
			&$\pihat$: quad & 0.814 (0.083)  & 0.797 (0.084) & 0.743 (0.077) \\
			\hline
			\multirow{2}{*}{$\mhat$: SIM}&	$\pihat$: logit & 1.050 (0.110)  & 0.797 (0.084) & 0.743 (0.077)   \\
			&$\pihat$: quad &  1.038 (0.109) & 0.797 (0.084) & 0.743 (0.077)  \\
			\hline
	\end{tabular}}
	(c) DGP: ``SIM-SIM'' for $\pi(\cdot)$ and $m(\cdot)$.
	\vspace{0.05in}
	\resizebox{\textwidth}{!}{
		\begin{tabular}{ll||ccc}\hline
			\multicolumn{2}{l}{Working nuisance model}& $\bthetahatDR$  & $ \widehat{\btheta}_{orac} $ &  $\widehat{\btheta}_{full}$  \\
			\hline
			\multirow{2}{*}{$\mhat$: linear}&$\pihat$: logit &  1.026 (0.166) & 1.001 (0.153) & 0.974 (0.151) \\
			& $\pihat$: quad & 1.016 (0.159) & 1.001 (0.153) & 0.974 (0.151)     \\
			\hline
			\multirow{2}{*}{$\mhat$: quad} &$\pihat$: logit &1.029 (0.162)  & 1.001 (0.153) & 0.974 (0.151) \\
			&$\pihat$: quad & 1.019 (0.157) & 1.001 (0.153) & 0.974 (0.151)   \\
			\hline
			\multirow{2}{*}{$\mhat$: SIM}&	$\pihat$: logit &  0.961 (0.162)  & 1.001 (0.153) & 0.974 (0.151)  \\
			&$\pihat$: quad &  0.952 (0.158) & 1.001 (0.153) & 0.974 (0.151) \\
			\hline
	\end{tabular}}
\end{table}

\begin{table}[!ht]%[!htbp]
	\centering
	\caption{Average ($A$-CovP) and median ($M$-CovP) of the empirical coverage probabilities (CovPs) for the (coordinatewise) 95\% CIs of $\btheta_0$ obtained via $\bthetatilDR$ (based on various combinations of the nuisance estimators $\{\pihat(\cdot), \mhat(\cdot)\}$) for $n = 1000$, $\bSigma_p = \mbox{AR1}$ and all three choices of the \emph{true} DGPs. Shown also are the corresponding average lengths of these CIs. All values are reported separately for the truly zero and non-zero coefficients of $\btheta_0$ (see Section \ref{Est. Implemented}).}
	\label{table: DGP_auto_n1000_p50_infer}
	{\bf (I)} $p=50$.  \\
	(a) DGP: ``Linear-linear'' for $\pi(\cdot)$ and $m(\cdot)$.
		\vspace{0.05in}
		\resizebox{\textwidth}{!}{
		\begin{tabular}{ll||ccc|ccc}\hline
				\multicolumn{2}{l}{Working nuisance model}  & \multicolumn{3}{c}{Zero coefficients} &  \multicolumn{3}{c}{Non-zero coefficients} \\
			\hline
			\multicolumn{2}{l}{} &$A$-CovP & $M$-CovP  & Length &$A$-CovP & $M$-CovP  & Length \\
			\hline
				\multirow{2}{*}{$\mhat$: linear}&$\pihat$: logit & $0.94_{0.01}$ & $(0.94_{0.01})$ & $0.17_{0}$ & $0.94_{0.01}$ & $(0.94_{0.02})$ & $0.17_{0}$   \\
				&$\pihat$: quad & $0.94_{0.01}$ & $(0.94_{0.01})$ & $0.17_{0}$ & $0.94_{0.01}$ & $(0.94_{0.02})$ & $0.17_{0}$  \\
				\hline
				\multirow{2}{*}{$\mhat$: quad}&$\pihat$: logit &  $0.94_{0.01}$ & $(0.94_{0.01})$ & $0.17_{0}$ & $0.94_{0.01}$ & $(0.94_{0.02})$ & $0.17_{0}$   \\
				&$\pihat$: quad &  $0.94_{0.01}$ & $(0.94_{0.01})$ & $0.17_{0}$ & $0.94_{0.01}$ & $(0.95_{0.02})$ & $0.17_{0}$  \\
				\hline
				\multirow{2}{*}{$\mhat$: SIM}&	$\pihat$: logit &   $0.94_{0.01}$ & $(0.94_{0.01})$ & $0.17_{0}$ & $0.94_{0.01}$ & $(0.94_{0.01})$ & $0.17_{0}$  \\
				&$\pihat$: quad & $0.94_{0.01}$ & $(0.94_{0.01})$ & $0.17_{0}$ & $0.94_{0.01}$ & $(0.94_{0.01})$ & $0.17_{0}$ \\
				\hline
		\end{tabular}}
	(b) DGP: ``Quad-quad'' for $\pi(\cdot)$ and $m(\cdot)$.
	\vspace{0.05in}
	\resizebox{\textwidth}{!}{
	\begin{tabular}{ll||ccc|ccc}\hline
		\multicolumn{2}{l}{Working nuisance model}  & \multicolumn{3}{c}{Zero coefficients} &  \multicolumn{3}{c}{Non-zero coefficients} \\
	\hline
	\multicolumn{2}{l}{} &$A$-CovP & $M$-CovP  & Length &$A$-CovP & $M$-CovP  & Length \\
	\hline
			\multirow{2}{*}{$\mhat$: linear}&$\pihat$: logit & $0.94_{0.01}$ & $(0.94_{0.01})$ & $0.42_{0}$ & $0.89_{0.14}$ & $(0.94_{0.02})$ & $0.47_{0.08}$  \\
			& $\pihat$: quad &$0.94_{0.01}$ & $(0.94_{0.01})$ & $0.42_{0}$ & $0.90_{0.12}$ & $(0.94_{0.02})$ & $0.47_{0.07}$   \\
			\hline
			\multirow{2}{*}{$\mhat$: quad}&$\pihat$: logit &  $0.94_{0.01}$ & $(0.94_{0.01})$ & $0.34_{0}$ & $0.95_{0.01}$ & $(0.95_{0.01})$ & $0.38_{0.05}$ \\
			&$\pihat$: quad & $0.94_{0.01}$ & $(0.94_{0.01})$ & $0.34_{0}$ & $0.94_{0.01}$ & $(0.94_{0.01})$ & $0.38_{0.05}$  \\
			\hline
			\multirow{2}{*}{$\mhat$: SIM}&	$\pihat$: logit &  $0.94_{0.01}$ & $(0.94_{0.01})$ & $0.42_{0}$ & $0.89_{0.14}$ & $(0.94_{0.01})$ & $0.47_{0.07}$  \\
			&$\pihat$: quad & $0.94_{0.01}$ & $(0.94_{0.01})$ & $0.42_{0}$ & $0.90_{0.12}$ & $(0.94_{0.02})$ & $0.47_{0.07}$  \\
			\hline
	\end{tabular}}
	(c) DGP: ``SIM-SIM'' for $\pi(\cdot)$ and $m(\cdot)$.
	\vspace{0.05in}
	\resizebox{\textwidth}{!}{
	\begin{tabular}{ll||ccc|ccc}\hline
		\multicolumn{2}{l}{Working nuisance model}  & \multicolumn{3}{c}{Zero coefficients} &  \multicolumn{3}{c}{Non-zero coefficients} \\
	\hline
	\multicolumn{2}{l}{} &$A$-CovP & $M$-CovP  & Length &$A$-CovP & $M$-CovP  & Length \\
	\hline
			\multirow{2}{*}{$\mhat$: linear}&$\pihat$: logit & $0.95_{0.01}$ & $(0.95_{0.01})$ & $0.43_{0}$ & $0.94_{0.01}$ & $(0.94_{0.01})$ & $0.48_{0.03}$ \\
			& $\pihat$: quad & $0.95_{0.01}$ & $(0.94_{0.01})$ & $0.43_{0}$ & $0.94_{0.01}$ & $(0.94_{0.01})$ & $0.48_{0.03}$  \\
			\hline
			\multirow{2}{*}{$\mhat$: quad}&$\pihat$: logit & $0.95_{0.01}$ & $(0.95_{0.01})$ & $0.42_{0}$ & $0.94_{0.01}$ & $(0.94_{0.01})$ & $0.47_{0.03}$   \\
			&$\pihat$: quad & $0.95_{0.01}$ & $(0.95_{0.01})$ & $0.42_{0}$ & $0.94_{0.01}$ & $(0.94_{0.01})$ & $0.47_{0.03}$ \\
			\hline
			\multirow{2}{*}{$\mhat$: SIM}&	$\pihat$: logit &  $0.95_{0.01}$ & $(0.95_{0.01})$ & $0.37_{0}$ & $0.94_{0.01}$ & $(0.95_{0.01})$ & $0.41_{0.02}$   \\
			&$\pihat$: quad & $0.95_{0.01}$ & $(0.95_{0.01})$ & $0.37_{0}$ & $0.94_{0.01}$ & $(0.95_{0.01})$ & $0.41_{0.02}$ \\
			\hline
	\end{tabular}}
\end{table}

\begin{table}[!ht]%[!htbp]
	\centering
	\caption{See caption of Table \ref{table: DGP_auto_n1000_p50_infer}. (Only change: $p= 500$ instead of $50$)}
	\label{table: DGP_auto_n1000_p500_infer}
	{\bf (II)} $p=500$. \\
	(a) DGP: ``Linear-linear'' for $\pi(\cdot)$ and $m(\cdot)$.
	\vspace{0.05in}
	\resizebox{\textwidth}{!}{
		\begin{tabular}{ll||ccc|ccc}\hline
				\multicolumn{2}{l}{Working nuisance model}  & \multicolumn{3}{c}{Zero coefficients} &  \multicolumn{3}{c}{Non-zero coefficients}\\
			\hline
			\multicolumn{2}{l}{} &$A$-CovP & $M$-CovP  & Length &$A$-CovP & $M$-CovP  & Length \\
			\hline
			\multirow{2}{*}{$\mhat$: linear}&$\pihat$: logit & $0.95_{0.01}$ & $(0.95_{0.01})$ & $0.17_{0}$ & $0.91_{0.02}$ & $(0.91_{0.01})$ & $0.17_{0}$ \\
			&$\pihat$: quad &  $0.95_{0.01}$ & $(0.95_{0.01})$ & $0.17_{0}$ & $0.91_{0.02}$ & $(0.91_{0.01})$ & $0.17_{0}$  \\
			\hline
			\multirow{2}{*}{$\mhat$: quad}&$\pihat$: logit & $0.95_{0.01}$ & $(0.95_{0.01})$ & $0.17_{0}$ & $0.91_{0.02}$ & $(0.91_{0.02})$ & $0.17_{0}$    \\
			&$\pihat$: quad &  $0.95_{0.01}$ & $(0.95_{0.01})$ & $0.17_{0}$ & $0.91_{0.02}$ & $(0.91_{0.02})$ & $0.17_{0}$  \\
			\hline
			\multirow{2}{*}{$\mhat$: SIM}&	$\pihat$: logit &  $0.95_{0.01}$ & $(0.95_{0.01})$ & $0.16_{0}$ & $0.91_{0.01}$ & $(0.91_{0.01})$ & $0.17_{0}$  \\
			&$\pihat$: quad & $0.95_{0.01}$ & $(0.95_{0.01})$ & $0.16_{0}$ & $0.91_{0.01}$ & $(0.91_{0.02})$ & $0.17_{0}$ \\
			\hline
	\end{tabular}}
	(b) DGP: ``Quad-quad'' for $\pi(\cdot)$ and $m(\cdot)$.
	\vspace{0.05in}
	\resizebox{\textwidth}{!}{
		\begin{tabular}{ll||ccc|ccc}\hline
				\multicolumn{2}{l}{Working nuisance model}  & \multicolumn{3}{c}{Zero coefficients} &  \multicolumn{3}{c}{Non-zero coefficients} \\
			\hline
			\multicolumn{2}{l}{} &$A$-CovP & $M$-CovP  & Length &$A$-CovP & $M$-CovP  & Length \\
			\hline
			\multirow{2}{*}{$\mhat$: linear}&$\pihat$: logit & $0.95_{0.01}$ & $(0.95_{0.01})$ & $0.44_{0}$ & $0.92_{0.03}$ & $(0.93_{0.02})$ & $0.46_{0.07}$  \\
			& $\pihat$: quad & $0.95_{0.01}$ & $(0.95_{0.01})$ & $0.43_{0}$ & $0.91_{0.03}$ & $(0.92_{0.02})$ & $0.46_{0.06}$   \\
			\hline
			\multirow{2}{*}{$\mhat$: quad}&$\pihat$: logit & $0.95_{0.01}$ & $(0.95_{0.01})$ & $0.33_{0}$ & $0.92_{0.02}$ & $(0.92_{0.02})$ & $0.35_{0.04}$ \\
			&$\pihat$: quad & $0.95_{0.01}$ & $(0.95_{0.01})$ & $0.33_{0}$ & $0.92_{0.02}$ & $(0.92_{0.02})$ & $0.35_{0.04}$   \\
			\hline
			\multirow{2}{*}{$\mhat$: SIM}&	$\pihat$: logit &  $0.95_{0.01}$ & $(0.95_{0.01})$ & $0.44_{0}$ & $0.91_{0.03}$ & $(0.92_{0.02})$ & $0.46_{0.07}$  \\
			&$\pihat$: quad & $0.95_{0.01}$ & $(0.95_{0.01})$ & $0.43_{0}$ & $0.91_{0.03}$ & $(0.91_{0.02})$ & $0.46_{0.06}$ \\
			\hline
	\end{tabular}}
	(c) DGP: ``SIM-SIM'' for $\pi(\cdot)$ and $m(\cdot)$.
	\vspace{0.05in}
	\resizebox{\textwidth}{!}{
		\begin{tabular}{ll||ccc|ccc}\hline
				\multicolumn{2}{l}{Working nuisance model}  & \multicolumn{3}{c}{Zero coefficients} &  \multicolumn{3}{c}{Non-zero coefficients} \\
			\hline
			\multicolumn{2}{l}{} &$A$-CovP & $M$-CovP  & Length &$A$-CovP & $M$-CovP  & Length \\
			\hline
			\multirow{2}{*}{$\mhat$: linear}&$\pihat$: logit & $0.95_{0.01}$ & $(0.95_{0.01})$ & $0.52_{0}$ & $0.88_{0.04}$ & $(0.88_{0.04})$ & $0.55_{0.03}$  \\
			& $\pihat$: quad & $0.95_{0.01}$ & $(0.95_{0.01})$ & $0.52_{0}$ & $0.87_{0.04}$ & $(0.87_{0.04})$ & $0.55_{0.03}$ \\
			\hline
			\multirow{2}{*}{$\mhat$: quad}&$\pihat$: logit & $0.95_{0.01}$ & $(0.95_{0.01})$ & $0.52_{0}$ & $0.88_{0.03}$ & $(0.88_{0.03})$ & $0.55_{0.03}$  \\
			&$\pihat$: quad & $0.95_{0.01}$ & $(0.95_{0.01})$ & $0.52_{0}$ & $0.88_{0.03}$ & $(0.88_{0.04})$ & $0.55_{0.03}$\\
			\hline
			\multirow{2}{*}{$\mhat$: SIM}&	$\pihat$: logit &$0.95_{0.01}$ & $(0.95_{0.01})$ & $0.48_{0}$ & $0.93_{0.01}$ & $(0.94_{0.01})$ & $0.51_{0.03}$  \\
			&$\pihat$: quad &  $0.95_{0.01}$ & $(0.95_{0.01})$ & $0.48_{0}$ & $0.93_{0.01}$ & $(0.94_{0.01})$ & $0.51_{0.03}$ \\
			\hline
	\end{tabular}}
\end{table}

For $\bSigma_p =$ CS, %the covariance matrix is CS matrix,
the corresponding results are given in Tables \ref{table: DGP_symm_n1000_p50} and \ref{table: DGP_symm_n1000_p50_infer} for  $p=50$, and Tables \ref{table: DGP_symm_n1000_p500} and \ref{table: DGP_symm_n1000_p500_infer} for $p=500$. Note that the CS matrix (and its inverse) is not sparse, so that the nodewise Lasso estimator of $\bOmega$ is not theoretically guaranteed to work when $p=500$. Nevertheless, the general pattern of the results stay the same as for $\bSigma =I_p$ or AR1, indicating that our procedures are fairly robust to the underlying correlation structure of $\bX$, as well as to the degree of sparsity of $\bOmega$, in high dimensional settings.
%We also note that for $p=500$, the increased errors in estimating the precision matrices and the influence function leads to slightly lower coverage probabilities. %even when the working models are correctly specified.

\begin{table}[!ht]%[!htbp]
	\centering
	\caption{Average $L_2$ errors of $\bthetahatDR$, obtained via various combinations of the nuisance estimators $\{\pihat(\cdot), \mhat(\cdot)\}$, and those of the oracle  estimators $\widehat{\btheta}_{orac}$ and $\widehat{\btheta}_{full}$, for $n = 1000$, $\bSigma_p = \mbox{CS}$ and all three choices of the \emph{true} DGPs.}
	\label{table: DGP_symm_n1000_p50}
	{\bf (I)} $p=50$.\\
	(a) DGP: ``Linear-linear'' for $\pi(\cdot)$ and $m(\cdot)$.
	\vspace{0.05in}
	\resizebox{\textwidth}{!}{
		\begin{tabular}{ll||ccc}\hline
			\multicolumn{2}{l}{Working nuisance model}& $\bthetahatDR$  & $ \widehat{\btheta}_{orac} $ &  $\widehat{\btheta}_{full}$   \\
			\hline
			\multirow{2}{*}{$\mhat$: linear}&$\pihat$: logit & 0.245 (0.039) &  0.247 (0.04) & 0.185 (0.03)  \\
			& $\pihat$: quad &0.245 (0.038) & 0.247 (0.04) & 0.185 (0.03)  \\
			\hline
			\multirow{2}{*}{$\mhat$: quad}&$\pihat$: logit & 0.248 (0.039) &  0.247 (0.04) & 0.185 (0.03)  \\
			&$\pihat$: quad & 0.247 (0.039) & 0.247 (0.04) & 0.185 (0.03)  \\
			\hline
			\multirow{2}{*}{$\mhat$: SIM}&	$\pihat$: logit & 0.246 (0.039) & 0.247 (0.04) & 0.185 (0.03) \\
			&$\pihat$: quad & 0.246 (0.038) & 0.247 (0.04) & 0.185 (0.03) \\
			\hline
	\end{tabular}}
	(b) DGP: ``Quad-quad'' for $\pi(\cdot)$ and $m(\cdot)$.
	\vspace{0.05in}
	\resizebox{\textwidth}{!}{
		\begin{tabular}{ll||ccc}\hline
			\multicolumn{2}{l}{Working nuisance model}& $\bthetahatDR$  & $ \widehat{\btheta}_{orac} $ &  $\widehat{\btheta}_{full}$    \\
			\hline
			\multirow{2}{*}{$\mhat$: linear}&$\pihat$: logit & 0.701 (0.126) &  0.513 (0.088) & 0.483 (0.083)   \\
			& $\pihat$: quad & 0.657 (0.118) & 0.513 (0.088) & 0.483 (0.083)   \\
			\hline
			\multirow{2}{*}{$\mhat$: quad}&$\pihat$: logit & 0.509 (0.087) & 0.513 (0.088) & 0.483 (0.083) \\
			&$\pihat$: quad & 0.509 (0.088) & 0.513 (0.088) & 0.483 (0.083) \\
			\hline
			\multirow{2}{*}{$\mhat$: SIM}&	$\pihat$: logit & 0.704 (0.126) &  0.513 (0.088) & 0.483 (0.083)  \\
			&$\pihat$: quad & 0.662 (0.119) & 0.513 (0.088) & 0.483 (0.083)  \\
			\hline
	\end{tabular}}
	(c) DGP: ``SIM-SIM'' for $\pi(\cdot)$ and $m(\cdot)$.
	\vspace{0.05in}
	\resizebox{\textwidth}{!}{
		\begin{tabular}{ll||ccc}\hline
			\multicolumn{2}{l}{Working nuisance model}& $\bthetahatDR$  & $ \widehat{\btheta}_{orac} $ &  $\widehat{\btheta}_{full}$  \\
			\hline
			\multirow{2}{*}{$\mhat$: linear}&$\pihat$: logit &  0.284 (0.052) &  0.272 (0.047) & 0.224 (0.042)  \\
			& $\pihat$: quad & 0.282 (0.052) & 0.272 (0.047) & 0.224 (0.042)      \\
			\hline
			\multirow{2}{*}{$\mhat$: quad} &$\pihat$: logit & 0.287 (0.052) & 0.272 (0.047) & 0.224 (0.042)  \\
			&$\pihat$: quad & 0.285 (0.052) & 0.272 (0.047) & 0.224 (0.042)     \\
			\hline
			\multirow{2}{*}{$\mhat$: SIM}&	$\pihat$: logit &   0.275 (0.048) & 0.272 (0.047) & 0.224 (0.042)  \\
			&$\pihat$: quad & 0.274 (0.048) & 0.272 (0.047) & 0.224 (0.042) \\
			\hline
	\end{tabular}}
\end{table}

\begin{table}[!ht]%[!htbp]
	\centering
	\caption{See caption of Table \ref{table: DGP_symm_n1000_p50}. (Only change: $p= 500$ instead of $50$)}
	\label{table: DGP_symm_n1000_p500}
	{\bf (II)} $p=500$. \\
	(a) DGP: ``Linear-linear'' for $\pi(\cdot)$ and $m(\cdot)$.
	\vspace{0.05in}
	\resizebox{\textwidth}{!}{
		\begin{tabular}{ll||ccc}\hline
			\multicolumn{2}{l}{Working nuisance model}& $\bthetahatDR$  & $ \widehat{\btheta}_{orac} $ &  $\widehat{\btheta}_{full}$   \\
			\hline
			\multirow{2}{*}{$\mhat$: linear}&$\pihat$: logit & 0.492 (0.055) & 0.466 (0.050) & 0.350 (0.032)  \\
			& $\pihat$: quad & 0.492 (0.055) & 0.466 (0.050) & 0.350 (0.032)   \\
			\hline
			\multirow{2}{*}{$\mhat$: quad}&$\pihat$: logit & 0.509 (0.059) & 0.466 (0.050) & 0.350 (0.032)   \\
			&$\pihat$: quad & 0.508 (0.059) & 0.466 (0.050) & 0.350 (0.032)   \\
			\hline
			\multirow{2}{*}{$\mhat$: SIM}&	$\pihat$: logit & 0.483 (0.053) & 0.466 (0.050) & 0.350 (0.032)  \\
			&$\pihat$: quad & 0.483 (0.053) &  0.466 (0.050) & 0.350 (0.032) \\
			\hline
	\end{tabular}}
	(b) DGP: ``Quad-quad'' for $\pi(\cdot)$ and $m(\cdot)$.
	\vspace{0.05in}
	\resizebox{\textwidth}{!}{
		\begin{tabular}{ll||ccc}\hline
			\multicolumn{2}{l}{Working nuisance model}& $\bthetahatDR$  & $ \widehat{\btheta}_{orac} $ &  $\widehat{\btheta}_{full}$    \\
			\hline
			\multirow{2}{*}{$\mhat$: linear}&$\pihat$: logit &  1.245 (0.131) & 0.949 (0.094) & 0.890 (0.087)   \\
			& $\pihat$: quad & 1.236 (0.130) & 0.949 (0.094) & 0.890 (0.087)  \\
			\hline
			\multirow{2}{*}{$\mhat$: quad}&$\pihat$: logit & 0.972 (0.100) & 0.949 (0.094) & 0.890 (0.087) \\
			&$\pihat$: quad & 0.973 (0.100) & 0.949 (0.094) & 0.890 (0.087) \\
			\hline
			\multirow{2}{*}{$\mhat$: SIM}&	$\pihat$: logit & 1.251 (0.128) & 0.949 (0.094) & 0.890 (0.087)   \\
			&$\pihat$: quad & 1.240 (0.128) & 0.949 (0.094) & 0.890 (0.087)  \\
			\hline
	\end{tabular}}
	(c) DGP: ``SIM-SIM'' for $\pi(\cdot)$ and $m(\cdot)$.
	\vspace{0.05in}
	\resizebox{\textwidth}{!}{
		\begin{tabular}{ll||ccc}\hline
			\multicolumn{2}{l}{Working nuisance model}& $\bthetahatDR$  & $ \widehat{\btheta}_{orac} $ &  $\widehat{\btheta}_{full}$  \\
			\hline
			\multirow{2}{*}{$\mhat$: linear}&$\pihat$: logit &  0.460 (0.055) & 0.463 (0.051) & 0.364 (0.036)  \\
			& $\pihat$: quad & 0.458 (0.055) & 0.463 (0.051) & 0.364 (0.036)      \\
			\hline
			\multirow{2}{*}{$\mhat$: quad} &$\pihat$: logit &0.473 (0.057) & 0.463 (0.051) & 0.364 (0.036)  \\
			&$\pihat$: quad & 0.472 (0.057) & 0.463 (0.051) & 0.364 (0.036)     \\
			\hline
			\multirow{2}{*}{$\mhat$: SIM}&	$\pihat$: logit &   0.466 (0.054) & 0.463 (0.051) & 0.364 (0.036)  \\
			&$\pihat$: quad & 0.465 (0.054) & 0.463 (0.051) & 0.364 (0.036)\\
			\hline
	\end{tabular}}
\end{table}

\begin{table}[!ht]%[!htbp]
	\centering
	\caption{Average ($A$-CovP) and median ($M$-CovP) of the empirical coverage probabilities (CovPs) for the (coordinatewise) 95\% CIs of $\btheta_0$ obtained via $\bthetatilDR$ (based on various combinations of the nuisance estimators $\{\pihat(\cdot), \mhat(\cdot)\}$) for $n = 1000$, $\bSigma_p = \mbox{CS}$ and all three choices of the \emph{true} DGPs. Shown also are the corresponding average lengths of these CIs. All values are reported separately for the truly zero and non-zero coefficients of $\btheta_0$ (see Section \ref{Est. Implemented}).}
	\label{table: DGP_symm_n1000_p50_infer}
	{\bf (I)} $p=50$. \\
	(a) DGP: ``Linear-linear'' for $\pi(\cdot)$ and $m(\cdot)$.
	\vspace{0.05in}
	\resizebox{\textwidth}{!}{
		\begin{tabular}{ll||ccc|ccc}\hline
			\multicolumn{2}{l}{Working nuisance model}  & \multicolumn{3}{c}{Zero coefficients} &  \multicolumn{3}{c}{Non-zero coefficients} \\
			\hline
			\multicolumn{2}{l}{} &$A$-CovP & $M$-CovP  & Length &$A$-CovP & $M$-CovP  & Length \\
			\hline
			\multirow{2}{*}{$\mhat$: linear}&$\pihat$: logit & $0.94_{0.01}$ & $(0.94_{0.01})$ & $0.18_{0}$ & $0.94_{0.01}$ & $(0.94_{0.01})$ & $0.18_{0}$  \\
			&$\pihat$: quad &  $0.94_{0.01}$ & $(0.94_{0.01})$ & $0.18_{0}$ & $0.94_{0.01}$ & $(0.94_{0.01})$ & $0.18_{0}$   \\
			\hline
			\multirow{2}{*}{$\mhat$: quad}&$\pihat$: logit &  $0.94_{0.01}$ & $(0.94_{0.01})$ & $0.18_{0}$ & $0.94_{0.01}$ & $(0.94_{0.01})$ & $0.18_{0}$  \\
			&$\pihat$: quad &  $0.94_{0.01}$ & $(0.94_{0.01})$ & $0.18_{0}$ & $0.94_{0.01}$ & $(0.94_{0})$ & $0.18_{0}$   \\
			\hline
			\multirow{2}{*}{$\mhat$: SIM}&	$\pihat$: logit &  $0.94_{0.01}$ & $(0.94_{0.01})$ & $0.18_{0}$ & $0.94_{0.01}$ & $(0.94_{0.01})$ & $0.18_{0}$  \\
			&$\pihat$: quad & $0.94_{0.01}$ & $(0.94_{0.01})$ & $0.18_{0}$ & $0.94_{0.01}$ & $(0.94_{0.01})$ & $0.18_{0}$ \\
			\hline
	\end{tabular}}
	(b) DGP: ``Quad-quad'' for $\pi(\cdot)$ and $m(\cdot)$.
	\vspace{0.05in}
	\resizebox{\textwidth}{!}{
		\begin{tabular}{ll||ccc|ccc}\hline
			\multicolumn{2}{l}{Working nuisance model}  & \multicolumn{3}{c}{Zero coefficients} &  \multicolumn{3}{c}{Non-zero coefficients} \\
			\hline
			\multicolumn{2}{l}{} &$A$-CovP & $M$-CovP  & Length &$A$-CovP & $M$-CovP  & Length \\
			\hline
			\multirow{2}{*}{$\mhat$: linear}&$\pihat$: logit & $0.94_{0.01}$ & $(0.94_{0.01})$ & $0.45_{0}$ & $0.90_{0.11}$ & $(0.94_{0.01})$ & $0.49_{0.07}$  \\
			& $\pihat$: quad & $0.94_{0.01}$ & $(0.95_{0.01})$ & $0.45_{0}$ & $0.90_{0.10}$ & $(0.93_{0.02})$ & $0.49_{0.06}$   \\
			\hline
			\multirow{2}{*}{$\mhat$: quad}&$\pihat$: logit &   $0.94_{0.01}$ & $(0.94_{0.01})$ & $0.37_{0}$ & $0.95_{0.01}$ & $(0.95_{0.01})$ & $0.41_{0.05}$  \\
			&$\pihat$: quad & $0.94_{0.01}$ & $(0.94_{0.01})$ & $0.37_{0}$ & $0.95_{0.01}$ & $(0.95_{0.01})$ & $0.41_{0.05}$   \\
			\hline
			\multirow{2}{*}{$\mhat$: SIM}&	$\pihat$: logit &  $0.94_{0.01}$ & $(0.95_{0.01})$ & $0.45_{0}$ & $0.90_{0.11}$ & $(0.94_{0.02})$ & $0.50_{0.07}$  \\
			&$\pihat$: quad & $0.94_{0.01}$ & $(0.95_{0.01})$ & $0.45_{0}$ & $0.90_{0.09}$ & $(0.93_{0.02})$ & $0.50_{0.06}$\\
			\hline
	\end{tabular}}
	(c) DGP: ``SIM-SIM'' for $\pi(\cdot)$ and $m(\cdot)$.
	\vspace{0.05in}
	\resizebox{\textwidth}{!}{
		\begin{tabular}{ll||ccc|ccc}\hline
			\multicolumn{2}{l}{Working nuisance model}  & \multicolumn{3}{c}{Zero coefficients} &  \multicolumn{3}{c}{Non-zero coefficients} \\
			\hline
			\multicolumn{2}{l}{} &$A$-CovP & $M$-CovP  & Length &$A$-CovP & $M$-CovP  & Length \\
			\hline
			\multirow{2}{*}{$\mhat$: linear}&$\pihat$: logit & $0.94_{0.01}$ & $(0.94_{0.01})$ & $0.21_{0}$ & $0.94_{0.01}$ & $(0.94_{0.01})$ & $0.22_{0.01}$  \\
			& $\pihat$: quad &$0.94_{0.01}$ & $(0.94_{0.01})$ & $0.21_{0}$ & $0.94_{0.01}$ & $(0.94_{0.01})$ & $0.22_{0.01}$  \\
			\hline
			\multirow{2}{*}{$\mhat$: quad}&$\pihat$: logit & $0.94_{0.01}$ & $(0.94_{0.01})$ & $0.21_{0}$ & $0.94_{0.01}$ & $(0.94_{0.01})$ & $0.23_{0.01}$  \\
			&$\pihat$: quad & $0.94_{0.01}$ & $(0.94_{0.01})$ & $0.21_{0}$ & $0.94_{0.01}$ & $(0.94_{0.01})$ & $0.22_{0.01}$ \\
			\hline
			\multirow{2}{*}{$\mhat$: SIM}&	$\pihat$: logit & $0.94_{0.01}$ & $(0.94_{0.02})$ & $0.20_{0}$ & $0.94_{0.01}$ & $(0.95_{0.01})$ & $0.21_{0.01}$ \\
			&$\pihat$: quad & $0.94_{0.01}$ & $(0.94_{0.01})$ & $0.20_{0}$ & $0.94_{0.01}$ & $(0.95_{0.01})$ & $0.21_{0.01}$ \\
			\hline
	\end{tabular}}
\end{table}

\begin{table}[!ht]%[!htbp]
	\centering
	\caption{See caption of Table \ref{table: DGP_symm_n1000_p50_infer}. (Only change: $p= 500$ instead of $50$)}
	\label{table: DGP_symm_n1000_p500_infer}
	{\bf (II)} $p=500$.\\
	(a) DGP: ``Linear-linear'' for $\pi(\cdot)$ and $m(\cdot)$.
	\vspace{0.05in}
	\resizebox{\textwidth}{!}{
		\begin{tabular}{ll||ccc|ccc}\hline
				\multicolumn{2}{l}{Working nuisance model}  & \multicolumn{3}{c}{Zero coefficients} &  \multicolumn{3}{c}{Non-zero coefficients} \\
			\hline
			\multicolumn{2}{l}{} &$A$-CovP & $M$-CovP  & Length &$A$-CovP & $M$-CovP  & Length \\
			\hline
			\multirow{2}{*}{$\mhat$: linear}&$\pihat$: logit & $0.94_{0.01}$ & $(0.94_{0.01})$ & $0.18_{0}$ & $0.91_{0.02}$ & $(0.92_{0.01})$ & $0.18_{0}$ \\
			&$\pihat$: quad &  $0.94_{0.01}$ & $(0.94_{0.01})$ & $0.18_{0}$ & $0.91_{0.02}$ & $(0.92_{0.01})$ & $0.18_{0}$  \\
			\hline
			\multirow{2}{*}{$\mhat$: quad}&$\pihat$: logit & $0.94_{0.01}$ & $(0.94_{0.01})$ & $0.18_{0}$ & $0.91_{0.02}$ & $(0.91_{0.01})$ & $0.19_{0}$   \\
			&$\pihat$: quad & $0.94_{0.01}$ & $(0.94_{0.01})$ & $0.19_{0}$ & $0.91_{0.02}$ & $(0.91_{0.01})$ & $0.19_{0}$  \\
			\hline
			\multirow{2}{*}{$\mhat$: SIM}&	$\pihat$: logit & $0.94_{0.01}$ & $(0.94_{0.01})$ & $0.18_{0}$ & $0.91_{0.01}$ & $(0.91_{0.01})$ & $0.18_{0}$   \\
			&$\pihat$: quad & $0.94_{0.01}$ & $(0.94_{0.01})$ & $0.18_{0}$ & $0.91_{0.01}$ & $(0.91_{0.01})$ & $0.18_{0}$  \\
			\hline
	\end{tabular}}
	(b) DGP: ``Quad-quad'' for $\pi(\cdot)$ and $m(\cdot)$.
	\vspace{0.05in}
	\resizebox{\textwidth}{!}{
		\begin{tabular}{ll||ccc|ccc}\hline
				\multicolumn{2}{l}{Working nuisance model}  & \multicolumn{3}{c}{Zero coefficients} &  \multicolumn{3}{c}{Non-zero coefficients} \\
			\hline
			\multicolumn{2}{l}{} &$A$-CovP & $M$-CovP  & Length &$A$-CovP & $M$-CovP  & Length \\
			\hline
			\multirow{2}{*}{$\mhat$: linear}&$\pihat$: logit & $0.94_{0.01}$ & $(0.95_{0.01})$ & $0.47_{0}$ & $0.91_{0.02}$ & $(0.92_{0.01})$ & $0.50_{0.06}$  \\
			& $\pihat$: quad & $0.94_{0.01}$ & $(0.94_{0.01})$ & $0.47_{0}$ & $0.91_{0.03}$ & $(0.91_{0.03})$ & $0.49_{0.06}$  \\
			\hline
			\multirow{2}{*}{$\mhat$: quad}&$\pihat$: logit & $0.94_{0.01}$ & $(0.94_{0.01})$ & $0.36_{0}$ & $0.92_{0.02}$ & $(0.92_{0.01})$ & $0.38_{0.04}$ \\
			&$\pihat$: quad &  $0.94_{0.01}$ & $(0.94_{0.01})$ & $0.36_{0}$ & $0.92_{0.02}$ & $(0.92_{0.02})$ & $0.38_{0.04}$   \\
			\hline
			\multirow{2}{*}{$\mhat$: SIM}&	$\pihat$: logit &  $0.94_{0.01}$ & $(0.94_{0.01})$ & $0.47_{0}$ & $0.91_{0.03}$ & $(0.92_{0.02})$ & $0.50_{0.06}$  \\
			&$\pihat$: quad & $0.94_{0.01}$ & $(0.94_{0.01})$ & $0.47_{0}$ & $0.91_{0.03}$ & $(0.91_{0.02})$ & $0.49_{0.06}$ \\
			\hline
	\end{tabular}}
	(c) DGP: ``SIM-SIM'' for $\pi(\cdot)$ and $m(\cdot)$.
	\vspace{0.05in}
	\resizebox{\textwidth}{!}{
		\begin{tabular}{ll||ccc|ccc}\hline
				\multicolumn{2}{l}{Working nuisance model}  & \multicolumn{3}{c}{Zero coefficients} &  \multicolumn{3}{c}{Non-zero coefficients} \\
			\hline
			\multicolumn{2}{l}{} &$A$-CovP & $M$-CovP  & Length &$A$-CovP & $M$-CovP  & Length \\
			\hline
			\multirow{2}{*}{$\mhat$: linear}&$\pihat$: logit & $0.94_{0.01}$ & $(0.94_{0.01})$ & $0.18_{0}$ & $0.92_{0.01}$ & $(0.92_{0.01})$ & $0.18_{0}$  \\
			& $\pihat$: quad & $0.94_{0.01}$ & $(0.94_{0.01})$ & $0.18_{0}$ & $0.92_{0.01}$ & $(0.92_{0.01})$ & $0.18_{0}$   \\
			\hline
			\multirow{2}{*}{$\mhat$: quad}&$\pihat$: logit & $0.94_{0.01}$ & $(0.94_{0.01})$ & $0.19_{0}$ & $0.92_{0.01}$ & $(0.92_{0.01})$ & $0.19_{0}$  \\
			&$\pihat$: quad &$0.94_{0.01}$ & $(0.94_{0.01})$ & $0.19_{0}$ & $0.92_{0.01}$ & $(0.92_{0.01})$ & $0.19_{0}$\\
			\hline
			\multirow{2}{*}{$\mhat$: SIM}&	$\pihat$: logit &$0.94_{0.01}$ & $(0.94_{0.01})$ & $0.18_{0}$ & $0.92_{0.01}$ & $(0.92_{0.01})$ & $0.18_{0}$  \\
			&$\pihat$: quad & $0.94_{0.01}$ & $(0.94_{0.01})$ & $0.18_{0}$ & $0.92_{0.01}$ & $(0.92_{0.01})$ & $0.18_{0}$ \\
			\hline
	\end{tabular}}
\end{table}

\section{Technical Tools}\label{techtools}

We collect here some useful definitions and supporting lemmas that serve throughout as key technical ingredients in the proofs of all our main results.

\subsection{Orlicz Norms, Sub-Gaussians and Sub-Exponentials}\label{subsec:orlicznorms}
We first introduce a few definitions and results regarding concentration bounds.

\begin{definition}[Orlicz norms]\label{orlicz:def}\emph{
For any $\alpha > 0$, let $\psialpha(\cdot)$ denote the function given by: $\psialpha(x) = \exp(x^{\alpha}) - 1 \; \forall \; x \geq 0$. Then, for any random variable $X$ and any $\alpha > 0$, the \emph{$\psialpha$-Orlicz norm} $\| X \|_{\psialpha}$ of $X$ is defined as:}
\begin{equation*}
\| X \|_{\psialpha} \; := \; \inf \left\{ c > 0: \; \E\{ \psialpha(|X|/c)\}\; \leq \; 1 \right\},
\end{equation*}
\emph{and $X$ is said to have a finite $\psialpha$-Orlicz norm, denoted as $\| X \|_{\psialpha} < \infty$ (if the set above is empty, then the infimum is simply defined to be $\infty$). }
\par\smallskip
\emph{For a \emph{random vector} $\bX \in \R^d$ $(d \geq 1)$, we define $\bX$ to have finite $\psialpha$-Orlicz norm if each coordinate of $\bX$ does %has finite $\psialpha$-Orlicz norm,
and we let $\|\bX\|_{\psialpha} := \max_{1 \leq j \leq d} \| \bX_{[j]}\|_{\psialpha}$.
}
\par\smallskip
\emph{A random variable (or random vector) is said to be \emph{sub-Gaussian} or \emph{sub-exponential} if it has finite $\psialpha$-Orlicz norm with $\alpha = 2$ or $\alpha = 1$ respectively. %, and it is said to be  if it has a finite $\psialpha$-Orlicz norm with $\alpha = 1$.
}
\end{definition}
Note that sub-Gaussians and sub-exponentials also possess other alternative definitions in terms of tail bounds, moment bounds or moment generating functions that are standard in the literature. All these definitions may be shown to be equivalent, upto constant factors in the parameters, to the one above. The $\psialpha$-Orlicz norms are more general norms allowing for any $\alpha > 0$ (not just $1$ or $2$), and hence, weaker tail behaviors. It is also worth noting that a bounded random variable $X$ has $\| X \|_{\psialpha} < \infty$ for \emph{any} $\alpha \in (0, \infty]$. % and hence, has finite $\psialpha$-Orlicz norm with $\alpha = \infty$.

\subsection{Properties of Orlicz Norms and Concentration Bounds}\label{subsec:probbound}

We enlist here some useful general properties of Orlicz norms along with a few specific ones for sub-Gaussians and sub-exponentials. These are all quite well known and routinely used. Their statements (possibly with slightly different constants) and proofs can be found in several relevant references, including \citet{VdVWellner_Book_1996, Pollard_2015, Vershynin_2012, Vershynin_2018, Rigollet_2017} and \citet{Wainwright_Book_2019} among others. We therefore skip their proofs here for the sake of brevity. %The proofs are therefore skipped here for brevity.

\begin{lemma}[General properties of Orlicz norms, sub-Gaussians and sub-exponentials]\label{lem:1:GenProp}
%\begin{itemize}
\hspace{-0.15in} Let $X, Y $ denote generic random variables and let $\mu := \E(X)$.
\begin{enumerate}[(i)]
\item \emph{(Basic properties).} For $\alpha \geq 1$, $\psialphanorm{\cdot}$ is a norm  %(and a quasinorm if $\alpha < 1$)
satisfying: (a) $\psialphanorm{X} \geq 0$ and $\psialphanorm{X} = 0 \Leftrightarrow X = 0$ a.s., (b) $\psialphanorm{c X} = |c| \psialphanorm{X}$ $\forall \; c \in \R$ and $\psialphanorm{|X|} = \psialphanorm{X}$, and (c) $\psialphanorm{X + Y} \leq \psialphanorm{X} + \psialphanorm{Y}$.
\par\smallskip
\item \emph{(Monotonicities).} (a) For any $0 < \alpha \leq \beta$, $(\log 2)^{1/\alpha} \psialphanorm{X} \leq  (\log 2)^{1/\beta}$ $\psibetanorm{X}$. %In particular, $\psionenorm{X} \leq (\log 2)^{-1/2} \psitwonorm{X}$.
(b) For any $\alpha > 0$, $\psionenorm{|X|^{\alpha}} \leq \psialphanorm{X}^{\alpha}$. (c) If $|X| \leq |Y|$ a.s., then $\psialphanorm{X} \leq \psialphanorm{Y}$ $\forall \; \alpha > 0$.  (d) If $X$ is bounded, i.e. $|X| \leq M$ a.s. for some constant $M$, then $\psialphanorm{X} \leq (\log 2)^{-1/\alpha} M$ for each $\alpha \in (0, \infty]$.
\par\smallskip
\item \emph{(Tail bounds and equivalences).} (a) If $\psialphanorm{X} \leq \sigma$, % for some $(\alpha, \sigma)$,
then $\P(|X| > \epsilon) \leq 2 \exp(- \epsilon^{\alpha}/ \sigma^{\alpha})$ $\forall \; \epsilon \geq 0$. (b) Conversely if $\P(|X| > \epsilon) \leq C \exp(- \epsilon^\alpha/ \sigma^\alpha)$ $\forall \; \epsilon \geq 0$, for some $(C, \sigma, \alpha) > 0$, then $\psialphanorm{X} \leq \sigma (1 + C/2)^{1/\alpha}$.
\par\smallskip

\item \emph{(Moment bounds).} If $\psialphanorm{X} \leq \sigma$ for some $(\alpha, \sigma) > 0$, then $\E(|X|^m) \leq C_{\alpha}^m \sigma^m m^{m/\alpha}$ $\forall \; m \geq 1$, for some constant $C_{\alpha}$ depending only on $\alpha$. (A converse also holds although not presented here). In particular, %for $\alpha = 1$ and $2$, we have:
    \begin{enumerate}[(a)]
    \item If $\psionenorm{X} \leq \sigma$, then for each $m \geq 1$, $\E(|X|^m) \; \leq \; \sigma^m  m! \; \leq \; \sigma^m m^m $.
    \item If $\psitwonorm{X} \leq \sigma$, then $\E(|X|^m) \leq 2 \sigma^m \Gamma(m/2 + 1)$ $\forall \; m \geq 1$, where $\Gamma(a) := \int_0^{\infty} x^{a - 1} exp(-x) dx $ $\forall \; a > 0$ denotes the Gamma function. Hence, $\E(|X|) \leq \sigma \sqrt{\pi}$ and $\E(|X|^m) \leq 2 \sigma^m (m/2)^{m/2}$ $\forall \; m \geq 2$.
    \end{enumerate}

\item \emph{(H\"{o}lder-type inequality for the Orlicz norm of products).} For any $\alpha, \beta > 0$, let $\gamma := (\alpha^{-1} + \beta^{-1})^{-1}$. Then, for any $X,Y$ %two random variables $X$ and $Y$
with $\psialphanorm{X} < \infty$ and $\psibetanorm{Y} < \infty$, $\psigammanorm{XY} < \infty$ and $\psigammanorm{XY} \leq  \psialphanorm{X} \psibetanorm{Y}$.
In particular, if $X$ and $Y$ are sub-Gaussian, then $XY$ is sub-exponential and $\psionenorm{XY} \leq \psitwonorm{X} \psitwonorm{Y}$. Further, if $Y$ is bounded with $Y \leq M$ a.s. and $\psialphanorm{X} < \infty$ for any $\alpha > 0$, then $\psialphanorm{XY} \leq M \psialphanorm{X}$.
\par\smallskip
\item \emph{(Orlicz norms and tail bounds for maximums).}  Let $\{X_i\}_{i=1}^n$ ($n \geq 1$) be random variables (possibly dependent) with $\underset{1 \leq i \leq n}{\max} \psialphanorm{X_i} \leq \sigma$ for some $(\alpha, \sigma)$ and let $Z_n := \underset{1 \leq i \leq n}{\max} |X_i|$. Then, $\psialphanorm{Z_n} \leq \sigma (\log n + 2)^{1/\alpha} \leq \sigma \{3 \log (n+1) \}^{1/\alpha}$ %,$\P(Z_n > \epsilon) \leq 2n \exp( - \epsilon^{\alpha}/\sigma^{\alpha})$ $\forall \; \epsilon \geq 0$,
    and $\P\{Z_n > c \sigma (\log n)^{1/\alpha}\} \leq 2 n^{- (c^\alpha-1)}$ $\forall \; c > 1$.
\par\smallskip

\item \emph{(MGF related properties of sub-Gaussians).} Let $\E[\exp\{t(X-\mu)\}]$ denote the moment generating function (MGF) of $X - \mu$ at $t \in \R$. Then:
\begin{enumerate}[(a)]
\item If $\psitwonorm{X - \mu} \leq \sigma$, %for some $\sigma \geq 0$,
then $\E[\exp\{t(X-\mu)\}] \leq \exp(2\sigma^2 t^2)$ $\forall \; t \in \R$.
\item Conversely, if $\E[\exp\{t(X-\mu)\}] \leq \exp(\sigma^2 t^2)$  $\forall \; t \in \R$,  %for some $\sigma \geq 0$,
then $\forall \; \epsilon \geq 0$, $\P(| X - \mu | > \epsilon) \leq 2 \exp(- \epsilon^2/ 4 \sigma^2)$ and hence, $\psitwonorm{X-\mu} \leq 2\sqrt{2} \sigma$.
    \end{enumerate}
\end{enumerate}
\end{lemma}

\begin{lemma}[Concentration bounds for sums of independent sub-Gaussian variables]\label{lem:2:SGconc}
Let $\{X_i\}_{i=1}^n$ ($n \geq 1$) be independent (but not necessarily i.i.d.) random variables with means $\{\mu_i\}_{i=1}^n$ such that $\psitwonorm{X_i - \mu_i} \leq \sigma_i $ for some $\{\sigma_i\}_{i=1}^n \geq 0$. Then, for any set of real numbers $\{a_i\}_{i=1}^n$, we have %and letting $\mathbf{a} = (a_1,\hdots,a_n) \in \R^n$, we have:
    \begin{eqnarray*}
    &&\E \left[ \exp\left\{t \sum_{i=1}^n a_i (X_i -\mu_i) \right\}\right]  \;\leq \; \exp\left(2 t^2 \sum_{i=1}^n\sigma_i^2 a_i^2 \right)\quad \forall \; t \in \R, \quad \mbox{and} \\
    &&\P\left\{ \left| \sum_{i=1}^n a_i (X_i - \mu_i) \right| \; > \; \epsilon \right\} \;\; \leq \; 2 \exp\left(\frac{- \epsilon^2} {8 \sum_{i=1}^n \sigma_i^2 a_i^2 } \right) \quad \forall \; \epsilon \geq 0.
    \end{eqnarray*}
    This further implies that $\psitwonorm{a_i (X_i -\mu_i)} \leq 4 (\sum_{i=1}^n\sigma_i^2 a_i^2)^{1/2}$. In particular, when $a_i = 1/n$ and $\sigma_i = \sigma$, $\psitwonorm{\frac{1}{n} \sum_{i=1}^n (X_i - \mu_i)} \leq (4\sigma)/\sqrt{n}$.
    %%and $ \P \left\{ \left|\frac{1}{n} \sum_{i=1}^n (X_i - \mu_i)\right| >  \epsilon \right\} \leq  2 \exp\left\{- n \epsilon^2 / \left(8 \sigma^2 \right)\right\} \forall \; \epsilon \geq 0.$
\end{lemma}

\begin{lemma}[Sub-Gaussian properties of binary random variables]\label{lem:3:SubgaussianBinary}\hspace{-0.05in}
Let $Z \in \{0,1\}$ be a binary random variable with $\E(Z) \equiv \P(Z = 1) = p \in [0,1]$ and let $\widetilde{Z} = (Z-p)$. %denote the corresponding centered version of $Z$.
Then, $\| \widetilde{Z} \|_{\psitwo} \leq 2\widetilde{p}$, where $\widetilde{p} = 0$ if $p \in \{0,1\}$, $\widetilde{p} = 1/2$ if $p = 1/2$, and $\widetilde{p} = [(p-1/2)/\log \{p/(1-p)\}]^{1/2}$ if $p \notin \{0,1,1/2\}$.
\end{lemma}

Lemma \ref{lem:3:SubgaussianBinary} explicitly characterizes the sub-Gaussian properties of (centered) binary random variables and its proof can be found in \citet{Buldygin_2013}. The statement therein uses a MGF based definition of sub-Gaussians. The statement above is appropriately modified with the factor $2$ multiplied in the $\psitwonorm{\cdot}$ norm bound to adapt to our definition.

Next, we present a version of the well known Bernstein's inequality. While Lemma \ref{lem:2:SGconc} is useful, it applies only to sub-Gaussians. However, Bersntein's inequality applies more generally to sub-exponentials that include as special cases: sub-guassian variables, bounded variables, as well as products of two sub-Gaussian and/or bounded variables (see Lemma \ref{lem:5:BMC}).

\begin{lemma}[Bernstein's inequality - adopted from \citet{VdG_2013}]\label{lem:4:Bernstein}
Let $\{Z_i\}_{i=1}^n$ be independent (but not necessarily i.i.d.) random variables and let $\mu_i := \E(Z_i)$ $\forall \; 1 \leq i \leq n$. Suppose $\exists$ constants $\sigma, K \geq 0$ such that $n^{-1} \sum_{i=1}^n \E(|Z_i - \mu_i|^m) \leq (m!/2) \sigma^2 K^{m-2} $ for each $m \geq 2$. Then, %the following concentration bound holds:
 \begin{equation*}
 \P \left( \left| \frac{1}{n} \sum_{i=1}^n (Z_i - \mu_i) \right| \; \geq \; \sqrt{2}\sigma \epsilon + K \epsilon^2\right) \; \leq \; 2 \exp\left(-n\epsilon^2\right) \quad \mbox{for any} \; \epsilon \geq 0.
 \end{equation*}
In particular, if $\{Z_i\}_{i=1}^n$ are i.i.d. realizations of a sub-exponential variable $Z$ with $\E(Z) = \mu$ and $\psionenorm{Z} \leq \sigma_Z$ for some $\sigma_Z \geq 0$, then $\psionenorm{Z - \mu} \leq 2 \sigma_Z$ and the bound above holds with $\sigma \equiv 2\sqrt{2}\sigma_Z$ and $K \equiv 2\sigma_Z$. Two important special cases of such a setting include: (a) $Z = XY$ with $X$ and $Y$ sub-Gaussian, in which case $\sigma_Z \leq \psitwonorm{X} \psitwonorm{Y}$, and (b) $Z = XY$ with $X$ sub-exponential and $|Y| \leq M$ a.s. for some $M > 0$, in which case $\sigma_Z \leq M\psionenorm{X}$.
\end{lemma}

\begin{lemma}[The Bernstein moment conditions and their verification]\label{lem:5:BMC}
 Consider the moment conditions required in Bernstein's inequality in Lemma \ref{lem:4:Bernstein}. Define a random variable $Z$ to satisfy the Bernstein moment conditions (BMC) with parameters $(\sigma, K) \geq 0$, denoted as $Z \sim \BMC(\sigma, K)$, if for each $m \geq 2$, $\E(|Z - \mu|^m) \leq (m!/2) \sigma^2 K^{m-2}$ where $\mu := E(Z)$. Then,
\par\smallskip
\noindent (a) If $Z$ is sub-exponential with $\psionenorm{Z} \leq \sigma_Z$, then $Z \sim \BMC(2\sqrt{2}\sigma_Z, 2\sigma_Z)$ and $|Z| \sim \BMC(2\sqrt{2}\sigma_Z, 2\sigma_Z)$.
\par\smallskip
\noindent (b) If $X$ and $Y$ sub-Gaussian variables, then $Z := XY \sim \BMC(2 \sqrt{2}\sigma_Z, 2\sigma_Z )$ with $\sigma_Z = \psitwonorm{X} \psitwonorm{Y}$.
\par\smallskip
\noindent (c) If $X$ is sub-exponential and $Y$ is a bounded random variable with $|Y| \leq M$ a.s., then $Z := XY \sim \BMC(2 \sqrt{2} \sigma_Z, 2\sigma_Z)$ with $\sigma_Z = M\psionenorm{X}$.
\end{lemma}

\begin{proof}
If $\psionenorm{Z} \leq \sigma_Z$, then using Lemma \ref{lem:1:GenProp} (i)(c) and (iv)(a), $\psionenorm{Z - \mu} \leq 2 \sigma_{Z}$ and $\E (|Z - \mu|^m) \leq (2 \sigma_Z)^m m! \equiv (m!/2) (2\sqrt{2}\sigma_Z)^2 (2\sigma_Z)^{m-2} $ for each $m \geq 1$. Hence, by definition, $Z \sim \BMC(2\sqrt{2}\sigma_Z, 2\sigma_Z)$. \qed

Similarly, $\psionenorm{|Z|} = \psionenorm{Z} \leq \sigma_Z$ and $\psionenorm{|Z| - \E\{|Z|\}} \leq 2 \sigma_{Z}$. Therefore, by identical arguments as above we again have: $|Z| \sim \BMC(2\sqrt{2}\sigma_Z, 2\sigma_Z)$. \qed

Finally, using Lemma \ref{lem:1:GenProp}, we have: for case (b), $\psionenorm{Z} \leq \psitwonorm{X} \psitwonorm{Y} \equiv \sigma_Z$, while for case (c), $\psionenorm{Z} \leq M\psionenorm{X} \equiv \sigma_Z$. The desired results then follow by using the same arguments used for proving the first result above. %\qed
\end{proof}

The following lemma is a useful concentration inequality that applies generally to any random variables with finite $\psi_{\alpha}$-Orlicz norm, preserves the right rate and tail behaviors and involves only the variance in the leading term.

\begin{lemma}[Concentration bounds with variance in the leading term - adopted from Theorem 3.4 of \citet{KC_MBSG_2018}]\label{lem:6:VarianceTailBound}
Suppose $\{\bX_i\}_{i=1}^n$ are independent mean zero random vectors in $\R^p$, for any $p \geq 1$ and $ n \geq 2$, such that for some $\alpha > 0$ and some $K_n > 0$,
\[
\max_{1\leq i\leq n}\max_{1\leq j\leq p}\,\psialphanorm{\bX_{i[j]}} \leq K_n, \;\; \mbox{and define} \;\; \Gamma_{n} := \max_{1\leq j\leq q}\,\frac{1}{n}\sum_{i=1}^n \E\left(\bX_{i[j]}^2\right).
\]
Then for any $t \geq 0,$ with probability at least $1 - 3e^{-t}$,
\begin{align*}
\left\|\frac{1}{n}\sum_{i=1}^n \bX_i\right\|_{\infty} \leq 7\sqrt{\frac{\Gamma_{n}(t + \log p)}{n}} + \frac{C_{\alpha}K_n (\log n)^{1/\alpha}(t + \log p)^{1/\alpha^*}}{n}, %% replace \log n by \log (2n) if n geq 1 only %%
\end{align*}
where $\alpha^* := \min\{\alpha,1\}$ and $C_{\alpha} > 0$ is some constant depending only on $\alpha$.
\end{lemma}

Finally, we end with a simple lemma that relates high probability bounds to sub-Gaussian type tail bounds with an extra probability correction term.

\begin{lemma}[High probability bounds to sub-Gaussian type tail bounds]\label{lem:7:hptosgtypebound}
Let $X_n$ be any sequence of random variables satisfying $|X_n| \leq a_n$ with probability at least $1 - q_n$ for some $a_n \in [0, \infty)$ and $q_n \in [0,1]$, $\forall \; n \geq 1$. Then, %for any $t \geq 0$,
\[
\P(|X_n| > t) \; \leq \; 2 \exp\left\{ -t^2/(2a_n^2)\right\} + q_n \;\; \mbox{for any} \;\; t \geq 0.
\]
\end{lemma}

\begin{proof}
Define the event $\Asc_n := \{ X_n \leq a_n \}$ and let $\Asc_n^c$ denote its complement event. Then, $\P (\Asc_n^c) \leq q_n$ by assumption. Furthermore, note that $|X_n 1(\Asc_n)| \leq a_n$ a.s. $[\P]$, where $1(\cdot)$ denotes the indicator function. Hence, using Lemma \ref{lem:1:GenProp} (ii) (d), we have: $\psitwonorm{X_n 1(\Asc_n)} \leq (\log 2)^{-1/2} a_n \leq \sqrt{2} a_n$.

Hence, using Lemma \ref{lem:1:GenProp} (iii) (a), $\P \{ |X_n 1(\Asc_n)| > t \} \leq 2 \exp\{-t^2/(2a_n^2)\}$ for any $t \geq 0$. Consequently, we have: for any $t \geq 0$,
\begin{align*}
\P(|X_n| > t) & \; = \; \P(|X_n| > t, \Asc_n) + \P(|X_n| > t, \Asc_n^c) \\
& \; \leq \; \P(|X_n 1(\Asc_n) | > t) + \P(\Asc_n^c) \; \leq \; 2 \exp\{-t^2/(2a_n^2)\} + q_n.
\end{align*}
This establishes the desired tail bound and completes the proof.
\end{proof}

\section{Technical Discussions on the Error Terms}\label{discussion:errorterms}
%\subsection[Some Key Technical Discussions on the Structure of the Error Terms]{Some Key Technical Discussions on the Structure of the Error Terms}\label{discussion:errorterms}

We note here a few useful %important
details regarding the structure and techniques for controlling the error terms $\bTpin$, $\bTmn$ and $\bRpimn$ accounting for the nuisance function estimators $\{ \pihat(\cdot), \mhat(\cdot) \}$ in the decomposition \eqref{decomp:eqn} of $\bT_n$.

\paragraph*{(a) The structure of $\bTpin$ and reasons for obtaining $\pihat(\cdot)$ solely from $\Xsc_n$} $\bTpin$ is simply the sample average of the random variables $\{\bTpi(\bZ_i)\}_{i=1}^n$ in \eqref{tpi}. However, this average is \emph{not} an i.i.d. average due to the presence of $\pihat(\cdot)$ which depends on all observations in $\Dsc_n$. A key property that is quite useful in this regard is that, by assumption, $\pihat(\cdot)$ is obtained solely from the subset $\Xsc_n := \{ (T_i,\bX_i): i = 1,\hdots,n\}$ of $\Dsc_n$. Hence, $\{\bTpi(Z_i)\}_{i=1}^n \medgiven \Xsc_n$ are \emph{conditionally} independent and centered with $\E\{\bTpi(\bZ_i)\} = \E[\E\{\bTpi(\bZ_i) \given \pihat(\cdot), \bX_i\}]$ $= \E[\E\{\bTpi(\bZ_i)\given \Xsc_n\}]$ $= \bzero$. %owing to the definitions of the underlying quantities involved, the nature of the construction of $\pihat(\cdot)$, and Assumption \ref{base:assmpns} (a).
The conditioning on $\Xsc_n$ ensures that $\pihat(\cdot)$, as well as all other components in $\bTpi(\bZ_i)$ which are functions of $(T_i,\bX_i)$ only, can now be treated as fixed and further, the conditional expectation being $\bzero$ follows from the fact that $\E[\{Y_i -m(\bX_i)\} \given \Xsc_n]$ $ \equiv \E\{\varepsilon(\Z_i) \given \Xsc_n\}$ $ = \E \{ \varepsilon(\Z_i) \given \T_i, \bX_i \}$ = $\E\{ \varepsilon(\Z_i) \given \bX_i\} = 0$, where the final step is due to  Assumption \ref{base:assmpns} (a).

Thus, $\bTpin$ is a centred average of (conditionally) independent variables. We exploit this and the structure of $\bTpi(\bZ)$ in Theorem \ref{TPI:THM} to control $\bTpin$. %But first, a few remarks are in order.

\paragraph*{(b) The structure of $\bTmn$ and the benefits of sample splitting/cross-fitting} $\bTmn$ is simply the sample average of the random variables $\{\bTm(\bZ)\}_{i=1}^n$ in (\ref{tm}). However, in the absence of sample splitting, this is \emph{not} an i.i.d. average due to the presence of $\mhat(\cdot)$ which depends on all observations in $\Dsc_n$. Further, unlike $\bTpin$ where $\{\bTpi(\bZ_i)\}_{i=1}^n\given \Xsc_n$ were at least (conditionally) independent and centered, $\bTmn$ possesses no such desirable features even if $\mhat(\cdot)$ is obtained solely from the subset $\Dsc_{n}^{(c)} := \{ (Y_i,\bX_i): T_i = 1, 1 \leq i \leq n\}$ of `complete cases' in $\Dsc_n$, as $\Dsc_{n}^{(c)}$ still (implicitly) depends on $\{T_i\}_{i=1}^n$ due to the restriction to the set with $T_i = 1$, and not just on $\{Y_i, \bX_i \}_{i=1}^n$. %Moreover, for any of the training points $\{\bZ_i\}_{i=1}^n$, it is also not clear whether $\bTm(\bZ_i)$ is even centered (unconditionally or conditionally).

Thus, in the absence of sample splitting, $\bTmn$ has no additional `structure' readily available that may lead to averages of variables which can be treated as conditionally independent and centered. In general, to control $\bTmn$ without sample splitting, one needs tools from empirical process theory. The corresponding analyses can be substantially involved and the conditions necessary can be quite strong, especially in high dimensional settings. However, these technical issues can be avoided through the sample splitting based estimates $\{\mtil(\bX_i)\}_{i=1}^n$ which `induces' a natural independence.

For any $\bZ \ind \mhat(\cdot)$, or more specifically, $\bZ \ind$ \{data used to obtain $\mhat(\cdot)\}$, $\E\{\bTm(\bZ) \given \mhat(\cdot), \bX\}= \E\{\bTm(\bZ) \medgiven \bX\} = \bzero$ due to Assumption \ref{base:assmpns} (a). Hence, $\E\{\bTm(\bZ) \given \mhat(\cdot)\} = \bzero$ and %$\E\{\bTm(\bZ) \} = \bzero$
for any i.i.d. collection $\{\bZ_k\}_{k=1}^K$  of $\bZ \ind \mhat(\cdot)$, $\{\bTm(\bZ_k)\}_{k=1}^K \given \mhat(\cdot)$ are (conditionally) independent and centered random variables. These serve as the main motivations behind the sample splitting.

In contrast to the `in-sample' estimates $\{ \mhat(\bX_i\}_{i=1}^n$, wherein $\mhat(\cdot)$ is obtained from $\Dscn$ and also evaluated at the same training points $\{\bX_i\}_{i=1}^n \in \Dscn$, thereby making them intractably dependent on $\mhat(\cdot)$, the cross-fitted estimates $\{ \mtil(\bX_i)\}_{i=1}^n$ ensure that for each $k \neq k' \in \{1,2\}$, the evaluation points $\{\bX_i \in \Dscnk\}$ used are independent of the estimator $\mhat^{(k')}(\cdot)$ obtained from $\Dscn^{(k')} \ind \Dscnk$, thus inducing a desirable `independence structure'. %between the training and the evaluation points.
This has substantial technical as well as practical benefits in reducing over-fitting. We exploit the technical benefits  greatly in Theorem \ref{TM:THM} to control $\bTmn$.

\paragraph*{(b) The structure of $\bRpimn$} Finally, note that $\bRpimn$ is essentially a second order (product-type) bias term involving the product of two error terms arising from the estimation of $\{\pi(\cdot), m(\cdot)\}$. Under reasonable assumptions on the convergence rates of the estimators $\{\pihat(\cdot), \mhat(\cdot)\}$, one can try to control the behavior of this term by `naive' techniques, as opposed to the more sophisticated analyses required for controlling $\bTpin$ and $\bTmn$. Such techniques and associated conditions are well known and standard in the literature for the special case of the mean estimation problem (or ATE estimation problem in CI), where a commonly adopted assumption is to have the product of the two convergence rates to be faster than $n^{-0.5}$ \citep{Farrell_2015, Chern_DDML_2018}. In general, such product conditions are typically reasonable and allows for much weaker (slower) convergence rates for one estimator as long as the other one has sufficiently fast enough rates. A stronger but familiar sufficient condition however is to have the convergence rates of both estimators to be faster than $n^{-0.25}$. In Theorem \ref{RMPI:THM}, we control $\bRpimn$ by adopting a similar condition with an additional logarithmic factor involved to account for the inherent high dimensionality of our error terms.

\section{Proof of Lemma \ref{DEV:BOUND}}\label{pf:dev:bound}

The proof relies substantially on a useful result of \citet{Negahban_2012}. We therefore adopt some of their basic notations and terminology at the beginning of the proof in order to facilitate the use of that result.

For any $\bu \in \R^p$, let $\Rsc(\bu) = \| \bu \|_1$ and let $\Rsc^*(\bu) \equiv %\underset{\bv \in \R^p \backslash \{\bzero\}}{\mbox{sup}}
\mbox{sup}_{\bv \in \R^p \backslash \{\bzero\}}
\{\bu'\bv/ \Rsc(\bv)\}$ be the `dual norm' for $\Rsc(\cdot)$. Further, for any subspace $\Msc \subseteq \R^p$, let $\Psi(\Msc) \equiv \mbox{sup}_{\bu \in \Msc \backslash \{\bzero \}}
%\underset{\bu \in \Msc \backslash \{{\bzero} \}}{\mbox{sup}}
\{\Rsc(\bu)/ \| \bu\|_2\}$ denote its %the
`subspace compatibility constant' %for $\Msc$
with respect to %(w.r.t.) %the norm
$\Rsc(\cdot)$. Then, with $\Jsc, \Msc_{\Jsc}$ and $\Msc^{\perp}_{\Jsc}$ as defined
in Section \ref{estimation}, %at the beginning of Section \ref{P2:psetup},
it is not difficult to show that: (i) $\Rsc(\cdot)$ is \emph{decomposable} with respect to the orthogonal subspace pair $(\Msc_{\Jsc},\Msc^{\perp}_{\Jsc})$ for any $\Jsc \subseteq  \{1,\hdots,p\}$, in the sense that $\Rsc(\bu + \bv) = \Rsc(\bu) + \Rsc(\bv)$ $\forall \; \bu \in \Msc_{\Jsc}, \bv \in \Msc^{\perp}_{\Jsc}$; (ii) $\Rsc^*(\bu) = \| \bu \|_{\infty} \; \forall \; \bu \in \R^p$; and (iii) with $\Jsc = \Asc(\bv)$ for any $\bv \in \R^p$, $\Psi^2(\Msc_{\Jsc}) = s_{\bv}$. (We refer %the %interested reader
to \citet{Negahban_2012} for further discussions and/or proofs of these facts). Lastly, let $P_{\Jsc}(\bv)$ and $P_{\Jsc}^{\perp} (\bv)$ respectively denote the orthogonal projections of any $\bv \in \R^p$ onto $\Msc_{\Jsc}$ and $\Msc_{\Jsc}^{\perp}$, for any $\Jsc$ as above.
%\par\smallskip

To establish the result, we consider the alternative representation \eqref{eq:simplealgo} of $\btheta_0$ based on regularized minimization of the pseudo loss $\LntilDR(\btheta)$ defined in \eqref{eq:pseudo:outcome:def}. Clearly, since $L(\cdot)$ is convex and differentiable in $\btheta$ as assumed, so is $\LntilDR (\btheta)$. Further, owing to \eqref{lossfn:splform1}-\eqref{lossfn:splform2}, we have: for any $\btheta, \bv \in \R^d$,
\begin{align}
\bnabla \LntilDR(\btheta) \; = \; \bnabla \LnDR(\btheta) \;\; \mbox{and} \;\;  \delta \LntilDR(\btheta,\bv) \; = \; \delta \LnDR(\btheta, \bv), \label{pf:basic:lem:eq1}
\end{align}
where $\delta \LntilDR(\btheta,\bv) := \LntilDR(\btheta + \bv) - \LntilDR(\btheta) - \bv'\bnabla \LntilDR(\btheta)$. Thus, under Assumption \ref{strngconv_assmpn}, $\LntilDR(\btheta)$ also satisfies the RSC property \eqref{strngconv_eqn} at $\btheta = \btheta_0$.

Hence, using the decomposability of $\Rsc(\cdot)$ over $(\Msc_{\Jsc},\Msc^{\perp}_{\Jsc})$ with $\Jsc$ chosen to be $\Asc(\btheta_0)$, and the RSC property
of $\LntilDR(\btheta)$ at $\btheta = \btheta_0$ under Assumption \ref{strngconv_assmpn} and \eqref{pf:basic:lem:eq1}, we have: by Theorem 1 of \citet{Negahban_2012}, for any realization of $\Dsc_n$ and any choice of $\lambda \equiv \lambda_n \geq 2 \|\bnabla \LnDR(\btheta_0\|_{\infty}$,
\begin{align}
& \left\| \bthetahatDR - \btheta_0 \right\|_2 \; \equiv \; \left\| \bthetahatDR(\lambda_n; \Dsc_n) - \btheta_0 \right\|_2 \;\; \leq \; 3 \sqrt{s} \frac{\lambda}{\kappaDR}  \label{pf:basic:lem:eq2}
\end{align}
where, while applying the result from \citet{Negahban_2012}, we chose the parameter $\btheta^*$, in their notation, as $\btheta^* = \btheta_0$, $\{\Rsc(\cdot), \Rsc^*(\cdot)\}$ as $\{ \| \cdot \|_1, \| \cdot ||_{\infty}\}$, and used:  $\Psi^2(\Msc_{\Jsc}) = \| \btheta_0 \|_0 \equiv s$, $\Rsc^*[ \boldsymbol{\nabla}\{\LntilDR(\btheta)\}] =  \Rsc^*[ \boldsymbol{\nabla}\{\LnDR(\btheta)\}] \equiv \|\bnabla \LnDR(\btheta_0)\|_{\infty}$ and $P_{\Asc(\btheta_0)}^{\perp} (\btheta_0) = \Pi_{\Asc^c(\btheta_0)}(\btheta_0) \equiv \Pi^c_{\btheta_0}(\btheta_0) = \bzero$.  \qed

\par\smallskip
Further, using Lemma 1 of \citet{Negahban_2012}, we also have that for $\lambda$ chosen as above, the error $\bDeltahat := (\bthetahatDR - \btheta_0)$ belongs to the set  $\C(\btheta_0)$ as defined in \eqref{strngconv_eqn}. Consequently, $\| \Pi_{\btheta_0}^c(\bDeltahat)\|_1 \leq 3 \| \Pi_{\btheta_0}(\bDeltahat)\|_1$. Hence we have:
\begin{align*}
& \left\| \bthetahatDR - \btheta_0 \right\|_1 \; \equiv \| \bDeltahat \|_1 \; = \;  \| \Pi_{\btheta_0}(\bDeltahat)\|_1 + \| \Pi_{\btheta_0}^c(\bDeltahat)\|_1  \;\leq \; 4  \| \Pi_{\btheta_0}(\bDeltahat)\|_1 \nonumber \\
& \;\; \leq \; 4 \sqrt{s} \| \Pi_{\btheta_0}(\bDeltahat)\|_1 \; \leq \; 4 \sqrt{s} \left\| \bthetahatDR - \btheta_0 \right\|_2 \; \leq \;  12 s \frac{\lambda}{\kappaDR},
\end{align*}
where the final step follows from using \eqref{pf:basic:lem:eq2}. This, along with \eqref{pf:basic:lem:eq2}, establishes the desired $L_2$ and $L_1$ error bounds for $\bthetahatDR$. %\qed
The rest of the informal claims in the second part of Lemma \ref{DEV:BOUND} are straightforward consequences of combining the deterministic error bounds proved above with the results of Theorems \ref{TZERO:THM}-\ref{RMPI:THM}. This completes the proof of Lemma \ref{DEV:BOUND}. \qed

\section{Proof of Theorem \ref{TZERO:THM}}\label{pf:tzero:thm}

Recalling from \eqref{decomp:eqn} and \eqref{tzero}, we note that $\bTzeron$ is simply a sum of two centered i.i.d. averages given by:
\begin{align}
& \bTzeron  \; = \; \bTzeroonen + \bTzerotwon \; \equiv \; \frac{1}{n}\sum_{i=1}^n \bTzeroone(\bZ_i) \; + \; \frac{1}{n}\sum_{i=1}^n \bTzerotwo(\bZ_i), \quad \mbox{where} \label{tzeron:decomp} %\\
 \end{align}
\vspace{-0.1in}
\begin{equation*}
 \bTzeroone(\bZ)  :=  \{m(\bX) - g(\bX,\btheta_0)\} \bh(\bX) \;\; \mbox{and} \; \hspace{0.03in} \bTzerotwo(\bZ)  :=  \frac{T}{\pi(\bX)} \{ Y - m(\bX)\} \bh(\bX), \nonumber
\end{equation*}
with $\E\{\bTzeroone(\bZ)\} = \bzero$ and $\E\{\bTzerotwo(\bZ)\} = \bzero$ since $\E \{ \bnabla \bphi(\bX, \btheta_0) \} = \bzero$ and $\E\{\epsilon(\Z) \medgiven \bX\} = 0$, by definition,  and $\epsilon(\Z) \ind T \given \bX$ due to Assumption \ref{base:assmpns} (a).

Now, using Assumption \ref{subgaussian:assmpn} (a) and Lemma \ref{lem:5:BMC} (a), we have:
\begin{equation}
\bTzeroonej(\bZ) \; \equiv \; \psi(\bX) \bhj(\bX) \; \sim \; \BMC(\sigmabarone, \Kbarone) \quad \forall \; j \in \{1, \hdots , d\}, \label{tzeroone:BMC:bound}
\end{equation}
for some constants $\sigmabarone := 2\sqrt{2} \sigmapsi\sigmabh \geq 0$ and $\Kbarone := 2 \sigmapsi\sigmabh  \geq 0$.

Next, using Assumption \ref{subgaussian:assmpn} (a) and Lemma \ref{lem:1:GenProp} (v), $\psionenorm{\varepsilon(\Z) \bhj(\bX)} \leq \sigmaeps \sigmabh$ for each $j \in \{1, \hdots, d \}$. Further, owing to Assumption \ref{base:assmpns} (b) and (\ref{pos:eqn}), $T/\pi(\bX) \leq \deltapi^{-1}$ a.s. $[\P]$. Hence, using Lemma \ref{lem:5:BMC} (b), we have
\begin{equation}
\bTzerotwoj(\bZ) \; \equiv \; \frac{T}{\pi(\bX)} \varepsilon(\Z) \bhj(\bX) \; \sim \; \BMC(\sigmabartwo, \Kbartwo) \quad \forall \; j \in \{1, \hdots , d\},  \label{tzerotwo:BMC:bound}
\end{equation}
for some constants $\sigmabartwo := 2\sqrt{2} \sigmaeps\sigmabh \deltapi^{-1} \geq 0$ and $\Kbartwo := 2 \sigmaeps\sigmabh \deltapi^{-1}  \geq 0$

Hence, \eqref{tzeroone:BMC:bound} and \eqref{tzerotwo:BMC:bound} ensure that for each $j \in \{1, \hdots, d\}$, $\bTzeroonej(\bZ)$ and $\bTzerotwoj(\bZ)$ satisfy the required moment conditions for Bernstein's inequality (Lemma \ref{lem:4:Bernstein}) to apply. Using Lemma \ref{lem:4:Bernstein}, we then have: for any $\epsilon_1 \geq 0$,
\begin{eqnarray}
\nonumber && \P\left\{ \left\| \bTzeroonen \right\|_{\infty} \; \equiv \; \left\| \frac{1}{n} \sum_{i=1}^n \bTzeroone(\bZ_i)\right\|_{\infty} > \; \sqrt{2}\sigmabarone\epsilon_1 + \Kbarone \epsilon_1^2 \right\} \\
\nonumber && \;\;\; \leq \;\; \sum_{j=1}^d \P \left\{ \left| \frac{1}{n} \sum_{i=1}^n \bTzeroonej(\bZ_i)\right| > \sqrt{2}\sigmabarone\epsilon_1 + \Kbarone \epsilon_1^2 \right\}  \\
 && \;\;\; \leq \;\; \sum_{j=1}^d  2\exp\left( - n \epsilon_1^2\right) \; = \; 2d\exp\left( - n \epsilon_1^2\right) \;\; \equiv \;\; 2\exp\left( - n \epsilon_1^2 + \log d \right), \label{tzeroonen:bound}
\end{eqnarray}
where the second step uses the union bound (u.b.). Similarly, for any $\epsilon_2 \geq 0$,
\begin{eqnarray}
\nonumber && \P\left\{ \left\| \bTzerotwon \right\|_{\infty} \; \equiv \; \left\| \frac{1}{n} \sum_{i=1}^n \bTzerotwo(\bZ_i)\right\|_{\infty} > \; \sqrt{2}\sigmabartwo\epsilon_2 + \Kbartwo \epsilon_2^2 \right\}   \\
\nonumber && \;\;\; \leq \; \sum_{j=1}^d \P \left\{ \left| \frac{1}{n} \sum_{i=1}^n \bTzerotwoj(\bZ_i)\right| > \sqrt{2}\sigmabartwo\epsilon_2 + \Kbartwo \epsilon_2^2 \right\} \\
&& \;\;\; \leq \; \sum_{j=1}^d  2\exp\left( - n \epsilon_2^2\right)  \; = \; 2d\exp\left( - n \epsilon_2^2\right) \;\equiv \; 2\exp\left( - n \epsilon_2^2 + \log d \right). \label{tzerotwon:bound}
\end{eqnarray}
Hence, setting $\epsilon_1 = \epsilon_2 \equiv \epsilon$ for any $\epsilon \geq 0$, letting $\sigma_0 := \sigmabarone + \sigmabartwo$ and $K_0 := \Kbarone + \Kbartwo$, and using (\ref{tzeroonen:bound})-(\ref{tzerotwon:bound}) in the original decomposition (\ref{tzeron:decomp}) of $\bTzeron$, we have a tail bound for $\|\bTzeron\|_{\infty}$, as follows. For any $\epsilon \geq 0$,
\begin{eqnarray}
\nonumber && \P \left( \left\| \bTzeron \right\|_{\infty} \; \equiv \; \left\|  \bTzeroonen +  \bTzerotwon  \right\|_{\infty} > \; \sqrt{2} \sigma_0 \epsilon + K_0 \epsilon^2 \right) \\
%\nonumber && \;\;\;\; \leq \;\; \P \left( \left\|  \bTzeroonen \right\|_{\infty} + \left\|  \bTzerotwon  \right\|_{\infty} \; > \; \sqrt{2} \sigma_0 \epsilon + K_0 \epsilon^2 \right) \\
\nonumber && \;\;\;\; \leq\;\;  \P\left( \left\| \bTzeroonen \right\|_{\infty} > \sqrt{2}\sigmabarone\epsilon + \Kbarone \epsilon^2 \right)  +  \P\left( \left\| \bTzerotwon \right\|_{\infty} > \sqrt{2}\sigmabartwo\epsilon + \Kbartwo \epsilon^2 \right) \\
&& \quad \leq \; 4 \exp\left( - n \epsilon^2 + \log d \right). \label{btzeron:finalbound}
\end{eqnarray}
(\ref{btzeron:finalbound}) therefore establishes a general tail bound for $\| \bTzeron\|_{\infty}$ and also establishes its rate of convergence. This completes the proof of Theorem \ref{TZERO:THM}. \qed

\section{Proof of Theorem \ref{TPI:THM}}\label{pf:tpi:thm}
To establish Theorem \ref{TPI:THM}, we first state and prove a more general result that gives an explicit tail bound for $\|\bTpin \|_{\infty}$.

\begin{theorem}[Tail bound for $\| \bTpin \|_{\infty}$]\label{tpi:thm:tailbound}
Let Assumptions \ref{base:assmpns}, \ref{subgaussian:assmpn} and \ref{tpicont:assmpn} hold with the sequences $(\vnpi, q_{n,\pi})$ and the constants $(\deltapi,\sigmaeps, \sigmabh, C)$ as defined therein, %and let $\bmusqbhinf := \max\{ \E\{\bhj^2(\bX)\}: j = 1, \hdots, d\}$.
Then, for any $\epsilon, \epsilon_1, \epsilon_2, \epsilon_3 \geq 0$, with $\epsilon_2 < \deltapi$ small enough,
\begin{alignat*}{2}
& && \P \left( \left\| \bTpin \right\|_{\infty} > \epsilon \right)  \; \leq \;  2\exp\left\{ \frac{ - n \epsilon^2}{d_n(\epsilon_1, \epsilon_2, \epsilon_3)}  +  \log d\right\}  +   4 \exp\left( - n \epsilon_3^2 + \log d \right)  \\
& && \qquad +  2C \exp \left\{ \frac{ - \epsilon_1^2}{v_{n,\pi}^2} + \log (nd)\right\}   +  2C \exp \left\{ \frac{- \epsilon_2^2}{v_{n,\pi}^2} + \log(nd) \right\} + 4q_{n,\pi}(nd) ,
\end{alignat*}
where, for any $(\epsilon_1, \epsilon_2,\epsilon_3) \geq 0$ as above, $d_n(\epsilon_1,\epsilon_2,\epsilon_3) \geq 0$ is given by:
\begin{equation*}
d_n(\epsilon_1, \epsilon_2, \epsilon_3) \; := \; \frac{8 \sigmaeps^2 \epsilon_1  ^2}{(\deltapi - \epsilon_2 )^2} \left(\frac{\bmusqbhinf }{\deltapi} + \sqrt{2} \sigma_{\pi} \epsilon_3 + K_{\pi} \epsilon_3^2 \right), \;\; \mbox{with}
\end{equation*}
$\bmusqbhinf := \max_{1 \leq j \leq d}\E\{\bhj^2(\bX)\}$, $\sigma_{\pi} := 2 \sqrt{2}\sigmabh^2 \deltapi^{-2}$ and $K_{\pi} := 2\sigmabh^2 \deltapi^{-2}$.
\end{theorem}

\subsection{Proof of Theorem \ref{tpi:thm:tailbound}}\label{pf:tpi:thm:tailbound}
%\paragraph*{Some notations}
Let $\Xsc_n := \{ (T_i,\bX_i): i = 1,\hdots,n\}$. Let $\E_{\Xscn}(\cdot)$ and $\P_{\Xscn}(\cdot)$ respectively denote expectation and probability w.r.t. $\Xscn$ and $\P(\cdot \given \Xscn)$ denote conditional probability given $\Xscn$. Next, let us define:
\begin{alignat}{4}
& \Deltapin(\bX) && \; := \; \pihat(\bX) - \pi(\bX),  \;\; && \Deltapininfn  && \; := \; \underset{1 \leq i \leq n}{\max} \left| \Deltapin(\bX_i) \right|,   \label{tpicont:notn1} \\
& \pitiln(\bX) && \; := \; - \; \frac{1}{\pihat(\bX)} \;\;\; \mbox{and} \;\; && \pitilninfn &&  \; := \; \underset{1 \leq i \leq n}{\max} \left| \pitiln(\bX_i) \right|. \label{tpicont:notn2}
\end{alignat}
Further, for each $j \in \{1,\hdots, d\}$, let us define:
\begin{eqnarray}%\label{}
 && \bphij(T, \bX)  :=  \frac{T}{\pi(\bX)}\bhj(\bX), \;\; \bphibarsqnj \equiv \bphibarsqnj (\Xscn)  :=  \frac{1}{n} \sum_{i=1}^n \bphij^2(T_i, \bX_i), \label{tpicont:notn3} \\
 && \bmusqbphij  := \E \left\{\bphij^2(T, \bX) \right\} \equiv  \E \left\{ \bphibarsqnj (\Xscn) \right\} \; \mbox{and} \;  \bmusqbhj :=  \E\left\{ \bhj^2(\bX) \right\}. \label{tpicont:notn4}
\end{eqnarray}
Using (\ref{tpicont:notn1})-(\ref{tpicont:notn3}) in \eqref{tpi} and recalling that $\varepsilon(\Z) = Y - m(\bX)$, we have:
\begin{equation}
\bTpi(\bZ) %&\equiv& \left\{\frac{T}{\pihat(\bX)} - \frac{T}{\pi(\bX)} \right\} \{ Y - m(\bX)\}  \bh(\bX) \\
\; = \; \Deltapin(\bX) \pitiln(\bX) \bvphi(T,\bX) \varepsilon(\Z), \quad \mbox{where} \label{tpi:newnotn}
\end{equation}
$\bvphi(T,\bX) \in \R^d$ denotes the vector with $j^{th}$ entry $ = \bphij(T, \bX)$ $\; \forall \; 1 \leq j \leq d$.
\par\smallskip
Under Assumptions \ref{base:assmpns} (a) and \ref{subgaussian:assmpn} (b), $\E\{\varepsilon(\Z) \given \bX\} \equiv \E\{\varepsilon(\Z) \given T, \bX\} = 0$ and $ \psitwonorm{\varepsilon(\Z) \given \bX}$ $\equiv \psitwonorm{\varepsilon(\Z) \given (T,\bX)} $ $\leq \sigmaeps(\bX) \leq \sigmaeps < \infty$. Hence, $ \varepsilon(\Z_i) \given \Xscn$ are (conditionally) independent random variables satisfying: $\E\{\varepsilon(\Z_i) \given \Xscn\} = 0$ and $\psitwonorm{\varepsilon(\bZ_i) \given \Xscn} \leq \sigmaeps \; \forall \; 1 \leq i \leq n$. Further, conditional on $\Xscn$, $\phi(T_i,\bX_i)$, $\Deltapin(\bX_i)$ and $\bh_{[j]}(\bX_i)$ are all constants $\forall \; i, j$. Using these facts along with (\ref{tpicont:notn1})-(\ref{tpicont:notn3}), we have: $\forall \; 1 \leq i \leq n$ and $1 \leq j \leq d$,
\begin{align*}
& \bigpsitwonorm{\bTpij (\bZ_i) \; \given[\big] \Xscn}  \; \equiv \;  \bigpsitwonorm{\Deltapin(\bX_i) \pitiln(\bX_i) \bphij(T_,\bX_i) \varepsilon(\Z_i) \given \Xscn}  \\
&  \; \leq  \Deltapin(\bX_i) \pitiln(\bX_i) \bphij(T_i,\bX_i) \sigmaeps(\bX_i) \; \leq \sigmaeps \Deltapininfn \pitilninfn \bphij(T_i,\bX_i).
\end{align*}
Further, $\forall \; 1 \leq j \leq d$, $\{\bTpij (\bZ_i)\}_{i=1}^n \given \Xscn$ are (conditionally) independent and centered random variables. Hence, using Lemma \ref{lem:2:SGconc}, we have: $\forall \; 1 \leq j \leq d$,
\def\bTpinj{\bT_{\pi, n [j]}}
\begin{eqnarray}
&& \bigpsitwonorm{\frac{1}{n} \sum_{i=1}^n \bTpij(\bZ_i) \; \given[\bigg] \Xscn}  \; \leq \; \frac{4c_{n,j}(\Xscn)}{\sqrt{n}}, \;\; \mbox{where} \nonumber \\
&& \;\; c_{n,j}(\Xscn) \; := \;  \sigmaeps \Deltapininfn \pitilninfn \left(\bphibarsqnj\right)^{1/2} \label{tpicont:cnjdefn}
\end{eqnarray}
and all notations are as defined in (\ref{tpicont:notn1})-(\ref{tpicont:notn3}). Using Lemma \ref{lem:2:SGconc} again, it now follows that for any $\epsilon \geq 0$,
\begin{eqnarray}
&& \quad \P\left\{ \left| \frac{1}{n} \sum_{i=1}^n \bTpij(\bZ_i)\right| > \epsilon  \given[\bigg] \Xscn \right\} \; \leq \; 2 \exp\left\{ \frac{- n \epsilon^2}{8 c_{n,j}^2(\Xscn)} \right\} \;\; \forall \; 1 \leq j \leq d. \label{tpicont:condsgbound}
\end{eqnarray}

\paragraph*{The fundamental bound for $\| \bTpin \|_{\infty}$} Using (\ref{tpicont:condsgbound}), the union bound (u.b.) and the law of iterated expectations (l.i.e.), we then have: for any $\epsilon \geq 0$,
\begin{eqnarray}
&& \P \left\{ \left\| \bTpin \right\|_{\infty} \; \equiv \; \left\| \frac{1}{n} \sum_{i=1}^n \bTpi (\bZ_i) \right\|_{\infty} \; > \; \epsilon \right\} \nonumber \\
&& \qquad \leq \; \sum_{j=1}^d \P \left\{ \left| \frac{1}{n} \sum_{i=1}^n \bTpij (\bZ_i) \right| > \epsilon \right\} \quad \mbox{[using the u.b.]}, \nonumber \\
&& \qquad = \;  \sum_{j=1}^d \E_{\Xscn}\left[\P \left\{ \left| \frac{1}{n} \sum_{i=1}^n \bTpij (\bZ_i) \right| > \epsilon \given [\bigg] \Xscn \right\} \right] \quad \mbox{[using the l.i.e.]}, \nonumber \\
&& \qquad \leq \; \sum_{j=1}^d 2 \; \E_{\Xscn}\left[\exp\left\{\frac{- n \epsilon^2}{8 c_{n,j}^2(\Xscn)} \right\} \right] \quad \mbox{[using (\ref{tpicont:condsgbound})]}. \label{tpicont:fundbound}  \qed
\end{eqnarray}

Next, we aim to control the behavior of the random variable $c_{n,j}^2(\Xscn)$ appearing in the bound (\ref{tpicont:fundbound}). Based on the definition of $c_{n,j}(\Xscn)$ in (\ref{tpicont:cnjdefn}), it suffices to separately control the variables $\Deltapininfn^2$, $\pitilninfn^2$ and $\bphibarsqnj$.

\paragraph*{Controlling $\Deltapininfn^2$} Using (\ref{pi:tailbound2}) in Assumption \ref{tpicont:assmpn} along with the u.b., and recalling %the constant $\gpi$ defined therein with $\|\gpi(\cdot) \|_{\infty} \leq \gpi$, as well as
all notations defined in (\ref{tpicont:notn1})-(\ref{tpicont:notn2}), we have: for any $\epsilon_1 \geq  0$,
\begin{eqnarray}
 \nonumber && \P \left\{ \Deltapininfn^2  \equiv   \underset{1 \leq i \leq n}{\max} \left| \Deltapin(\bX_i) \right|^2  \; > \; \epsilon_1 ^2  \right\} \\
% \nonumber & \quad \leq \; \sum_{i=1}^n \P\left\{ \left| \pihat(\bX_i) - \pi(\bX_i) \right| \; > \; \epsilon_1 \right\} \quad \mbox{[using the u.b.]}, \\
&& \quad \;\; \leq \; \sum_{i=1}^n \P\left\{ \left| \pihat(\bX_i) - \pi(\bX_i) \right|  >    \epsilon_1 \right\}  \; \leq \; C n \exp \left(\frac{- \epsilon_1^2}{v_{n,\pi}^2} \right) + n q_{n,\pi}.
\label{tpicont:deltapinbound} \qquad \qed
\end{eqnarray}

\paragraph*{Controlling $\pitilninfn^2$} Using similar arguments, along with (\ref{pos:eqn}), we have: $\forall \; \epsilon_2 \geq 0$ small enough such that $\epsilon _2  < \deltapi$ with $\deltapi$ as in (\ref{pos:eqn}),
\begin{eqnarray}
\nonumber && \P\left[ \pitilninfn^2 \equiv \underset{1 \leq i \leq n}{\max}\left|\pitiln(\bX_i) \right| ^2 \; > \; \left(\deltapi - \epsilon_2 \right)^{-2} \right] \\
\nonumber && \quad\;\; \leq \; \sum_{i=1}^n \P \left\{ \pihat^{-1}(\bX_i)  >   \left( \deltapi - \epsilon_2 \right)^{-1} \right\} \; \leq \; \sum_{i=1}^n \P \left\{ \pihat(\bX_i)   \; < \; \pi(\bX_i) - \epsilon_2 \right\}  \\
%\nonumber && \quad \leq \; \sum_{i=1}^n \P \left\{ \pihat(\bX_i)   \; < \; \pi(\bX_i) - \epsilon_2 \right\} \quad \mbox{[using (\ref{pos:eqn})]}, \\
&& \;\;\quad \leq \; \sum_{i=1}^n \P \left\{ \left|\pihat(\bX_i) -  \pi(\bX_i) \right| \; > \; \epsilon_2 \right\} \;\; \leq \; C n \exp \left( \frac{- \epsilon_2^2}{v_{n,\pi}^2} \right) + n q_{n,\pi} \label{tpicont:pitilnbound} \quad \qed
\end{eqnarray}

\paragraph*{Controlling $\bphibarsqnj$} Finally, in order to control $\bphibarsqnj(\Xscn)$ which is an average of the i.i.d. random variables $\{\bphij^2(T_i, \bX_i)\}_{i=1}^n$, we first recall all notations from (\ref{tpicont:notn3})-(\ref{tpicont:notn4}) and note that under Assumption \ref{subgaussian:assmpn} (a), $\psionenorm{\bhj^2(\bX)} \leq \sigmabh^2$ $\forall \; j \in \{1, \hdots, d\}$ owing to Lemma \ref{lem:1:GenProp} (v). Further, $T^2/\pi^2(\bX) \leq \deltapi^{-2}$ a.s. $[\P]$. Hence, using Lemma \ref{lem:5:BMC} (b), we have: $\forall \; j \in \{1, \hdots, d\}$, and for some constants $\sigma_{\pi} \equiv \sigmabarbvphi := 2 \sqrt{2}\sigmabh^2 \deltapi^{-2}$ and $K_{\pi} \equiv \Kbarbvphi := 2\sigmabh^2 \deltapi^{-2}$,
\begin{eqnarray}
&& \bphij^2(T, \bX) \; \equiv \; \frac{T^2}{\pi^2(\bX)} \bhj^2(\bX) \; \sim \; \BMC(\sigmabarbvphi, \Kbarbvphi) \quad \mbox{and further,} \label{tpicont:phisqbmc1}%\\
\end{eqnarray}
\begin{eqnarray}
&& \bmusqbphij \; \equiv \; \E\left\{\bphij^2(T, \bX)\right\} \; = \; \E \left\{\frac{\bhj^2(\bX)}{\pi(\bX)} \right\} \; \leq \; \frac{\bmusqbhj}{\deltapi} \; \leq \; \frac{\bmusqbhinf}{\deltapi}, \label{tpicont:phisqbmc2}
\end{eqnarray}
where $\bmusqbhinf := \max \{ \bmusqbhj : j = 1, \hdots, d \} < \infty$ and $\bmusqbhj$ is as in (\ref{tpicont:notn4}).

Using (\ref{tpicont:phisqbmc1})-(\ref{tpicont:phisqbmc2}) along with Lemma \ref{lem:4:Bernstein}, we then have: for any $\epsilon_3 > 0$ and for each $j \in \{1, \hdots, d\}$,
\begin{eqnarray}
\nonumber && \P \left\{  \bphibarsqnj \; \equiv \; \frac{1}{n} \sum_{i=1}^n \bphij^2(T_i, \bX_i) \; > \;  \frac{\bmusqbhinf}{\deltapi} + \sqrt{2} \sigmabarbvphi \epsilon_3 + \Kbarbvphi \epsilon_3^2 \right\} \\
\nonumber && \quad \leq \;  \P \left\{  \left| \frac{1}{n} \sum_{i=1}^n \bphij^2(T_i, \bX_i) -  \bmusqbphij \right| \; > \;  \sqrt{2} \sigmabarbvphi \epsilon_3 + \Kbarbvphi \epsilon_3^2 \right\} \\
&& \quad \leq \; 2 \exp\left( - n \epsilon_3^2\right). \label{tpicont:bphibarsqnjbound} \qquad \qed
\end{eqnarray}

For any $\epsilon_1, \epsilon_3 > 0$, and any $\epsilon_2 > 0$ such that $\epsilon_2 < \deltapi$, let us now define
the event $\Asc_{\pi, n,j}(\epsilon_1, \epsilon_2, \epsilon_3)$, for each $j \in \{1, \hdots, d\}$, as follows.
\begin{align}
 \Asc_{\pi, n,j}(\epsilon_1, \epsilon_2, \epsilon_3) &  := \; \left\{ 8 c_{n,j}^2 (\Xscn)
 >  d_n(\epsilon_1, \epsilon_2, \epsilon_3)\right\}, \; 1 \leq j \leq d, \; \mbox{where}
\label{tpicont:eventandconstdefn}
\end{align}
\begin{equation*}
d_n(\epsilon_1, \epsilon_2, \epsilon_3)  \; := \; \frac{8 \sigmaeps^2 \epsilon_1 ^2}{(\deltapi - \epsilon_2 )^2} \left(\frac{\bmusqbhinf }{\deltapi} + \sqrt{2} \sigmabarbvphi \epsilon_3 + \Kbarbvphi \epsilon_3^2 \right). \nonumber
\end{equation*}

Then, recalling from (\ref{tpicont:cnjdefn}) that $c_{n,j}^2(\Xscn) \equiv \sigmaeps^2 \Deltapininfn^2 \pitilninfn^2 \bphibarsqnj $ and using the bounds (\ref{tpicont:deltapinbound}), (\ref{tpicont:pitilnbound}) and (\ref{tpicont:bphibarsqnjbound}) for $\Deltapininfn^2$, $\pitilninfn^2$ and $\bphibarsqnj$ respectively, along with the union bound, we have:
\begin{eqnarray}
&& \P\left(\Asc_{\pi, n,j}\right) \; \equiv \; \P_{\Xscn}\left(\Asc_{\pi, n,j}\right) \; \equiv \; \P_{\Xscn} \left\{ 8 c_{n,j}^2 (\Xscn)  >  d_n(\epsilon_1, \epsilon_2, \epsilon_3) \right\} \nonumber \\
&& \quad \leq \; C n \exp \left(\frac{- \epsilon_1^2}{v_{n,\pi}^2} \right) + C n \exp \left(  \frac{- \epsilon_2^2}{v_{n,\pi}^2} \right) + 2n q_{n,\pi} + 2 \exp\left( - n \epsilon_3^2\right). \label{tpicont:compeventbound}
\end{eqnarray}
Therefore, it now follows that for each $j \in \{1, \hdots, d\}$ and any $\epsilon \geq 0$,
\begin{eqnarray}
&& \E_{\Xscn}\left[\exp\left\{ \frac{- n \epsilon^2}{8 c_{n,j}^2(\Xscn)} \right\} \right]  \; = \; \E\left[\exp\left\{ \frac{- n \epsilon^2}{8 c_{n,j}^2(\Xscn)} \right\} \given[\bigg]\Asc_{\pi, n, j}^c\right] \P\left( \Ascpinjc \right) \nonumber \\ %\equiv   \E\left[\exp\left\{ \frac{ - n \epsilon^2}{8 c_{n,j}^2(\Xscn)} \right\} \left(1_{\Ascpinjc} + 1_{\Ascpinj} \right)\right] \nonumber \\
%&& \quad = \; \E\left[\exp\left\{ \frac{- n \epsilon^2}{8 c_{n,j}^2(\Xscn)} \right\} \given[\bigg]\Asc_{\pi, n, j}^c\right] \P\left( \Ascpinjc \right) \nonumber \\
&& \qquad \qquad \qquad \qquad \qquad \qquad + \; \E\left[\exp\left\{  \frac{- n \epsilon^2}{8 c_{n,j}^2(\Xscn)} \right\} \given[\bigg]\Asc_{\pi, n, j}\right] \P\left( \Ascpinj \right) \nonumber \\
&& \;\;\quad \leq \; \exp\left\{ \frac{- n \epsilon^2}{d_n(\epsilon_1, \epsilon_2, \epsilon_3)} \right\} \; + \; 2 \exp\left( - n \epsilon_3^2\right) \; + \; 2n q_{n, \pi} \label{tpicont:finalbound1} \\
&& \;\; \qquad + \; C n \exp \left( \frac{- \epsilon_1^2}{v_{n,\pi}^2} \right) \; + \; C n \exp \left( \frac{- \epsilon_2^2}{v_{n,\pi}^2} \right) \quad \mbox{[using (\ref{tpicont:eventandconstdefn})-(\ref{tpicont:compeventbound})]}. \nonumber
\end{eqnarray}
\paragraph*{The final bound for $\|\bTpin \|_{\infty}$} Using (\ref{tpicont:finalbound1}) in the fundamental bound (\ref{tpicont:fundbound}) for $\| \bTpin \|_{\infty}$, we finally have: for any $\epsilon \geq 0$,
\begin{eqnarray}
&&  \P \left( \left\| \bTpin \right\|_{\infty} > \epsilon \right) \; \leq \; \sum_{j=1}^d 2 \; \E_{\Xscn}\left[\exp\left\{  \frac{- n \epsilon^2}{8 c_{n,j}^2(\Xscn)} \right\} \right] \nonumber \\
&& \quad \; \leq \; 2d\exp\left\{ \frac{- n \epsilon^2}{d_n(\epsilon_1, \epsilon_2, \epsilon_3)} \right\} \; + \; 4d \exp\left( - n \epsilon_3^2\right) \; + \; 4 q_{n,\pi}(nd) \nonumber \\
&& \qquad + \; 2C (nd) \exp \left( \frac{- \epsilon_1^2}{v_{n,\pi}^2} \right) \; + \; 2C (nd) \exp \left( \frac{- \epsilon_2^2}{v_{n,\pi}^2} \right) \quad \mbox{[using (\ref{tpicont:finalbound1})],} \nonumber \\
&& \quad \; \equiv \; 2\exp\left\{ \frac{- n \epsilon^2}{d_n(\epsilon_1, \epsilon_2, \epsilon_3)} + \log d\right\}  +  4 \exp\left( - n \epsilon_3^2 + \log d \right)  + 4 q_{n,\pi}(nd) \label{tpicont:finalbound2} \\
&& \qquad +  2C \exp \left\{ \frac{- \epsilon_1^2}{v_{n,\pi}^2} + \log (nd)\right\}  +  2C \exp \left\{ \frac{- \epsilon_2^2}{v_{n,\pi}^2} + \log(nd) \right\}. \nonumber
\end{eqnarray}
This leads to the desired bound and completes the proof of Theorem \ref{tpi:thm:tailbound}.
\qed

%\paragraph*{Implications of the bound in Theorem \ref{tpi:thm:tailbound} and  characterization of the rate}
\subsection{Completing Proof of Theorem \ref{TPI:THM}}
We next evaluate the general tail bound for $\|\bTpin  \|_{\infty}$ in Theorem \ref{tpi:thm:tailbound} under a specific family of choices for $(\epsilon, \epsilon_1, \epsilon_2, \epsilon_3) > 0$ in order to understand its behavior and also establish the convergence rate of $\|\bTpin  \|_{\infty}$. Let $(c_1, c_2, c_3) > 1$ be any universal constants and set $\epsilon_1 = c_1 \vnpi \sqrt{\log (nd)} $, $\epsilon_2 = c_2 \vnpi \sqrt{\log (nd)} $ and $\epsilon_3 = c_3 \sqrt{(\log d)/n}$, where we assume w.l.o.g. that $\epsilon_3 < 1$ and $\epsilon_2  \leq \deltapi/2$, so that $(\deltapi - \epsilon_2) \geq \deltapi/2$. Further with a choice of $\epsilon_3$ as above, note that
\begin{equation*}
\frac{\bmusqbhinf }{\deltapi} +
\sqrt{2} \sigmabarbvphi \epsilon_3 + \Kbarbvphi \epsilon_3^2 \; \leq \; \frac{\bmusqbhinf }{\deltapi} + \left(\sqrt{2}\sigmabarbvphi + \Kbarbvphi\right) c_3 \sqrt{\frac{\log d}{n}}.
\end{equation*}
Using these in the definition (\ref{tpicont:eventandconstdefn}) and letting $C_{\bvphi} := (\sqrt{2}\sigmabarbvphi + \Kbarbvphi)$, we get %$d_n(\epsilon_1, \epsilon_2, \epsilon_3)$
\begin{align*}
& d_n(\epsilon_1, \epsilon_2, \epsilon_3) \; \leq \; 8\sigmaeps^2 \frac{4c_1^2}{\deltapi^2}\{\vnpi \sqrt{\log (nd)} \}^2 \left( \frac{\bmusqbhinf}{\deltapi} + c_3 C_{\bvphi} \sqrt{\frac{\log d}{n}} \right).
\end{align*}
Given these choices of $\{\epsilon_j\}_{j=1}^3$, let us now set $\epsilon = c\sqrt{ \{(\log d)/n \} d_n(\epsilon_1, \epsilon_2, \epsilon_3)}$ for any universal constant $c > 1$. Using Theorem \ref{tpi:thm:tailbound}, we then have:
\begin{equation*}
 \mbox{With probability at least} \; \; 1 - \frac{2}{d^{c^2-1}} - \frac{4}{d^{c_3^2-1}} - \sum_{j=1}^2 \frac{2C}{(nd)^{c_j^2-1}}  - 4 q_{n,\pi} (nd),%\\
\end{equation*}
\begin{equation*}
\|\bTpin\|_{\infty}  \; \leq \; c\sqrt{\frac{\log d}{n}} \{ \vnpi \sqrt{\log (nd)}\} C_1 \left( \frac{\bmusqbhinf}{\deltapi} + C_2 \sqrt{\frac{\log d}{n}} \right)^{\half},
\end{equation*}
where $C_1 := c_1(4\sqrt{2}\sigmaeps/\deltapi)$ and $C_2 := c_3 C_{\bvphi} \equiv c_3(\sqrt{2}\sigmabarbvphi + \Kbarbvphi)$, with $\sigmabarbvphi$ and $\Kbarbvphi$ being as in \eqref{tpicont:phisqbmc1}. This completes the proof of Theorem \ref{TPI:THM}. \qed %

\section{Proof of Theorem \ref{TM:THM}}\label{pf:tm:thm}

To show Theorem \ref{TM:THM}, we first state and prove a more general result that gives an explicit tail bound for $\|\bTmn \|_{\infty}$.

\begin{theorem}[Tail bound for $\| \bTmn \|_{\infty}$]\label{tm:thm:tailbound}
Let Assumptions \ref{base:assmpns}, \ref{subgaussian:assmpn} (a) and \ref{tmcont:assmpn} hold with the sequences $(\vnbarm, q_{\nbar,m})$, $\nbar \equiv n/2$ and the constants $(\deltapi, \sigmabh, C)$ as defined therein.
Then, for any $\epsilon, \epsilon_1, \epsilon_2 \geq 0$,
\begin{eqnarray*}
\P\left( \left\| \bTmn \right\|_{\infty}  > \; \epsilon \right) & \leq & 4\exp \left\{\frac{- \nbar \epsilon^2}{\tnbar(\epsilon_1, \epsilon_2)}  + \log d \right\} \; + \; 8 \exp(- \nbar \epsilon_2^2 + \log d) \\
&&  + \; 4C   \exp\left\{ \frac{- \epsilon_1^2}{\vnbarm^2} + \log(\nbar d) \right\} + 4 q_{\nbar,m} (\nbar d), \quad \mbox{where} \\
%\end{eqnarray*}
%\begin{equation*}
\tnbar(\epsilon_1, \epsilon_2) & := & 8 \deltapibar^2 \epsilon_1 ^2 \left(\bmusqbhinf + \sqrt{2} \sigma_m \epsilon_2  + K_m \epsilon_2^2\right), \;\; \mbox{with}
%\end{equation*}
\end{eqnarray*}
$\bmusqbhinf := \max_{1\leq j \leq d} \E\{\bhj^2(\bX)\}$, $\deltapibar \leq \deltapi^{-1}$, $\sigma_m := 2 \sqrt{2} \sigmabh^2$ and $K_m := 2 \sigmabh^2$. %depending only on the constants introduced in the assumptions.
\end{theorem}

\subsection{Proof of Theorem \ref{tm:thm:tailbound}}\label{pf:tm:thm:tailbound}

We first rewrite $\bTmn$ from (\ref{decomp:eqn}) as:
\begin{eqnarray}
\nonumber \bTmn &\equiv& \frac{1}{n} \sum_{i=1}^n \left\{ \frac{T_i}{\pi(\bX_i)} - 1 \right\} \left\{\mtil(\bX_i) - m(\bX_i)\right\} \bh(\bX_i) \\
\nonumber &=& \frac{1}{2\nbar} \sum_{k \neq k' = 1}^2 \sum_{i\in \Isc_{k'}} \left\{ \frac{T_i}{\pi(\bX_i)} - 1 \right\} \left\{\mhatk(\bX_i) - m(\bX_i)\right\} \bh(\bX_i) \\
&=:& \frac{1}{2} \sum_{k \neq k' = 1}^2 \bTmkkpn, \;\; \mbox{where} \;\; \bTmkkpn \; := \; \frac{1}{\nbar} \sum_{i\in \Isc_{k'}} \bTmk(\bZ_i) \;\; \mbox{and}
\label{tm:redef:eq1} \\
\bTmk(\bZ) &:=& \left\{ \frac{T}{\pi(\bX)} - 1 \right\} \left\{\mhatk(\bX) - m(\bX)\right\} \bh(\bX) \quad \forall \; k \neq k' \in \{1,2\}. \nonumber
\end{eqnarray}

Define $\Xscnk := \{ \bX_i : i \in \Isck \} \; \forall \; k  \in \{1,2\}$, and let $\EXscnk (\cdot)$ and $\P(\cdot \given \Xscnk)$ respectively denote expectation w.r.t. $\Xscnk$ and conditional probability given $\Xscnk$. Further, for each $k \neq k' \in \{1,2\}$, let $\Ekkp(\cdot)$ and $\P(\cdot \given \Dscnk, \Xscnkp)$ respectively denote expectation w.r.t. $\{ \Dscnk, \Xscnkp\}$ and conditional probability given $\{ \Dscnk, \Xscnkp\}$. With $\Dscnk \ind \Xscnkp$ $\forall \; k \neq k' \in \{1,2\}$, we  note that $\Ekkp(\cdot) = \E_{\Xscnkp}\{\E_{\Dscnk}(\cdot)\}$. Next, let us define: $\forall \; k \neq k' \in \{1,2\}$,
\begin{align}
&   \Deltamnk(\bX) \; := \; \mhatk(\bX) - m(\bX), \;\; \Deltamnkkpin  := \; \underset{i \in \Isckp}{\max} \left| \Deltamnk(\bX_i) \right|, \label{tmcont:notn1}  \\
&  \bhsqbarkpnj  \; := \; \frac{1}{\nbar} \sum_{i \in \Isckp} \bhj^2 (\bX_i) \;\; \mbox{and let} \;\; \psi(T,\bX) \; := \; \frac{T}{\pi(\bX)} - 1.  \label{tmcont:notn2}
\end{align}

Further, for any $a \in (0,1]$, let $\bar{a} := 2\widetilde{a}/a$, where $\widetilde{a} := 1/2$ if $a = 1/2$, $\widetilde{a} := 0$ if $a =1$ and $\widetilde{a} := [(a-1/2)/\log\{a/(1-a)\}]^{1/2}$ if $a \notin \{1/2,1\}$. Let $\{\pibar(\bX), \pitil(\bX)\}$ and $\{\deltapibar, \deltapitil\}$ denote the corresponding versions of $\{\bar{a},\widetilde{a}\}$ for $a \equiv \pi(\bX)$ and $a \equiv \deltapi$ respectively, with $\deltapi$ being as in (\ref{pos:eqn}). We note that $\bar{a}$ is decreasing in $a \in (0,1]$ and $\widetilde{a} \leq 1/2$, so that $\bar{a} \leq 1/a$ $\forall \; a \in (0,1]$. Using this and (\ref{pos:eqn}), we therefore have: $\pibar(\bx) \leq \deltapibar \leq 1/\deltapi$ $\forall \; \bx \in \Xsc$.

Using the notations from (\ref{tmcont:notn1}) and \eqref{tmcont:notn2}, we have: for each $k \in \{1,2\}$,
\begin{equation*}
\bTmk(\bZ) \equiv \left\{\frac{T}{\pi(\bX)}- 1 \right\} \{ \mhatk(\bX) - m(\bX) \} \bh(\bX)
= \psi(T,\bX) \Deltamnk(\bX) \bh(\bX).
\end{equation*}

Now, for each $k \in \{1, 2\}$ and $ k' \neq k \in \{1,2\}$, $\Dscnk \ind \Xscnkp$ and we have: $\{\psi(T_i,\bX_i) \given \Dscnk, \Xscnkp \}_{i\in \Isckp} \equiv \{\psi(T_i,\bX_i) \given \Xscnkp \}_{i\in \Isckp} \equiv \{ \psi(T_i, \bX_i) \given \bX_i \}_{i \in \Isckp}$ are (conditionally) independent sub-Gaussian random variables that satisfy:
\begin{eqnarray}
&& \;\; \forall \; i \in \Isckp, \;\; \E\{\psi(T_i,\bX_i) \given \Dscnk, \Xscnkp\} %\equiv \E\{\psi(T_i,\bX_i) \given \Xscnkp\}
\; \equiv \; \E\{\psi(T_i, \bX_i) \given \bX_i\} = 0 \quad \mbox{and} \nonumber \\
&& \;\; \psitwonorm{\psi(T_i, \bX_i) \given \Dscnk, \Xscnkp} %\equiv \psitwonorm{\psi(T_i, \bX_i) \given \Xscnkp}
\; \equiv \; \psitwonorm{\psi(T_i, \bX_i) \given \bX_i} \; \leq \; \pibar^2(\bX_i) \; \leq \; \deltapibar^2, \label{tmcont:psisgnormbound}
\end{eqnarray}
where the bounds on the $\psitwonorm{\cdot}$ norm follow from using Lemma \ref{lem:3:SubgaussianBinary} and Lemma \ref{lem:1:GenProp} (i)(b) along with the definitions of $\pibar(\cdot)$ and $\deltapibar$ given earlier.
Further, conditional on $\Dscnk$ and $\Xscnkp$, $\{\Deltamnk(\bX_i)\}_{i \in \Isckp}$ and $\{\bhj(\bX_i)\}_{i \in \Isckp}$, for each $ j \in \{1, \hdots, d\}$, are all constants.
Hence, using Lemma \ref{lem:2:SGconc} and \eqref{tmcont:psisgnormbound}, along with \eqref{tm:redef:eq1}-\eqref{tmcont:notn2}, we have: $\forall \; k \neq k' \in \{1,2\}$ and $j \in \{ 1, \hdots, d\}$, %and $|\Deltamnk(\bX_i)| \leq \Deltamnkkpin$
\begin{eqnarray}
&& \bigpsitwonorm{\frac{1}{\nbar}\sum_{i \in \Isckp} \bTmkj(\bZ_i) \given[\bigg] \Dscnk, \Xscnkp} \; \leq \; \frac{4 \dnbarj\left(\Dscnk, \Xscnkp\right)}{\sqrt{\nbar}}, \quad \mbox{where} \label{tmcont:condsgdefn} \\
&& \dnbarj \left(\Dscnk, \Xscnkp\right) \; := \; \deltapibar \Deltamnkkpin \left(\bhsqbarkpnj\right)^{1/2}. \nonumber
\end{eqnarray}
Using Lemma \ref{lem:2:SGconc}, we then have: $\forall \; k \neq k' \in \{1,2\}$, $1\leq j \leq d$ and $\epsilon \geq 0$, %j \in \{1, \hdots, d\}$ and any $\epsilon \geq 0$,
\begin{equation*}
\P \left\{ \left|\frac{1}{\nbar}\sum_{i \in \Isckp} \bTmkj(\bZ_i)\right| > \epsilon \given[\bigg] \Dscnk, \Xscnkp \right\} \; \leq \; 2 \exp \left\{\frac{- \nbar \epsilon^2}{8 \dnbarj^2\left(\Dscnk, \Xscnkp\right)}\right\}.
\end{equation*}
\paragraph*{The fundamental bound for $\|\bTmkkpn\|_{\infty}$} Using the bound obtained above for $\bT_{m,\nbar [j]}^{(k,k')} \given \Dscnk, \Xscnkp$, we then have the following (unconditional) probabilistic bound for $\|\bTmkkpn\|_{\infty}$. For any $\epsilon \geq 0$ and $k \neq k' \in \{1,2\}$,
\begin{eqnarray}
&& \P \left\{ \left\| \bTmkkpn \right\|_{\infty} \; \equiv \; \left\|\frac{1}{\nbar} \sum_{i \in \Isckp} \bTmk(\bZ_i) \right\|_{\infty} \; > \; \epsilon \right\} \nonumber \\
&& \qquad \leq \; \sum_{j=1}^d \P \left\{\left|\frac{1}{\nbar} \sum_{i \in \Isckp} \bTmkj(\bZ_i) \right| \; > \; \epsilon \right\} \quad \mbox{[using the u.b.],} \nonumber \\
&& \qquad = \; \sum_{j=1}^d \Ekkp\left[ \P \left\{\left|\frac{1}{\nbar} \sum_{i \in \Isckp} \bTmkj(\bZ_i) \right| \; > \; \epsilon  \given[\bigg] \Dscnk, \Xscnkp\right\} \right] \nonumber \\
&& \qquad \leq \; 2 \sum_{j=1}^d  \Ekkp \left[ \exp \left\{\frac{- \nbar \epsilon^2}{8 \dnbarj^2\left(\Dscnk, \Xscnkp\right)}\right\} \right]. \label{tmcont:condsg:fundbound} \qquad \qed
\end{eqnarray}

Next, we aim to control the random variable $\dnbarj^2(\Dscnk, \Xscnkp)$ appearing in \eqref{tmcont:condsg:fundbound}. Based on the definition \eqref{tmcont:condsgdefn} of $\dnbarj^2(\Dscnk, \Xscnkp)$, it suffices to separately control $\Deltamnkkpin^2$ and $\bhsqbarkpnj$. To this end, let $\EDscnk(\cdot)$ and $\PDscnk(\cdot)$ denote expectation and probability w.r.t $\Dscnk$ $\forall \; k \in \{1, 2\}$.

With $\Dscnk \ind \Xscnkp$ for each $k \neq k' \in \{1, 2\}$, we note that for any event $A \equiv A(\Dscnk, \Xscnkp)$, $\P(A) \equiv \Pkkp(A) = \EXscnkp [\EDscnk\{1(A) \given \Xscnkp\}] \equiv \EXscnkp[\PDscnk\{A(\Dscnk, \Xscnkp) \given \Xscnkp \}] = \EXscnkp[\PDscnk\{A(\Dscnk, \Xscnkp)\}]$, where the final step holds since $\PDscnk(\cdot \given \Xscnk) = \PDscnk(\cdot)$ as $\Dscnk \ind \Xscnkp$.
\paragraph*{Controlling $\Deltamnkkpin^2$} Using \eqref{m:tailbound2} in Assumption \ref{tmcont:assmpn} along with the u.b. and the notations and facts discussed above, we have: $\forall$ $k \neq k' \in \{1,2\}$,
\begin{eqnarray}
&& \P\left\{ \Deltamnkkpin^2  \equiv  \; \underset{i \in \Isckp}{\max} \left| \Deltamnk(\bX_i) \right|^2  \; > \; \epsilon_1 ^2\right\} \nonumber \\
&& \;\;\; \leq  \; \sum_{i \in \Isckp} \P\left\{ \left| \Deltamnk(\bX_i) \right|  > \epsilon_1 \right\}  \; \leq \; \sum_{i \in \Isckp} \EXscnkp \left\{ C \exp\left(\frac{- \epsilon_1^2}{\vnbarm^2} \right) + q_{\nbar,m}\right\} \nonumber \\
&& \;\;\; \equiv \; C \nbar  \exp\left( \frac{- \epsilon_1^2}{\vnbarm^2}\right) + \nbar q_{\nbar,m} \quad \mbox{for any} \; \epsilon_1 \geq 0, \label{tmcont:deltakkpbound}
\end{eqnarray}
where we also used that $\Dscnk \ind \Xscnkp$ which ensures $\PDscnk(\cdot \given \Xscnk) = \PDscnk(\cdot)$ and makes (\ref{m:tailbound2}) in Assumption \ref{tmcont:assmpn} applicable conditional on $\Xscnkp$. \qed
\paragraph*{Controlling $\bhsqbarkpnj$} We first recall that $\bmusqbhinf = \max_{1 \leq j \leq d} \; \bmusqbhj$, where $\bmusqbhj \equiv \E\{\bhj^2(\bX)\}$.
Now, $ \forall \; k' \in \{1,2\}$ and $j \in \{1, \hdots, d\}$, $\bhsqbarkpnj$ is simply an average of the i.i.d. random variables $\{\bhj^2(\bX_i)\}_{i \in \Isckp}$. Further, using Assumption \ref{subgaussian:assmpn} (a) and Lemma \ref{lem:5:BMC} (a), $\bhj^2(\bX) \sim \BMC(\sigmabarbhsq, \Kbarbhsq)$ for some constants $\sigma_m \equiv \sigmabarbhsq := 2 \sqrt{2} \sigmabh^2 $ and $K_m \equiv \Kbarbhsq := 2 \sigmabh^2$. Hence, using Lemma \ref{lem:4:Bernstein}, we have: for each $k' \in \{1,2\}$ and $j \in \{1, \hdots, d\}$, and for any $\epsilon_2 \geq 0$,
\begin{eqnarray}
&& \quad\;\; \P\left\{ \bhsqbarkpnj \; \equiv \; \frac{1}{\nbar} \sum_{i \in \Isckp} \bhj^2(\bX_i) \; > \; \bmusqbhinf + \sqrt{2} \sigmabarbhsq \epsilon_2  + \Kbarbhsq \epsilon_2^2 \right\} \label{tmcont:bhsqbarkpbound} \\
&& \leq \; \P\left\{ \left| \frac{1}{\nbar} \sum_{i \in \Isckp} \bhj^2(\bX_i)  - \bmusqbhj  \right|  >  \sqrt{2} \sigmabarbhsq \epsilon_2  + \Kbarbhsq \epsilon_2^2 \right\}  \hspace{0.03in} \leq \; 2 \exp(- \nbar \epsilon_2^2).\qed \nonumber
\end{eqnarray}
\paragraph*{The final bound for $\left\| \bTmkkpn \right\|_{\infty}$} For any $\epsilon_1, \epsilon_2 \geq 0$, let us now define:
\begin{equation}
 \tnbar(\epsilon_1, \epsilon_2) \; := \; 8 \deltapibar^2\epsilon_1 ^2 \left(\bmusqbhinf + \sqrt{2} \sigmabarbhsq \epsilon_2  + \Kbarbhsq \epsilon_2^2\right). \label{tmcont:tnbar:def}
\end{equation}
Then, using the bounds (\ref{tmcont:deltakkpbound}) and (\ref{tmcont:bhsqbarkpbound}) in the definition of $\dnbarj^2(\Dscnk, \Xscnkp)$ in (\ref{tmcont:condsgdefn}), we have: for each $ k \neq k' \in \{1, 2\}$, $j \in \{1, \hdots, d\}$ and $\epsilon_1, \epsilon_2 \geq 0$,
\begin{align}
& \P\left\{ 8 \dnbarj^2( \Dscnk, \Xscnkp)\; > \; \tnbar \left( \epsilon_1, \epsilon_2\right) \right\}  \nonumber  \\
&  \leq \; C \nbar  \exp\left( \frac{- \epsilon_1^2}{\vnbarm^2}\right) + \nbar q_{\nbar,m} + 2 \exp(- \nbar \epsilon_2^2). \label{tmcont:dnbarj:probbound}
\end{align}
Using \eqref{tmcont:dnbarj:probbound} in the fundamental bound (\ref{tmcont:condsg:fundbound}) for $\|\bTmkkpn \|_{\infty}$, we then have: for each $k \neq k' \in \{1,2\}$ and for any $\epsilon, \epsilon_1, \epsilon_2 \geq 0$,
\begin{eqnarray}
\nonumber && \P \left\{ \left\| \bTmkkpn \right\|_{\infty} > \; \epsilon \right\} \; \leq \; 2 \sum_{j=1}^d  \Ekkp \left[ \exp \left\{\frac{- \nbar \epsilon^2}{8 \dnbarj^2\left(\Dscnk, \Xscnkp\right)}\right\} \right] \\
\nonumber && \quad\; \equiv \; 2 \sum_{j=1}^d \E \left[ \exp \left\{\frac{- \nbar \epsilon^2}{8 \dnbarj^2\left(\Dscnk, \Xscnkp\right)}\right\} 1_{\left\{8\dnbarj^2( \Dscnk, \Xscnkp) \; \leq \; \tnbar \left( \epsilon_1, \epsilon_2\right)\right\}} \right] \\
\nonumber && \qquad +  \; 2 \sum_{j=1}^d \E \left[ \exp \left\{\frac{- \nbar \epsilon^2}{8 \dnbarj^2\left(\Dscnk, \Xscnkp\right)}\right\} 1_{\left\{8\dnbarj^2( \Dscnk, \Xscnkp)\; > \; \tnbar \left( \epsilon_1, \epsilon_2\right)\right\}} \right] \\
\nonumber && \quad\;\; \leq \; 2d \left[\exp \left\{\frac{- \nbar \epsilon^2}{\tnbar(\epsilon_1, \epsilon_2)} \right\} \; + \; \P\left\{ 8\dnbarj^2( \Dscnk, \Xscnkp)\; > \; \tnbar \left( \epsilon_1, \epsilon_2\right) \right\} \right]\\
&& \quad \;\; \leq \; 2d \left[ \exp \left\{\frac{- \nbar \epsilon^2}{\tnbar(\epsilon_1, \epsilon_2)} \right\} + C \nbar  \exp\left( \frac{- \epsilon_1^2}{\vnbarm^2}\right) +  \nbar q_{\nbar,m} + 2 \exp(- \nbar \epsilon_2^2) \right].
\label{tmcont:tkkpn:finalbound}
\end{eqnarray}
Thus, (\ref{tmcont:tkkpn:finalbound}) establishes an explicit tail bound for $\left\| \bTmkkpn \right\|_{\infty}$. \qed
\paragraph*{The final bound for $\| \bTmn\|_{\infty}$} A tail bound for $\| \bTmn\|_{\infty}$ now follows easily using (\ref{tm:redef:eq1}) and (\ref{tmcont:tkkpn:finalbound}) along with the u.b. For any $\epsilon, \epsilon_1, \epsilon_2 \geq 0$, we have:
\begin{eqnarray}
 \label{tmcont:finalbound} && \quad \;\; \P \left( \left\| \bTmn \right\|_{\infty}  \; > \; \epsilon \right)  \; \leq \; \P\left( \left\| \bTmonetwon  \right\|_{\infty} > \; \epsilon \right) \; + \; \P\left( \left\| \bTmtwoonen  \right\|_{\infty} > \; \epsilon \right)  \\
 \nonumber && \leq \; 4d \exp \left\{\frac{- \nbar \epsilon^2}{\tnbar(\epsilon_1, \epsilon_2)} \right\} + 4C \nbar d  \exp\left( \frac{- \epsilon_1^2}{\vnbarm^2}\right) + 4 \nbar d q_{\nbar,m} + 8d \exp(- \nbar \epsilon_2^2).
\end{eqnarray}
This leads to the desired bound and concludes the proof of Theorem \ref{tm:thm:tailbound}. \qed
%

%
%\paragraph*{Implications of the bound in Theorem \ref{tm:thm:tailbound} and characterization of the rate}
\subsection{Completing the Proof of Theorem \ref{TM:THM}}
Given the general tail bound for $\|\bTmn  \|_{\infty}$ in Theorem \ref{tm:thm:tailbound}, we next evaluate it for a specific set of choices of $(\epsilon, \epsilon_1, \epsilon_2) > 0$ in order to understand its behavior and also establish the convergence rate of $\|\bTmn  \|_{\infty}$. To this end, let $(c_1, c_2) > 1$ be any universal constants and set $\epsilon_1 = c_1 \vnbarm \sqrt{\log (\nbar d)} $ and $\epsilon_2 = c_2 \sqrt{(\log d)/\nbar}$, where we further assume w.l.o.g. that $\epsilon_2 < 1$ so that
\begin{equation*}
\bmusqbhinf + \sqrt{2} \sigmabarbhsq \epsilon_2 + \Kbarbhsq \epsilon_2^2 \; \leq \; \bmusqbhinf  + \left(\sqrt{2}\sigmabarbhsq + \Kbarbhsq\right) c_2 \sqrt{\frac{\log d}{\nbar}}.
\end{equation*}
Using these in the definition (\ref{tmcont:tnbar:def}) and letting $C_{\bh} := (\sqrt{2}\sigmabarbhsq + \Kbarbhsq)$, we have:
\begin{equation*}
\tnbar(\epsilon_1, \epsilon_2) \; \leq \; 8 c_1^2 \deltapibar^2 \{\vnbarm \sqrt{\log (\nbar d)}\}^2 \left\{ \bmusqbhinf + c_2 C_{\bh} \sqrt{\frac{\log d}{\nbar}} \right\}.
\end{equation*}
Given these choices of $\{\epsilon_j\}_{j=1}^2$, let us now set $\epsilon = c\sqrt{ \{ (\log  d)/\nbar\} \tnbar(\epsilon_1, \epsilon_2)}$ for any $c > 1$. Using Theorem \ref{tm:thm:tailbound} and with $\nbar  \equiv n/2 \leq n$, we then have:
\begin{equation*}
 \mbox{With probability at least} \;\; 1 - \frac{4}{d^{c^2-1}} - \frac{8}{d^{c_2^2-1}} - \frac{4C}{(\nbar d)^{c_1^2-1}} - 4q_{\nbar,m}(\nbar d),\\
\end{equation*}
\begin{equation*}
\|\bTmn\|_{\infty} \; \leq \; c\sqrt{\frac{\log d}{n}} \{ \vnbarm \sqrt{\log (n d)} \} C_1^*\left( \bmusqbhinf + C_2^* \sqrt{\frac{\log d}{n}} \right)^{\half},
\end{equation*}
where $C_1^* := 4c_1 \deltapibar$ and $C_{2,n}^* := \sqrt{2}c_2 C_{\bh} \equiv \sqrt{2}c_2(\sqrt{2}\sigmabarbhsq + \Kbarbhsq)$, with $\sigmabarbhsq$ and $\Kbarbhsq$ being as in \eqref{tmcont:bhsqbarkpbound}. This completes the proof of Theorem \ref{TM:THM}. \qed

\section{Proof of Theorem \ref{RMPI:THM}}\label{pf:rmpi:thm}

To show Theorem \ref{RMPI:THM}, we first state and prove a more general result that gives an explicit tail bound for $\|\bRpimn \|_{\infty}$.

\begin{theorem}[Tail bound for $\| \bRpimn \|_{\infty}$]\label{rmpi:thm:tailbound}
Let Assumptions \ref{base:assmpns}, \ref{subgaussian:assmpn}, \ref{tpicont:assmpn} and \ref{tmcont:assmpn} hold with the sequences $(\vnpi, q_{n,\pi})$, $(\vnbarm, q_{\nbar,m}, \nbar)$ and the constants $(\deltapi, \sigmabh, C)$ as defined therein, and let $\bmumodbhinf := \max\{ \E\{ |\bhj(\bX)|\} : j = 1, \hdots, d\}$. %and the constants $(\sigmabarbxi, \Kbarbxi) > 0$ be as in (\ref{rmpicont:eqn5}).
Then, for any $\epsilon_1, \epsilon_2, \epsilon_3 , \epsilon_4 \geq 0$ with $\epsilon_2 < \deltapi$ small enough,
\begin{eqnarray*}
&& \P \left\{ \left\| \bRpimn \right\|_{\infty} \; > \; \frac{\epsilon_1 \epsilon_3}{\deltapi - \epsilon_2} \rbxi(\epsilon_4)\right\} \; \leq \;  2d \exp( - n \epsilon_4^2) \nonumber \\
&&  \;\; + C n \left\{\exp\left( \frac{-\epsilon_1^2}{\vnpi^2}\right) + \exp\left( \frac{-\epsilon_2^2}{\vnpi^2}\right) + \exp\left(\frac{- \epsilon_3^2}{\vnbarm^2}\right) \right\}  + 2nq_{n,\pi} + n q_{\nbar,m}, \nonumber
\end{eqnarray*}
where $\rbxi(\epsilon_4) := \bmumodbhinf + \sqrt{2}{\sigma_{\pi,m}}\epsilon_4 + K_{\pi,m} \epsilon_4^2$ with $\sigma_{\pi,m} := 4 \sigmabh\deltapi^{-1} $ and $K_{\pi,m} := 2 \sqrt{2} \sigmabh \deltapi^{-1}$ being constants. %that depend only on the constants introduced in the assumptions.
\end{theorem}

\subsection{Proof of Theorem \ref{rmpi:thm:tailbound}}\label{pf:rmpi:thm:tailbound}
%Note that $\bRpimn$ is essentially a `second order' term since it involves a product of the two error terms arising from the estimation of $\pi(\cdot)$ and $m(\cdot)$ and under reasonable assumptions on $\pihat(\cdot) - \pi(\cdot)$ and $\mhat(\cdot) -m(\cdot)$, one can attempt to control the behavior of this term by `naive' techniques, as opposed to the more sophisticated analyses required for controlling $\bTpin$ and $\bTmn$.
Recalling from the notations in \eqref{decomp:eqn},
\begin{equation}
\bRpimn \; = \; \frac{1}{n} \sum_{i=1}^n \left\{\frac{T_i}{\pihat(\bX_i)} - \frac{T_i}{\pi(\bX_i)}\right\} \left\{ \mtil(\bX_i) - m(\bX_i)\right\} \bh(\bX_i). \label{rmpin:redefn}
\end{equation}

Hence, with $\Deltapininfn$ and $\pitilninfn$ as in (\ref{tpicont:notn1}) and (\ref{tpicont:notn2}) respectively, and with $\Deltamnkkpin$ as in (\ref{tmcont:notn1}) for any $k \neq k' \in \{1, 2\}$, we have:
\begin{eqnarray}
&& \quad \left\| \bRpimn \right\|_{\infty} \; \leq \; \pitilninfn \Deltapininfn \Deltastarmninfn \|\bxibarn\|_{\infty}, \quad \mbox{where} \label{rmpicont:eqn1} \\
&& \Deltastarmninfn \; := \; \max \left\{ \Deltamnonetwoin, \Deltamntwoonein \right\}  \;\; \mbox{and} \nonumber \\
&& \bxibarn \; := \; \frac{1}{n} \sum_{i=1}^n \bxi(T_i, \bX_i), \;\; \mbox{with} \;\; \bxi(T, \bX) \; := \; \left\{\frac{T}{\pi(\bX)} \left|\bhj(\bX) \right| \right\}_{j=1}^d \in \R^d. \nonumber
\end{eqnarray}
For most of the quantities appearing in the bound (\ref{rmpicont:eqn1}), we already have their explicit tail bounds. Specifically, using (\ref{tpicont:deltapinbound}), we have: for any $\epsilon_1 \geq 0$,
\begin{equation}
\P\left\{ \Deltapininfn \; > \; \epsilon_1 \right\} \; \leq \; C n \exp\left( \frac{-\epsilon_1^2}{\vnpi^2}\right) + n q_{n,\pi}, \quad \mbox{where} \label{rmpicont:eqn2} \\
\end{equation}
and using (\ref{tpicont:pitilnbound}), for any $\epsilon_2 \geq 0$ small enough such that $\epsilon_2 < \deltapi$,
\begin{equation}
\P\left\{\pitilninfn \; > \; \left(\deltapi - \epsilon_2\right) ^{-1}\right\} \; \leq \; C n  \exp\left( \frac{-\epsilon_2^2}{\vnpi^2}\right) + n q_{n,\pi}. \label{rmpicont:eqn3}
\end{equation}
Next, using (\ref{tmcont:deltakkpbound}) and recalling that $\nbar = n/2$, we have: for any $\epsilon_3 \geq 0$,
\begin{eqnarray}
&& \P \left\{ \Deltastarmninfn  >  \epsilon_3 \right\} \; \leq \sum_{k \neq k' \in \{1,2\}} \P \left\{ \Deltamnkkpin \; > \; \epsilon_3 \right\}  \nonumber \\
&& \quad \leq \; 2 C \nbar \exp\left(\frac{- \epsilon_3^2}{\vnbarm^2}\right) + 2 \nbar q_{\nbar,m} \;\; \equiv \; Cn \exp\left(\frac{- \epsilon_3^2}{\vnbarm^2}\right) + n q_{\nbar,m}. \label{rmpicont:eqn4}
\end{eqnarray}

Finally, $\bxibarn$ is a simple i.i.d. average defined by the random vector $\bxi(T,\bX)$ and can be controlled as follows. Under Assumption \ref{subgaussian:assmpn} (a) and Lemma \ref{lem:1:GenProp} (ii)(a), $\psionenorm{|\bhj(\bX)|} = \psionenorm{\bhj(\bX)} \leq \sqrt{2} \psitwonorm{\bhj(\bX)} \leq \sqrt{2} \sigmabh$ $\forall \; 1 \leq j \leq d$. Further, due to (\ref{pos:eqn}), $T/\pi(\bX) \leq \deltapi^{-1}$ a.s. $[\P]$. Hence, using Lemma \ref{lem:5:BMC} (ii), we have: for constants $\sigma_{\pi,m} \equiv \sigmabarbxi := 4 \sigmabh\deltapi^{-1} $ and $K_{\pi,m} \equiv \Kbarbxi := 2 \sqrt{2} \sigmabh \deltapi^{-1}$,
\begin{equation}
\bxij (T,\bX) \; \equiv\;  \frac{T}{\pi(\bX)}|\bhj(\bX)| \; \sim \; \BMC(\sigmabarbxi, \Kbarbxi) \quad \forall \; j \in \{1, \hdots, d\}. \label{rmpicont:eqn5}
\end{equation}
Further, $\E\{ \bxij(\T,\bX)\} = \E\{ |\bhj(\bX)|\} \equiv \bmumodbhj$ (say) $\forall \; j \in \{1, \hdots, d\}$, and recall that $\bmumodbhinf = \max\{\bmumodbhj : j = 1, \hdots, d\}$. Using \eqref{rmpicont:eqn5} and Lemma \ref{lem:4:Bernstein} along with the u.b., we then have: for any $\epsilon_4 \geq 0$,
\begin{eqnarray}
&& \P\left\{ \left\| \bxibarn \right\|_{\infty} \; > \; \rbxi(\epsilon_4) \; \equiv \; \bmumodbhinf + \sqrt{2}{\sigmabarbxi}\epsilon_4 + \Kbarbxi \epsilon_4^2 \right\}\nonumber  \\
&& \leq \; \sum_{j=1}^d \P \left\{ \left|  \frac{1}{n} \sum_{i=1}^n \bxij (T_i, \bX_i)  - \bmumodbhj  \right| \; > \; \sqrt{2}{\sigmabarbxi}\epsilon_4 + \Kbarbxi \epsilon_4^2 \right\} \nonumber \\
&& \leq \; 2d \exp(- n \epsilon_4^2) \; \equiv \; 2\exp( - n \epsilon_4^2 + \log d). \label{rmpicont:eqn6}
\end{eqnarray}
Using the bounds (\ref{rmpicont:eqn2}), (\ref{rmpicont:eqn3}), (\ref{rmpicont:eqn4}) and (\ref{rmpicont:eqn6}), along with the u.b., in the original bound (\ref{rmpicont:eqn1}) for $\| \bRpimn \|_{\infty}$, we then have: for any $\epsilon_1, \epsilon_2, \epsilon_3, \epsilon_4 \geq 0$,
\begin{eqnarray}
&& \quad \P \left\{ \left\| \bRpimn \right\|_{\infty} \; > \; \frac{\epsilon_1 \epsilon_3}{\deltapi - \epsilon_2} \rbxi(\epsilon_4)\right\} \; \leq \; 2d \exp( - n \epsilon_4^2) \label{rmpicont:finalbound} \\
&&  \quad\; + C n \left\{\exp\left( \frac{-\epsilon_1^2}{\vnpi^2}\right) + \exp\left( \frac{-\epsilon_2^2}{\vnpi^2}\right) + \exp\left(\frac{- \epsilon_3^2}{\vnbarm^2}\right) \right\} + 2n q_{n,\pi} + n q_{\nbar,m}, \nonumber
\end{eqnarray}
where we assume that $\epsilon_2 < \deltapi$. The proof of Theorem \ref{rmpi:thm:tailbound} is complete. \qed
%

%\paragraph*{Implications of the bound in Theorem \ref{rmpi:thm:tailbound} and characterization of the rate}
\subsection{Completing the Proof of Theorem \ref{RMPI:THM}}

Given the general tail bound for $\|\bRpimn  \|_{\infty}$ in Theorem \ref{rmpi:thm:tailbound}, we next evaluate it under a specific set of choices for $\epsilon_1, \epsilon_2, \epsilon_3, \epsilon_4$ $ > 0$ to understand its behavior and to establish the convergence rate of $\|\bRpimn  \|_{\infty}$. Let $c_1, c_2, c_3, c_4 > 1$ be universal constants, and set $\epsilon_1 = c_1 \vnpi \sqrt{\log n} $, $\epsilon_2 = c_2 \vnpi \sqrt{\log n} $, $\epsilon_3 = c_3 \vnbarm \sqrt{\log n} $ and $\epsilon_4= c_4 \sqrt{(\log d)/n}$, where we assume w.l.o.g. that $\epsilon_2 \leq \deltapi/2 $ and $\epsilon_4 < 1$, so that
\begin{equation*}
\rbxi(\epsilon_4) \; \leq \; \bmumodbhinf + c_4C_{\bxi}\sqrt{\frac{\log d}{n}}, \;\; \mbox{where} \;\; C_{\bxi} := \sqrt{2}\sigmabarbxi + \Kbarbxi
\end{equation*}
 with $\sigmabarbxi$ and $\Kbarbxi$ as in \eqref{rmpicont:eqn5}. Using Theorem \ref{rmpi:thm:tailbound}, we then have: with probability at least $1 - \sum_{j=1}^3 Cn^{-(c_j^2 - 1)} - 2d^{-(c_4^2 - 1)} - 2nq_{n,\pi} - n q_{\nbar,m}$,
\begin{equation*}
%&& \mbox{With probability} \; \geq \; 1 - \sum_{j=1}^3\frac{C}{n^{c_j^2 - 1}}  - \frac{2}{d^{c_4^2 - 1}} - 2nq_{n,\pi} - n q_{\nbar,m}, \\
 \left\| \bRpimn \right\|_{\infty}\; \leq  \; \frac{2 c_1 c_3}{\deltapi} \{\vnpi \vnbarm (\log n)\} \left(\bmumodbhinf + c_4C_{\bxi}\sqrt{\frac{\log d}{n}}\right),  \;\; \mbox{where}
\end{equation*}
This leads to the desired bound and completes the proof of Theorem \ref{RMPI:THM}. \qed

\section{Proof of Theorem \ref{HDINF:THM}}\label{pf:HDinf:thm}

Under the assumed form of $L(\cdot)$ and recalling the definition of $\bSigmahat$ and that $\bnabla \LnDR(\btheta) = \bnabla \LntilDR(\btheta)$, we first note that the gradient $\bnabla \LnDR(\btheta)$ satisfies:
\begin{equation*}
\bnabla \LnDR(\bthetahatDR) - \bnabla \LnDR(\btheta_0) \; = \; %- \frac{2}{n} \sum_{i=1}^n \bPsi(\bX_i) \bPsi(\bX_i)' (\bthetahatDR - \btheta_0) =
 2 \bSigmahat (\bthetahatDR - \btheta_0).
\end{equation*}
Using the definition \eqref{desparsify:est:def} of $\bthetatilDR$ and the notations in \eqref{eq:HDinf:Deltandefn}, we then have:
\begin{eqnarray}
&& (\bthetatilDR - \btheta_0) \; = \; (\bthetahatDR - \btheta_0) - \frac{1}{2} \bOmegahat \{\bnabla \LnDR(\btheta_0) +  2 \bSigmahat (\bthetahatDR - \btheta_0) \} \nonumber \\
&&  \;\;\; = \; -\frac{1}{2} \bOmega \bnabla \LnDR(\btheta_0) - \frac{1}{2}(\bOmegahat - \bOmega) \bnabla \LnDR(\btheta_0) %
 + (I_d - \bOmegahat \bSigmahat) (\bthetahatDR - \btheta_0) \nonumber \\
&& \;\;\; \equiv \; -\frac{1}{2} \bOmega \bnabla \LnDR(\btheta_0)  + \bR_{n,1} + \bR_{n,3} \qquad \mbox{[using \eqref{eq:HDinf:Deltandefn}]}.  \label{eq1:pf:HDinf:thm} %\nonumber
\end{eqnarray}
Next, recall from \eqref{decomp:eqn} that $\bnabla \LnDR(\btheta_0) \equiv \bT_n  = \bTzeron + \bTpin - \bTmn - \bRpimn$, with all notations as in \eqref{tzero}-\eqref{rmpi}. Further, with our choice of $L(\cdot)$, we have:
\begin{equation*}
\bTzeron \; \equiv \; \frac{1}{n} \sum_{i=1}^n \bT_0(\bZ_i)  \; = \; -\frac{2}{n} \sum_{i=1}^n \bpsi_0(\bZ_i), \;\; \mbox{with} \; \bpsi_0(\bZ) \; \mbox{as in the ALE \eqref{desparsifiedest:ALE}}.
\end{equation*}
Applying these facts in \eqref{eq1:pf:HDinf:thm} and using the notations in \eqref{eq:HDinf:Deltandefn}, we then have:
\begin{eqnarray}
&& (\bthetatilDR - \btheta_0) \; = \; -\frac{1}{2} \bOmega(\bTzeron + \bTpin - \bTmn - \bRpimn) + \bR_{n,1} + \bR_{n,3} \nonumber \\
&& \;\;\; = \;  -\frac{1}{2} \bOmega \bTzeron -\frac{1}{2} \bOmega(\bTpin - \bTmn - \bRpimn) + \bR_{n,1} + \bR_{n,3} \nonumber  \\
&& \;\;\; \equiv \; \frac{1}{n} \sum_{i=1}^n \bOmega  \bpsi_0(\bZ_i) + \bR_{n,1} + \bR_{n,2} + \bR_{n,3} \; \equiv \; \frac{1}{n} \sum_{i=1}^n \bOmega \bpsi_0(\bZ_i) + \bDelta_n. \label{eq2:pf:HDinf:thm}
\end{eqnarray}
Now, under Assumptions \ref{base:assmpns}, \ref{subgaussian:assmpn}, \ref{tpicont:assmpn} and \ref{tmcont:assmpn}, all of Theorems \ref{TZERO:THM}-\ref{RMPI:THM} apply, and under Assumption \ref{strngconv_assmpn} and with $L(\cdot)$ being convex and differentiable in $\btheta$  trivially, Lemma \ref{DEV:BOUND} applies as well. Using these results, we then have: %Thus, with $\lambda_n \asymp \sqrt{(\log d)/n}$ (as assumed), we have
\begin{align}
& \| \bnabla \LnDR(\btheta_0) \|_{\infty} = O_{\P}\left( \sqrt{\frac{\log d}{n}}\right) \; \mbox{and} \;\; \| \bthetahatDR(\lambda_n) - \btheta_0\|_1  =  O_{\P}\left( s\sqrt{\frac{\log d}{n}}\right) \label{eq:extra:pf:HDinf:thm} %% Had to eventually label this as I needed to refer to this eqn. at one place in the std. error estimator's consistency's proof. %%
\end{align}
for any choice of $\lambda_n \asymp \sqrt{(\log d)/n}$, as assumed. Using these facts along with Assumption \ref{HDinf:assmpn} (a) and multiple uses of $L_1$-$L_{\infty}$ type bounds, we then have:
\begin{eqnarray}
&& \quad \| \bR_{n,1} \|_{\infty} \; \leq \; \frac{1}{2} \| \bOmegahat - \bOmega \|_1  \| \bnabla \LnDR(\btheta_0) \|_{\infty} \; = \; O_{\P} \left( r_n \sqrt{\frac{\log d}{n}} \right), \;\; \mbox{and} \label{eq3:pf:HDinf:thm} \\
&& \quad \| \bR_{n,3} \|_{\infty} \; \leq \; \| I_d - \bOmegahat \bSigmahat\|_{\max} \| \bthetahatDR(\lambda_n) - \btheta_0\|_1 \; = \; O_{\P} \left( \omega_n s\sqrt{\frac{\log d}{n}} \right). \label{eq4:pf:HDinf:thm}
\end{eqnarray}
Next, to control $\bR_{n,2} \equiv - \half \bOmega(\bTpin - \bTmn - \bRpimn)$, observe that each of the variables $- \half \bOmega\bTpin$, $- \half \bOmega \bTmn $ and $- \half \bOmega \bRpimn$ admit exactly the \emph{same form} as $\bTpin$, $\bTmn$ and $\bRpimn$ in \eqref{decomp:eqn}, respectively, but with a \emph{different choice} of the function $\bh(\bX)$ in the definitions \eqref{tpi}-\eqref{rmpi} of the underlying summands for these terms. In this particular case, the summands correspond to the forms \eqref{tpi}-\eqref{rmpi} with $h(\bX)$ replaced by $\widetilde{h}(\bX) = \bOmega \bPsi(\bX) \equiv \bUpsilon(\bX)$.

Further under Assumption \ref{HDinf:assmpn} (b), $\widetilde{h}(\bX)$  is sub-Gaussian with $\psitwonorm{\widetilde{h}(\bX)} \leq \sigmabUps$, as required in Assumption \ref{subgaussian:assmpn} (a). %$\psitwonorm{\widetilde{h}(\bX)} = \psitwonorm{\bUpsilon(\bX)} \leq \sigmabUps$.
Hence, under Assumptions \ref{base:assmpns}, \ref{subgaussian:assmpn}, \ref{tpicont:assmpn}, \ref{tmcont:assmpn} and \ref{HDinf:assmpn}, Theorems \ref{TPI:THM}, \ref{TM:THM} and \ref{RMPI:THM} certainly apply to $\bOmega\bTpin$, $\bOmega \bTmn $ and $\bOmega \bRpimn$ with this `modified choice' $\widetilde{h}(\bX)$ of $\bh(\bX)$, using which we have:
\begin{align*}
& \left\| \bOmega \bTpin \right\|_{\infty} + \left\| \bOmega \bTmn \right\|_{\infty}   \; = \; O_{\P}\left( (\vnpi + \vnbarm) \sqrt{\frac{ (\log d) \log (nd)}{n}} \right)   \nonumber  \\  %\sigmabUps
\mbox{and} \;\; &   \| \bOmega \bRpimn \|_{\infty}  = \; O_{\P} \left(  \vnpi \vnbarm \log n \right), \nonumber %\sigmabUps
\end{align*}
where both results follow directly from the non-asymptotic bounds in Theorems \ref{TPI:THM}-\ref{RMPI:THM}. %Note further that for both the rates, we have made their dependence on the sub-Gaussian norm $\sigmabUps$ of $\bUpsilon(\bX)$ explicit, whereby it is allowed to depend on $d$ and not necessarily be a constant (although the latter case is what will be typically assumed, implicitly or explicitly).
Combining these and recalling the definition of $v_n^*$ in Assumption \ref{HDinf:assmpn} (b) along with the rate condition on $v_n^*$ assumed therein, we have:
\begin{equation}
\| \bR_{n,3}\|_{\infty} \; \equiv \; \half \| \bOmega(\bTpin - \bTmn - \bRpimn) \|_{\infty} \; = \; O_{\P} \left( v_n^* \nnhalf \right). \label{eq5:pf:HDinf:thm} %
\end{equation}

Combining \eqref{eq3:pf:HDinf:thm}, \eqref{eq4:pf:HDinf:thm} and \eqref{eq5:pf:HDinf:thm} along with the definition of $\bDelta_n$ in \eqref{eq:HDinf:Deltandefn}, and using these in the original decomposition  \eqref{eq2:pf:HDinf:thm} of $(\bthetatilDR - \btheta_0)$, we have:
\begin{eqnarray}
&& \quad (\bthetatilDR - \btheta_0) \; = \;  \frac{1}{n} \sum_{i=1}^n \bOmega \bpsi_0(\bZ_i) + \bDelta_n, \quad \mbox{where} \;\; \bDelta_n \;\; \mbox{satisfies:} \nonumber \\
&& \| \bDelta_n \|_{\infty} \; \equiv \; \|  \bR_{n,1} +  \bR_{n,2} +  \bR_{n,2}  \|_{\infty} \leq \; \|  \bR_{n,1} \|_{\infty} +  \|  \bR_{n,2} \|_{\infty} + \|  \bR_{n,3} \|_{\infty} \nonumber \\
&& \qquad\qquad = \; O_{\P} \left( r_n \sqrt{\frac{\log d}{n}} + v_n^* \nnhalf + \omega_n s\sqrt{\frac{\log d}{n}} \right) \; = \; o_{\P}(\nnhalf).  \label{eq6:pf:HDinf:thm} \qed
\end{eqnarray}
\eqref{eq6:pf:HDinf:thm} therefore establishes the desired ALE \eqref{desparsifiedest:ALE}. Note further that the claim $\E\{ \bpsi_0(\bZ) \} = \bzero$ holds as a simple consequence of the definition of $\btheta_0$ and Assumption \ref{base:assmpns} (b). Specifically, recalling the notations $\varepsilon(\Z) = Y - m(\bX)$ and $\psi(\bX) = m(\bX) - g(\bX, \btheta_0)$ from Assumption \ref{subgaussian:assmpn} (a), with  $g(\bX,\btheta_0) = \bPsi(\bX)'\btheta_0$ for our choice of $L(\cdot)$, we have: $\E\{ \varepsilon(\Z) \medgiven \bX \} = 0$, by definition of $m(\bX)$, and hence, $\E\{ \psi(\bX) \bPsi(\bX)  \} = \E[ \bPsi(\bX)\{Y - \bPsi(\bX)'\btheta_0\}] - \E\{ \bPsi(\bX)\varepsilon(\Z)\}= \bzero$, by definition of $\btheta_0$ and $L(\cdot)$. Further, $T \ind Y \medgiven \bX$ by Assumption \ref{base:assmpns} (a). Thus,
\begin{equation*}
 \E\{ \bpsi_0(\bZ)\} \; \equiv \; \E\{\bPsi(\bX)\psi(\bX)\} + \E_{\bX}[\E\{ T \pi^{-1}(\bX) \medgiven \bX\}\E\{\varepsilon(\Z) \medgiven \bX\}] \; = \; \bzero.
\end{equation*}
This therefore completes the proof of the first part of Theorem \ref{HDINF:THM}. \qed
\par\smallskip
To establish the (coordinatewise) asymptotic normality results claimed in the second part, we simply use the established ALE \eqref{desparsifiedest:ALE} or \eqref{eq6:pf:HDinf:thm} and invoke Lyapunov's Central Limit Theorem (CLT) along with Slutsky's Theorem. To apply Lyapunov's CLT, we need to verify the Lyapunov moment conditions for $\bGamma_0(\bZ) \equiv \bOmega \bpsi_0(\bZ)$. We establish this by first showing that $\bGamma_0(\bZ)$ is, in fact, sub-exponential (as per Definition \ref{orlicz:def} with $\alpha = 1$) under our assumptions.

To this end, under Assumptions \ref{subgaussian:assmpn} (a),  \ref{base:assmpns} (b) and \ref{HDinf:assmpn} (b), we have: %with $(\sigmaeps, \sigmapsi)$, $\deltapi$ and $\sigmabUps$ as defined therein and with $\bUpsilon(\bX) \equiv \bOmega \bPsi(\bX)$, we have:
\begin{eqnarray}
&& \quad \; \psionenorm{\bGamma_0(\bZ) }  \; \equiv \; \psionenorm{ \bOmega \bpsi_0(\bZ) }  \; = \; \psionenorm{\bOmega \bPsi(\bX) \{\psi(\bX) + T \pi^{-1}(\bX) \varepsilon(\Z)\}} \label{eq7:pf:HDinf:thm} \\
&& %\;\; \equiv \psionenorm{\bUpsilon(\bX) \{\psi(\bX) + T \pi^{-1}(\bX) \varepsilon(\Z)\}}
 \leq \;  \psitwonorm{\bOmega \bPsi(\bX)} \{\psitwonorm{\psi(\bX)} + \psitwonorm{\varepsilon(\Z)} \deltapi^{-1} \} \; \leq \;   \sigmabUps (\sigmapsi + \deltapi^{-1} \sigmaeps)  \; =:  \sigma_{\bGamma_0},  \nonumber
\end{eqnarray}
where the steps follow from using Lemma \ref{lem:1:GenProp} (v) and (i) (c). Consequently, using \eqref{eq7:pf:HDinf:thm} and Lemma \ref{lem:1:GenProp} (iv) (a), we have: uniformly in $j \in \{1, \hdots, d\}$,
\begin{equation*}
\; \rho_{\bGamma_0, j} \; := \; \E\{|\bGamma_{0[j]}(\bZ)|^3 \} \; \leq \; 6 \sigma_{\bGamma_0}^3 < \;  \infty \;\; \mbox{and} \;\; \sigma_{0,j}^2 \; := \; \E\{|\bGamma_{0[j]}(\bZ)|^2  \} \; > \; c_0^2,
\end{equation*}
where the second result is due to the lower bound condition assumed on $ \sigma_{0,j}$ with the constant $c_0 > 0$ as defined there. Hence, $\rho_{\bGamma_0, j} /\sigma_{0,j}^3 \leq 6 \sigma_{\bGamma_0}^3 / c_0^3 < \infty$ uniformly in $j \in \{1, \hdots, d\}$. Thus, the Lyapunov moment conditions are now verified (uniformly) for each coordinate of $\bGamma_0(\bZ) \equiv \bOmega \bPsi_0(\bZ)$. Note also that $\E\{\bGamma_0(\bZ) \} = \bzero$ since $\E\{ \bpsi_0(\bZ) \} = \bzero$, as shown earlier. Finally, observe that $\sigma_{0,j}^{-1} |\bDelta_{n [j]}| \leq c_0^{-1} \| \bDelta_n \|_{\infty} = o_{\P}(\nnhalf)$. Hence, by Lyapunov's CLT along with multiple uses of Slutsky's Theorem, we have: for each  $1 \leq j \leq d$,
\begin{eqnarray}
&& \quad \sqrt{n}\sigma_{0,j}^{-1}\left(\bthetatil_{\DDR [j]} - \btheta_{0 [j]}\right) \; = \; \frac{1}{\sqrt{n}\sigma_{0,j}}\sum_{i=1}^n \bGamma_{0 [j]}(\bZ_i) + \sqrt{n} \sigma_{0,j}^{-1} \bDelta_{n [j]|} \label{eq8:pf:HDinf:thm}  \\
&& \quad = \; \frac{1}{\sqrt{n}\sigma_{0,j}}\sum_{i=1}^n \bGamma_{0 [j]}(\bZ_i) + o_{\P}(1) \;\; \convd \; \Nsc(0,1) + o_{\P}(1) \;\; \convd \; \Nsc(0,1). \nonumber \qed
\end{eqnarray}
This establishes the first of the two (coordinatewise) asymptotic normality claims in Theorem \ref{HDINF:THM}. For the second claim, we mainly need to establish the consistency of the estimator $\sigmahat_{0,j}^2$ of $\sigma_{0,j}^2$, uniformly in $1 \leq j \leq d$, as claimed. The asymptotic normality then follows directly from Slutsky's Theorem and \eqref{eq8:pf:HDinf:thm}. To establish the consistency, we first note that for all $1 \leq j \leq d$,
\begin{align}
& \qquad \sigma_{0,j}^2 - \sigma_{0,j}^2 \; \equiv \; \frac{1}{n} \sum_{i=1}^n \bGammahat_{0 [j]}^2 (\bZ_i) - \E\{\bGamma_{0[j]}^2(\bZ)\} \label{eq9:pf:HDinf:thm} \\
&  = \left\{\frac{1}{n} \sum_{i=1}^n \bGammahat_{0 [j]}^2 (\bZ_i) - \frac{1}{n} \sum_{i=1}^n \bGamma_{0 [j]}^2 (\bZ_i)\right\} + \left\{ \frac{1}{n} \sum_{i=1}^n \bGamma_{0 [j]}^2 (\bZ_i) - \E\{\bGamma_{0[j]}^2(\bZ)\} \right\},  \nonumber %\\
\end{align}
where $\bGamma_0(\bZ) =  \bOmega \bpsi_0(\bZ)$ and $\bGammahat_0(\bZ)  =  \bOmegahat \bpsihat_0(\bZ)$ with $\bpsihat_0(\bZ)$ given by:
\begin{equation*}
\bpsihat_0(\bZ) \; := \; \left[\{\mhat(\bX) - \bPsi(\bX)'\bthetahatDR\}+ \frac{T}{\pihat(\bX)} \{Y - \mhat(\bX)\}\right] \bPsi(\bX). \nonumber
\end{equation*}
Next, recall from \eqref{decomp:eqn} the terms $\bTzero(\bZ), \bTpi(\bZ), \bTm(\bZ)$ and $\bRpim(\bZ)$ defined in \eqref{tzero}-\eqref{rmpi}, with $g(\bX,\btheta_0) = \bPsi(\bX)'\btheta_0$ and $h(\bX) = - 2\bPsi(\bX)$ in this case, and let $\bTzero^*(\bZ), \bTpi^*(\bZ), \bTm^*(\bZ)$ and $ \bRpim^*(\bZ)$ respectively denote their versions with $h(\bX)$ replaced by $h^*(\bX) = \bPsi(\bX)$. Then, we have: $\bpsi_0(\bZ) = \bTzero^*(\bZ)$ and
\begin{equation*}
 \bpsihat_0(\bZ) \; = \; \bTzero^*(\bZ) + \bTpi^*(\bZ) - \bTm^*(\bZ) - \bRpim^*(\bZ) - \bPsi(\bX)\bPsi(\bX)' (\bthetahatDR - \btheta_0). \nonumber
\end{equation*}
Hence for all $1 \leq i \leq n$, $\bGammahat_0(\bZ_i) - \bGamma_0(\bZ_i)$ satisfies:
\begin{eqnarray}
%& \mbox{Hence} \;\; \forall \; i, \;\;
&& \quad\;\; \bGammahat_0(\bZ_i) - \bGamma_0(\bZ_i) \; \equiv \; \bOmegahat \bpsihat_0(\bZ_i) - \bOmega\bpsi_0(\bZ_i)
\; = \; (\bOmegahat - \bOmega)\bTzero^*(\bZ_i) \label{eq10:pf:HDinf:thm}  \\
&&  \quad \quad + \bOmegahat\{\bTpi^*(\bZ_i) - \bTm^*(\bZ_i) - \bRpim^*(\bZ_i)\} - \bOmegahat\bPsi(\bX_i)\bPsi(\bX_i)' (\bthetahatDR - \btheta_0). \nonumber
\end{eqnarray}
Under Assumption \ref{subgaussian:assmpn} (a) and \ref{base:assmpns} (b), similar to the proof of \eqref{eq7:pf:HDinf:thm}, we have using Lemma \ref{lem:1:GenProp} (v) and (i) (c): $\bTzero^*(\bZ) \equiv - \half \bTzero(\bZ)$ is sub-exponential with
\begin{equation*}
 \psionenorm{\bTzero^*(\bZ)}  \; \leq \;  \psitwonorm{\bPsi(\bX)} \{\psitwonorm{\psi(\bX)} + \psitwonorm{\varepsilon(\Z)} \deltapi^{-1} \} \; \leq \;   \sigmabh (\sigmapsi + \deltapi^{-1} \sigmaeps).  \nonumber
\end{equation*}
Hence, $\max_{1 \leq i \leq n} \|\bTzero^*(\bZ_i) \|_{\infty} \equiv \max_{1 \leq i \leq n, 1 \leq j \leq d } |\bT_{0 [j]}^*(\bZ_i)| = O_{\P}(\log (nd))$  due to Lemma \ref{lem:1:GenProp} (vi). Using this along with Assumption \ref{HDinf:assmpn} (a), we have:
\begin{align}
\max_{1 \leq i \leq n} \| (\bOmegahat - \bOmega) \bTzero^*(\bZ_i)\|_{\infty} \; \leq \;  \| \bOmegahat - \bOmega \|_1 \max_{i} \|\bTzero^*(\bZ_i) \|_{\infty}  \; =  O_{\P}\left(r_n \log(nd)\right). \label{eq11:pf:HDinf:thm}
\end{align}
Now, since $\bPsi(\bX)$, $\varepsilon(\Z)$ and $\bOmega \bPsi(\bX)$ are all sub-Gaussian due to Assumptions \ref{subgaussian:assmpn} (a) and \ref{HDinf:assmpn} (b), using Lemma \ref{lem:1:GenProp} (vi), we have:
\begin{equation}
\max_{1 \leq i \leq n} \{\| \bPsi(\bX_i) \|_{\infty} + \| \bOmega \bPsi(\bX_i) \|_{\infty} + |\varepsilon(\Z_i)| \} \; = \; O_{\P}\left(\sqrt{\log (nd)}\right). \label{eq12:pf:HDinf:thm}
\end{equation}
Next, recalling the proof techniques and notations introduced in the proofs of Theorems \ref{TPI:THM}, \ref{TM:THM} and \ref{RMPI:THM}, as well as using Assumption \ref{base:assmpns} (b),  we have:
\begin{align}
& \max_{1 \leq i \leq n} \| \bTpi^*(\bZ_i) \|_{\infty}  \;\;\; \leq \; \deltapi^{-1} \pitilninfn \Deltapininfn  \max_{1 \leq i \leq n} \{\| \bPsi(\bX_i) \|_{\infty} |\varepsilon(\Z_i)|\}, \label{eq13:pf:HDinf:thm}\\
%& \max_{1 \leq i \leq n} \| \bOmega \bTpi^*(\bZ_i) \|_{\infty}  \; \leq \; \deltapi^{-1} \pitilninfn \Deltapininfn \max_{1 \leq i \leq n} \{\| \bOmega \bPsi(\bX_i) \|_{\infty} |\varepsilon(\Z_i)|\}; \\
& \max_{1 \leq i \leq n} \| \bTm^*(\bZ_i) \|_{\infty} \;\; \leq \; (1 + \deltapi^{-1}) \Deltastarmninfn \max_{1 \leq i \leq n} \| \bPsi(\bX_i)\|_{\infty} \;\; \mbox{and} \nonumber \\
%& \max_{1 \leq i \leq n} \| \bOmega \bTm^*(\bZ_i) \|_{\infty} \; \leq \; (1 + \deltapi^{-1}) \Deltastarmninfn \max_{1 \leq i \leq n} \| \bOmega \bPsi(\bX_i)\|_{\infty}; \;\; \mbox{and} \\
& \max_{1 \leq i \leq n} \| \bRpim^*(\bZ_i) \|_{\infty}  \leq \; \deltapi^{-1} \pitilninfn \Deltapininfn \Deltastarmninfn \max_{1 \leq i \leq n} \| \bPsi(\bX_i) \|_{\infty}, \nonumber %\\
%& \max_{1 \leq i \leq n} \| \bOmega \bRpim^*(\bZ_i) \|_{\infty} \; \leq \; \deltapi^{-1} \pitilninfn \Deltapininfn \Deltastarmninfn \max_{1 \leq i \leq n} \| \bOmega \bPsi(\bX_i) \|_{\infty} \nonumber
\end{align}
where $\pitilninfn$  and $\Deltapininfn$ are as in \eqref{tpicont:notn1}-\eqref{tpicont:notn2} and $\Deltastarmninfn$ is as defined in \eqref{tmcont:notn1} and \eqref{rmpicont:eqn1}. Using \eqref{rmpicont:eqn2}, \eqref{rmpicont:eqn3} and \eqref{rmpicont:eqn4}, we further have:
\begin{align}
\pitilninfn \Deltapininfn  =  O_{\P}(\vnpi \sqrt{\log n}) \;\; \mbox{and} \;\; \Deltastarmninfn  =  O_{\P}( \vnbarm\sqrt{\log n}). \label{eq14:pf:HDinf:thm}
\end{align}
Using \eqref{eq12:pf:HDinf:thm} and \eqref{eq14:pf:HDinf:thm} in \eqref{eq13:pf:HDinf:thm}, we then have:
\begin{align}
& \max_{1 \leq i \leq n} \{ \| \bTpi^*(\bZ_i) \|_{\infty} + \| \bTm^*(\bZ_i) \|_{\infty} + \| \bRpim^*(\bZ_i) \|_{\infty} \} \; = \; O_{\P}(\widetilde{v}_n),  \label{eq15:pf:HDinf:thm}  \\
& \mbox{where} \;\; \widetilde{v}_n \; := \; \{(\vnpi + \vnbarm) \sqrt{\log n} + \vnpi \vnbarm (\log n) \} \log (nd). \nonumber
\end{align}
Using similar arguments as above, with $\bPsi(\bX)$ replaced by $\bOmega \bPsi(\bX)$ in \eqref{eq13:pf:HDinf:thm} throughout, and using \eqref{eq12:pf:HDinf:thm} and \eqref{eq14:pf:HDinf:thm}, we also have:
\begin{equation}
\max_{1 \leq i \leq n} \{ \| \bOmega \bTpi^*(\bZ_i) \|_{\infty} + \| \bOmega \bTm^*(\bZ_i) \|_{\infty} + \|\bOmega \bRpim^*(\bZ_i) \|_{\infty} \} \; = \; O_{\P}(\widetilde{v}_n).  \label{eq16:pf:HDinf:thm}  \\
\end{equation}
Combining \eqref{eq15:pf:HDinf:thm} and \eqref{eq16:pf:HDinf:thm} along with Assumption \ref{HDinf:assmpn} (a), we have:
\begin{align}
& \max_{1 \leq i \leq n} \| \bOmegahat\{\bTpi^*(\bZ_i) - \bTm^*(\bZ_i) - \bRpim^*(\bZ_i) \} \|_{\infty} \nonumber\\
& \;\; \leq \; \max_{1 \leq i \leq n} \{\| \bOmega \bTpi^*(\bZ_i) \|_{\infty} +  \| \bOmega \bTm^*(\bZ_i) \|_{\infty} + \|\bOmega \bRpim^*(\bZ_i) \|_{\infty}\} \nonumber \\
& \;\; \quad + \|\bOmegahat - \bOmega \|_1 \max_{1 \leq i \leq n} \{ \| \bTpi^*(\bZ_i) \|_{\infty} + \| \bTm^*(\bZ_i) \|_{\infty} + \| \bRpim^*(\bZ_i) \|_{\infty} \} \nonumber \\
& \;\; = \; O_{\P}\left( \widetilde{v}_n + r_n \widetilde{v_n}\right) \; = \; O_{\P}\left(\widetilde{v}_n\right),  \quad \mbox{since} \;\; r_n = o(1). \label{eq17:pf:HDinf:thm}
\end{align}
Now turning to the third term in \eqref{eq10:pf:HDinf:thm}, under Assumption \ref{HDinf:assmpn}, and using \eqref{eq:extra:pf:HDinf:thm} and \eqref{eq12:pf:HDinf:thm} along with multiple uses of $L_1$-$L_{\infty}$ type bounds, we have:
\begin{align}
& \|\bOmegahat \bPsi(\bX_i) \bPsi(\bX_i)' (\bthetahatDR - \btheta_0) \|_{\infty} \; \leq \; \|\bOmega \bPsi(\bX_i)\|_{\infty} \|\bPsi(\bX_i)\|_{\infty} \|\bthetahatDR - \btheta_0 \|_1   \nonumber \\
&\qquad +  \|\bOmegahat - \bOmega\|_1  \|\bPsi(\bX_i)\|_{\infty} \|\bPsi(\bX_i)\|_{\infty} \|\bthetahatDR - \btheta_0\|_1  \;\;\; \forall \; 1 \leq i \leq n, \;\;\; \mbox{so that} \nonumber\\
& \max_{1 \leq i \leq n} \|\bOmegahat \bPsi(\bX_i) \bPsi(\bX_i)' (\bthetahatDR - \btheta_0) \|_{\infty} \; \leq \; O_{\P}\left( s \sqrt{\frac{\log d}{n}} \log(nd) (1+ r_n) \right). \label{eq18:pf:HDinf:thm}
\end{align}
Applying \eqref{eq11:pf:HDinf:thm}, \eqref{eq17:pf:HDinf:thm} and \eqref{eq18:pf:HDinf:thm} in \eqref{eq10:pf:HDinf:thm} via triangle inequality, we get
\begin{align}
& \max_{1 \leq i \leq n} \| \bGammahat_0(\bZ_i) - \bGamma_0(\bZ_i) \|_{\infty} \; = \; O_{\P} \left( r_n \log(nd) +  \widetilde{v}_n +  s \sqrt{\frac{\log d}{n}} \log(nd) \right). \label{eq19:pf:HDinf:thm}
\end{align}
Finally, note that owing to \eqref{eq7:pf:HDinf:thm}, $\bGamma_0(\bZ)$ is sub-exponential with $\psionenorm{\bGamma_0(\bZ)}$ $\leq \sigma_{\bGamma_0} < \infty$. Hence, using Bernstein's Inequality (Lemma \ref{lem:4:Bernstein}), we have:
\begin{align}
& \max_{1 \leq j \leq d} \left\{\frac{1}{n} \sum_{i=1}^n | \bGamma_{0 [j]} (\bZ_i)|\right\} \; \leq \; \max_{1 \leq j \leq d} \E\{|\bGamma_{0 [j]} (\bZ)|\} + O_{\P}\left( \sqrt{\frac{\log d}{n}} + \frac{\log d}{n}\right), \label{eq20:pf:HDinf:thm}
\end{align}
which is $O_{\P}(1)$ since $ \E\{|\bGamma_{0 [j]} (\bZ)|\} \leq \sigma_{\bGamma_0}$ $\forall \; j$ owing to Lemma \ref{lem:1:GenProp} (iv) (a).

Applying \eqref{eq19:pf:HDinf:thm} and \eqref{eq20:pf:HDinf:thm} to the first term in \eqref{eq9:pf:HDinf:thm} via several uses of the triangle inequality and that $a^2 - b^2 = (a-b)(a+b)$ $\forall\; a, b \in \R$, we have:
\begin{eqnarray}
&& \quad \max_{1 \leq j \leq d} \left|\frac{1}{n} \sum_{i=1}^n \bGammahat_{0 [j]}^2 (\bZ_i) - \frac{1}{n} \sum_{i=1}^n  \bGamma_{0 [j]}^2 (\bZ_i)\right| \label{eq21:pf:HDinf:thm} \\
&& = \; \max_{1 \leq j \leq d} \frac{1}{n} \sum_{i=1}^n |\bGammahat_{0 [j]}(\bZ_i) - \bGamma_{0 [j]}(\bZ_i)| \; |\bGammahat_{0 [j]}(\bZ_i) - \bGamma_{0 [j]} (\bZ_i) + 2\bGamma_{0 [j]} (\bZ_i)| \nonumber \\
&& \leq \; \max_{1 \leq i \leq n} \| \bGammahat_0(\bZ_i) - \bGamma_0(\bZ_i) \|_{\infty} \left[\max_{1 \leq j \leq d} \left\{\frac{2}{n} \sum_{i=1}^n | \bGamma_{0 [j]} (\bZ_i)|\right\} + o_{\P}(1)\right] \nonumber \\
&& = \; O_{\P} \left( r_n \log(nd) +  \widetilde{v}_n +  s \sqrt{\frac{\log d}{n}} \log(nd) \right). \nonumber
\end{eqnarray}
Furthermore, since $\psionenorm{\bGamma_0(\bZ)}$ $\leq \sigma_{\bGamma_0}$, we have: $\max_{1 \leq j \leq d} \psialphanorm{\Gamma_{0[j]}^2(\bZ)}\leq \sigma_{\bGamma_0}^2$ with $\alpha = \half$ owing to Lemma \ref{lem:1:GenProp} (v). Hence, using Lemma \ref{lem:6:VarianceTailBound}, we get
\begin{align}
\max_{1 \leq j \leq d} \left|\frac{1}{n} \sum_{i=1}^n \bGamma_{0 [j]}^2 (\bZ_i) - \E\{ \bGamma_{0 [j]}^2 (\bZ) \}\right| \; \leq \; O_{\P}\left( \sqrt{\frac{\log d}{n}} + \frac{(\log n)^2 (\log d)^2}{n}\right). \label{eq22:pf:HDinf:thm}
\end{align}

Hence, combining \eqref{eq21:pf:HDinf:thm} and \eqref{eq22:pf:HDinf:thm} via a triangle inequality and applying them in \eqref{eq9:pf:HDinf:thm}, and recalling $\widetilde{v}_n$ from \eqref{eq15:pf:HDinf:thm}, we finally have: %\eqref{
\begin{eqnarray}
& \quad \max_{1 \leq j \leq d} | \sigmahat_{0,j}^2 - \sigma_{0,j}^2 |   \; = \; O_{\P} (\tau_n) \; = \; o_{\P}(1), \quad \mbox{where} \label{eq23:pf:HDinf:thm} \\
& \qquad \tau_n \; := \; r_n \log(nd) +  \widetilde{v}_n +  s \sqrt{\frac{\log d}{n}} \log(nd) +  \sqrt{\frac{\log d}{n}} + \frac{(\log n)^2 (\log d)^2}{n}. \nonumber
\end{eqnarray}
Note that we have implcitly assumed $\tau_n$ to be $o(1)$ here. A careful analysis will reveal that this entails essentially the same rate conditions as those needed for the ALE \eqref{desparsifiedest:ALE} in Theorem \ref{HDINF:THM} to hold, upto an additional factor of $\sqrt{\log (nd)}$ appearing in the first three terms of $\tau_n$, as well as the presence of the last term in $\tau_n$ (which is expected to be of lower order than the rest).

\eqref{eq23:pf:HDinf:thm} therefore establishes the desired (uniform) consistency of the standard error estimators $\{ \sigmahat_{0,j}\}_{j=1}^d$, and also establishes the second asymptotic normality result in Theorem \ref{HDINF:THM} through use of \eqref{eq8:pf:HDinf:thm}, \eqref{eq23:pf:HDinf:thm} and Slutsky's Theorem, as discussed earlier. This completes the proof of Theorem \ref{HDINF:THM}. \qed

\section{Proofs of all Results in Appendix \ref{SEC:NUISANCE:SEPARATE:SUPP}}\label{pfs:nuisance} %Section \ref{NUISANCE}}\label{pfs:nuisance}
%We present here the proofs of Theorems \ref{parametric:thm}-\ref{KS:mainthm2}, as well as the assumptions required for Theorems \ref{KS:mainthm1} and \ref{KS:mainthm2}. We begin with the proof of Theorem \ref{parametric:thm}.

%We present here the proofs of Theorems \ref{parametric:thm}-\ref{KS:mainthm2} stated in Appendix \ref{SEC:NUISANCE:SEPARATE:SUPP}.

\subsection{Proof of Theorem \ref{parametric:thm}}\label{pf:parametric:thm}

Under the assumed conditions, we have:
\begin{align}
& \sup_{\bx \in \Xsc} | g\{\bbetahat'\bPsi(\bx) \} - g\{\bbeta_0'\bPsi(\bx) \}| \; \leq \; C_g \sup_{\bx \in \Xsc} | (\bbetahat - \bbeta_0)'\bPsi(\bx) |  \nonumber \\
 & \qquad \leq \; C_g \|\bbetahat - \bbeta_0 \|_1 \sup_{\bx \in \Xsc} \|\bPsi(\bX) \|_{\infty} \;\; \leq \; C_g C_{\bPsi} \| \bbetahat - \bbeta_0 \|_1. \label{eq1:pf:parametric:thm}
\end{align}
where the steps follow from the Lipschitz continuity of $g(\cdot)$ and the boundedness of $\bPsi(\cdot)$ along with an $L_1$-$L_{\infty}$ bound.
Now, under the $L_1$ error bound assumed for $\bbetahat$ and using a simple union bound argument, we have: $\forall \; \epsilon \geq 0$,
\begin{align*}
& \P(\|\bbetahat - \bbeta_0\|_1 > \epsilon) \nonumber \\
&   = \; \P(\|\bbetahat - \bbeta_0\|_1 > \epsilon, \|\bbetahat - \bbeta_0\|_1 \leq a_n) + \P(\|\bbetahat - \bbeta_0\|_1 > \epsilon, \|\bbetahat - \bbeta_0\|_1 > a_n) \nonumber \\
& \; \leq \; \P(\|\bbetahat - \bbeta_0\|_1 > \epsilon , \|\bbetahat - \bbeta_0\|_1 \leq a_n) + \P(\|\bbetahat - \bbeta_0\|_1 > a_n) \nonumber \\
& \; \leq \; \P(\|\bbetahat - \bbeta_0\|_1 > \epsilon  \given \|\bbetahat - \bbeta_0\|_1 \leq a_n)\P(\|\bbetahat - \bbeta_0\|_1 \leq a_n) +q_n \nonumber \\
& \; \leq \; 2 \exp\{-\epsilon^2/(2a_n^2)\} (1- q_n)+ q_n \; \leq \; 2 \exp\{-\epsilon^2/(2a_n^2)\}  + q_n, \nonumber
\end{align*}
where the final bounds follow from an application of Hoeffding's inequality for bounded random variables (or using Lemma \ref{lem:1:GenProp} (ii)(d) and (iii)(a)). Using this bound along with that in \eqref{eq1:pf:parametric:thm}, we then have: for any $\epsilon \geq 0$,
\begin{align*}
& \P[\sup_{\bx \in \Xsc} | g\{\bbetahat'\bPsi(\bx) \} - g\{\bbeta_0'\bPsi(\bx) \}| > C_g C_{\bPsi} \epsilon ] \; \leq \; 2 \exp\{-\epsilon^2/(2a_n^2)\} + q_n. \nonumber
\end{align*}
The desired result then follows by setting $\epsilon = \sqrt{2}a_n t$ for any $t \geq 0$. \qed

\subsection{Proof Sketch for Theorems \ref{KS:mainthm1} and \ref{KS:mainthm2}}\label{pfs:KS:mainthms} We first introduce two key supporting lemmas regarding tail bounds for $\lhat(\bbetahat,\bx)$ both of which will be useful for proving Theorems \ref{KS:mainthm1} and \ref{KS:mainthm2}.
We begin with a few notations and a sketch of our analysis to set up and prove these lemmas, and subsequently, use them to complete the proofs of the main theorems.

To analyze the behavior of $\lhat(\bbetahat,\bx)$, we first introduce the corresponding \emph{hypothetical} version of the estimator where the index parameter $\bbeta$ is treated as known. Specifically, for any $\bx \in \Xsc$, let us define the `oracle' `estimator':
\begin{equation*}
\ltil(\bbeta, \bx) \; := \; \frac{1}{nh} \sum_{i=1}^n Z_i K \left( \frac{\bbeta'\bX_i - \bbeta'\bx}{h} \right) \; \equiv \; \frac{1}{nh} \sum_{i=1}^n Z_i K \left( \frac{W_i - w_{\bx} }{h} \right). \nonumber
\end{equation*}
Then, we note that the error $\lhat(\bbetahat, \bx) - l(\bbeta,\bx)$ of the original estimator $\lhat(\cdot)$ admits the following decomposition. For any $\bx \in \Xsc$,
\begin{align}
& | \lhat(\bbetahat, \bx) - l(\bbeta, \bx) |  \; \leq \; | \ltil(\bbeta, \bx)  - l(\bbeta,\bx) | + |\lhat(\bbetahat,\bx) - \ltil(\bbeta, \bx)| \nonumber \\
& \qquad \leq \; | \ltil(\bbeta, \bx)  - \E \{ \ltil(\bbeta, \bx) \} | + | \E \{ \ltil(\bbeta, \bx) \}  - l(\bbeta,\bx) | + |\lhat(\bbetahat,\bx) - \ltil(\bbeta, \bx)| \nonumber \\
& \qquad  =: \; |\Stil_n(\bx)| + |\Sbar_n(\bx)| + |\Rhat_n(\bx)| \;\; \mbox{(say)}. \nonumber
\end{align}
Thus, to analyze the behavior of $ | \lhat(\bbetahat, \bx) - l(\bbeta, \bx) | $, it suffices to control each of the quantities $\Stil_n(\bx)$,  $\Sbar_n(\bx)$ and $\Rhat_n(\bx)$. We now proceed towards obtaining non-asymptotic pointwise tail bounds for these quantities. We first focus on $\Stil_n(\bx)$ and $\Sbar_n(\bx)$ which involve only the hypothetical estimator $\ltil(\cdot)$.
\begin{lemma}[Characterizing the tail bounds for $\Stil_n(\bx)$ and $\Sbar_n(\bx)$]\label{KS:thm1}
Under Assumption \ref{KS:assmpn1} (a)-(c), we have: for any fixed $\bx \in \Xsc$ and any $t \geq 0$,
\begin{equation*}
\P \left\{ |\Stil_n(\bx)| \; > \; C_1 \frac{t}{\sqrt{nh}}  + C_2 \frac{t^2 \sqrt{\log n} }{nh} \right\} \; \leq \; 3 \exp(-t^2),
\end{equation*}
where $C_1 := 7(B_1 C_K M_K)^{1/2}$ and $C_2 := D \sigma_Z M_K$ for some absolute constant $D > 0$. Further, under Assumption \ref{KS:assmpn1} (d), we have:
\begin{equation*}
|\Sbar_n(\bx)| \; \leq \; C_3 h^2 \quad \mbox{uniformly in} \;\; \bx \in \Xsc, \quad \mbox{where} \;\; C_3 \; := \; B_2 R_K.
\end{equation*}
\end{lemma}
Hence, for any $\bx \in \Xsc$ and $t \geq 0$, with probability at least $1 - 3 \exp(-t^2)$,
\begin{equation}
|\ltil(\bbeta, \bx) - l(\bbeta, \bx)| \; \leq \; C_1\frac{t}{\sqrt{nh}}  + C_2 \frac{t^2 \sqrt{\log n}}{nh} + C_3 h^2, \quad \forall \; \bx \in \Xsc. \qed \label{LKA:tailbound1}
\end{equation}

Next, we aim to control the term $\Rhat_n(\bx)$ whose behavior signifies the nature and extent of the additional price one pays due to estimation of $\bbeta$.
%We first enlist the assumptions we further require for obtaining the bounds on $\Rhat_n(\bx)$.

Using a first order Taylor series expansion of $\lhat(\bbetahat,\bx)$ around $\lhat(\bbeta,\bx) \equiv \ltil(\bbeta,\bx)$, we first rewrite $\Rhat_n(\bx) \equiv \lhat(\bbetahat,\bx) - \ltil(\bbeta,\bx)$  as:
\begin{align*}
& \Rhat_n(\bx) = (\bbetahat - \bbeta)' \left\{ \frac{1}{nh}\sum_{i=1}^n Z_i \frac{(\bX_i - \bx)}{h} K'\left( \frac{W^*_i - w_{\bx}^*}{h} \right) \right\}, \;\; \mbox{where}
\end{align*}
$\{W^*_i \}_{i=1}^n$ and $w_{\bx}^*$ are `intermediate' points that satisfy, for each $i = 1,\hdots, n$,  $ | (W^*_i  - w_{\bx}^*) - (W_i - w_{\bx}) | \leq |(\What_i - \what_{\bx}) - (W_i - w_{\bx}) | \equiv |(\bbetahat - \beta)'(\bX_i - \bx)|$.
\par\smallskip
We now rewrite the expansion above as: $\Rhat_n(\bx) \equiv \Rhat_{n,1}(\bx) + \Rhat_{n,2}(\bx)$, where
\begin{align}
& \Rhat_{n,1}(\bx) \; := \; (\bbetahat - \bbeta)' \left\{ \frac{1}{nh}\sum_{i=1}^n Z_i \frac{(\bX_i - \bx)}{h} K'\left( \frac{W_i - w_{\bx}}{h} \right) \right\} \nonumber \\
& \qquad =: \; (\bbetahat- \bbeta)'\bThat_n(\bx) \;\; \mbox{(say),} \quad \mbox{and} \;\; \Rhat_{n,2}(\bx) \; := \; \Rhat_n(\bx) - \Rhat_{n,1}(\bx). \nonumber
\end{align}
In the result below, we now characterize the tail bounds for $\Rhat_{n}(\bx)$.
\begin{lemma}[Characterizing the tail bounds for $\Rhat_{n,1}(\bx) $ and $\Rhat_{n,2}(\bx) $]\label{KS:thm2}
Under Assumption \ref{KS:assmpn2} (a), (b) and (d), and Assumption \ref{KS:assmpn1} (a) and (c), we have: for any $t \geq 0$, with probability at least  $1  - 3 \exp(-t^2) - q_n$,
\begin{align}
& |\Rhat_{n,1}(\bx) | \; \leq \; C^*_1 a_n + C_2^* \frac{a_n(t + \sqrt{\log p})}{\sqrt{nh^3}} +  C^*_3 \frac{a_n (t^2  + \log p)\sqrt{\log n}}{nh^2}, \;\; \mbox{where} \nonumber
\end{align}
$C^*_1, C^*_2, C^*_3 > 0$ are constants depending only on the constants introduced in Assumptions \ref{KS:assmpn2} and \ref{KS:assmpn1}, and $\bx \in \Xsc$ is any fixed evaluation point.

Further, under the additional condition in Assumption \ref{KS:assmpn2} (c), we have: for any $t \geq 0$, with probability at least $1  - 3 \exp(-t^2) - q_n$,
\begin{align}
& |\Rhat_{n,2}(\bx) | \;\leq \; 4 M_{\bX}^2 C^*_4 \frac{a_n^2}{h^2} + 4 M_{\bX}^2 \left(C^*_5 \frac{t a_n^2}{\sqrt{nh^5}} + C^*_6 \frac{t^2 a_n^2 \sqrt{\log n}}{nh^3} \right), \;\; \mbox{where} \nonumber \\
& \qquad \leq \; 3 \exp(-t^2) + q_n, \quad \mbox{for any fixed} \; \bx \in \Xsc \; \mbox{and any given} \; t \geq 0, \;\; \mbox{where} \nonumber
\end{align}
$C^*_4, C^*_5, C^*_6 > 0$ are constants depending only on the constants introduced in Assumptions \ref{KS:assmpn1} and \ref{KS:assmpn2}, and $\bx \in \Xsc$ is any fixed evaluation point.
\end{lemma}
With $a_n/h = o(1)$ as assumed, %$h = o(1)$, $(nh)^{-1} = o(1)$ and $a_n/h = o(1)$,
note that the second and the third terms in the bound for $\Rhat_{n,2}(\bx)$ are each dominated by the respective terms in the bound for $\Rhat_{n,1}(\bx)$ in Lemma \ref{KS:thm2}. Using this, we obtain a bound for $\Rhat_n(\bx)$ as follows: for any $t \geq 0$, with probability at least $ 1 - 6 \exp(-t^2) - 2 q_n$,
\begin{align}
| \Rhat_n(\bx) | & \; \equiv \; | \lhat(\bbetahat,\bx)  - \ltil(\bbeta,\bx)| \nonumber \\
& \; \leq \; C_1^* (a_n + a_n^2 h^{-2}) + C_2^* \frac{a_n (t + \sqrt{\log p})}{\sqrt{nh^3}} + C^*_3 \frac{a_n (t^2  + \log p)\sqrt{\log n}}{nh^2}, \label{LKA:tailbound2}
\end{align}
for some constants $C^*_1, C^*_2, C^*_3 > 0$ (possibly different from those in Lemma \ref{KS:thm2}) depending only on the constants defined in Assumptions \ref{KS:assmpn1}-\ref{KS:assmpn2}. \qed %and \ref{KS:assmpn2}. \qed

\subsection{Completing the Proof of Theorem \ref{KS:mainthm1}}\label{pf:KS:mainthm1}
\hspace{-0.08in} Combining the bounds \eqref{LKA:tailbound1} and \eqref{LKA:tailbound2} via a union bound, we then have: for any $\bx \in \Xsc$ and for any $t \geq 0$, with probability at least $1 - 9 \exp(-t^2) -2 q_n$,
\begin{align}
&  | \lhat(\bbetahat, \bx) - l(\bbeta, \bx)|  \; \leq \; | \ltil(\bbeta, \bx) - l(\bbeta, \bx)| + | \Rhat_n(\bx)| \; \leq \; C_1\frac{t}{\sqrt{nh}}  + C_2 \frac{t^2 \sqrt{\log n}}{nh} \nonumber \\
& \;\;\; + C_3 h^2 + C_1^* (a_n + a_n^2 h^{-2}) + C_2^* \frac{a_n (t + \sqrt{\log p})}{\sqrt{nh^3}} + C^*_3 \frac{a_n (t^2  + \log p)\sqrt{\log n}}{nh^2}  \nonumber %\\
\end{align}    %%% NOTE: TWO ALIGNS DELIBERATELY ADDED HERE TO BREAK UP THE EQUATIONS TO BETTER FORMAT THE FINAL SUPPLEMENT BEFORE SUBMISSION. %%%
\begin{align}
& \;\; \equiv \; D_1 \frac{t}{\sqrt{nh}} \left(1 + \frac{a_n}{h} \right) + D_2 \frac{t^2 \sqrt{\log n}} {nh} \left(1 + \frac{a_n}{h}\right) + D_3 b_n, \;\; \mbox{where} \label{LKA:tailbound:final1} \\
& \quad r_n \; := \; h^2 + a_n + \frac{a_n^2}{h^2} + \frac{a_n}{h} \sqrt{\frac{\log p}{nh}} + \frac{a_n}{h} \frac{\sqrt{\log n} \log p}{nh} \; = \; o(1) \quad \mbox{and}  \nonumber %& \;\; \leq \; D_1 \frac{t}{\sqrt{nh}} + D_2^* \frac{t^2 \sqrt{\log n}} {nh} + D_3^* b_n \label{LKA:tailbound:final}
\end{align}
$D_1, D_2, D_3 > 0$ are some constants depending on the constants $\{C_j, C_j^* \}_{j=1}^3$.

Further, with $(a_n \sqrt{\log p})/h = o(1)$  and $\{\log(np)\}/(nh) = o(1)$ by assumption, the fourth term in the definition of $r_n$ in \eqref{LKA:tailbound:final1} can be bounded as: $(a_n/h)\{\sqrt{\log p}/(nh)\} = o(1/\sqrt{nh})$ and the fifth term can be bounded as:
\begin{align*}
&\frac{a_n}{h} \frac{\sqrt{\log n} \log p}{nh}  \; \leq \; \frac{a_n\sqrt{\log p}}{h} \frac{\log(np)}{nh} = o\left( \frac{\log (np)}{nh}\right),
\end{align*}
where we used that $\sqrt{\log n} \sqrt{\log p} \leq (\log n + \log p)/2 \leq \log(np)$. Using these simplifications in \eqref{LKA:tailbound:final1} and that $a_n/h = o(1)$ by assumption, we finally have: for any $\bx \in \Xsc$ and for any $t \geq 0$, with probability at least $1 - 6 \exp(-t^2) - 2 q_n$,
\begin{align*}
& | \lhat(\bbetahat, \bx) - l(\bbeta,\bx)| \; \leq \; D_1^* \frac{t}{\sqrt{nh}} + D_2^* \frac{t^2 \sqrt{\log n}}{nh} + D_3^* b_n, \;\; \mbox{where} \\
& b_n \; := \; h^2 + a_n + \frac{a_n^2}{h^2} + \frac{1}{\sqrt{nh}} + \frac{\log(np)}{nh} \;\; \mbox{and}
\end{align*}
$D_1^*, D_2^*, D_3^* > 0$ are some constants depending only on those introduced in the assumptions. This completes the proof of Theorem \ref{KS:mainthm1}. \qed

\subsection{Completing the Proof of Theorem \ref{KS:mainthm2}}\label{pf:KS:mainthm2}

%For any $t \geq 0$, let us define:
%$$
%\epsilon_n(t) := C_1\frac{t+1}{\sqrt{nh}} + C_2 \frac{t^2 \sqrt{\log n}}{nh} + C_3 b_n, \;\; \mbox{where} \;\; b_n := h^2 + a_n + \frac{a_n^2}{h^2} + \frac{\log np}{nh}
%$$ and $C_1, C_2, C_3 > 0$ are the same constants as in Theorem \ref{KS:mainthm1}. Then
Using Theorem \ref{KS:mainthm1}, we have: for any fixed $\bx \in \Xsc$ and for any $t \geq 0$,
\begin{align}
& \P\left\{ |\lhat(\bbetahat, \bx) - l(\bbeta,\bx)| > \epsilon_n(t) \right\} \; \leq \; 9 \exp(-t^2) +2q_n \;\;\mbox{and} \nonumber \\
& \P\left\{ |\fhat(\bbetahat, \bx) - f(\bbeta,\bx)| > \epsilon_n(t) \right\} \; \leq \; 9 \exp(-t^2) +2q_n, \label{eq1:pf:KS:mainthm2}
\end{align}
where we recall that $\{\fhat(\bbetahat,\bx), f(\bbeta,\bx)\}$ is a special case of $\{\lhat(\bbetahat,\bx), l(\beta,\bx)\}$ with $Z \equiv 1$ so that Theorem \ref{KS:assmpn2} indeed applies to get both bounds above.

Next, note that $\mhat(\cdot) \equiv \lhat(\cdot)/\fhat(\cdot)$ and $m(\cdot) \equiv l(\cdot)/f(\cdot)$, so that
\begin{align*}
& |\fhat(\cdot)\{\mhat(\cdot) - m(\cdot)\}| \; = \; |\{\lhat(\cdot) - l(\cdot)\} - m(\cdot)\{ \fhat(\cdot) - f(\cdot)\}|\\
& \;\; \leq \;  |\lhat(\cdot) - l(\cdot)| + |m(\cdot)| |\fhat(\cdot) - f(\cdot)| \; \leq \; |\lhat(\cdot) - l(\cdot)| + \delta_m |\fhat(\cdot) - f(\cdot)|,
\end{align*}
where in the last step, we used $\|m(\cdot)\|_{\infty} \leq \delta_m$ by assumption. Using a simple union bound argument, we then have: for any $\bx \in \Xsc$ and for any $t \geq 0$,
\begin{align}
%\begin{eqnarray}
& \P \left\{ |\fhat(\bbetahat,\bx)\{\mhat(\bbetahat,\bx) - m(\bbetahat,\bx)\}| > (1 + \delta_m) \epsilon_n(t) \right\} \nonumber \\
& \quad \leq \; \P\left\{ |\lhat(\bbetahat, \bx) - l(\bbeta,\bx)| > \epsilon_n(t) \right\} + \P\left\{ |\fhat(\bbetahat, \bx) - f(\bbeta,\bx)| > \epsilon_n(t) \right\} \nonumber %\\
\end{align}     %%% NOTE: TWO ALIGNS DELIBERATELY ADDED HERE TO BREAK UP THE EQUATIONS TO BETTER FORMAT THE FINAL SUPPLEMENT BEFORE SUBMISSION. %%%
\begin{align}   %%% Also, this was originally written using a single eqnarray environment - apparently, that helped save a line from the numbering placement. %%
& \quad \leq \; 18 \exp(-t^2) + 4q_n, \label{eq2:pf:KS:mainthm2}
\end{align}
%\end{eqnarray}
where the final step follows from using the bounds in \eqref{eq1:pf:KS:mainthm2}.

Recall further that by assumption, $|f(\bbeta,\bx)| \equiv f(\bbeta,\bx) \geq \delta_f > 0$ $\forall \; \bx \in \Xsc$. Then, for any $\bx \in \Xsc$ and any $t_*\geq 0$ such that $\delta_f - \epsilon_n(t_*) > 0$, we have:
\begin{align}
& \P\{ |\fhat(\bbeta,\bx)| <  \delta_f - \epsilon_n(t_*)\} \; \leq \; \P\{ |\fhat(\bbeta,\bx)| <  |f(\bbeta,\bx)| - \epsilon_n(t_*)\} \nonumber \\
& \; \leq \; \P\{ |\fhat(\bbeta,\bx) - f(\bbeta,\bx)| | >  \epsilon_n(t_*)\} \; \leq \; 9 \exp(-t_*^2) + 2q_n,  \label{eq3:pf:KS:mainthm2}
\end{align}
where the penultimate bound follows since $|b| - |a| \leq | |a| -|b| | \leq |a -b| $ for any $a,b \in \R$, and the final bound follows from \eqref{eq1:pf:KS:mainthm2}. In particular, we have:
\begin{align*}
&\P\left\{ |\fhat(\bbeta,\bx)| <  \frac{\delta_f}{2}\right\} \; \leq \; 9 \exp(-t_*^2) + 2q_n, \;\; \forall \; t_* \geq 0 \;\; \mbox{such that} \;\; \epsilon_n(t_*) \leq \frac{\delta_f }{2}.
\end{align*}
Combining this bound along with \eqref{eq2:pf:KS:mainthm2}, we now have: for any $\bx \in \Xsc$ and for any $t, t_* \geq 0$ with $\epsilon_n(t_*) \leq \delta_f/2$,
\begin{align*}
& \P\left\{ |\mhat(\bbetahat,\bx) - m(\bbeta,\bx)| > \frac{2(1+\delta_m)}{\delta_f}\epsilon_n(t)\right\} \nonumber \\
&  = \; \P\left\{ |\mhat(\bbetahat,\bx) - m(\bbeta,\bx)| > \frac{2(1+\delta_m)}{\delta_f}\epsilon_n(t), |\fhat(\bbetahat,\bx)| \geq  \frac{\delta_f }{2}\right\} \nonumber \\
& \quad + \P\left\{ |\mhat(\bbetahat,\bx) - m(\bbeta,\bx)| > \frac{2(1+\delta_m)}{\delta_f}\epsilon_n(t), |\fhat(\bbetahat,\bx)| < \frac{\delta_f }{2}\right\} \nonumber \\
& \leq \; \P \left\{ |\fhat(\bbetahat,\bx)||\mhat(\bbetahat,\bx) - m(\bbetahat,\bx)| > (1 + \delta_m) \epsilon_n(t) \right\} + \P \left\{ |\fhat(\bbetahat,\bx)| < \frac{\delta_f }{2}\right\} \nonumber \\
& \leq \;  18 \exp(-t^2) + 9 \exp(-t_*^2) + 6q_n, \nonumber
\end{align*}
where the final bound follows from using \eqref{eq2:pf:KS:mainthm2}, \eqref{eq3:pf:KS:mainthm2} and the bound noted below \eqref{eq3:pf:KS:mainthm2} as a special case. This completes the proof of Theorem \ref{KS:mainthm2}. \qed

\subsection{Proof of Lemma \ref{KS:thm1}}\label{pf:ks:thm1}

Let $\bZ := (Z, \bX)$ and rewrite $\ltil(\bbeta,\bx)$ as:
$$
\ltil(\bbeta,\bx) \; = \; \frac{1}{n}\sum_{i=1}^n T_h(\bZ_i; \bx, \bbeta), \;\; \mbox{where} \;\; T_{h}(\bZ; \bx, \bbeta) := \frac{1}{h} Z K\left( \frac{W_i - w_{\bx}}{h} \right).
$$
Under Assumption \ref{KS:assmpn1} (a)-(b) and using Lemma \ref{lem:1:GenProp} (i)(b), (ii)(d) and (v), $T_{h}(\bZ; \bx, \bbeta)$ is sub-Gaussian with $\psitwonorm{T_{h}(\bZ; \bx, \bbeta)} \leq h^{-1}\sigma_{Z} M_{K}$. Hence, using Lemma \ref{lem:1:GenProp} (iv)(b) and (i)(c), we have:
$$
\psitwonorm{T_h(\bZ; \bx, \bbeta) - \E\{T_h(\bZ; \bx, \bbeta) \}} \leq 3 h^{-1}\sigma_{Z} M_{K} \quad \mbox{uniformly for all} \; \bx \in \Xsc.
$$
Further, under Assumption \ref{KS:assmpn1} (b)-(c), we have: uniformly for all $\bx \in \Xsc$,
\begin{align*}
& \Var\{T_{h}(\bZ;\bx,\bbeta) \} \; \leq \; \E\{T_{h}^2(\bZ;\bx,\bbeta) \}  \; = \; \E_{W}[\E\{T_{h}^2(\bZ;\bx,\bbeta) \medgiven W\}] \\
& \quad = \; h^{-2}\int_{\R} \E(Z^2 \medgiven W = w) K^2\{(w - w_{\bx})/h\} f_{\bbeta}(w) dw \\
& \quad \equiv \; h^{-2}\int_{\R} m_{\bbeta}^{(2)}(w) K^2\{(w - w_{\bx})/h\} f_{\bbeta}(w) dw \\
& \quad = \; h^{-1}\int_{\R} m_{\bbeta}^{(2)}(w_{\bx} + hu) f_{\bbeta}(w_{\bx} + hu) K^2(u) du \;\; \leq \; h^{-1}B_1 M_K C_K,
\end{align*}
where the penultimate step follows from a standard change of variable argument. We have thus verified all the conditions required for Lemma \ref{lem:6:VarianceTailBound} using which we now obtain: for any $t \geq 0$, with probability  at least $ 1 - 3 \exp(-t^2)$,
$$
\Stil_n(\bx) \equiv |\ltil(\bbeta, \bx) - \E\{ \ltil(\bbeta,\bx) \}| \leq \; 7 t \sqrt{\frac{B_1 M_K C_K}{nh}} + t^2 \frac{D \sigma_{Z} M_K}{nh} \sqrt{\log n},
$$
where while using Lemma \ref{lem:6:VarianceTailBound}, we set $\Gamma_n = h^{-1}B_1 M_K C_K$, $K_n = h^{-1}\sigma_{Z} M_K$, $p = 1$, $\alpha = 2$, and $D$ depends on the absolute constant $C_{\alpha}$ in the statement of the lemma. This completes the proof of the first part of Lemma \ref{KS:thm1}. \qed

\par\smallskip
For the second part regarding $\Sbar_n(\bx) \equiv \E\{\ltil(\bbeta,\bx)\} - l(\bbeta,\bx)$, observe that $\E\{\ltil(\bbeta,\bx)\} = \E\{T_h(\bZ;\bx, \bbeta) \}$ and $l(\bbeta,\bx) \equiv l_{\bbeta}(w_{\bx})$. We the have: $\forall \; \bx \in \Xsc$, % that $l(\bbeta,\bx) \equiv l_{\bbeta}(w_{\bx})$ by notation. We then have: , . We then have:
\begin{align*}
& \Sbar_n(\bx)  \; = \; \E\{T_h(\bZ;\bx, \bbeta) \} - l(\bbeta,\bx) \; = \; \E_{W}[\E\{T_h(\bZ;\bx, \bbeta) \medgiven W \}] - l_{\bbeta}(w_{\bx}) \\
& \; = \; h^{-1} \int_{\R} \E(Z \medgiven W = w) K\{(w - w_{\bx})/h\} f_{\bbeta}(w) dw - l_{\bbeta}(w_{\bx})  \\
& \; = \; h^{-1} \int_{\R} l_{\bbeta}(w) K\{(w - w_{\bx})/h\} dw -  l_{\bbeta}(w_{\bx})  \\
&  \; = \; \int_{\R} l_{\bbeta}(w_{\bx} + hu) K(u) du - l_{\bbeta}(w_{\bx}) \; = \; \int_{\R} \{l_{\bbeta}(w_{\bx} + hu) - l_{\bbeta}(w_{\bx})\}K(u) du \\
& \; = \; h l_{\bbeta}'(w_{\bx}) \underbrace{\int_\R u K(u) du}_{= \; 0} + h^2 R^*(\bx) \; := \; h^2 \int_{\R} l_{\bbeta}''(w_{\bx,u}^*) u^2 K(u) du, \;\; \mbox{where}
\end{align*}
$w_{\bx, u}^*$ is some `intermediate' point satisfying $|w_{\bx, u} - w_{\bx}| \leq h |u|$. The first two steps use $\E(Z \medgiven W = w) \equiv m_{\bbeta}(w)$ and $m_{\bbeta}(w) f_{\bbeta}(w) \equiv l_{\bbeta}(w)$. The next steps follow from a standard change of variable and Taylor series expansion argument under the assumed smoothness of $l_{\bbeta}(\cdot)$ in Assumption \ref{KS:assmpn1} (d) along with the conditions imposed therein on the kernel $K(\cdot)$. Using Assumption \ref{KS:assmpn1} (d), we further have: $\| l_{\bbeta}''(\cdot) \|_{\infty} \leq B_2$ and $\int |u^2 K(u)| du \leq R_K$. Hence, %therefore,
$$
|\Sbar_n(\bx)| \; \leq \; B_2 \int_{\R} u^2 |K(u)| du \; \leq \; B_2 R_K \;\; \mbox{uniformly for all} \;\; \bx \in \Xsc.
$$
This establishes the second part of Lemma \ref{KS:thm1} and completes the proof. \qed

\subsection{Proof of Lemma \ref{KS:thm2}}\label{pf:ks:thm2}
To control $\Rhat_{n,1}(\bx) \equiv (\bbetahat - \bbeta)'\bThat_n(\bx)$, note
\begin{align}
& |\Rhat_{n,1}(\bx)| \; \leq \; \| \bbetahat - \bbeta\|_1 \left[ \| \bThat_n(\bx) - \E\{\bThat_n(\bx) \}\|_{\infty} +  \| \E\{\bThat_n(\bx) \}\|_{\infty} \right] \label{eq1:pf:ks:thm2}
\end{align}
In the light of \eqref{eq1:pf:ks:thm2} and the assumed high probability bound for $\| \bbetahat - \bbeta\|_1$ in Assumption \ref{KS:assmpn2} (d), it now suffices to bound $\|\bThat_n(\bx) - \E\{ \bThat_n(\bx)\} \|_{\infty} $ and $\|\E\{\bThat_n(\bx)\}\|_{\infty}$. To this end, for each $\bx \in \Xsc$, define
$$
\bT_h^*(\bZ; \bx) %\; \equiv \; \bT_h^*(\bZ; \bx, \bbeta)
\; := \; \frac{1}{h^2}Z (\bX - \bx) K'\left(\frac{W - w_{\bx}}{h}\right) \;\; \mbox{so that} \; \bThat_n(\bx) \equiv \frac{1}{n} \sum_{i=1}^n \bT_h^*(\bZ_i; \bx).
$$
Now under Assumptions \ref{KS:assmpn1} (a), \ref{KS:assmpn1} (c), \ref{KS:assmpn2} (a), \ref{KS:assmpn2} (d) and using Lemma \ref{lem:1:GenProp} (i)(b)-(c), (iv)(b) and (v) at appropriate places, we have: for all $\bx \in \Xsc$,
\begin{align*}
& \max_{1 \leq j \leq p} \psitwonorm{\bT_{h[j]}^*(\bZ;\bx)} \; \leq \;  2 h^{-2} M_{\bX} M_{K'} \sigma_Z \;\; \mbox{and therefore,}\\
& \max_{1 \leq j \leq p} \psitwonorm{\bT_{h[j]}^*(\bZ;\bx) - \E\{ \bT_h^*(\bZ;\bx)\}}  \; \leq \;  6 h^{-2} M_{\bX} M_{K'} \sigma_Z.
\end{align*}
Further, under Assumptions \ref{KS:assmpn2} (d), \ref{KS:assmpn1} (c), \ref{KS:assmpn2} (a) and with $\E\{Z^2 (\bX_{[j]} - \bx_{[j]})^2 \medgiven W\} \leq 4M_{\bX}^2 \E_{W}(Z^2 \medgiven W) \equiv 4M_{\bX}^2 m_{\bbeta}^{(2)}(W) \; \forall j$, we have: for all $\bx \in \Xsc$,
\begin{align*}
& \max_{1 \leq j \leq p} \E [\{\bT_{h [j]}^*(\bZ;\bx)\}^2]  \; \leq \; \frac{4}{h^4} M_{\bX}^2 \int _{\R} m_{\bbeta}^{(2)}(w)[K'\{(w - w_{\bx_j})/h\}]^2 f_{\bbeta}(w) dw \\
& \;\; \leq \;  \frac{4}{h^3} M_{\bX}^2 M_{K'} B_1 \int_{\R}  m_{\bbeta}^{(2)}(w_{\bx} + hu)f_{\bbeta}(w_{\bx} + hu) \{K'(u)\}^2  du   \\
& \;\; \leq \; \frac{4}{h^3} M_{\bX}^2 M_{K'} B_1 \int_{\R} |K'(u)| du \;\; \leq \;  \frac{4}{h^3} B_1 M_{\bX}^2 M_{K'} C_{K'},
\end{align*}
where the second step follows from a change of variable argument and the final two bounds follow from using the assumptions mentioned above.

Using Lemma \ref{lem:6:VarianceTailBound} with the parameters therein set to: $\alpha = 2$,  $\Gamma_n \propto h^{-3}$ and $K_n \propto h^{-2}$, all in the light of the two bounds above, we then have: %for each $j \in \{1, \hdots, p\}$,
for any fixed $\bx \in \Xsc$  and for any $t \geq 0$, with probability at least $1 - 3 \exp (-t^2)$,
\begin{align}
\left\| \bThat_n(\bx) - \E\{\bThat_n(\bx)\} \right\|_{\infty}  & \equiv \; \left\| \frac{1}{n} \sum_{i=1}^n \bT_{h}^*(\bZ_i;\bx) - \E\{\bT_{h}^*(\bZ;\bx)\} \right\|_{\infty} \nonumber \\
&  \leq \; C_1 \frac{(t + \sqrt{\log p})}{\sqrt{nh^3}} + C_2 \frac{(t^2 + \log p) \sqrt{\log n}}{nh^2}, \label{eq2:pf:ks:thm2}
\end{align}
for some constants $C_1, C_2 > 0$ depending only on those introduced in the assumptions. Here, we further used $\sqrt{a + b} \leq \sqrt{a} + \sqrt{b}$ for any $a, b \geq 0$ to obtain the bound \eqref{eq2:pf:ks:thm2} from the one originally provided by Lemma \ref{lem:6:VarianceTailBound}.

Next, we focus on controlling $\|\E\{\bT_{h}^*(\bZ;\bx) \}\|_{\infty}$. To this end, recall the definitions of $\boldsymbol{\eta}_{\bbeta}(\cdot) \in \R^p$ and $l_{\bbeta}(\cdot) \in \R$, and let $\boldsymbol{\eta}_{\bbeta}'(w) := \frac{d}{dw} \boldsymbol{\eta}_{\bbeta}(w) \in \R^p$. Then, under Assumption \ref{KS:assmpn2} (a)-(b), we have: uniformly in $\bx \in \Xsc$,
\begin{align}
& \E\{ \bT_{h}^*(\bZ;\bx)\} \; = \; \frac{1}{h^2}\E_{W}[\E\{ (ZX - Zx)\medgiven W\} K'\{ (W - w_{\bx})/h\}] \nonumber \\
& \;\; \equiv \; \frac{1}{h^2}\int_{\R} \{\boldsymbol{\eta}_{\bbeta}(w) - \bx l_{\bbeta}(w)\} K'\{ (W - w_{\bx})/h\}  dw \nonumber \\
& \;\; = \; \frac{1}{h}\int_{\R} \{\boldsymbol{\eta}_{\bbeta}(w_{\bx} + hu) - \bx l_{\bbeta}(w_{\bx} + hu)\} K'(u)  du \nonumber  \\
& \;\; = \; \int_{\R} \{\boldsymbol{\eta}_{\bbeta}'(w_{\bx} + hu) - \bx l_{\bbeta}'(w_{\bx} + hu)\} K(u)  du, \nonumber
\end{align}
where the last two steps follow from a change of variable and integration by parts argument, where the latter is applicable under Assumption \ref{KS:assmpn2} (a)-(b). Under Assumptions \ref{KS:assmpn2} (a), \ref{KS:assmpn2} (b) and \ref{KS:assmpn2} (d), we then have:
\begin{align}
 \| \E\{ \bT_{h}^*(\bZ;\bx)\} \|_{\infty} & \leq \; \left\{\max_{1 \leq j \leq p} \|\boldsymbol{\eta}_{\bbeta [j]}'(\cdot) \|_{\infty} + \| \bx \|_{\infty} \| l_{\bbeta}'(\cdot) \|_{\infty} \right\} \int_{\R} |K(u)| du \nonumber \\
&  \leq \; (B_1^* + M_{\bX} B_2^*) C_K \quad \mbox{uniformly in} \; \bx \in \Xsc. \label{eq3:pf:ks:thm2}
\end{align}
Finally, recall that from Assumption \ref{KS:assmpn2} (d), we have $\bbetahat - \bbeta \|_1 \leq a_n$ with probability at least $1  - q_n$. Combining this with the bounds \eqref{eq2:pf:ks:thm2} and \eqref{eq3:pf:ks:thm2} and applying them in \eqref{eq1:pf:ks:thm2} through a simple union bound, we have: for any fixed $\bx \in \Xsc$ and for $t \geq 0$, with probability at least $1 - 3 \exp(-t^2) - q_n$,
\begin{align*}
|\Rhat_{n,1}(\bx)| \; \leq \; a_n \left\{ C_1^* + C_2^* \frac{(t + \sqrt{\log p})}{\sqrt{nh^3}} + C_3^* \frac{(t^2 + \log p) \sqrt{\log n}}{nh^2} \right\},
\end{align*}
for some constants $C_1^*, C_2^*, C_3^*$ depending only on those introduced in our assumptions. This establishes the first part of Lemma \ref{KS:thm2}. \qed

\par\smallskip
To establish the second part of Lemma \ref{KS:thm2} regarding bounds for $\Rhat_{n,2}(\bx)$, first recall that for some `intermediate' points $\{W^*_i \}_{i=1}^n$ and $w_{\bx}^*$ satisfying $ | (W^*_i  - w_{\bx}^*) - (W_i - w_{\bx}) | \leq |(\What_i - \what_{\bx}) - (W_i - w_{\bx}) | \equiv |(\bbetahat - \beta)'(\bX_i - \bx) |$,
\begin{align}
& |\Rhat_{n,2}(\bx)|  \equiv   \left|\frac{(\bbetahat - \bbeta)'}{nh^2}\sum_{i=1}^n Z_i (\bX_i - \bx)  \left\{K'\left( \frac{W^*_i - w_{\bx}^*}{h} \right) - K'\left( \frac{W_i - w_{\bx}}{h} \right) \right\} \right|  \nonumber \\
& \;\; \leq \;  \frac{\| \bbetahat - \bbeta \|_1}{nh^2}\sum_{i=1}^n \|\bX_i - \bx \|_{\infty} |Z_i|  \left|K'\left( \frac{W^*_i - w_{\bx}^*}{h} \right) - K'\left( \frac{W_i - w_{\bx}}{h} \right) \right|  \nonumber %\\
\end{align} %%% NOTE: TWO ALIGNS DELIBERATELY ADDED HERE TO BREAK UP THE EQUATIONS TO BETTER FORMAT THE FINAL SUPPLEMENT BEFORE SUBMISSION. %%%
\begin{align}
& \;\; \leq \; 2 M_{\bX}\| \bbetahat - \bbeta \|_1  \left\{ \frac{1}{nh^2}\sum_{i=1}^n |Z_i| \left|K'\left( \frac{W^*_i - w_{\bx}^*}{h} \right) - K'\left( \frac{W_i - w_{\bx}}{h} \right) \right| \right\}, \label{eq4:pf:ks:thm2}
\end{align}
where the steps follow from an $L_1$-$L_{\infty}$ bound along with a triangle inequality and using the boundedness of $\bX$ from Assumption \ref{KS:assmpn2} (d).

Let $\Asc_n$ denote the event $\Lonenorm{\bbetahat - \bbeta} \leq a_n$ and let $\Asc_n^c$ denote the complement event of $\Asc_n$. Then, from Assumption \ref{KS:assmpn2} (d), we have $\P(\Asc_n) \geq 1 - q_n$. Further, on the event $\Asc_n$, $(\bbetahat - \beta)'(\bX_i - \bx)/h \leq 2 M_{\bX} (a_n/h) \leq L$ under Assumption \ref{KS:assmpn2} (d) and consequently, using Assumption \ref{KS:assmpn2} (c) with the function $\varphi(\cdot)$ as defined therein, we have: on the event $\Asc_n$,
\begin{align}
&  \left|K'\left( \frac{W_i - w_{\bx}}{h} \right)  - K'\left( \frac{W^*_i - w_{\bx}^*}{h} \right) \right| \; \leq \;  \frac{1}{h}|(\bbetahat - \bbeta)'(\bX_i - \bx)|\varphi\left( \frac{W_i - w_{\bx}}{h} \right) \nonumber \\
& \;\; \leq \; \frac{1}{h}\|(\bbetahat - \bbeta)\|_1\|(\bX_i - \bx)\|_{\infty} \varphi\left( \frac{W_i - w_{\bx}}{h} \right) \;\; \leq \; \frac{2 M_{\bX} a_n}{h}\varphi\left( \frac{W_i - w_{\bx}}{h} \right), \label{eq5:pf:ks:thm2} %\;\; \forall 1 \leq i \leq n, \bx \in \Xsc.
\end{align}
and consequently, combining \eqref{eq4:pf:ks:thm2} and \eqref{eq5:pf:ks:thm2}, we have: on the event $\Asc_n$,
\begin{align}
& | \Rhat_{n,2}(\bx)| \; \leq \; \frac{2M_{\bX}^2 a_n^2} {nh^3} \sum_{i=1}^n |Z_i| \varphi\left( \frac{W_i - w_{\bx}}{h} \right) \quad \forall \; \bx \in \Xsc. \label{eq6:pf:ks:thm2}
\end{align}
Thus, we have: for any $\epsilon \geq 0$ and for any $\bx \in \Xsc$,
\begin{align}
 \P(| \Rhat_{n,2}(\bx)| > \epsilon) \; & \leq \; \P(| \Rhat_{n,2}(\bx)| > \epsilon, \Asc_n) + \P(| \Rhat_{n,2}(\bx)| > \epsilon, \Asc_n^c) \nonumber \\
& \leq \; \P\left\{\frac{4M_{\bX}^2 a_n^2} {nh^3} \sum_{i=1}^n |Z_i| \varphi\left( \frac{W_i - w_{\bx}}{h} \right) > \epsilon, \Asc_n \right\} + \P(\Asc_n^c) \nonumber \\
& \leq \; \P\left\{\frac{4M_{\bX}^2 a_n^2} {nh^3} \sum_{i=1}^n |Z_i| \varphi\left( \frac{W_i - w_{\bx}}{h} \right) > \epsilon \right\} + q_n, \label{eq7:pf:ks:thm2}
\end{align}
where the steps follow from \eqref{eq6:pf:ks:thm2} and that $\P(\Asc_n^c) \leq q_n$ by assumption.

Next, define: $\Tsc_h(\bZ; \bx) \equiv \Tsc_h(\bZ; \bx, \bbeta) := h^{-3} |Z| \varphi\{ (W - w_{\bx})/h\}$ and recall that $m_{\bbeta}^{(2)}(W) \equiv \E(Z^2 \medgiven W)$. Then, using the boundedness conditions from Assumptions \ref{KS:assmpn1} (c) and \ref{KS:assmpn2}(c), along with use of iterated expectations, we bound the first and second moments of $\Tsc_h(\bZ; \bx)$ $\forall \; \bx \in \Xsc$ as follows.
\begin{align*}
&\E \{\Tsc_h^2(\bZ; \bx)\}   \; = \; \frac{1}{h^6} \int_{\R} m_{\bbeta}^{(2)} (w) \varphi^2\left( \frac{W - w_{\bx}}{h} \right) f_{\bbeta}(w) dw  \nonumber \\
&  \quad = \; \frac{1}{h^5} \int_{\R} m_{\bbeta}^{(2)} (w_{\bx} + hu) f_{\bbeta}(w_{\bx} + hu) \varphi^2(u) du \; \leq \; \frac{B_1 M_{\varphi} C_{\varphi}}{h^5}, \;\; \mbox{and} \nonumber %\\
\end{align*}%%% NOTE: TWO ALIGNS DELIBERATELY ADDED HERE TO BREAK UP THE EQUATIONS TO BETTER FORMAT THE FINAL SUPPLEMENT BEFORE SUBMISSION. %%%
\begin{align*}
& \E \{\Tsc_h(\bZ; \bx)\}  \; = \; \frac{1}{h^3} \int_{\R} \E(|Z| \medgiven W = w) \varphi\left( \frac{W - w_{\bx}}{h} \right) f_{\bbeta}(w) dw \nonumber \\
& \quad \leq \; \frac{1}{h^3} \int_{\R} \{m_{\bbeta}^{(2)} (w)\}^{\half} \varphi\left( \frac{W - w_{\bx}}{h} \right) f_{\bbeta}(w) dw  \nonumber \\
& \quad \leq \; \frac{1}{h^2} \int_{\R} \{m_{\bbeta}^{(2)} (w_{\bx} + hu)\}^{\half} \varphi(u) f_{\bbeta}(w_{\bx} + hu) du \; \leq \; \frac{(B_1 C_f)^{\half} C_{\varphi}}{h^2}, \nonumber
\end{align*}
where $C_f > 0$ is a constant such that $\| f_{\bbeta}(\cdot) \|_{\infty} \leq C_f$. Further, under Assumptions \ref{KS:assmpn1} (a) and \ref{KS:assmpn2} (c), using various parts of Lemma \ref{lem:1:GenProp}, we have:
$$
\psitwonorm{ \Tsc_h(\bZ; \bx) - \E\{ \Tsc_h(\bZ; \bx)\} } \; \leq \; 3 \psitwonorm{\Tsc_h(\bZ; \bx)} \; \leq \; 3 h^{-3}\sigma_Z M_{\varphi} \quad \forall \bx \in \Xsc.
$$
Hence, using Lemma \ref{lem:6:VarianceTailBound}, with all required conditions verified now, we have: for any $\bx \in \Xsc$ and for any $t \geq 0$, with probability at least $1  - 3\exp(-t^2)$,
\begin{align}
\left|\frac{1}{n} \sum_{i=1}^n \Tsc_h(\bZ_i; \bx) \right| \; & \leq \; \left|\frac{1}{n} \sum_{i=1}^n \Tsc_h(\bZ_i; \bx) - \E\{\Tsc_h(\bZ; \bx) \}\right| + \left|\E\{\Tsc_h(\bZ; \bx) \right| \nonumber\\
\; & \leq \; C_3 \frac{t}{n h^5} + C_4 \frac{t^2 \sqrt{\log n}}{n h^3} + \frac{C_5}{h^2}, \label{eq8:pf:ks:thm2}
\end{align}
for some constants $C_3, C_4, C_5 > 0$ depending only on those in the assumptions. Hence, using \eqref{eq8:pf:ks:thm2} in \eqref{eq7:pf:ks:thm2}, we now have: for any $t \geq 0$,
\begin{align*}
& \P\left\{|\Rhat_{n,2}(\bx)| \geq 4M_{\bX}^2a_n^2 \left( C_3 \frac{t}{n h^5} + C_4 \frac{t^2 \sqrt{\log n}}{n h^3} + \frac{C_5}{h^2}\right) \right\} \\
& \leq \; \P\left\{ \frac{1}{nh^3} \sum_{i=1}^n |Z_i| \varphi\left( \frac{W_i - w_{\bx}}{h} \right) >  C_3 \frac{t}{n h^5} + C_4 \frac{t^2 \sqrt{\log n}}{n h^3} + \frac{C_5}{h^2}\right\} + q_n \\
& \equiv \; \P\left( \left|\frac{1}{n} \sum_{i=1}^n \Tsc_h(\bZ_i; \bx) \right|  \; > \; C_3 \frac{t}{n h^5} + C_4 \frac{t^2 \sqrt{\log n}}{n h^3} + \frac{C_5}{h^2}\right) + q_n \\
& \leq \; 3 \exp(-t^2) + q_n  \quad \mbox{for any} \; \bx \in \Xsc.
\end{align*}
This establishes the desired bound for $\Rhat_{n,2}(\bx)$ and completes the proof. \qed

%\bibliographystyle{imsart-nameyear}
%%\bibliographystyle{apalike}
%\bibliography{P1-HDME-Abhishek-Biblio}

%\end{document}